\documentclass[openright]{book}

\usepackage[Sonny]{fncychap} %style paquete
\usepackage{amsmath,amssymb,array,hhline,graphics}
\usepackage{lscape}
\usepackage{makeidx}% Para crear \'{\i}ndice alfab\'{e}tico
\usepackage[all]{xy} %For commutative diagrams
\usepackage{amssymb} %For double arrows
\usepackage{graphicx}
\usepackage{eucal}%Letras caligr\'{a}ficas
\usepackage{hyperref}% Hiperv\'{\i}nculos
\usepackage{color}
\usepackage{nicefrac}
\usepackage{rotating}%para rotar tablas
\usepackage{tabularx}%para tablas
\usepackage{multirow}
\usepackage{appendix}
\usepackage{tikz}
\usepackage[a4paper, margin=2.5cm]{geometry}
\newtheorem{theorem}{Theorem}[chapter]
\newtheorem{definition}[theorem]{Definition}
\newtheorem{prop}[theorem]{Proposition}
\newtheorem{lemma}[theorem]{Lemma}
\newtheorem{corol}[theorem]{Corollary}
\newtheorem{remark}[theorem]{Remark}
\newtheorem{example}[theorem]{Example}

\newcommand{\dem}{\noindent{\normalsize \textsl{Proof}:\;}}% Entorno demostraciones
\newcommand{\proof}{\noindent{\normalsize \textsl{Proof}:\;}}% Entorno demostraciones
\newcommand{\qed}{\hfill$\square$}% Fin demostraci\'{o}n proposici\'{o}n y teorema
\newcommand{\rqed}{\hfill$\diamond$}% Final observaciones
\newcommand{\ds}{\displaystyle}% Estilo exposici\'{o}n
\newcommand{\R}{\mathbb{R}} %Numeros reales
\newcommand{\ak}{1\leq \alpha \leq k}
\newcommand{\n}{1\leq i \leq n}

\renewcommand{\emph}[1]{{\bfseries\itshape{#1}}}

\def\r{\ensuremath{\mathbb{R}}}% N\'{u}meros reales
\def\rk{{\mathbb R}^{k}}% Espacio euclideo de dimensi\'{o}n k
\def\tkq{T^1_kQ}
\def\tkqh{(T^1_k)^*Q}
\def\rktkq{\rk\times T^1_kQ}
\def\rktkqh{\rk\times (T^1_k)^*Q}
\def\rkq{\rk \times Q}

\def\derpar#1#2{\ds\frac{\partial{#1}}{\partial{#2}}}
\def\derpars#1#2#3{\ds\frac{\partial^2{#1}}{\partial{#2}{\partial
{#3}}}}
\def\vf{\mathfrak{X}}% Campos de vectores
\def\d{{\rm d}}
\def\df{{\mit\Omega}}

\def\Lie#1#2{{\mathcal{L}}_{#1}#2}% Derivada de Lie
\def\Cinfty{{\rm C}^\infty} % Espacio de funciones de clase infinito

\def\longto{\longrightarrow}

\def\Lag{\mathcal L}

\newcounter{epigrafe}[subsection]

\newcounter{subepigrafe}[epigrafe]

\def\r{\ensuremath{\mathbb{R}}}% N\'{u}meros reales
\def\rk{{\mathbb R}^{k}}% Espacio euclideo de dimensi\'{o}n k
\def\tkq{T^1_kQ}
\def\tkqh{(T^1_k)^*Q}
\def\rktkq{\rk\times T^1_kQ}
\def\rktkqh{\rk\times (T^1_k)^*Q}
\def\rkq{\rk \times Q}
\def\rkq{\rk \times Q}

\def\beq{\begin{equation}}
\def\eeq{\end{equation}}
\def\bea{\begin{eqnarray}}
\def\eea{\end{eqnarray}}
\def\beann{\begin{eqnarray*}}
\def\eeann{\end{eqnarray*}}
\def\beasn{\begin{sneqnarray}}
\def\eeasn{\end{sneqnarray}}
\def\ben{\begin{enumerate}}
\def\een{\end{enumerate}}
\def\bit{\begin{itemize}}
\def\eit{\end{itemize}}

    \title{Methods of Differential Geometry \\ in Classical Field
    Theories:\\ $k$-symplectic and $k$-cosymplectic approaches}

    \author{M. de Le\'{o}n, M. Salgado $\&$ S. Vilari\~{n}o}

\makeindex

\begin{document}

\parskip=8pt

    %%%%%%%%%%%%%%%%%%%%%%%%%%%%%%%%%%%%%%%%%%%%%%%%%%%%%%%%%%%%%%%%%%%%%%%%%%%%%%%%%%%%%%%%%%%%%%%%%%%%%%%%%%%%%%%%%%%%%%%%%%
    %%%%            Frontmatter
    %%%%%%%%%%%%%%%%%%%%%%%%%%%%%%%%%%%%%%%%%%%%%%%%%%%%%%%%%%%%%%%%%%%%%%%%%%%%%%%%%%%%%%%%%%%%%%%%%%%%%%%%%%%%%%%%%%%%%%%%%%

    \frontmatter

            \pagestyle{empty}% con esta orden se elimina la numeraci\'{o}n

            \maketitle

            \newpage\mbox{}\newpage

            \pagestyle{plain}

\chapter*{Glossary of symbols}

\begin{tabular}{ll}
 $V, W,\ldots $ & Vector spaces\\
 $Q, M, N,...$ & $\mathcal{C}^\infty$ finite-dimensional manifolds \\
 $\mathfrak{X}(M)$ & Set of vector fields on $M$\\
 $\mathfrak{X}^k(M)$ & Set of $k$-vector fields on $M$\\
 $\Lambda^lM\to M$ & Bundle of $l$-forms\\
 $\tau\colon TQ\to Q$ & Tangent bundle\\
 $\pi\colon T^*Q\to Q$ & Cotangent bundle\\
 $(q^i)$ & local coordinate system on $Q$\\
 $(q^i,p_i)$ & local coordinate system on $T^*Q$\\
 $(q^i,v^i)$ & local coordinate system on $TQ$\\
 $f\colon U\subset N\to M$ & smooth $(\mathcal{C}^\infty)$ mapping\\
 $f_*$ or $Tf$ & Tangent map to $f\colon M\to N$\\
 $d$ & Exterior derivative\\
 $\iota_v$ & Inner product\\
 $\mathcal{L}_X$ & Lie derivative\\
 $\theta$ & Liouville $1$-form\\
 $\omega$ & Canonical symplectic form\\
 $H$ & Hamiltonian function\\
 $X_H$ & Hamiltonian vector field\\
 $(X_q)^v_{v_q}$ & Vertical lift of $X_q$ to $TQ$ at $v_q$\\
 $\Delta$ & Liouville vector field\\
 $J$ & Vertical endomorphism\\
 $\Gamma$ & Second order differential equation\\
 $L$ & Lagrangian function\\
 $E_L$ & Energy\\
 $FL$ & Legendre transformation\\
 $\theta_L$ & Pullback of $\theta$ by $FL$\\
 $\omega_L$ & Pullback of $\omega$ by $FL$\\
 $\pi^k\colon(T^1_k)^*Q\to Q$ & Cotangent bundle of $k^1$-covelocities\\
 $\pi^{k,\alpha}\colon (T^1_k)^*Q\to T^*Q$ & Projection over $\alpha$ copy of $T^*Q$\\
 $\tau^k\colon T^1_kQ\to Q$ & Tangent bundle of $k^1$-velocities\\
 $(q^i,p_i^\alpha)$ & local coordinate system on $(T^1_k)^*Q$\\
 $(q^i,v^i_\alpha)$ & local coordinate system on $T^1_kQ$\\
 $\phi^{(1)}$ & First prolongation of maps to $T^1_kQ$\\
 $\phi^{[1]}$ & First prolongation to $\rktkq$\\
 $\{J^1,\ldots, J^k\}$ & $k$-tangent structure\\
 $Z^C$ & Complete lift of a vector field $Z$\\

\end{tabular}

 \newpage
\mbox{}
\thispagestyle{empty} % para que no se numere esta p\'{a}gina
        \pagestyle{plain}

        \tableofcontents

            \pagestyle{plain}

\chapter{Introduction}

As it is well know, symplectic geometry is the natural arena to develop classical mechanics; indeed, a symplectic manifold is locally as a cotangent bundle $T^*Q$ of a manifold $Q$, so that the canonical coordinates $(q^i,p_i)$ can be used as coordinates for the position $(q^i)$ and the momenta $(p_i)$. The symplectic form is just $\omega=dq^i\wedge dp_i$, and a simple geometric tool permits to obtain the Hamiltonian vector field $X_H$ for a Hamiltonian $H=H(q^i,p_i)$. The integral curves of $X_H$ are just the solution of the Hamilton equations
\[
    \frac{dq^i}{dt} = \frac{\partial H}{\partial p_i},\quad \frac{dp_i}{dt}=-\frac{\partial H}{\partial q^i}\,.
\]

In classical field theory, the Hamiltonian function is of the form
\[
    H=H(x^\alpha, q^i,p^\alpha_i)
\]
where $(x^1,\ldots, x^k)\in \mathbb{R}^k$, $q^i$ represent the components of the fields and $p^\alpha_i$ are the conjugate momenta. In the Lagrangian description, the Lagrangian function is
\[
    L=L(x^\alpha, q^i, v^i_\alpha)
\]
where now $v^i_\alpha$ represent the derivations of the fields with respect to the space-time variables $(x^\alpha)$.

At the end of the sixties and the beginning of the seventies of the past century, there are some attempts to develop a convenient geometric framework to study classical field theories. This geometric setting was the so-called multisymplectic formalism, developed in a parallel but independent way by the Polish School led by W.M. Tulczyjew (see, for instance, \cite{{KT-79}Snia,Tulczy1}); H. Goldschmidt and S. Sternberg \cite{gs} and the Spanish School by P.L. Garcia and A. P\'{e}rez-Rend\'{o}n \cite{GP1, GP2}.

The idea was to consider, instead of the cotangent bundle $T^*Q$ of a manifold $Q$, its bundle of $k$-forms, $\bigwedge^kQ$. Indeed, $ \bigwedge^k Q$ is equipped with a tautological $k$-form where its differential is just a multisymplectic form. This approach was revised, among others, by G. Martin \cite{mar, mar2} and M. Gotay \textit{et al} \cite{Go1,Go2,Go3,Gymmsy,Gymmsy2} and, more recently, by F. Cantrijn \textit{et al} \cite{CIL-1996,CIL-1999} or M. Mu\~{n}oz-Lecanda \textit{et al} \cite{EMR-1996, EMR-1998,EMR-1999,EMR-00, ELMMR-2004}, among others.

An alternative approach is the so-called $k$-symplectic geometry, which is based on the Whitney sum of $k$ copies of the cotangent bundle $T^*Q$ instead of the bundle of exterior $k$-forms $\bigwedge^kQ$. The $k$-symplectic formalism is a natural generalization to field theories of the standard symplectic formalism in Mechanics. This formalism was developed in a parallel way in equivalent presentations by  C. G\"{u}nther in \cite{Gu-1987}, A. Awane \cite{Awane-1992, Awane-1994, Awane-2000}, L. K. Norris \cite{MN-2000,No2,No3,No4,No5} and de M. de Le\'{o}n \textit{et al} \cite{LMM-2002,LMS-88, LMS-91}.  In this sense, the $k$-symplectic formalism is used to give a geometric description of certain kinds of field theories: in a local description, those theories whose Lagrangian does not depend on the base coordinates,  $(x^1,\ldots, x^k)$ (typically, the space-time coordinates); that is, the $k$-symplectic formalism is only valid for Lagrangians $L(q^i,v^i_\alpha)$ and Hamiltonians $H(q^i,p^\alpha_i)$ that depend on the field coordinates $q^i$ and on the partial derivatives of the field $v^i_\alpha$, or the corresponding momenta $p^\alpha_i$.

G\"{u}nther's paper \cite{Gu-1987} gave a geometric Hamiltonian formalism for field theories. The crucial device is the introduction of a vector-valued generalization of a symplectic form called a polysymplectic form. One of the advantages of this formalism is that one only needs the tangent and cotangent bundle of a manifold to develop it. In \cite{MRS-2004} this formalism  has been revised and clarified.

Let us remark here that the polysymplectic formalism developed by I.V. Kanatchikov \cite{Kana} and the polysymplectic formalism developed by G. Sardanashvily {\it et al} \cite{Sarda2, Sd-95, ms,Sarda, Sd-95b}, based on a vector-valued form defined on some associated fiber bundle, is a different description of classical field theories of first order that the polysymplectic (or $k$-symplectic) formalism proposed by C. G\"{u}nther.

This book is devoted to review two of the most relevant approaches to the study of classical field theories of first order, say $k$-symplectic and $k$-cosymplectic geometry.

The book is structured as follows. Chapter \ref{chapter: Mechanics} is devoted to review the fundamentals of Hamilton and Lagrangian Mechanics; therefore, the Hamilton and Euler-Lagrange equations are derived on the cotangent and tangent bundles of the configurations manifold, and both are related through the Legendre transformation.

In Part \ref{Part2} we develop the geometric machinery  behind the classical field theories of first order when the Hamiltonian or the Lagrangian function do not depend explicitly on the space-time variables. The geometric scenario is the so-called $k$-symplectic geometry. Indeed, instead to consider the cotangent bundle $T^*Q$ of a manifold $Q$, we take the Whitney sum of $k$-copies of  $T^*Q$ and investigate its geometry. This study led to the introduction to a $k$-symplectic structure as a family of $k$ closed $2$-forms and a distribution satisfying some compatibility relations.

$k$-symplectic geometry allows us to derive the Hamilton-de Donder-Weyl equations. A derivation of these equations using a variational method is also included for the sake of completeness.

This part of the book also discusses the case of Lagrangian classical theory. The key geometric structure here is the so-called tangent bundle of $k^1$-covelocities, which can be defined using theory of jets, or equivalently as the Whitney sum of $k$ copies of the tangent bundle $TQ$ of a manifold $Q$. This geometric bundle $TQ\oplus \stackrel{k)}{\ldots}\oplus TQ$ lead us to define a generalization of the notion of vector fields, that is, a $k$-vector field on $Q$ as a section of the canonical fibration  $TQ\oplus \stackrel{k)}{\ldots}\oplus TQ\to Q$. $k$-vector fields will play in classical field theories the same role that vector fields on classical mechanics.

Additionally, tangent bundles of $k$-velocities have its own geometry, which is a natural extension of the canonical almost tangent structures on tangent bundles. Both descriptions, Hamiltonian and Lagrangian ones, can be related by an appropriate extension of the Legendre transformation.

In this part we also include a recent result on the Hamilton-Jacobi theory for classical field theories in the framework of $k$-symplectic geometry.

Part \ref{part k-cosymp} is devoted to extend the results in Part \ref{Part2} to the case of Hamiltonian and Lagrangian functions depending explicitly on the space-time variables.

The geometric setting is the so-called $k$-cosymplectic manifolds, which is a natural extension of cosymplectic manifold. Let us recall that cosymplectic manifolds are the odd-dimensional counterpart of symplectic manifolds.

Finally, in Part \ref{relation_k-cosym_multi} we relate the $k$-symplectic and $k$-cosymplectic formalism with the multisymplectic theory.

The book ends with two appendices where the fundamentals notions on symplectic and cosymplectic manifolds are presented.

Along this book, manifolds are smooth, real, paracompact, connected and $\mathcal{C}^\infty$. Maps are $\mathcal{C}^\infty$. Sum over crossed repeated indices is understood.

    %%%%%%%%%%%%%%%%%%%%%%%%%%%%%%%%%%%%%%%%%%%%%%%%%%%%%%%%%%%%%%%%%%%%%%%%%%%%%%%%%%%%%%%%%%%%%%%%%%%%%%%%%%%%%%%%%%%%%%%%%%
    %%%%            Mainmatter
    %%%%%%%%%%%%%%%%%%%%%%%%%%%%%%%%%%%%%%%%%%%%%%%%%%%%%%%%%%%%%%%%%%%%%%%%%%%%%%%%%%%%%%%%%%%%%%%%%%%%%%%%%%%%%%%%%%%%%%%%%%

    \mainmatter
        \pagestyle{myheadings}

            \renewcommand{\chaptermark}[1]
                {\markboth{{\rm\thechapter\quad #1}}{}}
            \renewcommand{\sectionmark}[1]
                {\markright{{\rm\thesection\quad #1}}}
            \renewcommand{\subsectionmark}[1]
                {\markright{{\rm\thesubsection\quad #1}}}
            %\renewcommand{\thesubsection}{\Alph{subsection}}
            %\renewcommand{\appendixname}[1]{Ap\'{e}ndice #1}

            %%%%%%%%%%%%%%%%%%%%%%%%%%%%%%%%%%%%%%%%%%%%%%%%%%%%%%%%%%%%%%%%%
            %           chapters
            %%%%%%%%%%%%%%%%%%%%%%%%%%%%%%%%%%%%%%%%%%%%%%%%%%%%%%%%%%%%%%%%%

\part{A review of Hamiltonian and Lagrangian Mechanics}\label{Part1}

\chapter{Hamiltonian and Lagrangian Mechanics}\label{chapter: Mechanics}

In this chapter we present a brief review of Hamiltonian and Lagrangian Mechanics; firstly on the cotangent bundle of an arbitrary manifold $Q$ (the Hamiltonian formalism) and then on the tangent bundle (the Lagrangian formalism). Finally, we consider the general theory on an arbitrary symplectic manifold.

In the last part of this chapter we give a review of the non autonomous Mechanics using cosymplectic structures.

A complete description of Hamiltonian and Lagrangian Mechanics can be found in \cite{AM-1978,Arnold-1978,{Arnold-1998}, god, Goldstein, HSS-2009, Holm-2008,lm,lr}. There exists an alternative description of the Lagrangian and Hamiltonian dynamics using the notion of Lagrangian submanifold, this description can be found in \cite{T1,T2}.

\section{Hamiltonian Mechanics}\label{Ham mec}
\index{Hamiltonian Mechanics}

In this section we present a review of the Hamiltonian Mechanics on the cotangent bundle of an arbitrary manifold $Q$. Firstly we review some results on vector spaces.

\subsection{Algebraic preliminaries}

\index{Form!on a vector space}
    By an \emph{exterior form} (or simply a form) on a vector space $V$, we mean an alternating multilineal function on that space with values in the field of scalars. The contraction of a vector $v\in V$ and an exterior form $\omega$ on $V$ will be denoted by $\iota_v\omega$.

Let $V$ be a real vector space of dimension $2n$, and  $\omega:V\times V \to \r$ a skew-symmetric bilinear form.  This   form allows us to define the map
$$\begin{array}{rcl} \flat : V &\to & V^* \\ \noalign{\medskip}
                        v & \to &\flat (v)=\iota_v\omega=\omega(v,-)\,.
\end{array}$$

\index{Form!Symplectic}

If $\omega$ is non degenerate $(\textstyle{i.e.,}\, \omega(v,w)=0,\;  \forall w \Rightarrow v=0)$ then $\omega$ is called a \emph{symplectic form} and, $V$ is said to be a \emph{symplectic vector space}.

\index{Symplectic!Vector space}
\index{Symplectic!Form}

Let us observe that when $\omega$ is non degenerate, the map $\flat$ is injective. In fact,
 $$\flat(v)=0 \Leftrightarrow \omega(v,w)=0,  \quad \forall w\in V \Leftrightarrow v=0\, . $$

 In this case, since $\flat$ is an injective mapping between vector spaces of the same dimension, we deduce that it is an isomorphism.   Let us observe that the matrix of $\flat$ coincide with the matrix $(\omega_{ij})$ of  $\omega$ with respect to an arbitrary basis $\{e_i\}$ of $V$. The inverse isomorphism will be denoted by $ \sharp : V^* \to  V $.

The proof of the following proposition is a direct computation.

\begin{prop}\label{algprop}
Let $(V,\omega)$ be a symplectic vector space. Then there exists a basis (Darboux basis) $\{e_1, \ldots, e_n,u_1, \ldots, u_n\}$ of $V$, such that

 \begin{enumerate}
    \item $\omega =\ds\sum_{i=1}^n e^i\wedge u^i$.
     \item The isomorphisms $\flat$  and $\sharp$ associated with $\omega$ are characterized  by
$$
 \begin{array}{lcrlcr}
   \flat(e_i)&=&u^i, & \flat (u_i)&=&-e^i, \\ \noalign{\medskip}
   \sharp(e^i)&=&- u_i, & \sharp (u^i)&=&e_i\,.
     \end{array}
 $$
\end{enumerate}

\end{prop}

\subsection{Canonical forms on the cotangent bundle.}

\index{Cotangent bundle}
\index{Cotangent bundle!Canonical projection}
Let $Q$ be a manifold of dimension $n$ and  $T^*Q$ the cotangent bundle of  $Q$, with canonical projection $\pi:T^*Q
\rightarrow Q$  defined by $ \,\pi(\nu_q)=q$.

If  $(q^i)$  is a coordinate system on  $U
\subseteq Q$, the induced fiber   coordinate system $(q^i,p_i)$ on $T^*U$ is defined as follows
\index{Cotangent bundle!coordinate system}
\begin{equation}
q^i \left( \nu_q \right)  =   q^i \left( q \right) \, , \quad
p_i \left( \nu_q \right)  =  \nu_q \left(
\frac{\partial}{\partial q^i} \Big\vert_q \right)\, , \quad 1 \leq i
\leq n \,,
\end{equation}
being  $\nu_q \in T^*U$.

\index{Form!Liouville $1$-form}
\index{Cotangent bundle!Liouville form}
The canonical \emph{Liouville $1$-form} $\theta  $ on $T^* Q$ is
defined by
\begin{equation}\label{Liouville}
\theta \left( \nu_q \right) \left( X_{\nu_q} \right)  =
\nu_q \left(\pi_* \left( \nu_q \right)
\left( X_{\nu_q} \right)\right)
\end{equation}
where $\nu_q \in T^*Q$, $X_{\nu_q} \in T_{\nu_q}
\left( T^*Q \right)$ and $ \pi_* \left( \nu_q
\right) : T_{\nu_q} \left( T^* Q \right) \rightarrow T_qQ$ is the tangent mapping of the canonical projection
 $\pi:T^*Q \to Q   $ at $\nu_q \in
T_q^*Q$.

In canonical coordinates, the Liouville $1$-form  $\theta$ is given by
\begin{equation}\label{theta0local}
\theta  =  p_i \, dq^i \, .
\end{equation}

\index{Cotangent bundle!Symplectic form}
\index{Form!Canonical symplectic form}
\index{Symplectic!Form}
The Liouville $1$-form  let us define  the closed two form
\begin{equation}\label{omega0}
\omega  =  -d\theta
\end{equation}
which is non degenerate (at each point of $T^*Q$), such that $(T_{\nu_q}(T^*Q), \omega(\nu_q))$ is a  symplectic vector space.
This $2$-form is called  the
\emph{canonical symplectic form} on the   cotangent bundle. From
(\ref{theta0local}) and (\ref{omega0})  we deduce that the local expression of $\omega$ is
\begin{equation}\label{omega0local}
\omega  =  dq^i \, \wedge \, dp_i \, .
\end{equation}

The manifold $T^*Q$ with its canonical symplectic form $\omega$ is the geometrical model of the {\bf symplectic manifolds} which will be studied in Appendix \ref{symma}.
\index{Symplectic!Manifold}

For each  $\nu_q\in T^*Q$,     $ \omega(\nu_q)$ is a bilinear form on the vector space  $T_{\nu_q}(T^*Q)$, and therefore we can define a vector bundle isomorphism
$$\begin{array}{rcl} \flat : T(T^*Q) &\to &  T^*(T^*Q) \\ \noalign{\medskip}
                        Z_{\nu_q }& \to &\flat_{\nu_q}  (Z_{\nu_q })=\iota_{ Z_{\nu_q }}\omega(\nu_q)=\omega(\nu_q)(Z_{\nu_q },-)
\end{array}$$
with inverse $ \sharp   :  T^*(T^*Q )  \to   T(T^*Q)\, .
 $

Thus we have an isomorphism of $C^\infty(T^*Q)$-modules  between the corresponding spaces of sections
\[
\begin{array}{rccl}
\flat \colon &  \mathfrak{X}(T^*Q) &\longrightarrow & \textstyle\bigwedge^1(T^*Q) \\\noalign{\medskip} & Z &\mapsto & \flat(Z)= \iota_Z\omega\end{array}\]
  and its inverse is denoted by
 $\quad \sharp: \bigwedge^1(T^*Q) \longrightarrow T(T^*Q)$.

 Taking into account Proposition \ref{algprop} (or by a direct computation) we deduce the following Lemma
  \begin{lemma} The isomorphisms $\flat$ and $\sharp$ are locally characterized by
 \begin{equation}\begin{array}{lcrlcr}
 \flat\big(\ds\frac{\partial}{\partial q^i}\big)&=&dp_i, &\flat\big(\ds\frac{\partial}{\partial p_i}\big)&=&-dq^i\,, \\ \noalign{\medskip}
 \sharp(dq^i)&=& -\ds\frac{\partial}{\partial p_i}, & \sharp(dp_i)&=&\ds\frac{\partial}{\partial q^i}\,.
 \end{array}\label{qp}
 \end{equation}
 \end{lemma}

\subsection{Hamilton equations}\label{hamquaions}

\index{Hamiltonian!Function}
  Let $H:T^*Q \to \r$ be a function, usually called \emph{Hamiltonian function}. Then there exists
an unique  vector field   $X_H\in\mathfrak{X}(T^*Q)$  such that
\medskip
\begin{equation}\label{geoham}\flat(X_H)=\iota_{X_H}\omega = dH \end{equation}
or, equivalently, $X_H=\sharp(dH)$.

 From (\ref{omega0local}) and (\ref{geoham}) we deduce the local expression of $X_H$
\begin{equation}\label{locX_Hh}
X_H= \ds\frac{\partial H}{\partial p_i}\ds\frac{\partial}{\partial q^i} -\ds\frac{\partial H}{\partial q^i}
\ds\frac{\partial}{\partial p_i} \, .
\end{equation}

$X_H$ is called the \emph{Hamiltonian vector field}    corresponding to the Hamiltonian function $H$.
\index{Hamiltonian!Vector Field}

From (\ref{locX_Hh}) we obtain the following theorem.

\begin{prop} Let
     $c:\r \rightarrow T^*Q$ be a curve with local expression $c(t)=(q^i(t),p_i(t))$. Then $c$ is an integral curve of the
      vector field  $X_H$ if and only if $c(t)$ is solution of the following system of differential equations.
       \begin{equation}\label{ecHam}
\frac{d q^i}{d t}\Big\vert_{t}  =  \frac{\partial H}{\partial p_i}\Big\vert_{c(t)}  \, , \quad \frac{d p_i}{d t}\Big\vert_{t}   =  - \frac{\partial
H}{\partial q^i}\Big\vert_{c(t)}   \, , \quad 1 \leq i \leq n
\end{equation}
which are known as the \textbf{Hamilton equations of the Classical Mechanics}.
\index{Hamilton equations!Classical Mechanics}

\end{prop}
\index{Hamilton equations}
So equation   (\ref{geoham}) is considered the geometric version of Hamilton equations.

 We recall that Hamilton equations can be also obtained from the Hamilton Principle, for more details see for instance \cite{AM-1978}.

\section{Lagrangian Mechanics}

\index{Lagrangian Mechanics}
The Lagrangian Mechanics allows us to obtain the Euler-Lagrange equations from a geometric approach. In this case we work over the tangent bundle of the configuration space. In this section we present a brief summary of the Lagrangian Mechanics; a complete description can be found in \cite{AM-1978,Arnold-1978, god, Goldstein, HSS-2009, Holm-2008}.

\subsection{Geometric preliminaries.}

\index{Tangent bundle}

In this section we recall the canonical geometric ingredients on the
tangent bundle, $TQ$, of a manifold $Q$, as well as other objects defined from a Lagrangian $L$. We denote by $\tau:TQ \rightarrow Q$ the
  canonical projection $\tau(v_q)=q$.

\index{Tangent bundle!Canonical coordinates}

  If $(q^i)$
 is a coordinate system on  $U \subseteq Q$ the induced
coordinate system $(q^i,v^i)$ on $TU
\subseteq TQ$ is given by
\begin{equation}\label{coortang}
q^i \left( v_q \right)  =   q^i \left( q \right) \, , \quad v^i
\left( v_q \right)  =  \left( dq^i \right)_q \left( v_q
\right)\,=v_q(q^i) , \quad 1 \leq i \leq n
\end{equation}
being $v_q \in TU$.

We now recall the definition of some geometric elements which are necessary for the geometric description of the Euler-Lagrange equations.

\noindent{\bf  Vertical lift of vector fields.}

\index{Tangent bundle!Vertical lift of vector fields}
The structure of vector space of each fibre $T_qQ$  of  $TQ$ allows us to define the vertical lifts of tangent vectors.

\begin{definition}
Let $X_q  \in T_qQ$ be a tangent vector  at the point $q \in Q$. We define the mapping
$$
\begin{array}{ccl}
T_qQ &  \longrightarrow &T_{v_q}(TQ)\\ \noalign{\medskip}
 v_q & \to & \, (X_q)^{\textsc{v}}_{v_q}=\displaystyle \frac{d}{dt}
\Big\vert_{t=0} \left( v_q + t X_q \right) \, .
\end{array}
$$
Then, the tangent vector $(X_q)^{\textsc{v}}_{v_q} $ is called
the \emph{vertical lift} of  $X_q$ to $TQ$ at the point
$v_q \in TQ$, and it is the tangent vector at $0\in\r$ to the curve $\alpha(t)=v_q+tX_q\in T_qQ\subset TQ$.
\end{definition}

In local coordinates, if $X_q=a^i\derpar{}{q^i}\Big\vert_{q}$, then
\begin{equation}\label{vvv}
(X_q)^{\textsc{v}}_{v_q}=a^i \, \, \ds\frac{\partial}{\partial v^i}\Big\vert_{v_q} \, .
\end{equation}

The definition can be extended for a vector field $X$ on $Q$ in the obvious manner.

\noindent{\bf The Liouville vector field.}

\index{Tangent bundle!Liouville vector field}

\begin{definition}\label{defcvl}
The \textbf{Liouville vector field} $\triangle$ on $TQ$ is the
 infinitesimal generator of the flow given by dilatations on each fiber, it is
$\Phi:( t, v_q )\in \r \times TQ  \longrightarrow e^t \, v_q\in TQ $
\end{definition}
Since  $\,\Phi_{v_q}(t)=(q^i,e^t\, v^i)$ we deduce that, in bundle coordinates,
   the Liouville  vector field  is given by
\begin{equation}\label{locliov}
\triangle  =  v^i \, \ds\frac{\partial}{\partial v^i} \, .
\end{equation}

%Also $\triangle$ can be defined using the vertical lifts of vector fields.
%\begin{equation}\label{cvectliouv}
%\triangle \left( v_q \right)  =  \left( v_q
%\right)^{\textsc{v}}_{v_q}
%\end{equation}
%for all $v_q \in TQ$.

\noindent{\bf Canonical tangent structure on    $TQ$.}

\index{Tangent bundle!Canonical tangent structure}
\index{Canonical tangent structure}
The vertical lifts let us construct a canonical tensor field   of type  $(1,1)$ on
$TQ$ in the following way

\begin{definition}\label{defj0}
A tensor field $J$ of type $(1,1)$ on $TQ$ is defined as follows
\begin{equation}\label{jtang}
\begin{array}{cccl}
  J(v_q): & T_{v_q}(TQ) & \to & T_{v_q}(TQ)\\ \noalign{\medskip}
   & Z_{v_q} & \to  & J(v_q) \left( Z_{v_q} \right)  =  \left(  \tau_*
\left( v_q \right) \left( Z_{v_q} \right) \right)^{\textsc{v}}(v_q)
\end{array}
\end{equation}
where $Z_{v_q} \in T_{v_q} \left( TQ \right)$ and $v_q \in
T_qQ$.
\end{definition}
This tensor field is called   the \emph{canonical tangent structure} or \emph{vertical endomorphism} of the tangent bundle
 $TQ$.
 \index{Tangent bundle!Vertical endomorphism}
 \index{Vertical endomorphism}

From (\ref{vvv}) and (\ref{jtang}) we deduce that in canonical coordinates $J$ is given by
\begin{equation}\label{j0local}
J\, = \frac{\partial}{\partial v^i} \, \otimes \,  dq^i \,     \,.
\end{equation}
%because $$
%S(\ds\frac{\partial}{\partial q^i})=(\ds\frac{\partial}{\partial q^i})^V=\ds\frac{\partial}{\partial v^i} ,
%\quad S(\ds\frac{\partial}{\partial v^i})=0
%$$

\subsection{Second order differential equations.}

\index{Tangent bundle!SODE}
 In this section we shall describe a special kind of vector fields on $TQ$, known as second order differential equations, semisprays and semigerbes (in French) \cite{grif1,grif2,grif3,szilasi}. For short, we will use the term {\sc sode}s.

\begin{definition}
Let $\Gamma$ be a vector field on $TQ$, i.e. $\Gamma\in \mathfrak{X}(TQ)$.
$\Gamma$ is a  {\sc sode} if and only if it is a section of the map
  $ \tau_*:T(TQ) \rightarrow TQ$, that is
\begin{equation}\label{sode}
 \tau_* \circ \Gamma = id_{TQ}
\end{equation}
  where  $id_{TQ}$ is the  identity function on $TQ$ and $\tau:TQ
\to Q$ the canonical projection.
\end{definition}
The tangent lift of a curve $\alpha:I\subset\r \to Q$ is the curve $\dot{\alpha}:I \to TQ$ where  $\dot{\alpha}(t)$
is the tangent vector to the curve $\alpha$. Locally if  $\alpha(t)=(q^i(t))$ then $\dot{\alpha}(t)=(q^i(t),dq^i/dt)$.

 A direct computation show that the local expression of a  {\sc sode} is
   $$\Gamma=v^i \ds\frac{\partial}{\partial q^i}+\Gamma^i
\ds\frac{\partial}{\partial v^i}  ,$$ and as consequence of this local expression one obtains that its integral curves are tangent lifts of curves on $Q$.

\begin{prop}\label{defsemisp}
Let $\Gamma$ be a  vector field  on $TQ$.   $\Gamma$ is   a
   {\sc sode} if and only if its integral curves are tangent lifts of curves on $Q$.
\end{prop}

\proof

Let us suppose $\Gamma$ a {\sc sode}, then locally  $$\Gamma =  v^i   \ds\frac{\partial}{\partial q^i}  + \Gamma^i \ds\frac{\partial}{\partial v^i}  $$
where $\Gamma^i\in \mathcal{C}^\infty(TQ)$, and let $\phi(t)=(q^i(t),v^i(t))$ be an integral curve of $\Gamma$. Then
$$
\ds\frac{dq^i}{dt}\Big\vert_{t}\derpar{}{q^i}\Big\vert_{\phi(t)}+
\ds\frac{dv^i}{dt}\Big\vert_{t}\derpar{}{v^i}\Big\vert_{\phi(t)}= \Gamma(\phi(t))
=
\phi_*(t)\big(\ds\frac{d}{dt}\Big\vert_{t}\big)=
v^i(\phi(t))\derpar{}{q^i}\Big\vert_{\phi(t)}+\Gamma^i(\phi(t)) \derpar{}{v^i}\Big\vert_{\phi(t)}\,,
$$
thus
$$
\ds\frac{dq^i}{dt}\Big\vert_{t}=v^i(\phi(t))=v^i(t) \, , \quad \Gamma^i(\phi(t))=\ds\frac{d^2q^i}{dt}\Big\vert_{t}
$$
and we deduce that $\phi(t)=\dot{\alpha}(t)$ where $\alpha(t)=(\tau\circ \phi)(t)=(q^i(t))$, and this curve $\alpha(t)$ is a solution of the  following
second order differential system
\begin{equation}\label{nn10}
\frac{\displaystyle d^2q^i} {\displaystyle dt^2}\Big\vert_{t}=
\Gamma^i\Big(q^i(t), \ds\frac{dq^i}{dt}\Big\vert_{t}\Big)\,,\quad 1\leq
i\leq n\, .
\end{equation}
 The converse is proved in an analogous way. \qed

As a consequence of (\ref{locliov}) and (\ref{j0local}), a {\sc sode} can be characterized using the tangent structure as follows.
\begin{prop}
A vector field  $X$ on $TQ$ is a  {\sc sode} if and only if
\begin{equation}
J \, X = \triangle
\end{equation}
  where  $\triangle$ is the  Liouville vector field and $J$ the vertical endomorphism on $TQ$.

\end{prop}
%\proof It is an immediate consequence of (\ref{locliov}) and (\ref{j0local}).
\qed

%%%%%%%%%%%%%%%%%%%%%%%%%%%%%%%%%%%%%%%%%%%%%%%%%%%%%%%%%%%%%%%%%%%%%%%%%%

\subsection{Euler-Lagrange equations.}

\index{Euler-Lagrange equations}
In this subsection we shall give a geometric description of the Euler-Lagrange equations. Note that these equations can be also obtained from a variational principle.

\index{Lagrangian Mechanics!Poincar\'{e}-Cartan forms}
\index{Form!Poincar\'{e}-Cartan}

\noindent{\bf  The Poincar\'{e}-Cartan forms on $TQ$.}

Given a Lagrangian function, that is, a function $L\colon TQ \rightarrow \r$, we consider the
 $1$-form on $TQ$
\begin{equation}
\theta_L  =  dL\circ J
\end{equation}
that is
$$
\xymatrix{\theta_L(v_q): T_{v_q}(TQ) \ar[r]^-{J(v_q)} & T_{v_q}(TQ) \ar[r]^-{dL(v_q)}& \mathbb{R}}
$$
at each point  $v_q\in TQ$.

Now we  define the $2$-form on $TQ$
\begin{equation}\label{defll}
\omega_L  =  - \, d \theta_L \, .
\end{equation}

From (\ref{coortang}) and (\ref{j0local}) we deduce that
\begin{equation}\label{lllocal}
\theta_L  =  \frac{\partial L}{\partial v^i} \, dq^i  ,
\end{equation}
and from   (\ref{defll}) y (\ref{lllocal}) we obtain that
\begin{equation}\label{omegaLlocal}
\omega_L = \, dq^i\wedge d\Big(\ds\frac{\partial L}{\partial v^i }\Big) =  \frac{\partial^2 L}{\partial v^i \partial q^j} \,
dq^i \wedge dq^j \, + \, \frac{\partial^2 L}{\partial v^i
\partial v^j} \, dq^i \wedge dv^j \, .
\end{equation}

This   $2$-form   $\omega_L$ is closed,
and it is non degenerate if and only if the matrix $\left( \displaystyle
\frac{\partial^2 L}{\partial v^i \partial v^j} \right)$ is non singular; indeed the matrix of  $\omega_L$ is just
$$
\left(
  \begin{array}{cc}
   \ds\frac{\partial^2 L}{\partial q^j
\partial v^i}-\ds\frac{\partial^2 L}{\partial q^i
\partial v^j} & \ds\frac{\partial^2 L}{\partial v^i
\partial v^j} \\
    -\ds\frac{\partial^2 L}{\partial v^i
\partial v^j} & 0 \\
  \end{array}
\right).
$$

\begin{definition}
A Lagrangian function $L\colon TQ \rightarrow \r$ is said to be
\textbf{regular} if the matrix $\left( \displaystyle
\frac{\partial^2 L}{\partial v^i \partial v^j} \right)$ is non
singular.% In other case the Lagrangian is called \textbf{singular}.
\end{definition}
\index{Lagrangian!On tangent bundle}
\index{Lagrangian!Regular}
\index{Lagrangian Mechanics!Lagrangian function}

 When $L$ is regular, $\omega_L$ is non degenerate (and hence, symplectic) and thus we can  consider  the  isomorphism
$$\begin{array}{rccl}\flat_L\colon & \mathfrak{X}(TQ) & \longrightarrow &\textstyle\bigwedge^1TQ \\\noalign{\medskip} & Z& \mapsto & \flat_L(Z)=\iota_Z\omega_L\end{array}$$
with inverse mapping  $\sharp: \bigwedge^1(TQ) \longrightarrow \mathfrak{X}(TQ)$.

\index{Lagrangian Mechanics!Energy function}
\index{Energy function}
 \begin{definition} Given a Lagrangian function $L$, we    define the \emph{energy} function $E_L$ as the function
$$E_L=\Delta(L)-L\colon TQ \to \r.$$
\end{definition}
From  (\ref{locliov}) we deduce that $E_L$ has the local expression
\begin{equation}\label{locener}
E_L= v^i \ds\frac{\partial L}{\partial v^i} -L .
\end{equation}

 We now consider the equation
\begin{equation}\label{ecsimlag}
\flat_L(X_L)=\iota_{X_L}  \omega_L  =  d E_L\,.
\end{equation}

If we write locally $X_L$ as
\begin{equation}
X_L =  A^i \frac{\partial}{\partial  q^i} +
B^i\frac{\partial}{\partial v^i}\, ,
\end{equation}
where $A^i, B^i\in \mathcal{C}^\infty(TQ)$
then $X_L$ is solution of the equation (\ref{ecsimlag}) if and only if  $A^i$ and $B^i$ satisfy the following system of equations:
\begin{equation}\label{0alfa}
\begin{array}{rll}
\left(\ds \frac{\partial^2 L}{\partial q^i \partial v^j} -
\ds\frac{\partial^2 L}{\partial q^j \partial v^i} \right) \, A^j -
\ds \frac{\partial^2 L}{\partial v^i \partial v^j} \, B^j & = &
v^j \ds\frac{\partial^2 L}{\partial q^i\partial v^j} -
\ds\frac{\partial  L}{\partial q^i } \, ,
\\ \noalign{\medskip}
\ds\frac{\partial^2 L}{\partial v^i\partial v^j} \, A^j & = &
\ds\frac{\partial^2 L}{\partial v^i\partial v^j} \, v^j \, .
\end{array}
\end{equation}

If the Lagrangian is regular, then $A^i=v^i$ and we have
\begin{equation}\label{locel40}
\begin{array}{rll}
\ds\frac{\partial^2 L}{\partial q^j
\partial v^i} v^j + \ds\frac{\partial^2 L}{\partial v^i\partial
v^j}B^j = \ds\frac{\partial  L}{\partial q^i}
 \, .
\end{array}
\end{equation}

Therefore when $L$ is regular there exists an unique solution $X_L$,   and it  is a {\sc sode}. Let $\dot{\alpha}(t)=(q^i(t),\ds\frac{dq^i}{dt})$ be  an integral curve of $X_L$
where  $\alpha:t\in\r \to \alpha(t)=(q^i(t))\in Q$.

From (\ref{nn10}) we know that $$
\frac{\displaystyle d^2q^i} {\displaystyle dt^2}\Big\vert_{t} =
B^i(q^j(t), \ds\frac{dq^j}{d t }\Big\vert_{t} )
$$
and from (\ref{nn10})  and  (\ref{locel40}) we obtain that the curve $ \alpha(t)$
 satisfies the following system of equations
\begin{equation}\label{e-l-eq}
 \ds\frac{\partial^2 L}{\partial q^j
\partial v^i}\Big\vert_{\dot{\alpha}(t)} \ds\frac{dq^j}{dt}\Big\vert_{t} + \ds\frac{\partial^2 L}{\partial v^i\partial
v^j}\Big\vert_{\dot{\alpha}(t)} \frac{\displaystyle d^2q^j} {\displaystyle dt^2}\Big\vert_{t} = \ds\frac{\partial  L}{\partial q^i}\Big\vert_{\dot{\alpha}(t)} \quad 1\leq i\leq n\end{equation}
The above equations are known as the \emph{Euler-Lagrange equations}. Let us observe that   its solutions are curves on $Q$.

\index{Lagrangian Mechanics!Euler-Lagrange equations}
\index{Equations!Euler-Lagrange}
\index{Euler-Lagrange!Classical Mechanics}

\begin{prop}
If $L$ is regular then the vector field
  $X_L$   solution of  (\ref{ecsimlag})  is a
{\sc sode}, and its solutions are the solutions of the Euler-Lagrange equations.
\end{prop}

Usually the
  Euler-Lagrange equations defined by  $L$ are written as
\begin{equation}\label{eceullag}
 \frac{d}{dt}\Big\vert_{t} \left( \frac{\partial L}{\partial v^i}\circ \dot{\alpha}
\right) \, - \, \frac{\partial L}{\partial q^i}\circ \alpha
   =  0 \, , \quad 1 \leq i \leq n
\end{equation}whose solutions are curves $ \alpha:\r \to Q$. Let us observe that (\ref{e-l-eq}) are just the same equations that (\ref{eceullag}), but written in an extended form.

Equation (\ref{ecsimlag}) is the geometric version of the Euler-Lagrange equations, which can be obtained from  Hamilton's principle, see for instance \cite{AM-1978}.

\section{Legendre transformation}

The Hamiltonian and   Lagrangian formulations of Mechanics are related by the Legendre transformation.

\index{Legendre transformation!Classical Mechanics}

\begin{definition}
Let $L\colon TQ \rightarrow \r$ be a  Lagrangian function; then  the
\emph{Legendre transformation} associated to   $L$ is the map
  $$\begin{array}{cccl}FL  :& TQ & \rightarrow & T^*Q \\ \noalign{\medskip}
                                & v_q & \to & FL(v_q):T_qQ\to \r
\end{array}$$ defined by
\begin{equation}
[FL(v_q)] \left( w_q \right)  =  \frac{d}{dt} \Big\vert_0 L \left(
v_q \, + \, t \, w_q \right)
\end{equation}
where $v_q,w_q \in TQ$.
\end{definition}

A direct  computation  shows that locally
\begin{equation}\label{FLlocal}
FL \left( q^i , v^i \right)  =  \left( q^i , \ds\frac{\partial
L}{\partial v^i} \right)
\end{equation}

From (\ref{omega0local}),
(\ref{omegaLlocal}) and (\ref{FLlocal}) we deduce the following relation between the canonical symplectic form and the Poincar\'{e}-Cartan $2$-form.

\begin{prop}
If $\omega$ is the canonical symplectic $2$-form of the cotangent bundle $T^*Q$ and $\omega_L$ is the Poincar\'{e}-Cartan $2$-form
defined  in (\ref{defll}) then
\begin{equation}
FL^* \, \omega  =  \omega_L \, .
\end{equation}
\end{prop}

\begin{prop}
The following statements are equivalent
\begin{enumerate}
\item $L\colon TQ \rightarrow \r$ is a regular Lagrangian.

\item $FL\colon TQ \rightarrow T^*Q$ is a local diffeomorphism.
\item $\omega_L$ is a nondegenerate, and then, a symplectic form.
\end{enumerate}
\end{prop}
\proof
The Jacobian matrix of  $FL$ is
$$
\left(
  \begin{array}{cc}
    I_n & * \\ \noalign{\medskip}
    0 & \ds\frac{\partial^2
L}{\partial v^i \partial v^j} \\
  \end{array}
\right)
$$ thus $FL$ local diffeomorphism if and only if  $L$ is
regular.

On the other hand we know that $\omega_L$ is non degenerate if and only if  $L$ is
regular.\qed

\index{Lagrangian!hyperregular}

\begin{definition}
A Lagrangian $L\colon TQ \rightarrow \r$ is said to be
\emph{hyperregular} if the Legendre transformation
$FL\colon TQ \rightarrow T^*Q$ is a global diffeomorphism.
\end{definition}

The following result connects the Hamiltonian and Lagrangian formulations.

\begin{prop}\label{propeqhl}
Let $L\colon TQ \rightarrow \r$ be a hyperregular Lagrangian, then we define the Hamiltonian
$H:T^*Q \rightarrow \r$ by $H \circ FL = E_L$. Therefore, we have
\begin{equation}\label{tflh}
FL_* (X_L)  =  X_H \, .
\end{equation}

\noindent Moreover if $\alpha: \r \rightarrow TQ$  is an integral curve
of  $X_L$ then $FL \circ \alpha$ is an integral curve of
 $X_H$.
\end{prop}
\proof  (\ref{tflh})  is a  consequence of the following: \textit{the Euler-Lagrange equation (\ref{ecsimlag}) transforms into the Hamilton equation (\ref{geoham}) via the Legendre transformation}, and conversely.\qed

%%%%%%%%%%%%%%%%%%%%%%%%%%%%%%%%%%%%%%%%%%%%%%%%%%%%%%%%%%%%%%%%
\section{Non autonomous Hamiltonian and La\-gran\-gian  Mechanics}
%%%%%%%%%%%%%%%%%%%%%%%%%%%%%%%%%%%%%%%%%%%%%%%%%%%%%%%%%%%%%%%%%%

In this section we consider the case of time-dependent Mechanics. Now we shall give a briefly review of the geometric description of the dynamical equations in this case. As in the autonomous case this description can be extended to general cosymplectic manifolds. Thus, in Appendix \ref{cosymma} we recall the notion of cosymplectic manifolds.

\subsection{Hamiltonian Mechanics}

\index{Hamiltonian!time dependent}
Let $H \colon \r\times T^*Q \to \r$ be a time-dependent Hamiltonian. If $\pi:\r\times T^*Q\to T^*Q$ denotes the canonical projection, we consider
$\widetilde{\omega}=\pi^*\omega$ the pull-back of the canonical symplectic $2$-form on $T^*Q$. We shall consider bundles coordinates $(t,q^i,p_i)$ on
$\r\times T^*Q$.

 Let us take the equations
  \begin{equation}\label{h02}\iota_{E_H}dt  =   1 \;, \quad \iota_{E_H}\Omega =0\,,\end{equation}
where $\Omega=\widetilde{\omega}+dH\wedge dt$.

   A direct computation using that locally $\widetilde{\omega}=dq^i\wedge dp_i$  shows that
  $$
E_H= \ds\frac{\partial }{\partial t}+\ds\frac{\partial H}{\partial p_i}\ds\frac{\partial}{\partial q^i} -\ds\frac{\partial H}{\partial q^i}
\ds\frac{\partial}{\partial p_i} \, .
$$

\index{Evolution vector field}
$E_H$ is called the \emph{evolution vector field}    corresponding to Hamiltonian function $H$.
Consider now an integral curve $c(s)=(t(s),q^i(s), p_i(s))$ of the evolution vector field $E_H$: this implies that $c(s)$ should satisfy the following system of differential equations
\[
    \ds\frac{dt}{ds}=1,\quad \ds\frac{dq^i}{ds}=\ds\frac{\partial H}{\partial p_i},\quad \ds\frac{dp_i}{ds}= -\ds\frac{\partial H}{\partial q^i}\,.
\]

Since $\ds\frac{dt}{ds}=1$ implies $t(s)=s+ constant$, we deduce that
\[
    \ds\frac{dq^i}{dt} = \ds\frac{\partial H}{\partial p_i},\quad \ds\frac{dp_i}{dt}= -\ds\frac{\partial H}{\partial q^i}\,,
\]
since $t$ is an affine transformation of $s$, which are the \emph{ Hamilton equations} for  a non-autonomous Hamiltonian $H$.
\index{Hamilton equations!non-autonomous Mechanics}

\subsection{Lagrangian Mechanics}

\index{Lagrangian!time-dependent}
Let us consider that the Lagrangian $L(t,q^i,v^i)$ is time-dependent, then $L$ is a function   $\r\times TQ\to \r$.

\index{Liouville vector field}
 Let us denote also by $\Delta$   the \emph{canonical vector field (Liouville
vector field)} on $\r \times TQ$. This vector field
  is the infinitesimal generator of  the following flow
$$
\begin{array}{ccc}
\r \times ( \r\times TQ) & \longrightarrow & \r\times
TQ  \\ \noalign{\medskip} (s,(t,{v_1}_q,\ldots , {v_k}_q)) &
\longrightarrow & (t, e^s{v_1}_q, \ldots,e^s{v_k}_q)\, ,
\end{array}
$$
  and in local coordinates it has the form
 $\Delta =     v^i
\frac{\displaystyle\partial}{\displaystyle\partial v^i}\, .
$

Now we shall characterize the
 vector fields on $\r \times TQ$ such that their integral
curves are canonical prolongations of curves on $Q$.

\index{First prolongation}
\begin{definition}
 Let $\alpha:\r \rightarrow Q$  be  a curve, we define the \emph{first prolongation}
$\alpha^{[1]}$ of $\alpha$ as the map
$$
\begin{array}{rcl}
\alpha^{[1]}:\r  & \longrightarrow &   \r  \times TQ \\ t &
\longrightarrow &
(t,\dot{\alpha}(t))
\end{array}$$
  \end{definition}

In an obvious way we shall consider the extension of the tangent structure $J$    to $\r\times TQ$
which we   denote by $J$
and it has the same local expression $J=\derpar{}{v^i}\otimes dq^i$.

\index{SODE}
\index{Second order partial differential equation}
\begin{definition}\label{sode2}
A vector field  $X$ on   $\r \times
TQ$
   is  said to be a \emph{second order partial differential equation}  (SODE for short)  if  :
$$
\iota_X dt=1, \quad J(X)=\Delta\,.
$$
\end{definition}

   From a direct
computation in local coordinates  we obtain  that the local
expression of a {\sc sode} $X $ is
\begin{equation}\label{0localsode2r}
X(t,q^i,v^i)=\frac{\partial}{\partial
t}+v^i\frac{\displaystyle
\partial} {\displaystyle
\partial q^i}+
X^i \frac{\displaystyle\partial} {\displaystyle
\partial v^i}.
\end{equation}

   As in the autonomous case, one can prove the following
 \begin{prop}
$X$ is a {\sc sode} if and only if its integral curves are prolongations of curves on $Q$.
\end{prop}

In fact, if $\phi:\r \to \r\times TQ$ is an integral curve of $X$ then $\phi$ is the first prolongation of $\tau\circ \phi$.

The tensor $J$ allows us to introduce the forms $\Theta_L$ and $\Omega_L$ on $\r\times TQ$  as follows: $\Theta_L=dL\circ J$
and $\Omega_L=-d\Theta_L$ with local expressions
\begin{equation}\label{chicf}
\Theta_L= \derpar{L}{v^i}\, dq^i \quad  ,\quad \Omega_L= dq^i\wedge d\left(\derpar{L}{v^i}\right)\,.
 \end{equation}

 Let us consider the  equations
 \begin{equation}\label{lform_mec}\iota_Xdt=1 \;, \quad \iota_{X_L}\widetilde{\Omega}_L =0\,,\end{equation}
where $\widetilde{\Omega}_L=\Omega_L+dE_L\wedge dt$ is the \emph{Poincar\'{e}-Cartan $2$-form}.
\index{Poincar\'{e}-Cartan $2$-form}
The Lagrangian is said to be regular if $(\partial^2L/\partial v^i \partial v^j)$ is not singular. In this case, equations (\ref{lform_mec}) has a unique solution $X$.

 \begin{theorem} Let $L$ be a non-autonomous regular Lagrangian on
$\r\times TQ$ and $X$ the vector field given by (\ref{lform_mec}). Then $X$ is a {\sc sode} whose integral curves $\alpha^{[1]}(t)$ are the solutions of
$$
\ds\frac{d}{dt}\left(  \derpar{L}{v^i}\circ \alpha^{[1]}\right) =\derpar{L}{q^i}\circ \alpha^{[1]} \, ,
$$
which are \textbf{Euler-Lagrange equations} for  $L$.
   \end{theorem}

\begin{remark}
{\rm
The Lagrangian and Hamiltonian Mechanics can be obtained from the unified Skinner-Rusk approach, \cite{CMC}. On the other hand, in \cite{MR-2001} the authors study the non-autonomous Lagrangian invariant by a vector field.
\rqed}
\end{remark}

\part{$k$-symplectic formulation of Classical Field Theories}\label{Part2}

    The symplectic geometry allows us to give a geometric description of  Classical Mechanics (see chapter \ref{chapter: Mechanics}). On the contrary, there exist several alternative models for describing geometrically first-order Classical Field Theories. From a conceptual point of view, the simplest one is the $k$-symplectic formalism, which is a natural generalization to field theories of the standard symplectic formalism.

    The $k$-symplectic formalism (also called polysymplectic formalism of C. G\"{u}n\-ther in \cite{Gu-1987}) is used to give a geometric description of certain kind of Classical Field Theories: in a local description, those whose Lagrangian and Hamiltonian functions do not depend on the coordinates on the basis (that is, the space-time coordinates). Then, the $k$-symplectic formalism is only valid for Lagrangians and Hamiltonians that depend on the field coordinates $(q^i)$ and on the partial derivatives of the field $(v^i_\alpha)$ or the corresponding momenta $(p^\alpha_i)$. The foundations of the $k$-symplectic formalism are the $k$-symplectic manifolds introduced by A. Awane in \cite{Awane-1992, Awane-1994, Awane-2000}, the $k$-cotangent structures introduced by M. de Le\'{o}n {\it et al.} in \cite{LMM-2002,LMS-93, LMeS-97} or the $n$-symplectic structures on the frame bundle introduced by M. McLean and L.K. Norris \cite{MN-2000,No2,No3,No4,No5}.

    In a first chapter of this part of the book, we shall introduce the notion of $k$-symplectic manifold using as a model the cotangent bundle of $k^1$-covelocities of a manifold, that is, the Whitney sum of $k$-copies of the cotangent bundle. Later in chapter \ref{chapter: k-symplectic formalism} we shall describe the geometric equations using the $k$-symplectic structures. This formulation can be applied to the study of Classical Field Theories as we shall see in chapters \ref{chapter: K-SympHamaCFT} and \ref{chapter: K-SympLagCFT}. We present these formulations and several physical examples which can be described using this approach. Finally, we establish the equivalence between  the Hamiltonian and Lagrangian formulations when the Lagrangian function satisfies some regularity property. Moreover, we shall discuss the Hamilton-Jacobi equation in the $k$-symplectic setting (see chapter \ref{ksymp-HJ}).

 %
%    In a first chapter of this part of the book, we introduce the spaces and the necessary geometric elements for describing the $k$-symplectic formulation which allows us to obtain a geometric description of Hamiltonian and Lagrangian Classical Field Theories. We present these formulations and several physical examples which can be described using this approach. Finally, we establish the equivalence between both the Hamiltonian and Lagrangian formulations when the Lagrangian function satisfies certain regularity property.
\newpage
\mbox{}
\thispagestyle{empty} % para que no se numere esta p\'{a}gina

\chapter{$k$-symplectic geometry}\label{chapter: k-symplectic manifolds}

The $k$-symplectic formulation is based in the so-called $k$-symplectic geometry. In this chapter we introduce the $k$-symplectic structure which is a generalization of the notion of symplectic structure.

\index{$k$-symplectic manifold}
        We first describe the geometric model of the called $k$-symplectic manifolds, that is the cotangent bundle of $k^1$-covelocities and we introduce the notion of canonical geometric structures on this manifold. The formal definition of the $k$-symplectic manifold is given in  Section \ref{section k-symp-manifold}.

  \section{The cotangent bundle of $k^1$-covelocities}\protect\label{section k-cotangent}

\index{Cotangent bundle of $k^1$-covelocities}

            We denote by $(T^1_k)^*Q$ the Whitney sum with itself of $k$-copies of the cotangent bundle of a manifold $Q$ of dimension $n$, that is,
                \[
                    (T^1_k)^*Q= T^*Q\oplus_Q \stackrel{k}{\dots} \oplus_Q T^*Q\,.
                \]

            An element $\nu_q$ of $\tkqh$ is a family $({\nu_1}_q,\ldots, {\nu_k}_q)$ of $k$ covectors at the same base point $q\in Q$. Thus one can consider the canonical projection
                \begin{equation}\label{pi k}
                    \begin{array}{rccl}
                        \pi^k\colon & \tkqh& \to & Q\\\noalign{\medskip}
                         &({\nu_1}_q,\ldots,{\nu_k}_q) & \mapsto & \pi^k({\nu_1}_q,\ldots,{\nu_k}_q)=q\,.
                    \end{array}
                \end{equation}

            If $(q^i)$, with ${1\leq i\leq n}$, is a local coordinate system defined on an open set $U \subseteq Q$,
            the induced local (bundle) coordinates system $(q^i, p^{\alpha}_i)$  on $(T^1_k)^*U=(\pi^k)^{-1}(U)$ is given by
                \begin{equation}\label{tkqh: natural coord}
                    q^i({\nu_1}_q,\ldots, {\nu_k}_q)=q^i( q),\qquad p^{\alpha}_i({\nu_1}_q,\ldots, {\nu_k}_q)= \nu_{\alpha_{q}}\left(\frac{\partial}{\partial q^i}\Big\vert_{q}\right)\,,
                \end{equation}
         for $\ak$ and $\n$.

\index{Cotangent bundle of $k^1$-covelocities!canonical coordinates}
            These coordinates are called the \emph{canonical coordinates} on $\tkqh$. Thus, $\tkqh$ is endowed with a smooth  structure of differentiable manifold of  dimension $n(k+1)$.

            The following diagram shows the notation which we shall use along this book:
                \[
                    \xymatrix@=15mm{
                        \tkqh\ar[r]^-{\pi^{k,\alpha}}\ar[rd]_-{ \pi^k}& T^*Q\ar[d]^-{\pi}\\ & Q
                    }
                \]
            where
                \begin{equation}\label{pikalpha}
                    \begin{array}{rccl}
                        \pi^{k,\alpha}\colon &\tkqh & \to & T^*Q\\\noalign{\medskip}
                         & ({\nu_1}_q,\ldots, {\nu_k}_q) & \mapsto & {\nu_\alpha}_q
                    \end{array}\,,
                \end{equation}
            is the canonical projection on each copy of the cotangent bundle $T^*Q$, for each $\ak$.

            \begin{remark}\label{j1qrk}
            {\rm
                 The manifold $\tkqh$  can be described using $1$-jets, (we refer to  \cite{Saunders-89} for more details about jets).

                Let $\sigma:U_q\subset Q \to   \rk  $  and $\tau:V_q\subset Q\to \rk $ be two maps defined in an open neighborhoods   $U_q$ and $V_q$ of $q\in Q$, respectively, such that $\sigma(q)=\tau(q)=0$.  We say that $\sigma$ and $\tau$ are related   at $0\in  \rk$ if $\sigma_*(q)=\tau_*(q)$, which means that the partial derivatives of $\sigma$ and $\tau$ coincide up to order one at $q\in Q$.

                         The equivalence classes determined by this relationship are called {\it jet of order 1}, or, simply, $1$-jets with source $q\in Q$ and the same target.

                         The $1$-jet of a map $\sigma:U_q\subset Q \to   \rk  $  is denoted by $j^1_{q,0}\sigma$ where $\sigma(q)=0$. The set of all $1$-jets at $q$ is denoted by
                         $$J^1(Q,\rk)_0= \ds\bigcup_{q\in Q}J^1_{q,\,0}(Q,\rk)=\ds\bigcup_{q\in Q}\{j^1_{ q,0}\sigma\, \vert\,\sigma:Q\to \rk\makebox{ smooth, }\sigma( q)=0\}\,.$$

                         The canonical projection $\beta:J^1(Q,\rk)_0\to Q$ is defined by $\beta( j^1_{q,0}\sigma )=q$ and $J^1(\rk,Q)_0$                  is called the   \emph{cotangent bundle of $k^1$-covelocities}, \cite{Ehresmann,{kms}}.

                         Let us observe that for $k=1$, $J^1(Q,\rk)_0$ is diffeomorphic to $T^*Q$.

                         We shall now describe the local coordinates on $J^1(\rk,Q)_0$. Let $U$ be a chart of $Q$ with local coordinates $(q^i)$,
                         $1\leq i\leq n$, $\sigma:U_0\subset Q \to \rk$ a mapping such that $q\in U$ and $\sigma^\alpha =x^\alpha                       \circ \sigma$. Then the $1$-jet $j^1_{q,0}\sigma$ is uniquely represented in $\beta^{-1}(U)$ by
                         $$
                         (q^i,p^1_i, \ldots , p^k_i) \; , \quad 1\leq i\leq n
                         $$
                         where
                         \begin{equation}\label{jetcoorl}
                         q^i(j^1_{q,0}\sigma)=q^i(q) \; ,  \quad p^\alpha_i(j^1_{q,0}\sigma)=
                         \derpar{\sigma^\alpha}{q^i}\Big\vert_{q}
                         =d\sigma^\alpha(q)\left( \derpar{}{q^i}\Big\vert_{q} \right)
                           .
                          \end{equation}

                        The manifolds $\tkqh$ and $J^1(\rk,Q)_0$ can be identified, via the diffeomorphism
                            \begin{equation}\label{difeo jiqrk-tkq}
                                \begin{array}{ccc}
                                    J^1(\r^k,Q)_0 & \equiv & T^*Q \oplus \stackrel{k}{\dots} \oplus T^*Q=\tkqh \\
                                    j^1_{q,0}\sigma  & \equiv & (d\sigma^1(q),\ldots, d\sigma^k(q))
                                \end{array}
                            \end{equation}
                              where $\sigma^\alpha= \pi^\alpha \circ \sigma:Q \longrightarrow \r$ is the $\alpha$-th component of $\sigma$ and $\pi^\alpha\colon \r^k \to \r$ the canonical projections for each $\ak$.
                  \rqed
                  }
            \end{remark}

%        \subsection{Canonical forms}\label{section tkqh: canonical forms}

            We now introduce certain canonical geometric structures on $\tkqh$. These structures will be used in the  description of the Hamiltonian $k$-symplectic formalism, see chapter  \ref{chapter: k-symplectic formalism}.

\index{Cotangent bundle of $k^1$-covelocities!canonical forms}
            \begin{definition}\label{section tkqh: canonical forms}
                 We define the \emph{canonical $1$-forms} $\theta^1,\ldots, \theta^k$ on $\tkqh$ as the pull-back of Liouville's $1$-form $\theta$ (see (\ref{Liouville})), by the canonical projection $\pi^{k,\alpha}$ (see (\ref{pikalpha})), that is, for each $\ak$
                    \[
                        \theta^\alpha = (\pi^{k,\alpha})^*\theta\,;
                    \]
                the \emph{canonical $2$-forms} $\omega^1,\ldots, \omega^k$ are defined by
                \[
                    \omega^\alpha=-d\theta^\alpha
                \]
                or equivalently by $\omega^\alpha=(\pi^{k,\alpha})^*\omega$ being $\omega$ the canonical symplectic form on the cotangent bundle $T^*Q$.
            \end{definition}

            If we consider the canonical coordinates $(q^i,p^{\alpha}_i)$ on $\tkqh$ (see (\ref{tkqh: natural coord})), then  the canonical forms $\theta^\alpha, \omega^\alpha$ have the following local expressions:
                \begin{equation}\label{Locformcan}
                    \theta^\alpha=p^\alpha_idq^i\,,\quad \omega^\alpha=dq^i\wedge dp^\alpha_i\,,
                \end{equation}
            with $\ak$.

            \begin{remark}{\rm
                An alternative definition of the canonical $1$-forms $\theta^1,\ldots, \theta^k$ is through  the composition:
                \[
                        \xymatrix@C=29mm{
                            T_{\nu_q}(\tkqh)\ar[r]^-{(\pi^k)_*(\nu_q)} \ar@/^{10mm}/[rr]^-{\theta^\alpha(\nu_q)}& T_{q}Q\ar[r]^-{{\nu_\alpha}_{q}}& \r
                        }
                \]
                That is,
                    \begin{equation}\label{theta}
                        \theta^\alpha(\nu_q)\big(X_{\nu_q}\big):={\nu_\alpha}_q\big((\pi^k)_*(\nu_q)(X_{\nu_q})\big)
                    \end{equation}
                for $X_{\nu_q}\in T_{\nu_q}(\tkqh)$, $\nu_q=({\nu_1}_q,\ldots, {\nu_k}_{q})\in \tkqh$ and $q\in Q$.
            \rqed}
            \end{remark}

        Let us observe that the canonical $2$-forms $\omega^1,\ldots, \omega^k$ are closed forms (indeed, they are exact). An interesting property of these forms is the following: for each $\ak$,  we consider the kernel of each $\omega^\alpha$,  i.e., the set
            \[
                \ker \omega^\alpha=\{X\in T(\tkqh)\,\vert\,\iota_X\omega^\alpha=0\}\, ;
            \]
         then from (\ref{Locformcan}) it is easy to check that
                 \begin{equation}\label{modelksim}
                        \omega^\alpha\Big\vert_{V\times V}=0 \makebox{ and } \ds\bigcap_{\alpha=1}^k\ker \omega^\alpha=\{0\}\,,
                 \end{equation}
            where $V=\ker (\pi^k)_*$ is the vertical distribution of dimension $nk$ associated to $\pi^k\colon \tkqh\to Q$. This vertical distribution is locally spanned by the set
                 \begin{equation}\label{vertdis}
                        \left\{ \derpar{}{p^1_1},\ldots, \derpar{}{p^k_1},\derpar{}{p^1_2},\ldots, \derpar{}{p^k_2} ,\ldots, \derpar{}{p^1_n},\ldots,                   \derpar{}{p^k_n}\right\}\,.
                 \end{equation}

            The properties (\ref{modelksim}) are interesting because the family of the manifold $\tkqh$ with the $2$-forms $\omega^1,\ldots, \omega^k$ and the distribution $V$ is the model for a $k$-symplectic manifold, which will be introduced in the following section.

  \section{$k$-symplectic geometry}\label{section k-symp-manifold}

\index{$k$-symplectic manifold}
    A natural generalization of a symplectic manifold is the notion of the so-called $k$-symplectic manifold. The canonical model of a symplectic manifold is the cotangent bundle $T^*Q$, while the canonical model of a $k$-symplectic manifold is the bundle of $k^1$-covelocities, that is, $(T^1_k)^*Q$.

     The notion of $k$-symplectic structure was independently introduced  by A. Awane \cite{Awane-1992, Awane-2000}, G. G\"{u}nther \cite{Gu-1987},  M. de Le\'{o}n \textit{et al.} \cite{LMM-2002,LMS-88, LMS-93, LMeS-97}, and L.K. Norris \cite{MN-2000, No2}. Let us recall that $k$-symplectic manifolds provide a natural arena to develop Classical Field Theory as an alternative to other geometrical settings which we shall comment in the last part of this book.

     A characteristic of the $k$-symplectic manifold is the existence of a theorem of  Darboux type, therefore all $k$-symplectic manifolds are locally as the canonical model.
     %Thus the reader can  skip this chapter, considering the canonical model, i.e, the cotangent bundle of $k^1$-covelocities, when discussing a  k-symplectic manifold.

\subsection{$k$-symplectic vector spaces}\label{section k-symp vspaces}

\index{$k$-symplectic vector space}
    As we have mentionated above, the $k$-symplectic manifolds constitute the arena for the geometric study of Classical Field Theories. This subsection considers the linear case as a preliminary step for the next subsection.

%    By an exterior form (or simply a form) on a vector space, we mean an alternating multilinear function on that space with values in the field of scalars. The contraction of a vector $v\in V$ and an exterior form $\omega$ on $V$ will be denoted by $\iota_v\omega$.

\index{$k$-symplectic vector space}
    \begin{definition}
        A \emph{\(k\)-symplectic vector space} \((\mathcal{V},\omega^1,\ldots,\omega^k, \mathcal{W})\) is a vector space \(\mathcal{V}\) of dimension \(n(k+1)\), a family of \(k\) skew-symmetric bilinear forms \(\omega^1,\ldots,\omega^k\)  and a vector subspace $\mathcal{W}$ of dimension $nk$ such that
        \begin{equation}\label{k-symp vspaces: nondeg_cond}
            \bigcap_{\alpha=1}^k\ker\,\omega^\alpha=\{0\}\,,
        \end{equation}
        where \[\ker\,\omega^\alpha=\{u\in \mathcal{V} |\, \omega^\alpha(u,v)=0,\,\forall v\in \mathcal{V} \}\] denotes the kernel of \(\omega^\alpha\) and
            \[
                \omega^\alpha\big\vert_{\mathcal{W}\times \mathcal{W}}=0\,,
            \]
        for $\ak$.
    \end{definition}

     The condition (\ref{k-symp vspaces: nondeg_cond}) means that the induced linear map
    \begin{equation}\label{bemol}
        \begin{array}{rccl}
            \sharp_\omega\colon & \mathcal{V} & \to & \mathcal{V}^*\times \stackrel{k}{\cdots}\times \mathcal{V}^*\\\noalign{\medskip}
             & v &\mapsto & (\iota_v\omega^1,\ldots, \iota_v\omega^k)
        \end{array}
    \end{equation}
    is injective, or equivalently, that it has maximal rank, that is,  $rank\, \sharp_\omega=\dim\,\mathcal{V}=n(k+1)$.

   % \begin{remark}
       % The condition $\dim\mathcal{V}=n(k+1)$ is not a restriction. In fact, if $\dim \mathcal{V}=m>1$ is a primer number, then $m=n(k+1)$ with $n=1$ and $k=m-1$, in other case, we consider %$n$ a divisor of $m$ and $k=\frac{m}{n}-1$.
    %\end{remark}

    Note that for \(k=1\) the above definition reduces to that of a symplectic vector space with a given Lagrangian subspace $\mathcal{W}$\footnote{A subspace $\mathcal{W}$ of $\mathcal{V}$  is called a Lagrangian subspace if $\mathcal{W}\subset \mathcal{W}^{\perp}$, there exits another subspace $\mathcal{U}$ such that $ \mathcal{U}\subset\mathcal{U}^{\perp}$ and $\mathcal{V}=\mathcal{W}\oplus \mathcal{U}$, (for more details see \cite{LV-2012}).}.
    \begin{example}\label{euclidean space}
    {\rm
        We consider the vector space $\mathcal{V}=\mathbb{R}^3$ with the family of skew-symmetric bilinear forms
                \[
                    \omega^1=e^1\wedge e^3 \quad \makebox{and} \quad \omega^2=e^2\wedge e^3\,,
                \]
        and the subspace
            \[
                \mathcal{W}=span\{e_1,e_2\},
            \]
        where $\{e_1,e_2,e_3\}$ is the canonical basis of $\mathbb{R}^3$ and $\{e^1,e^2,e^3\}$ its dual basis. It is easy to check that
        \[
            \omega^\alpha\big\vert_{\mathcal{W}\times \mathcal{W}}=0\,,\quad \alpha=1,2\,.
        \]
        Moreover,
        \[
            \ker\,\omega^1=span\{e_2\} \makebox{ and } \ker \,\omega^2=span\{e_1\}
        \] and therefore
        \(
            \ker\,\omega^1\cap\ker\,\omega^2=\{0\}
        \), that is, $(\omega^1,\omega^2, \mathcal{W})$ is a $2$-symplectic structure on $\mathbb{R}^3$.
}
    \end{example}

    \begin{example}\label{r6}
    {\rm
        We consider the vector space $\mathcal{V}=\mathbb{R}^6$ with the subspace
            \[
                \mathcal{W}=span\{e_1,e_2,e_4.e_5\}
            \]
        and the family of skew-symmetric bilinear forms
        \[
            \omega^1= e^1\wedge e^3 + e^4\wedge e^6 \makebox{ and } \omega^2= e^2\wedge e^3 + e^5\wedge e^6
        \]
        where $\{e_1,e_2,e_3,e_4,e_5,e_6\}$ is the canonical basis of $\mathbb{R}^6$ and $\{e^1,e^2,e^3,e^4,e^5,e^6\}$ the dual basis. It is easy to check that
        \[
            \ker\,\omega^1=span\{e_2, e_5\} \makebox{ and } \ker \,\omega^2=span\{e_1, e_4\}
        \] and therefore
        \(
            \ker\,\omega_1\cap\ker\,\omega_2=\{0\}
        \). Moreover
        \[
            \omega^\alpha\big\vert_{\mathcal{W}\times \mathcal{W}}=0\,,\quad \alpha=1,2\,.
        \]
        That is, $(\omega^1,\omega^2,\mathcal{W})$ is a $2$-symplectic structure on $\mathbb{R}^6$.

        Another $k$-symplectic structure on $\mathbb{R}^6$ is given by the family of $2$-forms $\omega^\alpha=e^\alpha\wedge e^6,\,$ with $1\leq \alpha\leq 5$, and $\mathcal{W}=span\{e_1,e_2,e_3,e_4,e_5\}$ which is a $5$-symplectic structure on $\mathbb{R}^6$.
    }
    \end{example}
    \begin{example}\label{k-symp vspaces: canonical model}
    {\rm
        It is well-known that for any vector space $V$, the space $V\times V^*$ admits a canonical symplectic form $\omega_V$ given by
        \[
            \omega_V\left( (v,\nu), (w,\eta)\right)= \eta(v)-\nu(w)\,,
        \]
        for $v, w\in V$ and $\nu,\eta\in V^*$ (see for instance \cite{AM-1978}). This structure has the following natural extension to the $k$-symplectic setting. For any $k$, the space $\mathcal{V}=V\times V^*\times\stackrel{k}{\cdots}\times V^*$ can be equipped with a family of $k$ canonical skew-symmetric bilinear forms $(\omega_V^1,\ldots, \omega_V^k)$ given by
        \begin{equation}\label{canonical_omega}
            \omega_V^\alpha\left( (v,\nu_1,\ldots, \nu_k),(w,\eta_1,\ldots, \eta_k)\right)=\eta_\alpha(v)-\nu_\alpha(w)\,,
        \end{equation}
        for $v,w\in V$ and $(\nu_1,\ldots, \nu_k), (\eta_1,\ldots,\eta_k)\in V^*\times\stackrel{k}{\cdots}\times V^*$. Now if we consider the subspace $\mathcal{W}=\{0\}\times V^*\times \stackrel{k}{\ldots}\times V^*$ a simple computation  shows that $(V\times V^*\times\stackrel{k}{\ldots}\times V^*, \omega_V^1,\ldots,\omega_V^k, \mathcal{W})$ is a $k$-symplectic vector space. In fact, this is a direct consequence of the computation of the kernel of $\omega^\alpha_V$ for $\ak$, i.e.,
            \[
                \ker \omega^\alpha_V=\{(v,\nu_1,\ldots, \nu_k)\in \mathcal{V}\,\vert\,v=0\,\makebox{ and } \nu_\alpha=0\}\,.
            \]

        Let us observe that if we consider the natural projection
            \[
            \begin{array}{rccl}
                pr_\alpha\colon & V\times V^*\times \stackrel{k}{\cdots}\times V^* & \to &  V\times V^*\\\noalign{\medskip}
                & (v,\nu_1,\ldots,\nu_k)&\mapsto  &(v,\nu_\alpha),
            \end{array}
            \]
           then  the $2$-form $\omega_V^\alpha$ is exactly $(pr_\alpha)^*\omega_V$.
           }
    \end{example}
    \index{$k$-symplectomorphism}
    \begin{definition}
        Let $(\mathcal{V}_1,\omega^1_1,\ldots, \omega^k_1,\mathcal{W}_1)$ and $(\mathcal{V}_2,\omega^1_2,\ldots, \omega^k_2, \mathcal{W}_2)$ be two $k$-sym\-plectic vector space and let $\phi\colon \mathcal{V}_1\to \mathcal{V}_2$ be a linear isomorphism.  The map $\phi$ is called a \emph{$k$-symplecto\-mor\-phism} if it preserves the $k$-symplectic structure, that is
        \begin{enumerate}
            \item  $\phi^*\omega_2^\alpha=\omega^\alpha_1$; for each $ \ak , $
            \item $\phi(\mathcal{W}_1)=\mathcal{W}_2$
        \end{enumerate}
    \end{definition}

    An important property of the $k$-symplectic structures is the following proposition, which establish a theorem of type Darboux for this generalization of the symplectic structure. A proof of the following result can be found in \cite{Awane-1992, LV-2012}.

    \begin{prop}\label{linear Darboux th}
        Let $(\omega^1,\ldots, \omega^k,\mathcal{W})$ be a $k$-symplectic structure on the vector space $\mathcal{V}$. Then there exists a basis (Darboux basis) $(e^i,f^\alpha_i)$  of $\mathcal{V}$ (with $1\leq i\leq n$ and $1\leq \alpha\leq k$), such that for each $ \ak $
        \[
            \omega^\alpha=e^i\wedge f^\alpha_i\,.
        \]
    \end{prop}

%%%%%%%%%%%%%%%%%%%555
\subsection{$k$-symplectic manifolds}\label{section k-symp manifolds}
%%%%%%%%%%%%%%%%%%%

\index{$k$-symplectic manifold}

    We turn now to the globalization of the ideas of the previous section to $k$-symplectic manifolds.

    \begin{definition}\label{k-simest}
        Let $M$ be  a smooth manifold of dimension $n(k+1)$,  $V$
        be an integrable distribution of dimension $nk$ and $\omega^1,\ldots,
        \omega^k$ a family of closed differentiable $2$-forms defined on $M$. In such a case $(\omega^1,\ldots, \omega^k,V)$ is called \emph{a $k$-symplectic structure on $M$} if and only if

                \begin{enumerate}
                    \item $\omega^\alpha\Big\vert_{V\times V}=0,\quad 1\leq \alpha\leq k\,$,
                    \item $\ds\bigcap_{\alpha=1}^k\ker \omega^\alpha=\{0\}\,.$
                 \end{enumerate}

        A manifold $M$ endowed with a $k$-symplectic structure is said to be a \emph{$k$-symplectic manifold}.
    \end{definition}
    \begin{remark}{\rm
    In the above definition, the condition $\dim M=n (k+1)$ with $n,k\in\mathbb{N}$ implies that, for an arbitrary manifold $M$ of dimension $m$, only can exist a $k$-symplectic structure if there is a  couple $(n,k)$ such that $M=n(k+1)$. Thus, for instance, if $M=\mathbb{R}^6$ there isn't a $3$-symplectic structure for instance; in fact,  only can exist $k$-symplectic structures if $k\in\{ 1, 2,5\}$.
    \rqed}
    \end{remark}
    \index{$k$-symplectomorphism}

    \begin{definition}\label{k-symplectomorphism}
        Let $(M_1, \omega_1^1,\ldots, \omega_1^k,V_1)$ and $(M_2, \omega_2^1,\ldots, \omega_2^k,V_2)$ be two $k$-symplectic manifolds and let $\phi\colon M_1\to M_2$ be a diffeomorphism. $\phi$ is called a \emph{$k$-symplectomorphism} if it preserves the $k$-symplectic structure, that is if
        \begin{enumerate}
            \item   $
            \phi^*\omega_2^\alpha = \omega_1^\alpha\,$; for each $\ak$,
            \item $\phi_*(V_1)=V_2$.
        \end{enumerate}
    \end{definition}
\begin{remark}
{\rm
Note that if $(M,\omega^1,\ldots, \omega^k, V)$ is a $k$-symplectic manifold then  $(T_xM,\omega^1(x),\ldots, \omega^k(x),$ $ T_xV)$ is a $k$-symplectic vector space for all $x\in M$.
\rqed}
    \end{remark}
    \begin{example}
    {\rm
        Let $(T^1_k)^*Q$ be the cotangent bundle of $k^1$-covelocities, then from (\ref{modelksim}) and (\ref{vertdis}) one easy checks that $(T^1_k)^*Q$, equipped with the canonical forms and the distribution $V=\ker (\pi^k)_*$,  is a $k$-symplectic manifold.
}
    \end{example}
    \begin{remark}\label{relation_canonical_models}
    {\rm
        For each $\nu_q\in (T^1_k)^*Q= T^*Q\oplus\stackrel{k}{\cdots}\oplus T^*Q$, the $k$-symplectic vector space $(T_{\nu_q}(\tkqh),$ $\omega^1(\nu_q),\ldots, \omega^k(\nu_q), T_{\nu_q}V)$ associated to the $k$-sym\-plec\-tic manifold $((T^1_k)^*Q,\omega^1,\ldots, \omega^k, V)$ is $k$-symplec\-to\-morphic to the canonical $k$-sym\-plec\-tic structure on $T_qQ\times T_q^*Q\times\stackrel{k}{\cdots}\times T_q^*Q$ described in example \ref{k-symp vspaces: canonical model} with $V=T_qQ$.\rqed}
    \end{remark}

\index{$k$-symplectic manifold!Darboux theorem}
\index{Darboux theorem!$k$-symplectic}
 The following theorem is the differentiable version of Theorem \ref{linear Darboux th}. This theorem has been proved in \cite{Awane-1992,{LMS-93}}.
    \begin{theorem}[$k$-symplectic Darboux theorem]\label{Darboux k-simp}
        Let $(M,\omega^1,\ldots, \omega^k,V)$ be a  $k$-symplec\-tic manifold. About every point of $M$ we can find a local coordinate system $(x^i,y^\alpha_i),\, 1\leq i\leq n,\, 1\leq \alpha\leq k$, called adapted coordinate system, such that
        \[
            \omega^\alpha=\ds\sum_{i=1}^ndx^i\wedge dy^\alpha_i
        \]
        for each $1\leq \alpha\leq k$, and
        \[
            V=span\Big\{ \ds\frac{\partial}{\partial y^\alpha_i},\, 1\leq i\leq n,\, 1\leq \alpha\leq k\Big \}\,.
        \]
    \end{theorem}
    \begin{remark}
    {\rm
        Notice that the notion of $k$-symplectic manifold introduced in this chapter coincides with the  one given by A. Awane \cite{Awane-1992,Awane-2000}, and it is equivalent to the notion of standard polysymplectic structure\footnote{A $k$-polysymplectic form on an $n(k+1)$-dimensional manifold $N$ is an $\mathbb{R}^k$-valued closed nondegenerated two-form on $N$ of the form
         \[
            \Omega=\displaystyle\sum_{i=1}^k\eta^i\otimes e_i\,,
         \]where $\{e_1,\ldots, e_k\}$ is any basis of $\mathbb{R}^k$. The pair $(N,\Omega)$ is called a $k$-polysymplectic manifold.} of C. G\"{u}nter \cite{Gu-1987} and integrable $p$-almost cotangent structure introduced by M. de Le\'{o}n \textit{et al.} \cite{LMS-88, LMS-93}.

        Observe that when $k=1$, Awane's definition reduces to the notion of polarized symplectic manifold, that is a symplectic manifold with a Lagrangian submanifold. For that, in \cite{LV-2012} we distinguish between \textit{$k$--symplectic} and \textit{polarized $k$--symplectic manifolds}.

        By taking a basis $\{e^1, \ldots, e^k\}$ of $\mathbb{R}^k$, every $k$-symplectic manifold $(N,\omega^1,\ldots, \omega^k)$ gives rise to a polysymplectic manifold $(N,\Omega=\sum_{i=1}^k \omega^i\otimes e_i)$. As $\Omega$ depends on the chosen basis, the polysymplecic manifold $(N,\Omega)$ is not canonically constructed. Nevertheless, two polysymplectic forms $\Omega_1$ and $\Omega_2$ induced by the same $k$-symplectic manifold and different bases for $\mathbb{R}^k$ are the same up to a change of basis on $\mathbb{R}^k$. In this case, we say that $\Omega_1$ and $\Omega_2$ are \textit{gauge equivalent}. In a similar way, we say that $(N,\omega^1,\ldots, \omega^k)$ and $(N,\tilde{\omega}^1,\ldots, \tilde{\omega}^k)$ are gauge equivalent if they give rise to gauge equivalent polysymplectic forms, \cite{LV-2014}.

        \rqed
        }
    \end{remark}

 \newpage
\mbox{}
\thispagestyle{empty} % para que no se numere esta p\'{a}gina

\chapter{$k$-symplectic formalism}\label{chapter: k-symplectic formalism}

In this chapter we shall describe the $k$-symplectic formalism. As we shall see in the following chapters, using this formalism we can study Classical Field Theories in the Hamiltonian and Lagrangian cases.

 One of the most important elements in the $k$-symplectic approach  is the notion of $k$-vector field. Roughly speaking, it is a family of $k$ vector fields. In order to introduce this notion in section \ref{section k-vector field}, we previously consider the tangent bundle of $k^1$-velocities of a manifold, i.e. the Whitney sum of $k$ copies of its tangent bundle with itself. In section \ref{section k-tangent} we shall describe this manifold with more details.

 Here we shall introduce a geometric equation, called the $k$-symplectic Hamiltonian equation, which  allows us to describe Classical Field Theories when the $k$-symplectic manifold is the cotangent bundle of $k^1$-covelocities or its Lagrangian counterpart under some regularity condition satisfied by the Lagrangian function.

   \section{$k$-vector fields and integral sections}\label{section k-vector field}

        %In the Hamiltonian and Lagrangian geometric
%        description of Classical Mechanics the solutions
%        of the Hamilton equations and Euler-Lagrange equations
%        are interpreted as integral curves of certain vector field.
%
%        However, for Classical Field Theories, the Hamilton-De Donder-Weyl equations  are interpreted in terms of $k$-vector field.

        We shall devote this  section to introduce the notion of $k$-vector field and discuss its integrability. This notion is fundamental in the $k$-symplectic and $k$-cosymplectic approaches.
 %       In the $k$-symplectic
%        description of Classical Field Theories appears a similar
%        interpretation of the fields solutions of the Hamilton-De
%         Donder-Weyl field equations and of the Euler-Lagrange field equations.

         %In this case we define the notion
%          of integral section of a   $k$-vector field as
%           a generalization of the notion of integral curve of a
%           vector field; and the field
%           solutions of the Hamilton-De Donder-Weyl and Euler-Lagrange
%            equations will be  shown
%            as integral sections of certain   $k$-vectors fields.
%
%
%               Therefore, the notion
%               of $k$-vector field is fundamental in this book.

Consider the tangent bundle $\tau\colon TM\to M$ of an arbitrary $n$-dimensional smooth manifold $M$ and  consider the space\footnote{A completed description of this space $T^1_kM$ can be found in section \ref{section k-tangent}.}.
                \[
                    T^{1}_{k}M=TM\oplus\stackrel{k}{\ldots}\oplus TM\,,
                \]
            as the Whitney sum of $k$ copies of the tangent bundle $TM$. Let us observe that an element ${\rm v}_p$ of $T^1_kM$ is a family of $k$ tangent vectors $({v_1}_p,\ldots, {v_k}_p)$ at the same point $p\in M$. Thus one can consider the canonical projection
                \begin{equation}\label{tkq: tauk}
                    \begin{array}{rccl}
                        \tau^k\colon & T^1_kM& \to & M\\\noalign{\medskip}
                         & {\rm v}_p=({v_1}_p,\ldots, {v_k}_p) & \mapsto & p\,.
                    \end{array}
                \end{equation}

                \index{$k$-vector field}

                \begin{definition}\label{k-vector field def}
                    A \emph{$k$-vector field} $\mathbf{X}$ on $M$ is a section of the canonical projection $\tau^k\colon T^1_kM\to M$. We denote by $\vf^k(M)$ the set of $k$-vector fields on $M$.
                \end{definition}

            Since $T^1_kM$ is the Whitney sum $TM\oplus \stackrel{k}{\dots} \oplus TM$ of $k$ copies of $TM$, a $k$-vector field $\mathbf{X}$ on $M$ defines a family of $k$ vector fields $(X_{1}, \dots, X_{k})$ on $M$ through  the projection of $\mathbf{X}$ onto every factor $TM$ of the $T^1_kM$,  as it is showed in the following diagram for each $\ak$:
                \[
                    \xymatrix@=15mm{
                        & T^1_kM\ar[d]^-{\tau^{k,\alpha}}\\ M\ar[ru]^-{\bf{X}}\ar[r]^-{X_\alpha} & TM
                    }
                \]
            where $\tau^{k,\alpha}$ denotes the canonical projection over the $\alpha^{th}$ component of $T^1_kM$, i.e.
            \[
            \begin{array}{lccl}
                \tau^{k,\alpha}\colon & T^1_kM & \to & TM\\\noalign{\medskip}
                    & ({v_1}_p,\ldots, {v_k}_p) & \mapsto & {v_\alpha}_p
            \end{array}\,.
            \]
            In what follows, we shall use indistinctly the notation $\mathbf{X}$ or $(X_1, \ldots, X_k)$ to refer a $k$-vector field.

            Let us recall that given a vector field, we can consider the notion of integral curve. In this new setting we now introduce the generalization of this concept for $k$-vector fields: integral sections of a $k$-vector field.\index{Integral section}
                \begin{definition} \label{integral section def}
                    An \emph{integral section} of a $k$-vector field $\mathbb{X}=(X_{1},\dots, X_{k})$, passing through  a point $p\in M$, is a map $\varphi:U_0\subset \r^k \rightarrow M$, defined in some neighborhood                $U_0$ of $0\in\rk$ such that
                        \begin{equation}\label{integral section expr}
                            \qquad\varphi(0)=p, \, \,   \varphi_{*}(x)\Big(\displaystyle\frac{\displaystyle\partial}{\displaystyle\partial x^\alpha}\Big\vert_{x}\Big) = X_{\alpha}(\varphi (x))\,,
                        \end{equation}
                    for all $x\in U_0$ and for all $\ak$.

                    If there exists an integral section passing  through    each point of $M$, then $(X_{1},\dots, X_{k})$ is called an \emph{integrable $k$-vector field}. \index{$k$-vector field!integrable}
            \end{definition}

Using local coordinates $(U,y^i)$ on $M$   we can write
$$
\varphi_* (x) \left( \displaystyle\frac{\partial}
                        {\partial x^\alpha}\Big\vert_{x} \right)
                        =\derpar{\varphi^i}{x^\alpha}\Big\vert_{x}\,
                        \derpar{}{y^i}\Big\vert_{\varphi(x)}
                        \;, \quad X_\alpha=X_\alpha^i\derpar{}{y^i}
$$
where $\varphi^i=y^i\circ \varphi$.

              Thus  $\varphi$ is an integral section of ${\bf
X}=(X_1,\ldots, X_k)$ if and only if the following system of partial
differential equations holds:
                \begin{equation}\label{integral section equivalence cond}
                    \derpar{\varphi^i}{x^\alpha}\Big\vert_{x} \, =\, X_\alpha^i(
                    \varphi (x))
                \end{equation}
         where $x\in U_0\subseteq\rk$, $\ak$ and $1\leq i\leq n$.

Let us remark that if $\varphi$ is an integral section of a $k$-vector field $\mathbf{X}=(X_1,\ldots, X_k)$, then each curve on $M$ defined by $\varphi_\alpha(s)=\varphi(se_\alpha)$, with $\{e_1,\ldots, e_k\}$ the canonical basis on $\rk$ and $s\in \mathbb{R}$, is an integral curve of the vector field $X_\alpha$ on $M$. However, given $k$ integral curves of $X_1,\ldots, X_k$ respectively, it is not possible in general to reconstruct an integral section of $(X_1,\ldots, X_k)$.

We remark that a
$k$-vector field ${\bf X}=(X_1,\ldots , X_k)$ with $\{X_1,\ldots, X_k\}$ linearly independent,  is integrable if and
only if $[X_\alpha,X_\beta] = 0$, for each $\alpha,\beta$, that is, $\mathbf{X}$ is integrable if and only if the distribution generated by $\{X_1,\ldots, X_k\}$ is integrable.
This is the geometric expression
of the integrability condition of the preceding differential equation
(see, for instance,  \cite{Die} or \cite{Lee}).

\begin{remark}
{\rm
$k$-vector fields in a manifold $M$ can also be
defined in a more general way as sections of the bundle $\Lambda^k
M\to M$ (i.e., the contravariant skew-symmetric
tensors of order $k$ in $M$). Starting from the $k$-vector fields
${\bf X}=(X_{1},\dots, X_{k})$ defined
in Definition \ref{k-vector field def}, and making the wedge product
$X_{1}\wedge\ldots\wedge X_{k}$, we obtain the particular class
of the so-called
{\sl decomposable} or {\sl homogeneous $k$-vector fields}, which can
be associated with distributions on $M$.
(See \cite{EMR-1998} for a detailed exposition on these
topics).\rqed}
\end{remark}

\begin{example}
{\rm
            Consider $M=(T^1_3)^*\r$ and a $3$-vector field $(X_1,X_2,X_3)$ with  local expression
            \[
                        X_\alpha= p^\alpha \derpar{}{q} + (X_\alpha)^\beta\derpar{}{p^\beta},\quad 1\leq \alpha\leq 3\,,
                    \]
                where the functions $(X_\alpha)^\beta$ with $1\leq \alpha,\beta\leq 3$ satisfy
                    \[
                        (X_1)^1+ (X_2)^2+ (X_3)^3=-4\pi r
                    \]
                $r$ being a constant.

            Then $\varphi\colon U_0\to (T^1_3)^*\r$ with components $\varphi(x)=(\psi(x),\psi^\alpha(x))$
            is an integral section of $(X_1,X_2,X_3)$ is and only if (see (\ref{integral section equivalence cond}))
            \begin{eqnarray*}
                        \psi^\alpha = \derpar{\psi}{x^\alpha}\,,\quad \alpha=1,2,3,\\\noalign{\medskip}
                        4\pi r = - \Big(\derpar{\psi^1}{x^1}+  \derpar{\psi^2}{x^2}+\derpar{\psi^3}{x^3}\Big)\,,
                    \end{eqnarray*}
            which are  the electrostatic equations (for more details about these equations, see section \ref{example k-symp: hamiltonian electrostatic}).
}
        \end{example}

              \section{$k$-symplectic Hamiltonian equation}\label{Section HDW eq: k-symp approach}

    \index{$k$-symplectic approach!Hamiltonian function}

            Let $\left( M,\omega^1, \ldots, \omega^k,V \right)$  a $k$-symplectic manifold and $H$  a Hamiltonian function defined on $M$, that is, a function $H:M \rightarrow \r$
                \begin{definition}\label{shks}
                    The family $(M,\omega^\alpha,H)$ is called \emph{$k$-symplectic Hamiltonian system}.
                \end{definition}
                \index{$k$-symplectic Hamiltonian system}
                \index{$k$-symplectic approach! Hamiltonian system}
%            The $k$-symplectic geometry allows us to give a geometric version of the Hamiltonian field equations. Indeed, we have the following

            Given a $k$-symplectic Hamiltonian system $(M,\omega^\alpha, H)$ we define a vector bundle morphisms $\flat_\omega$ as follows:
        \begin{equation}\label{bemol ksymp}
            \begin{array}{rccl}
                \flat_\omega\colon & T^1_kM & \to &T^*M\\\noalign{\medskip}
                 & (v_1,\ldots, v_k) & \mapsto & \flat_\omega(v_1,\ldots, v_k)=
                 {\rm trace}(\iota_{v_\beta}\omega^\alpha)=\ds\sum_{\alpha=1}^k\iota_{v_\alpha}\omega^\alpha\,.
            \end{array}
        \end{equation}
        The above morphism induce a morphism of $\mathcal{C}^\infty(M)$-modules
        between the corresponding space of sections
        $\flat_\omega\colon \mathfrak{X}^k(M)\to \bigwedge^1(M)$.

        \begin{lemma}
            The map $\flat_\omega$ is surjective.
        \end{lemma}
        \begin{proof}
            This result is a particular case of the following
            algebraic assertion: \textit{If $V$ is a vector space with
            a $k$-symplectic structure $(\omega^1,\ldots,
            \omega^k, \mathcal{W})$, then the map
            \[
            \begin{array}{rccl}
                \flat_\omega\colon & V\times\stackrel{k}{\ldots}\times V & \to &V^*\\\noalign{\medskip}
                 & (v_1,\ldots, v_k) & \mapsto & \flat_\omega(v_1,\ldots, v_k)=
                 {\rm trace}(\iota_{v_\beta}\omega^\alpha)=\ds\sum_{\alpha=1}^k\iota_{v_\alpha}\omega^\alpha\,
            \end{array}
        \] is surjective.}

        Indeed, we consider the identification
            \beq
                \begin{array}{rlcl}
                    F\colon &V^*\times \stackrel{k}{\ldots}\times V^* &
                     \cong &(V\times\stackrel{k}{\ldots}\times V)^*\\\noalign{\medskip}
                    &(\nu_1,\ldots,\nu_k) & \mapsto & F(\nu_1,\ldots,\nu_k) \ ,
                \end{array}
                \label{ident2}
            \eeq
            where $F(\nu_1,\ldots,\nu_k)(v_1,\ldots, v_k)={\rm trace}\,\big(\nu_\alpha(v_\beta)\big)=\ds\sum_{\alpha=1}^k\nu_\alpha(v_\alpha)$, and  we consider the map $\sharp_\omega$ defined in (\ref{bemol}).

        We recall that as $(\omega^1,\ldots, \omega^k, \mathcal{W})$ is a $k$-symplectic
        structure,  $\sharp_\omega$ is injective
        and therefore the dual map $\sharp_\omega^*$ is surjective.

        Finally, using the identification (\ref{ident2}) it is immediate to prove
        that $\flat_\omega=-\sharp_\omega^*$ and thus
        $\flat_\omega$ is surjective.\qed
        \end{proof}

                Let $(M,\omega^\alpha, H)$ be a $k$-symplectic Hamiltonian system and $\,{\bf X} \in \vf^k(M)$ a $k$-vector field solution of the geometric equation
                \begin{equation}\label{ecHksym}
                            \flat_\omega(\mathbf{X})=\ds\sum_{\alpha=1}^k\iota_{X_\alpha}\omega^\alpha=dH\,.
                \end{equation}

                Given a local coordinate system $\left( q^i,p^{\alpha}_i \right)$,    each  $X_\alpha$ is locally given by
                        $$
                            X_\alpha \, = \, (X_\alpha)^i \, \frac{\partial}{\partial q^i} \, + \, (X_\alpha)^\beta_i \, \frac{\partial}{\partial p^\beta_i}\; ,\quad 1 \leq \alpha \leq k\,.
                        $$

                    Now, since
                    $$
                            dH \, = \, \derpar{H}{q^i} \, dq^i  \, + \, \derpar{H}{p^\alpha_i} \, dp^\alpha_i \;,
                        $$
                        and        $$\omega^\alpha=dq^i\wedge dp^\alpha_i$$
                                            we deduce that
                    the equation (\ref{ecHksym}) is locally equivalent to the following equations
                        \begin{equation}\label{ecHDWloc}
                            \frac{\partial H}{\partial q^i} \, = \,- \displaystyle\sum_{\beta=1}^k (X_\beta)^\beta_i \,, \quad \frac{\partial H}{\partial p^{\alpha}_i} \, = \, (X_\alpha)^i \,,
                        \end{equation}
                    with $ 1 \leq i \leq n $ and $ 1 \leq \alpha \leq k\; .$

                     Let us suppose now that the $k$-vector field $\textbf{X}=\left( X_1,\ldots, X_k \right)$, solution of (\ref{ecHksym}), is integrable and
                    $$\begin{array}{cccl}\varphi: & \r^k & \longrightarrow  & M \\ \noalign{\medskip}
                       & x& \to & \varphi(x)=\left( \psi^i(x), \psi^\alpha_i(x) \right)
                                        \end{array}$$
                     is an integral section of  $\textbf{X}$, i.e. $\varphi$ satisfies
                      (\ref{integral section expr})  which in this case is locally equivalent to the following system of partial differential equations (condition (\ref{integral section equivalence cond}))
        \begin{equation}\label{tkqh integral section equivalence cond}
            \derpar{\psi^i}{x^\alpha}\Big\vert_{x}=(X_\alpha)^i(\varphi(x))\,,\quad \derpar{\psi^\beta_i}{x^\alpha}\Big\vert_{x}=(X_\alpha)^\beta_i(\varphi(x))
        \, .\end{equation}

                     From  (\ref{ecHDWloc}) and (\ref{tkqh integral section equivalence cond})   we obtain
                        \begin{equation}\label{local K-sympHeq}
                            \frac{\partial H}{\partial q^i}\Big\vert_{\varphi(x)} \, = \, -
                            \displaystyle\sum_{\beta=1}^k \frac{\partial \psi^\beta_i} {\partial
                            x^\beta}\Big\vert_{x} \,, \quad \frac{\partial H}{\partial
                            p^{\alpha}_i}\Big\vert_{\varphi(x)} \, = \, \frac{\partial \psi^i}{\partial
                            x^\alpha}\Big\vert_{x}
                        \end{equation}
                        where $  1 \leq i \leq n \,, \, 1 \leq \alpha \leq k$.

                        This theory can be summarized in the following
                \begin{theorem}\label{fhks}
                    Let $(M,\omega^\alpha, H)$ be a $k$-symplectic Hamiltonian system and $\,{\bf X} = (X_1, \ldots,  X_k)$ an integrable $k$-vector field on $M$ solution of the equation (\ref{ecHksym}).

                    If $\varphi:\rk \to M$ is an integral section of $\mathbf{X}$, then $\varphi$ is a solution of the following systems of partial differential equations
                    \[
                            \frac{\partial H}{\partial q^i}\Big\vert_{\varphi(x)} \, = \, -
                            \displaystyle\sum_{\beta=1}^k \frac{\partial \psi^\beta_i} {\partial
                            x^\beta}\Big\vert_{x} \,, \quad \frac{\partial H}{\partial
                            p^{\alpha}_i}\Big\vert_{\varphi(x)} \, = \, \frac{\partial \psi^i}{\partial
                            x^\alpha}\Big\vert_{x}\,.
                        \]
                \end{theorem}

                From now, we shall call this equation (\ref{ecHksym}) as \emph{$k$-symplectic Hamiltonian equation}.
\index{$k$-symplectic Hamiltonian equation}
                \begin{definition}
                    A $k$-vector field $\mathbf{X}=(X_1,\ldots, X_k)\in \vf^k(M)$ is called a \emph{$k$-symplectic Hamiltonian $k$-vector field} for a $k$-symplectic Hamiltonian system $(M,\omega^\alpha, H)$ if $\mathbf{X}$ is a solution of (\ref{ecHksym}).
\index{$k$-vector field!$k$-symplectic Hamiltonian}
                    We denote by $\vf^k_H(M)$ the set of $k$-vector fields which are solution of (\ref{ecHksym}), i.e.
                        \begin{equation}\label{vfkh}
                            \vf^k_H(M)\colon=\{\mathbf{X}=(X_1,\ldots, X_k)\in \vf^k(M)\,\vert \,  \flat_\omega(\mathbf{X})=dH\}\,.
                        \end{equation}
                \end{definition}

                One can guarantee the existence of the solution of the $k$-symplec\-tic Hamiltonian equation (\ref{ecHksym}), but the solution is not unique.
                        In fact, let $H\in\mathcal{C}^\infty(M)$ be a function on $M$.
        As $dH\in \Omega^1(M)$ and the map $\flat_\omega$ is surjective,
        then there exists a $k$-vector field $\mathbf{X}^H=(X_1^H,\ldots, X_k^H)$ satisfying
        \begin{equation}\label{poly ham eq}
            \flat_\omega(X_1^H,\ldots, X_k^H)=dH\,,
        \end{equation}
        i.e. $(X_1^H,\ldots, X_k^H)$ is a $k$-vector field solution of the $k$-symplectic Hamiltonian equation (\ref{ecHksym}).

       For instance one can define $\textbf{X}=\left(X_1,\ldots, X_k \right)$ locally as
                            \begin{equation}
                                \begin{array}{l}
                                    X_1 \, = \, \displaystyle\frac{\partial H}{\partial p^1_i}\,
                                    \frac{\partial}{\partial q^i} \, - \, \frac{\partial H}{\partial q^i} \, \frac{\partial}{\partial p^1_i} \\ \noalign{\medskip} X_\alpha \,= \, \displaystyle\frac{\partial H}{\partial p^{\alpha}_i}\,                             \frac{\partial}{\partial q^i} \,, \quad 2 \leq \alpha \leq k
                                \end{array}
                            \end{equation}
                        and using a partition of the unity one can find a $k$-vector field $\textbf{X}=\left( X_1,\ldots, X_k \right)$ defined globally and satisfying (\ref{ecHksym}).

                        Now we can assure the existence of solutions of (\ref{ecHksym}) but not its uniqueness. In fact, let us observe that given a particular solution $\left( X_1,\ldots, X_k \right)$ then any element of the set $\left( X_1,\ldots, X_k \right) + \ker\,\flat_\omega$ is also a solution, since given $\left( Y_1,\ldots, Y_k \right)\in \ker \flat_\omega$ then we have
                        \begin{equation}\label{kerflat}
                            Y^i_\beta=0,\quad \ds\sum_{\alpha=1}^k(Y_\alpha)^\alpha_i=0\,,
                        \end{equation}
where each $Y_\alpha$ is locally given by
\[
    Y_\alpha=Y^i_\alpha\derpar{}{q^i} + (Y_\alpha)^\beta_i\derpar{}{p^\beta_i}\,,
\]
for $\ak$.

Another interesting remark is that  a $k$-vector field solution of the equation (\ref{ecHksym}) is not necessarily integrable but in order to obtain the result of the theorem  \ref{fhks} it is necessary the existence of integral sections. We recall that an integrable $k$-vector field is equivalent to the condition $[X_\alpha,X_\beta]=0$ for all $1\leq \alpha,\beta\leq k$.

\begin{remark}{\rm

Using the $k$-symplectic formalism presented in this chapter we can study symmetries and conservation laws on first-order classical field theories, see \cite{RSV-2007, RSS-2010}. A large part of the discussion of the paper \cite{RSV-2007} is a generalization of the results obtained for non-autonomous mechanical systems (see, in particular, \cite{mssv,LM-96}).  The general problem of a group of symmetries acting on a $k$-symplectic manifold and the subsequent theory of reduction has been analyzed in \cite{MRSV-2013,MRS-2004}. We further remark that the problem of symmetries in field theory has been analyzed using other geometric frameworks, see for instance \cite{EMR-1999,Gymmsy,Gymmsy2,LMS-2004}. About this topic, Noether's theorem associates conservation laws to Cartan symmetries, however, these kinds of symmetries do not exhaust the set of symmetries. Different attempts have been made to extend Noether's theorem in order to include the so-called hidden symmetries and the corresponding conserved quantities, see for instance \cite{SC-81} in Mechanics, \cite{EMR-1999} in multisymplectic field theories or \cite{RSV-2013} in the $k$-symplectic setting.

The $k$-symplectic formalism described here can be extended to another geometrical approaches. For instance:
\begin{itemize}

\item The $k$-symplectic approach can be also studied when one consider classical field theories subject to nonholonomic constraints \cite{LMSV-2008}. The procedure developed in \cite{LMSV-2008} extends that by Bates and Sniatycki \cite{BS-1993} for the linear case. The interest of the study of nonholonomic constraints has been stimulated by its close connection to problems in control theory (see, for instance, \cite{Bloch,{BKMM-1996},{Cortes-2002}}. In the literature, one can distinguish mainly two different approaches in the study of systems subjected to a nonholonomic constraints. The first one is based on the d'Alembert's principle and the second is a constrained variational approach. As is well know, the dynamical equation generated by both approaches are in general not equivalent \cite{CLMM}. The nonholonomic field theory has been studied using another geometrical approaches, (see, for instance \cite{BLMS-2002,LMM-1997,LMM-1996,LM-1996,LMS-2003,Van-2007,Vanthesis,VCLM-2005,MV-2008}).

\item Another interesting setting is the category of the Lie algebroids \cite{Mack-1987,{Mack-1995}}. For further information on groupoids and lie algebroids and their roles in differential geometry see \cite{CW-1999,{HM-1990}}. Let us remember that a Lie algebroid is a generalization of both the Lie algebra and the integrable distribution. The idea of using Lie algebroids in mechanics is due to Weinstein \cite{we-1996}. His formulation allows a geometric unified description of dynamical systems with a variety of different kinds of phase spaces: Lie groups, Lie algebras, Cartesian products of manifolds, quotients manifolds,.... Two good surveys of this topic are \cite{CLMMM-2006, LMM-2005}. In \cite{LMSV-2009} we describe  the $k$-symplectic formalism on Lie algebroids.

\item The Skinner-Rusk approach \cite{skinner2} can be considered in the $k$-symplectic formalism. This topic was studied in \cite{RRS} in the $k$-symplectic approach and in \cite{RRSV-2012} in the $k$-cosymplectic approach.

\item Another interesting topic is the study of Lagrangian submanifolds in the $k$-symplectic setting \cite{LV-2012}. In this paper, we extend the well-know normal form theorem for Lagrangian submanifolds proved by Weinstein in symplectic geometry to the setting of $k$-symplectic manifolds.\rqed
\end{itemize}
}
\end{remark}
\section{Example: electrostatic equations}\label{section: electrostatic}

\index{Equations!electrostatic}
\index{Electrostatic equations}
Consider the  $3$-symplectic Hamiltonian equations
                    \begin{equation}\label{electroestatic_k-symp}
                        \iota_{X_1}\omega^1 + \iota_{X_2}\omega^2+\iota_{X_3}\omega^3=dH\,,
                    \end{equation}
                    where   $H$ is the Hamiltonian function given by
                    \begin{equation}\label{Helectros}
                    \begin{array}{ccccl}
                      H &:& (T^1_3)^*\r &  \longrightarrow & \r \\ \noalign{\medskip}
                                            & &   (q,p^1,p^2,p^3)  & \to &  4{\pi} rq+\ds\frac{1}{2}\ds\sum_{\alpha=1}^3(p^\alpha)^2\,.
                   \end{array} \end{equation}

                   Let us observe that in this example the $k$-symplectic manifold is the cotangent bundle of $3$-covelocities of the real line $(T^1_3)^*\r$ with its canonical $3$-symplec\-tic structure.

                                If $(X_1,X_2,X_3)$ is a solution of (\ref{electroestatic_k-symp}) then,  since \[
                        \ds\frac{\partial H}{\partial q} =4 \pi r\; , \quad   \ds\frac{\partial H}{\partial p^\alpha}= p^\alpha\;,
                    \]
                    and from (\ref{ecHDWloc}) we deduce that each $X_\alpha$, with $1\leq \alpha \leq 3$,  has the local expression
                    \[
                        X_\alpha= p^\alpha \derpar{}{q} + (X_\alpha)^\beta\derpar{}{p^\beta},
                    \]
                where  the functions components $(X_\alpha)^\beta$ with $1\leq \alpha,\beta \leq 3 $ satisfy the identity
                    \[
                        4\pi r=- \Big( (X_1)^1+ (X_2)^2+ (X_3)^3\Big) \, .
                    \]

                Let us suppose that   $(X_1,X_2,X_3)$ is integrable, that is, in this particular case,  the functions $(X_\alpha)^\beta$ with $1\leq \alpha,\beta\leq k$ satisfies
                \[
                    (X_\alpha)^\beta= (X_\beta)^\alpha
                \]
                and
                \[
                    X_1((X_2)^\beta)=X_2((X_1)^\beta), \;  X_1((X_3)^\beta)=X_3((X_1)^\beta), \; X_2((X_3)^\beta)=X_3((X_2)^\beta)\,.
                \]

                Under the assumption of integrability of $(X_1,X_2,X_3)$, if  \[
                    \begin{array}{ccccl}
                     \varphi &:& \r^3 &  \longrightarrow & (T^1_3)^*\r \\ \noalign{\medskip}
                                            & &   x  & \to &  \varphi(x)=(\psi(x),\psi^1(x),\psi^2(x),\psi^3(x))
                   \end{array} \]
                 is  an integral section of a $3$-vector field $(X_1,X_2,X_3)$  solution of (\ref{electroestatic_k-symp}), then we deduce that \[(\psi(x),\psi^1(x),\psi^2(x),\psi^3(x))\]  is a solution of
                 \begin{equation}\label{electrosEC}
                        \begin{array}{rcl}
                        \psi^\alpha &=&
                            \frac{\displaystyle\partial\psi}{\displaystyle\partial
                            x^\alpha},
                            \\ \noalign{\medskip}
                            -\Big(\frac{\displaystyle\partial\psi^1}{\displaystyle\partial
                            x^1}+
                            \frac{\displaystyle\partial\psi^2}{\displaystyle\partial x^2}+
                            \frac{\displaystyle\partial\psi^3}{\displaystyle\partial x^3}\Big)& =
                            &4{\pi} r\,.
                        \end{array}
                   \end{equation}
                 which is a particular case of the electrostatic equations (for a more detail description of these equations, see section \ref{example k-symp: hamiltonian electrostatic}).

                 \index{Electrostatic equations}

\newpage
\mbox{}
\thispagestyle{empty} % para que no se numere esta p\'{a}gina

\chapter{Hamiltonian Classical Field Theory}\label{chapter: K-SympHamaCFT}
\index{Hamiltonian Classical Field Theory}
\index{Hamilton-De Donder-Weyl equation}
In this chapter we shall study Hamiltonian Classical Field Theories, that is, we shall discuss the Hamilton-De Donder-Weyl equations (these equations will be called also the HDW equations for short) which have the following local expression
\begin{equation}\label{HDW_eq}
                \derpar{H}{q^i}\Big\vert_{\varphi(x)} = -\displaystyle\sum_{\alpha=1}^k\derpar{\psi^\alpha_i}{x^\alpha}\Big\vert_t
                \;, \quad \derpar{H}{p^\alpha_i}\Big\vert_{\varphi(x)}=\derpar{\psi^i}{x^\alpha}\Big\vert_{x}\,,
            \end{equation}
where $H\colon (T^1_k)^*Q\to \mathbb{R}$ is a Hamiltonian function.

A solution of these equations is a map
        \[
                    \begin{array}{ccccl}
                     \varphi &:& \r^k &  \longrightarrow & \tkqh\\ \noalign{\medskip}
                                            & &   x  & \to &  \varphi(x)=(\psi^i(x),\psi_i^\alpha(x))
                   \end{array} \]
        where $1 \leq i\leq n,\, 1\leq \alpha\leq k$.

In a classical view these equations can be obtained from a multiple integral variational problem. In this chapter we shall describe this variational approach and then we shall give a new geometric way of obtaining the HDW equations using the $k$-symplectic formalism described in chapter \ref{chapter: k-symplectic formalism} when the $k$-symplectic manifolds is the canonical model of these structures: the manifold $(\tkqh,\omega^1,\ldots, \omega^k, V)$ described in section \ref{section k-cotangent}.

 \section{Variational approach}\label{section HDW eq variational}

            In Hamiltonian Mechanics, the Hamilton equations are obtained from a variational principle. This can be generalized to Classical Field Theory,  where the problem consist to find the extremal of a variational problem associated to multiple integrals of Hamiltonian densities.

            In this subsection we shall see that the Hamilton-De Donder-Weyl  equations (\ref{HDW_eq})
                                 are equivalent  a one variational principle on the space of smooth maps  with compact support; we denote this set by $\mathcal{C}^\infty_C(\rk,\tkqh)$. %A more detail description of this variational principle is given in this section.

             To describe this variational principle we need the notion of prolongation  of diffeomorphisms and vector fields from $Q$ to the cotangent bundle of $k^1$-covelocities, which we shall introduce in the sequel.

             %%%%%%%%%%%%%%%%%%%%%%%%%%%%%%%%%%%%%%%%%%%%%%%%%%%%%%%%%%%%%%%%%%%%%%%%%%%%%%%%%%%%%%%%%%%%%%%%%%%%%

        \subsection{Prolongation of diffeomorphism and vector fields}\label{section tkqh: prolongation}

                Given a diffeomorphism between two manifolds $M$ and $N$ we can consider an induced map between $(T^1_k)^*N$ and $(T^1_k)^*M$. This map allows us to define the prolongation of vector fields from $Q$ to the cotangent bundle of $k^1$-covelocities.

                \index{Cotangent bundle of $k^1$-covelocities!Prolongation of maps}
                    \begin{definition}
                        Let $f\colon M\to N$ be a diffeomorphism. The \emph{natural or canonical prolongation} of $f$ to the corresponding bundles of $k^1$-covelocities is the map
                            \[
                                (T^1_k)^*f:(T^1_k)^*N \to (T^1_k)^*M
                            \]
                        defined as follows:
                            \[
                                \begin{array}{rl}
                                    (T^1_k)^*f({\nu_1}_{f(x)},\ldots,{\nu_k}_{f(x)}) &=
                                    (f^*({\nu_1}_{f(x)}),\ldots,f^*({\nu_k}_{f(x)}))\\\noalign{\medskip}
                                    &=({\nu_1}_{f(x)}\circ f_*(x),\ldots, {\nu_k}_{f(x)}\circ f_*(x))
                                 \end{array}
                            \]
                        where $({\nu_1}_{f(x)},\ldots,{\nu_k}_{f(x)})\in (T^1_k)^*N$ and $ m\in M\, . $
                    \end{definition}
\index{Cotangent bundle of $k^1$-covelocities!Prolongation of vector fields}
                The canonical prolongation of diffeomorphism allows us to introduce the canonical or complete lift of vector fields from $Q$ to $\tkqh$.
                    \begin{definition}
                        Let $Z$ be a vector field on $Q$, with $1$-parameter group of diffeomorphism $\{h_s\}$. The \emph{canonical or complete lift} of $Z$ to $(T^1_k)^*Q$ is the vector field
                        $Z^{C*}$ on $\tkqh$ whose local $1$-parameter group of diffeomorphism is $\{(T^1_k)^*(h_s)\}$.
                    \end{definition}

                Let $Z$ be a vector field on $Q$ with local expression $Z=Z^i\ds\frac{\partial}{\partial q^i}$. In the canonical coordinate system (\ref{tkqh: natural coord}) on $\tkqh$, the local expression of $Z^{C*}$ is
                    \begin{equation}\label{lcksh}
                        Z^{C*}=Z^i\frac{\partial}{\partial q^i} \, - \, p_j^\alpha \ds \frac{\partial Z^j} {\partial q^k} \frac{\partial}{\partial p_k^\alpha}\; .
                    \end{equation}

                The canonical prolongation of diffeomorphisms and vector fields from $Q$ to $(T^1_k)^*Q$ have the following properties.
                    \begin{lemma}\
                    \label{fiprop1h}
                        \begin{enumerate}
                            \item Let $\varphi\colon Q\to Q$ be a diffeomorphism and $\Phi=(T^1_k)^*\varphi$ the canonical prolongation of $\varphi$ to $(T^1_k)^*Q$. Then:
                                    \begin{equation}\label{tkqh: invariance}
                                        (i)\; \Phi^*\theta^\alpha=\theta^\alpha \quad \makebox{ and } \quad(ii)\,
                                        \Phi^*\omega^\alpha=\omega^\alpha,
                                    \end{equation}
                                     where $\ak$.
                            \item Let $Z\in\vf (Q)$ be  and  $Z^{C*}$ the complete lift of $Z$ to $(T^1_k)^*Q$. Then
                                    \begin{equation}\label{tkqh: infinitesimal invarianze}
                                        (i)\; \Lie{Z^{C*}}\theta^\alpha=0 \quad \makebox{ and } \quad (ii) \,\Lie{Z^{C*}}\omega^\alpha=0,
                                    \end{equation}
                                    with $\ak$.
                        \end{enumerate}
                    \end{lemma}
                    \proof
                        \begin{enumerate}
                            \item (i) is a consequence of the commutativity of the following diagram
                                    \[
                                        \xymatrix@=13mm{
                                            \tkqh\ar[r]^-{(T^1_k)^*\varphi}\ar[d]_-{\pi^{k,\alpha}}&\tkqh\ar[d]^-{\pi^{k,\alpha}}
                                            \\ T^*Q \ar[r]^-{\varphi^*}&T^*Q}
                                    \]
                                for each $\ak$, that is,
                                    \[
                                        \pi^{k,\alpha}\circ (T^1_k)^*\varphi=\varphi^*\circ \pi^{k,\alpha}.
                                    \]

                                In fact, using the above identity one has
                                    \[
                                        \begin{array}{lcl} \left[(T^1_k)^*\varphi\right]^*\theta^\alpha & = &
                                        \left[(T^1_k)^*\varphi\right]^*((\pi^{k,\alpha})^*\theta)= \big(\pi^{k,\alpha}\circ (T^1_k)^*\varphi\big)^*\theta
                                        \\\noalign{\medskip} &=&
                                        (\varphi^*\circ\pi^{k,\alpha})^*\theta=(\pi^{k,\alpha})^*((\varphi^*)^*\theta)=(\pi^{k,\alpha})^*\theta=\theta^\alpha\;,
                                        \end{array}
                                    \]
                                where we have used the identity $(\varphi^*)^*\theta=\theta$ (see \cite{AM-1978}, pag. 180).

                                The item (ii) is a direct consequence of (i) and the definition of the closed $2$-forms $\omega^1,\ldots, \omega^k$.
                        \item Since the infinitesimal generator of $Z^{C*}$ is the canonical prolongation of the infinitesimal generator of $Z$, then from item $(1)$ of this lemma one obtains that  (\ref{tkqh: infinitesimal invariance}) holds.
                    \end{enumerate} \qed

        \subsection{Variational principle}

        Now we are in conditions to describe the multiple integral problem from which one obtains the Hamilton-De Donder-Weyl equations.

            We denote by $d^kx$ the volume form on $\rk$ given by $dx^1\wedge\ldots\wedge dx^k$ and  $d^{k-1}x_\alpha$ is the $(k-1)$-form defined by
                \[
                    d^{k-1}x_\alpha=\iota_{\nicefrac{\partial}{\partial x^\alpha}}d^kx\,,
                \]
            for each $\ak$.

             Before  describing the variational problem in this setting we recall the following result:
             \begin{lemma}\label{variations}
                Let $G$ denote a fixed simply-connected domain in the $k$-dimen\-sional  space, bounded by a hypersurface $\partial G$. If $\Phi(x)$ is a continuous function in $G$ and if
                    \[
                        \int_{G}\Phi(x)\eta(x)d^kx=0
                    \]
                for all function $\eta(x)$ of class $C^1$ which vanish on the boundary $\partial G$ of $G$, then
                    \[
                        \Phi(x)=0
                    \]
                in $G$.
             \end{lemma}
             A proof of this lemma can be found in \cite{Rund}.

                    \begin{definition}\label{Alksim}
                        Denote by $\mathcal{C}^\infty_C(\rk,\tkqh)$ the set of maps
                            \[
                                \varphi:U_0\subseteq\rk\to \tkqh,
                            \]
                        with compact support defined on an open set $U_0$. Let $H:\tkqh\to \r$ be a  Hamiltonian function, then we define the integral action associated to $H$ by
                            \[
                                \begin{array}{lccl}
                                    {\mathcal H}:&\mathcal{C}^\infty_C(\rk,\tkqh) & \to &\r\\\noalign{\medskip}  & \varphi &\mapsto &\ds\int_{\rk}\Big( \sum_{\alpha=1}^k(\varphi^*\theta^\alpha)\wedge
                                    d^{k-1} x_\alpha- (\varphi^*H) d^k x \Big)\,.
                                \end{array}
                            \]
                    \end{definition}

                    \begin{definition}\label{extremales}
                        A map $\varphi \in \mathcal{C}^\infty_C(\rk,\tkqh)$  is an extremal of $\mathcal{H}$ if
                            \[
                                \ds\frac{d}{ds}\Big\vert_{s=0}\mathcal{H}(\tau_s\circ\varphi)=0
                            \]
                        for each flow $\tau_s$ on $\tkqh$ such that $\tau_s({\nu_1}_{q},\ldots, {\nu_k}_{q})=({\nu_1}_{q},\ldots, {\nu_k}_{q})$ for all $({\nu_1}_{q},\ldots, {\nu_k}_{q})$ on the boundary of  $\varphi(U_0)\subset\tkqh$, that is, we consider the variations of $\varphi$ given by the composition by elements of one-parametric group of diffeomorphism which leaves invariant the boundary of the image of $\varphi$.
                    \end{definition}

                    Let us observe that the flows $\tau_s:\tkqh\to\tkqh$ considered in the above definition are generated by vector fields on $\tkqh$ which are zero on the boundary of $\varphi(U_0)$.

                    The variational problem here considered consists in finding the extremals of the integral action $\mathcal{H}$.

                    The following proposition gives us a characterization of the extremals of the integral action $\mathcal{H}$ associated with the Hamiltonian $H$.
                        \begin{prop}
                        Let $H:\tkqh\to \r$ be a Hamiltonian function and                  $\varphi\in\mathcal{C}^\infty_C(\rk,\tkqh)$. The following statements are equivalents:
                        \begin{enumerate}
                            \item $\varphi:U_0\subset\rk\to \tkqh$ is an extremal of the variational problem associated to $H$.
                            \item For each vector field $Z$ on $Q$, such that its complete lift $\,Z^{C*}$ to $\tkqh$ vanishes on the boundary of $\varphi(U_0)$, the equality
                                \[
                                    \ds\int_{\rk} \left([\varphi^*(\mathcal{L}_{Z^{C^*}}\theta^\alpha)]\wedge
                                    d^{k-1}x_\alpha - [\varphi^*(\mathcal{L}_{Z^{C^*}} H) ]d^kx\right)=0\,,
                                \]
                            holds.
                            \item $\varphi$ is solution of the Hamilton-De Donder-Weyl equations, that is, if $\varphi$ is locally given by $\varphi(x)=(\psi^i(x),\psi^\alpha_i(x))$, then the functions $\psi^i,\,\psi^\alpha_i$ satisfy the  system of partial differential equations (\ref{HDW_eq}).
                    \end{enumerate}
                \end{prop}
                \proof
                    First  we shall prove the equivalence between items $\mathbf{(1)}$ and $\mathbf{(2)}$ $(1\Leftrightarrow 2)$.

                        Let $Z\in \vf(Q)$  be a vector field on  $Q$ satisfying the conditions in $\mathbf{(1)}$, and with one-parameter group of diffeomorphism $\{\tau_s\}$. Then, from the definition of the complete lift we know that $Z^{C*}$  generates the one-parameter group  $\{(T^1_k)^*\tau_s\}$.

                        Thus,
                            \[
                                \begin{array}{lcl}
                                    &&\ds\frac{d}{ds}\Big\vert_{s=0}\mathcal{H}((T^1_k)^*\tau_s\circ\varphi)
                                    \\\noalign{\medskip} =&&\ds\frac{d}{ds}\Big\vert_{s=0}\int_{\rk}
                                    \Big( \sum_{\alpha=1}^k([(T^1_k)^*\tau_s\circ\varphi]^*\theta^\alpha)\wedge
                                    d^{k-1} x_\alpha-([(T^1_k)^*\tau_s\circ\varphi]^*H)
                                    d^k x \Big)\\\noalign{\medskip}
%\end{array}\]
%
%\[
%                                \begin{array}{lcl}
                                     =&&
                                    \ds\lim_{s\to 0}\ds\frac{1}{s}\Big(\int_{\rk} \Big(
                                    \sum_{\alpha=1}^k([(T^1_k)^*\tau_s\circ\varphi]^*\theta^\alpha)\wedge
                                    d^{k-1} x_\alpha-([(T^1_k)^*\tau_s\circ\varphi]^*H)
                                    d^k x \Big) \\\noalign{\medskip} &-&\ds \int_{\rk} \Big(
                                    \sum_{\alpha=1}^k([(T^1_k)^*\tau_0\circ\varphi]^*\theta^\alpha)\wedge
                                    d^{k-1} x_\alpha-([(T^1_k)^*\tau_0\circ\varphi]^*H)
                                    d^k x \Big)\Big)
                                    \\\noalign{\medskip}
%\end{array}\]
%
%\[
%                                \begin{array}{lcl}
                                     =&&\ds\lim_{s\to 0}\ds\frac{1}{s}\left(\int_{\rk}
                                    \sum_{\alpha=1}^k([(T^1_k)^*\tau_s\circ\varphi]^*\theta^\alpha)\wedge
                                    d^{k-1} x_\alpha -\int_{\rk}  \sum_{\alpha=1}^k(\varphi^*\theta^\alpha)\wedge
                                    d^{k-1} x_\alpha\right)
                                    \\\noalign{\medskip}

                                     &-&\ds\lim_{s\to 0}\ds\frac{1}{s}\left(\int_{\rk}([(T^1_k)^*\tau_s\circ\varphi]^*H)
                                    d^k x -\int_{\rk}(\varphi^*H) d^k x \right)
                                    \\\noalign{\medskip}

                                    =&&\ds\lim_{s\to 0}\ds\frac{1}{s}\left(\int_{\rk}\sum_{\alpha=1}^k[\varphi^*
                                    \Big( ((T^1_k)^*\tau_s)^*\theta^\alpha-\theta^\alpha\Big)]\wedge
                                    d^{k-1} x_\alpha\right)
                                    \\\noalign{\medskip}

                                     &-& \ds\lim_{s\to 0}\ds\frac{1}{s}\left(\int_{\rk}[
                                     \varphi^*\Big(((T^1_k)^*\tau_s)^*H
                                    -H \Big)]d^k x \right)
                                    \\\noalign{\medskip}
                                    =&& \ds\int_{\rk}
                                    \left([\varphi^*(\mathcal{L}_{Z^{C^*}}\theta^\alpha)]\wedge
                                    d^{k-1} x_\alpha - [\varphi^*(\mathcal{L}_{Z^{C^*}} H)]
                                    d^k x \right)\,,
                                \end{array}
                            \]
                            where in the last identity we are using the definition of Lie derivative with respect to $Z^{C^*}$.

                            Therefore, $\varphi$ is and extremal of $\mathcal{H}$ if and only if
                            \[
                                \ds\int_{\rk}
                                    \left([\varphi^*(\mathcal{L}_{Z^{C^*}}\theta^\alpha)]\wedge
                                    d^{k-1} x_\alpha - [\varphi^*(\mathcal{L}_{Z^{C^*}} H)]
                                    d^k x \right)=0\,.
                            \]

                    We now prove the equivalence between  $\mathbf{(2)}$ and $\mathbf{(3)}$  $(2\Leftrightarrow 3)$.

                        Taking into account that
                            \[
                                \mathcal{L}_{Z^{C*}}\theta^\alpha=d\iota_{Z^{C*}}\theta^\alpha +\iota_{Z^{C*}}d\theta^\alpha
                            \]
                        one obtains
                            \[
                                \ds\int_{\rk}[\varphi^*(\mathcal{L}_{Z^{C^*}}\theta^\alpha)]\wedge d^{k-1} x_\alpha=\int_{\rk}[\varphi^*(d\iota_{Z^{C*}}\theta^\alpha)]\wedge d^{k-1} x_\alpha+ \int_{\rk}[\varphi^*(\iota_{Z^{C*}}d\theta^\alpha)]\wedge d^{k-1} x_\alpha\,.
                            \]

                        Since
                            \[
                                [\varphi^*(d\iota_{Z^{C*}}\theta^\alpha)]\wedge d^{k-1}x_\alpha =d\Big(\varphi^*(\iota_{Z^{C*}}\theta^\alpha) \wedge d^{k-1} x_\alpha\Big)
                            \]
                        then $[\varphi^*(d\iota_{Z^{C*}}\theta^\alpha)]\wedge d^{k-1}x_\alpha$ is a closed $1$-form on $\rk$. Therefore, applying Stoke's theorem one obtains:
                            \[
                                \int_{\rk}[\varphi^*(d\iota_{Z^{C*}}\theta^\alpha)]\wedge d^{k-1} x_\alpha=\int_{\rk}d\Big(\varphi^*(\iota_{Z^{C*}}\theta^\alpha) \wedge d^{k-1} x_\alpha\Big)=0\,.
                            \]

                        Then,
                            \[
                                \ds\int_{\rk} \left([\varphi^*(\mathcal{L}_{Z^{C^*}}\theta^\alpha)]\wedge d^{k-1} x_\alpha - [\varphi^*(\mathcal{L}_{Z^{C^*}} H)] d^k x \right)=0
                            \]
                        if and only if,
                            \[
                                \ds\int_{\rk} \left([\varphi^*\Big(\iota_{Z^{C*}}d\theta^\alpha \Big)]\wedge d^{k-1} x_\alpha - [\varphi^*(\mathcal{L}_{Z^{C^*}} H)] d^k x \right)=0\,.
                            \]

                        Consider now the canonical coordinate system such that $Z=Z^i\derpar{}{q^i}$;   taking into account the  local expression (\ref{lcksh}) for the complete lift $Z^{C^*}$ and that $\varphi(x)=(\psi^i(x),\psi^\alpha_i(x))$, we have
                            \[
                                \begin{array}{cl}
                                    & \varphi^*\Big(\iota_{Z^{C*}}d\theta^\alpha \Big)\wedge
                                    d^{k-1} x_\alpha - \varphi^*(\mathcal{L}_{Z^{C^*}} H)
                                    d^k x \\\noalign{\medskip}

                                    =&\left[-(Z^i(x))\left(\ds\sum_{\alpha=1}^k\ds\frac{\partial
                                    \psi^\alpha_i}{\partial x^\alpha}\Big\vert_{x}+\ds\frac{\partial H}{\partial
                                    q^i}\Big\vert_{\varphi(x)}\right)-\ds\sum_{\alpha=1}^k \psi^\alpha_j(x)
                                    \ds\frac{\partial Z^j}{\partial
                                    q^i}\Big\vert_{x}\left(\ds\frac{\partial \psi^i}{\partial
                                    x^\alpha}\Big\vert_{x}-\ds\frac{\partial H}{\partial
                                    p^{\alpha}_i}\Big\vert_{\varphi(x)}\right)\right]d^kx\,
                                \end{array}
                            \]
                        for each $Z\in\vf(Q)$ (under the conditions established in this theorem), where we are using the notation $Z^i(x):=(Z^i\circ\pi^k\circ\varphi)(x)$. From the last expression we deduce that $\varphi$ is an extremal of  $\mathcal{H}$ if and only if
                            \[
                                    \ds\int_{\rk}Z^i(x)\left(\ds\sum_{\alpha=1}^k\ds\frac{\partial
                                    \psi^\alpha_i}{\partial x^\alpha}\Big\vert_{x}+\ds\frac{\partial H}{\partial
                                    q^i}\Big\vert_{\varphi(x)}\right)d^kx + \ds \int_{\rk}\ds\sum_{\alpha=1}^k
                                    \psi^\alpha_j(x) \ds\frac{\partial Z^j}{\partial
                                    q^i}\Big\vert_{x}\left(\ds\frac{\partial \psi^i}{\partial
                                    x^\alpha}\Big\vert_{x}-\ds\frac{\partial H}{\partial
                                    p^{\alpha}_i}\Big\vert_{\varphi(x)}\right)d^kx=0
                            \]
                        for all $Z^i$. Therefore,
                            \begin{equation}\label{pvh}
                                \begin{array}{lcl}
                                    \ds\int_{\rk}(Z^i(x))\left(\ds\sum_{\alpha=1}^k\ds\frac{\partial
                                    \psi^\alpha_i}{\partial x^\alpha}\Big\vert_{x}+\ds\frac{\partial H}{\partial
                                    q^i}\Big\vert_{\varphi(x)}\right)d^kx & =&0\\\noalign{\medskip}

                                    \ds\int_{\rk}\ds\sum_{\alpha=1}^k \psi^\alpha_j(x) \ds\frac{\partial
                                    Z^j}{\partial q^i}\Big\vert_{x}\left(\ds\frac{\partial
                                    \psi^i}{\partial x^\alpha}\Big\vert_{x}-\ds\frac{\partial H}{\partial
                                    p^{\alpha}_i}\Big\vert_{\varphi(x)}\right)d^kx &= &0\,
                                \end{array}
                            \end{equation}
                        for all $Z\in \vf(Q)$ satisfying  the statements of this theorem, and, thus, for any values $Z^i( q)$ and $\derpar{Z^j}{q^i}\Big\vert_{q}$.

                        Applying  lemma \ref{variations}, from (\ref{pvh}) one obtains that,
                            \[
                                \ds\sum_{\alpha=1}^k\ds\frac{\partial \psi^\alpha_i}{\partial
                                x^\alpha}\Big\vert_{x}+\ds\frac{\partial H}{\partial
                                q^i}\Big\vert_{\varphi(x)} =0\quad,\quad \ds\sum_{\alpha=1}^k
                                \psi^\alpha_j(x) \left(\ds\frac{\partial \psi^i}{\partial
                                x^\alpha}\Big\vert_{x}-\ds\frac{\partial H}{\partial
                                p^{\alpha}_i}\Big\vert_{\varphi(x)}\right) = 0\,.
                            \]

                        The first group of equations gives us the first group of the Hamilton-De Donder-Weyl equations (\ref{HDW_eq}).

                        Now, consider the second set of above equations, it follows that
                            \[
                                \ds\frac{\partial \psi^i}{\partial x^\alpha}\Big\vert_{x}-\ds\frac{\partial H}{\partial p^{\alpha}_i}\Big\vert_{\varphi(x)}=0\,,
                            \]
                        which is the second set of the Hamilton-De Donder-Weyl equations (\ref{HDW_eq}).

                        The converse is obtained starting from the Hamilton-De Donder-Weyl equations and reversing the arguments in the above proof.
                        \qed

\section{Hamilton-De Donder-Weyl equations}

The above variational principle allows us to obtain the HDW equations but there exist  other methods to obtain these equations: one of them consist of using the $k$-symplectic Hamiltonian equation  when we consider the $k$-symplectic manifold $M=\tkqh$.

In this case, we take a Hamiltonian function $H\in\mathcal{C}^\infty(\tkqh)$. Thus, from  Theorem \ref{fhks}, one obtains that given  an integrable $k$-vector field $\mathbf{X}=(X_1,\ldots, X_k)\in \vf^k_H(\tkqh)$ and  an integral section $\varphi\colon U\subset\rk\to \tkqh$ of $\mathbf{X}$,  $\varphi$ is a solution of the following systems of partial differential equations
                    \[
                            \frac{\partial H}{\partial q^i}\Big\vert_{\varphi(x)} \, = \, -
                            \displaystyle\sum_{\beta=1}^k \frac{\partial \psi^\beta_i} {\partial
                            x^\beta}\Big\vert_{x} \,, \quad \frac{\partial H}{\partial
                            p^{\alpha}_i}\Big\vert_{\varphi(x)} \, = \, \frac{\partial \psi^i}{\partial
                            x^\alpha}\Big\vert_{x}\,,
                        \]
                        that is, $\varphi$ is a solution of the HDW equations (\ref{HDW_eq}).

Therefore, given an integrable $k$-vector field $X\in\vf^k_H(\tkqh)$, its integral sections are solutions of the HDW equations. Now it is natural to do the following question: \textit{Given a solution $\varphi\colon U\subset\rk\to \tkqh$ of the HDW equations, is there a $k$-vector field $X\in\vf^k_H(\tkqh)$ such that $\varphi$ is an integral section of $\mathbf{X}$?}

Here we give an answer to this question.

                \begin{prop}\label{int}
If  a map $\varphi:\rk\to (T^1_k)^*Q$
is a solution of the HDW equation (\ref{HDW_eq})
and $\varphi$ is an integral section of an integrable $k$-vector field ${\bf X}\in\vf^k((T^1_k)^*Q)$, then
 ${\bf X}=(X_1,\dots,X_k)$ is a solution of the equation
(\ref{ecHksym}) at the points of the image of $\varphi$.
\end{prop}
\begin{proof}
We must prove that
\begin{equation}\label{01}
    \ds\frac{\partial H}{\partial p^\alpha_i}(\varphi(x))=
  (X_\alpha)^i(\varphi(x)), \quad \ds\frac{\partial H}{\partial q^i}(\varphi(x))=
   -\ds\sum_{\alpha=1}^k(X_\alpha)^\alpha_i(\varphi(x)) \ .
\end{equation}
Now as $\varphi(x)=(\psi^i(x),\psi^\alpha_i(x))$ is an integral section of ${\bf X}$ we have that
(\ref{tkqh integral section equivalence cond}) holds; but,
as $\varphi$ is also a solution of the Hamilton-de Donder-Weyl equation (\ref{HDW_eq}), then we deduce (\ref{01}).\qed
 \end{proof}

We can not claim that  ${\bf X}\in\vf^k_H((T^1_k)^*Q)$
because we can not assure that ${\bf X}$ is a solution of the equations
 (\ref{ecHksym}) on the whole in $(T^1_k)^*Q$.
                \begin{remark}
                {\rm
                    It is also important to point out that the equations (\ref{HDW_eq}) and (\ref{ecHksym}) are not equivalent in the sense that not every solution of (\ref{HDW_eq}) is an integral section of some integrable $k$-vector field belonging to $\vf^k_H(\tkqh)$.}\rqed
                \end{remark}

\index{Admissible solution}
\begin{definition}
\label{hyp} A map $\varphi\colon\rk\to (T^1_k)^*Q$, solution of the
equation (\ref{HDW_eq}),
 is said to be an \emph{admissible solution} to the HDW-equation
for a $k$-symplectic Hamiltonian system $((T^1_k)^*Q,\omega^\alpha,H)$ if it
 is an integral section of some integrable $k$-vector
field ${\bf X}\in{\mathfrak{X}}^k((T^1_k)^*Q)$.
\end{definition}

If we consider only admissible solutions to the HDW-equations of $k$-symplectic Hamiltonian systems, we say that $((T^1_k)^*Q,\omega^\alpha,H)$ is an \emph{admissible $k$-symplectic Hamiltonian system}.

In this way, for admissible $k$-symplectic Hamiltonian systems,
 the geometric field equation (\ref{ecHksym}) for integrable $k$-vector fields
is equivalent to the HDW-equation (\ref{HDW_eq})
(as it is established in Theorem \ref{fhks} and Proposition \ref{int}).

%\newpage
%\mbox{}
%\thispagestyle{empty} % para que no se numere esta p\'{a}gina

\chapter{Hamilton-Jacobi theory in $k$-symplectic Field Theories}\label{ksymp-HJ}

\index{Hamilton-Jacobi theory}

The usefulness of Hamilton-Jacobi theory in Classical Mechanics
is well-known, giving an alternative  procedure to study and, in some  cases, to solve the evolution equations \cite{AM-1978}.
The use of symplectic geometry in the study of Classical Mechanics
has permitted to connect the Hamilton-Jacobi theory with the theory
of Lagrangian submanifolds and generating functions \cite{BLM-2012}.

At the beginning of the 1900s an analog of Hamilton-Jacobi equation for field theory has been developed
\cite{Rund}, but it has not been proved to be so powerful as the theory which is available for Mechanics \cite{BPP-2008,Bruno-2007,PR-2002-b,PR-2002,rosen,vita}.

\index{Hamilton-Jacobi equation}
\index{Equation!Hamilton-Jacobi}
Our goal in this chapter is to describe this equation
in a geometrical setting given by the $k$-symplectic geometry, that is, to extend the Hamilton-Jacobi theory to Field Theories just in
the context of $k$-symplectic manifolds (we remit to \cite{LMM-1996b,LMM-2009}
for a description in the multisymplectic setting). The dynamics for a given Hamiltonian function $H$ is interpreted
as a family of vector fields (a $k$-vector field) on the phase space $(T^1_k)^*Q$.
The  \emph{Hamilton-Jacobi equation} is of the form
$$
d(H \circ \gamma) = 0,
$$
where $\gamma = (\gamma_1, \dots, \gamma_k)$ is a family of closed
$1$-forms on $Q$. Therefore, we recover the classical form
$$
H\Big(q^i,\ds\frac{\partial W^1}{\partial q^i}, \ldots
,\ds\frac{\partial W^k}{\partial q^i}\Big) = constant\;.
$$
where $\gamma_i = dW_i$.
It should be noticed that our method is inspired in a  recent result by Cari\~{n}ena {\it et al} \cite{CGMMR} for Classical Mechanics (this method
has also used to develop a Hamilton-Jacobi theory for nonholonomic mechanical systems \cite{lmm1}; see also \cite{pepin2,LMM-2010}).

\section{The Hamilton-Jacobi equation}\label{HJproblem}

\index{Hamilton-Jacobi theory!Mechanics}
The standard formulation of the \textit{Hamilton-Jacobi problem for Hamiltonian Mechanics} consist of finding a function $S(t,q^i)$ (called the \emph{principal function}) such that
\begin{equation}\label{H-JeqHmech}
\ds\frac{\partial S}{\partial t} + H\Big(q^i,\ds\frac{\partial
S}{\partial q^j}\Big)=0\,.
\end{equation}

If we put $S(t,q^i)= W(q^i)-t \cdot  constant$, then $W\colon Q\to \R$ (called the \emph{characteristic function}) satisfies
\begin{equation}\label{H-JeqHmech2}
  H\Big(q^i,\ds\frac{\partial W}{\partial q^j}\Big)=constant\,.
\end{equation}

\index{Hamilton-Jacobi equation}
Equations (\ref{H-JeqHmech}) and (\ref{H-JeqHmech2}) are indistinctly referred as the \emph{ Hamilton-Jacobi equation} in Hamiltonian Mechanics.

\index{Hamilton-Jacobi theory!Field Theory}
In the framework of the $k$-symplectic description of Classical Field Theory, a Hamiltonian is a
function $H\in\mathcal{C}^\infty((T^1_k)^*Q)$. In this context, the
Hamilton-Jacobi problem consists of finding $k$ functions
$W^1,\ldots, W^k\colon Q\to \R$ such that
\begin{equation}\label{HJ-ksym}
H\Big(q^i,\ds\frac{\partial W^1}{\partial q^i}, \ldots
,\ds\frac{\partial W^k}{\partial q^i}\Big) = constant\;.
\end{equation}

In this subsection we give a geometric version of the Hamilton-Jacobi equation (\ref{HJ-ksym}).

Let $\gamma : Q \longrightarrow (T^1_k)^*Q$ be a closed section of $\pi^k :
(T^1_k)^*Q \longrightarrow Q$. Therefore, $\gamma = (\gamma^1,
\dots, \gamma^k)$ where each $\gamma^\alpha$ is an ordinary closed 1-form on
$Q$. Thus we have that
every point has an open neighborhood $U\subset Q$ where there exists
$k$ functions $W^\alpha\in\mathcal{C}^\infty(U)$ such that $\gamma^\alpha=dW^\alpha$.

Now, let $Z$ be a $k$-vector field on $(T^1_k)^*Q$. Using $\gamma$ we
can construct a $k$-vector field $Z^\gamma$ on $Q$ such that the following
diagram is commutative

\[
 \xymatrix{ (T^1_k)^*Q
\ar@/^1pc/[dd]^{\pi^k} \ar[rrr]^{Z}&   & &T_k^1((T^1_k)^*Q)\ar[dd]^{T^1_k\pi^k}\\
  &  & &\\
 Q\ar@/^1pc/[uu]^{\gamma}\ar[rrr]^{Z^{\gamma}}&  & & T_k^1Q }
\]
that is,
$$Z^\gamma:= T^1_k\pi^k\circ Z\circ \gamma\;.$$
\index{Prolongation of maps}
\index{Tangent bundle of $k^1$-velocities!Prolongation of maps}
Let us remember that for an arbitrary differentiable map $f:M\to N$, the induced
map $T^1_kf:T^1_kM\to T^1_kN$ is defined by
\begin{equation}\label{tkq: prolongation expr}
T^1_kf({v_1}_x,\ldots, {v_k}_x)=(f_*(x)({v_1}_x),\ldots
,f_*(x)({v_k}_x)) \;,
\end{equation}
where ${v_1}_x,\ldots,
{v_k}_x\in T_xM$, $x\in M$ and $f_*(x)\colon T_xM\to T_{f(x)}N$ is the tangent map to $f$ at the point $x$
\index{Tangent bundle of $k^1$-velocities!Prolongation of maps}
\index{Prolongation!Maps}

 Notice that the $k$-vector field $Z$ defines $k$ vector fields on $(T^1_k)^*Q$, say $Z
= (Z_1, \dots, Z_k)$. In the same way, the $k$-vector field
$Z^\gamma$ determines $k$ vector fields on $Q$, say $Z^\gamma =
(Z^\gamma_1, \dots, Z^\gamma_k)$.

In local coordinates, if each $Z_\alpha$ is locally given by
$$
Z_\alpha = Z^i_\alpha \, \ds\frac{\partial}{\partial q^i} + (Z_\alpha)^\beta_i\ds\frac{\partial}{\partial p^\beta_i}\,,
$$ then $Z^\gamma_\alpha$ has the following local expression:
\begin{equation}\label{zgamma}
 Z^\gamma_\alpha = (Z_\alpha^i \circ\gamma) \, \ds\frac{\partial}{\partial q^i}\,.
\end{equation}

Let us observe that if $Z$ is integrable, the $k$-vector  field $Z^\gamma$ is integrable.
\begin{theorem}[Hamilton-Jacobi Theorem]\label{hjth}
Let $Z$ be a solution of the $k$-symplectic Hamiltonian equation (\ref{ecHksym}) and
$\gamma : Q \longrightarrow (T^1_k)^*Q$ be a   closed section  of
$\pi^k : (T^1_k)^*Q \longrightarrow Q$, that is, $\gamma =
(\gamma^1, \dots, \gamma^k)$ where   each $\gamma^\alpha$ is an ordinary
closed 1-form on $Q$. If $Z$ is integrable then the following
statements are equivalent:
\begin{enumerate}
\item If $\sigma\colon U\subset \R^k\to Q$ is an integral section of $Z^\gamma$  then $\gamma\circ\sigma$ is a solution of the Hamilton-de Donder-Weyl field equations (\ref{HDW_eq});

\item $d(H\circ \gamma)=0$.

\end{enumerate}
\end{theorem}
\begin{proof}  The closeness of the $1$-forms $\gamma^\alpha=\gamma^\alpha_idq^i$ states that
 \begin{equation}\label{closeness}
 \ds\frac{\partial \gamma^\beta_i}{\partial q^j} = \ds \frac{\partial \gamma^\beta_j}{\partial q^i}\,.
 \end{equation}

In first place we assume that the item \textbf{(1)} holds, and then we shall check that $d(H\circ\gamma)=0$. In fact, let us  suppose that  $\gamma\circ \sigma(x)=(\sigma^i(x),\gamma^\alpha_i(\sigma(x)))$ is a solution of the Hamilton-De Donder-Weyl equations for $H$, then
 \begin{equation}\label{hegamma}
 \ds\frac{\partial \sigma^i}{\partial x^\alpha}\Big\vert_{x}=\ds\frac{\partial H}{\partial p^\alpha_i}\Big\vert_{\gamma(\sigma(x))} \quad\makebox{ and} \quad \ds\sum_{\alpha=1}^k \ds\frac{\partial (\gamma^\alpha_i\circ\sigma)}{\partial x^\alpha}\Big\vert_{x} = -\ds\frac{\partial H}{\partial q^i}\Big\vert_{\gamma(\sigma(x))}\;.
\end{equation}

 Now, we shall compute the differential of the function $H\circ\gamma\colon Q\to \R$:
 \begin{equation}\label{d(hgamma)}
d(H \circ \gamma) = \left(\ds\frac{\partial H}{\partial q^i}\circ \gamma +
(\ds\frac{\partial H}{\partial p^\alpha_j}\circ \gamma) \ds\frac{\partial
\gamma^\alpha_j}{\partial q^i}\right) \, dq^i\,.
\end{equation}

Then from (\ref{closeness}), (\ref{hegamma}) and (\ref{d(hgamma)}) we obtain
$$\begin{array}{cl}
&d(H\circ\gamma)(\sigma(x))\\\noalign{\medskip}  =& \left(\ds\frac{\partial H}{\partial q^i}\Big\vert_{\gamma(\sigma(x))} +
\ds\frac{\partial H}{\partial p^\alpha_j}\Big\vert_{\gamma(\sigma(x))} \ds\frac{\partial
\gamma^\alpha_j}{\partial q^i}\Big\vert_{\sigma(x)}\right)dq^i(\sigma(x))
\\\noalign{\medskip}
 =& \left(-\ds\sum_{\alpha=1}^k \ds\frac{\partial (\gamma^\alpha_i\circ\sigma)}{\partial x^\alpha}\Big\vert_{x} +\ds\frac{\partial \sigma^j}{\partial x^\alpha}\Big\vert_{x}\ds\frac{\partial
\gamma^\alpha_j}{\partial q^i}\Big\vert_{\sigma(x)}\right)dq^i(\sigma(x))
\\\noalign{\medskip}
 =& \left(-\ds\sum_{\alpha=1}^k \ds\frac{\partial (\gamma^\alpha_i\circ\sigma)}{\partial x^\alpha}\Big\vert_{x} +\ds\frac{\partial \sigma^j}{\partial x^\alpha}\Big\vert_{x}\ds\frac{\partial
\gamma^\alpha_i}{\partial q^j}\Big\vert_{\sigma(x)}\right)dq^i(\sigma(x))
=0\;,
\end{array}$$
the last term being zero by the chain rule.
 Since $Z$ is integrable, the $k$-vector field $Z^\gamma$ is integrable, then for each point $q\in Q$ we have an integral section $\sigma\colon U_0\subset \R^k\to Q$ of $Z^\gamma$ passing trough    this point,  then
 $$ d(H\circ\gamma)=0\,.$$

 Conversely, let us suppose that $d(H\circ\gamma)=0$ and  $\sigma$ is an integral section of $Z^\gamma$. Now we shall prove that $\gamma\circ\sigma$ is a solution of the Hamilton field equations, that is (\ref{hegamma}) is satisfied.

 Since $d(H\circ\gamma)=0$, from (\ref{d(hgamma)}) we obtain
 \begin{equation}\label{d(hgamma)=0}
 0= \ds\frac{\partial H}{\partial q^i}\circ \gamma +
(\ds\frac{\partial H}{\partial p^\alpha_j}\circ \gamma) \ds\frac{\partial
\gamma^\alpha_j}{\partial q^i}\;.
\end{equation}

From (\ref{ecHDWloc}) and (\ref{zgamma}) we know that
$$
  Z^\gamma_\alpha = (\ds\frac{\partial H}{\partial p^\alpha_i}\circ\gamma) \ds\frac{\partial}{\partial q^i}
$$
and then since $\sigma$ is an integral section of $Z^\gamma$ we obtain
\begin{equation}\label{sigmasi}
\ds\frac{\partial \sigma^i}{\partial x^\alpha} = \ds\frac{\partial H}{\partial p^\alpha_i}\circ\gamma\circ\sigma\;.
\end{equation}

On the other hand, from (\ref{closeness}), (\ref{d(hgamma)=0}) and (\ref{sigmasi}) we obtain
$$\begin{array}{lcl}
\ds\sum_{\alpha=1}^k\ds\frac{\partial (\gamma^\alpha_i\circ \sigma)}{\partial x^\alpha} &=& \ds\sum_{\alpha=1}^k (\ds\frac{\partial  \gamma^\alpha_i}{\partial q^j}\circ\sigma) \ds\frac{\partial \sigma^j}{\partial x^\alpha} = \ds\sum_{\alpha=1}^k (\ds\frac{\partial  \gamma^\alpha_i}{\partial q^j}\circ\sigma)(\ds\frac{\partial H}{\partial p^\alpha_j}\circ\gamma\circ\sigma)
\\\noalign{\medskip} &=&
\ds\sum_{\alpha=1}^k(\ds\frac{\partial  \gamma^\alpha_j}{\partial q^i}\circ\sigma)(\ds\frac{\partial H}{\partial p^\alpha_j}\circ\gamma\circ\sigma) = - \ds\frac{\partial H}{\partial q^i}\circ\gamma\circ\sigma\;.
\end{array}$$
and thus we have proved that $\gamma\circ\sigma $ is a solution of the Hamilton-de Donder-Weyl  equations. \qed
\end{proof}

\begin{remark}{\rm In the particular case $k=1$ we reobtain the theorem proved in  \cite{LMM-1996b,LMM-2009}.\rqed}\end{remark}

\begin{theorem}\label{hamjacobi1}
Let $Z$ be a solution of the $k$-symplectic Hamiltonian equations (\ref{ecHksym}) and
$\gamma : Q \longrightarrow (T^1_k)^*Q$ be a   closed section  of
$\pi^k : (T^1_k)^*Q \longrightarrow Q$, that is, $\gamma =
(\gamma^1, \dots, \gamma^k)$ where   each $\gamma^\alpha$ is an ordinary
closed 1-form on $Q$. Then, the following statements are equivalent:
\begin{enumerate}
\item  $Z\vert_{Im \gamma} - T^1_k\gamma(Z^\gamma) \in \ker \flat_\omega$, being $\flat_\omega$ the map defined in (\ref{bemol ksymp}).
\item $ d(H \circ \gamma) = 0$.
\end{enumerate}
\end{theorem}
\proof We know that if $Z_\alpha$ and $\gamma^\alpha$ are locally given by
$$ Z_\alpha= Z_\alpha^i\ds\frac{\partial }{\partial q^i} + (Z_\alpha)^\beta_i\ds\frac{\partial}{\partial p^\beta_i} \, , \quad \gamma^\alpha= \gamma^\alpha_idq^i\;,$$ then
$Z^\gamma_\alpha= (Z_\alpha^i\circ\gamma)\ds\frac{\partial}{\partial q^i}$. Thus   a direct computation  shows that
$Z\vert_{Im \gamma} - T^1_k\gamma(Z^\gamma) \in \ker \flat_\omega$ is locally written as
\begin{equation}\label{locexpr}
\left( (Z_\alpha)^\beta_i\circ\gamma - (Z_ A^j\circ\gamma ) \ds\frac{\partial \gamma^\beta_ i}{\partial q^j}
\right)\left(\ds\frac{\partial}{\partial p^\beta_i}\circ\gamma\right)=(Y_\alpha)^\beta_i\circ \gamma\left(\ds\frac{\partial}{\partial p^\beta_i}\circ\gamma\right)\;.
\end{equation}
where $\ds\sum_{\alpha=1}^k(Y_\alpha)^\alpha_i=0$.

Now, we are  ready to prove the result.

Assume that $\mathbf{(1)}$ holds, then from (\ref{ecHDWloc}), (\ref{kerflat}) and (\ref{locexpr}) we obtain that
\[\begin{array}{lcl}
0 &=& \ds\sum_{\alpha=1}^k\left( (Z_\alpha)^\alpha_i\circ\gamma - (Z_ A^j\circ\gamma ) \ds\frac{\partial \gamma^\alpha_ i}{\partial q^j}\right)
 \\\noalign{\medskip}  &=& -\left(
 (\ds\frac{\partial H}{\partial q^i}\circ\gamma) + (\ds\frac{\partial H}{\partial p^\alpha_ j}\circ\gamma)\ds\frac{\partial \gamma^\alpha_i}{\partial q^j}\right)
\\\noalign{\medskip} &=& - \left(
(\ds\frac{\partial H}{\partial q^i}\circ\gamma) + (\ds\frac{\partial
H}{\partial p^\alpha_ j}\circ\gamma)\ds\frac{\partial
\gamma^\alpha_j}{\partial q^i}\right)
\end{array}\]
where in the last identity we are using the closeness of $\gamma$ (see (\ref{closeness})).
Therefore, $d(H\circ\gamma)=0$ (see (\ref{d(hgamma)})).

The converse is proved in a similar way by reversing the arguments.
\qed

\begin{remark}
{\rm It should be noticed that if $Z$ and $Z^\gamma$ are
$\gamma$-related, that is, $Z_\alpha= T\gamma(Z^\gamma_\alpha)$, then $d(H\circ \gamma) = 0$, but the converse does
not hold. \rqed}
\end{remark}

\begin{corol}
Let $Z$ be a solution of (\ref{ecHksym}), and $\gamma$ a closed
section of $\pi^k : (T^1_k)^*Q \longrightarrow Q$, as in the above
theorem.  If $Z$ is integrable then the following statements are
equivalent:
\begin{enumerate}
\item  $Z_{Im\,\gamma} - T^1_k\gamma(Z^\gamma) \in \ker \flat_\omega$;
\item $ d(H \circ \gamma) = 0$;
\item If $\sigma\colon U\subset \R^k\to Q$ is an integral section of $Z^\gamma$  then $\gamma\circ\sigma$ is a solution of the Hamilton-De Donder-Weyl equations.
\end{enumerate}
 \end{corol}

\index{Hamilton-Jacobi equation}
The equation
\begin{equation}\label{hjeq}
d(H\circ\gamma)=0
\end{equation}
can be considered as the geometric version of the Hamilton-Jacobi equation for $k$-symplectic field theories. Notice that in local coordinates, equation (\ref{hjeq}) reads us
\[
H(q^i,\gamma^\alpha_i(q))=constant\;.
\]
which when $\gamma^\alpha=dW^\alpha$, where $W^\alpha\colon Q\to \R$ is a function, takes the more familiar form
\[
H(q^i,\ds\frac{\partial W^\alpha}{\partial q^i})=constant\,.
\]

\begin{remark}
    {\rm
        One can connect the Hamilton-Jacobi theory with the theory of Lagrangian submanifolds in the $k$-symplectic geometry. Let us observe that the Hamilton-Jacobi problem in the $k$-symplectic description consists of finding a closed section $\gamma$ of $\pi^k$ such that $d(H\circ \gamma)=0$, but the condition of closed section is equivalent to find a section $\gamma$ such that its image is a $k$-Lagrangian submanifold  of $(T^1_k)^*Q$. A proof of this equivalence and a completed description  of the Lagrangian submanifolds in the $k$-symplectic approach can be found in \cite{LV-2012}.
    \rqed}
\end{remark}
\section{Example: the vibrating string problem}
\index{Vibrating string}
 In this example we consider the vibrating string problem under the assumptions that the string is made up of individual particles that move vertically and $\psi(t,x)$ denotes the vertical displacement from equilibrium of the particle at horizontal position $x$ and at time $t$.
 \begin{figure}[h]
                    \centering
 \begin{tikzpicture}
    \draw (-3,0) -- (6,0);
    \draw (0,0)--(0,2.3);
    \draw (5.5,0) arc (0:180:4cm and 2.5cm);
    \draw (-2.5,-0.25) node {$0$};
    \draw (0,-0.25) node {$x$};
    \draw (5.5,-0.25) node {$L$};
    \draw (0.6,1) node {$\psi(t,x)$};
 \end{tikzpicture}
 \caption{Vibrating string at time $t$.}
\end{figure}
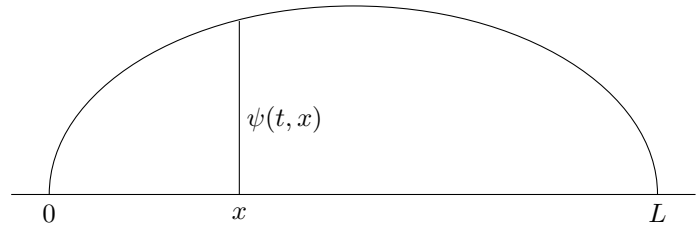

With a study of the tensile forces in this problem and using Newton's second Law one obtains the equation of motion for small oscillations of a frictionless string, that is the one-dimensional wave equation \index{Wave equation}\index{Equation!wave}
                \begin{equation}\label{one_wave_eq}
                        \sigma\frac{\displaystyle\partial^2 \psi}{\displaystyle\partial t^2}-\tau \frac{\displaystyle\partial^2\psi}{\displaystyle\partial x^2}=0 \,,
                \end{equation}
            where $\sigma$ and $\tau$ are certain constants of the problem, $\sigma$  represents the linear mass density, that is, a measure of mass per unit of length and $\tau$ is Young's module of the system related to the tension of the string, see for instance \cite{Goldstein}.

Let $\gamma\colon \R\to (T^1_2)^*\R $ be the  section of $\pi^2\colon T^*\r\oplus T^*\r\to \r$ defined by $\gamma(q)=(aq\,dq,bq\,dq)$
where $a$ and $b$ are two constants such that $\tau a^2=\sigma b^2$. This section $ \gamma$ satisfies
the condition $d(H\circ \gamma)=0$ with $H$ the Hamiltonian function defined by
               \begin{equation}\label{vibrating ham}
                    \begin{array}{rccl}
                      H \colon &T^*\r\oplus T^*\r &  \longrightarrow & \r \\ \noalign{\medskip}
                                             &   (q,p^1,p^2)  & \to &   \ds\frac{1}{2}\left(\ds\frac{(p^1)^2}{\sigma}- \ds\frac{(p^2)^2}{\tau}\right)\,.
                   \end{array} \end{equation}
Therefore, the condition $(2)$ of the Theorem \ref{hjth} holds.

Let $Z$ be a $2$-vector field solution of (\ref{ecHksym}) for the Hamiltonian (\ref{vibrating ham}), then the $2$-vector field $Z^\gamma=(Z^\gamma_1,Z^\gamma_2)$ is locally given by
$$
Z^\gamma_1=\ds\frac{a}{\sigma} q\ds\frac{\partial}{\partial q} \quad , \quad Z^\gamma_2=-\ds\frac{b}{\tau} q\ds\frac{\partial}{\partial q}\,.
$$
It is easy to check that $Z^\gamma$ is an integrable $2$-vector field.

If $\psi\colon \R^2\to \R$ is an integral section of $Z^\gamma$, then
$$\ds\frac{\partial \psi}{\partial x^1}=\ds\frac{a}{\sigma}\psi \quad , \quad
\ds\frac{\partial \psi}{\partial x^2}=-\ds\frac{b}{\tau}\psi,$$ thus
$$\psi(x^1,x^2)= C\,\hbox{exp }\left({\ds\frac{a}{\sigma}x^1-\ds\frac{b}{\tau}x^2 }\right),\quad C\in \R$$

By Theorem \ref{hjth} one obtains that the map $
\phi=\gamma\circ \psi$, locally given by $$(x^1,x^2) \mapsto (\psi(x^1,x^2),a\psi(x^1,x^2),b\psi(x^1,x^2)),$$
is a solution of the Hamilton-De Donder-Weyl equations associated to the Hamiltonian $H$ given by (\ref{vibrating ham}), that is,
$$\begin{array}{ccl}
0 &=& a\ds\frac{\partial \psi}{\partial x^1} + b\ds\frac{\partial \psi}{\partial x^2} \\\noalign{\medskip}
\ds\frac{a}{\sigma}\psi &=& \ds\frac{\partial \psi}{\partial x^1}
\\\noalign{\medskip}
-\ds\frac{b}{\tau}\psi &=& \ds\frac{\partial \psi}{\partial x^2}\end{array}
$$

Let us observe that from this system one obtains that $\psi$ is a solution of the motion equation of the vibrating string (\ref{one_wave_eq}).

%\newpage
%\mbox{}
%\thispagestyle{empty} % para que no se numere esta p\'{a}gina

\chapter{Lagrangian Classical Field Theories}\label{chapter: K-SympLagCFT}

\index{Lagrangian}
\index{Euler-Lagrange field equations}
The aim of this chapter is to give a geometric description of the \emph{Euler-Lagrange field equations}
\begin{equation}\label{solsopde3}
        \ds\sum_{\alpha=1}^k\left(\ds\frac{\partial^2
        L}{\partial q^j
        \partial v^i_\alpha}\Big\vert_{\psi(x)}
        \ds\frac{\partial \psi^j}{\partial x^\alpha}\Big\vert_{x}+
        \ds\frac{\partial^2 L}{\partial v^j_\beta
        \partial v^i_\alpha}\Big\vert_{\psi(x)}
        \ds\frac{\partial^2 \psi^j}{\partial x^\alpha \partial
        x^\beta}\Big\vert_{x}\right)=\ds\frac{\partial L}{\partial
        q^i}\Big\vert_{\psi(x)}\,,
    \end{equation}
  $1\leq i \leq n$, where $\psi\colon\mathbb{R}^k\to T^1_kQ$ and the Lagrangian function is a function  $L\colon \tkq\to \r$ defined on the tangent bundle of $k^1$-velocities $\tkq$ of an arbitrary manifold $Q$.

Let us observe that the above equations can be written in a equivalent way as follows:
    \begin{equation}\label{EcEL}
            \ds\sum_{\alpha=1}^k\ds\frac{\partial}{\partial x^\alpha}\Big\vert_{x}\left(\ds\frac{\partial L}{\partial
            v^i_\alpha}\Big\vert_{\psi(x)}\right)=\ds\frac{\partial
            L}{\partial q^i}\Big\vert_{\psi(x)}\,,\qquad v^i_\alpha (\psi(x))=\derpar{\psi^i}{x^\alpha}\Big\vert_{x}\,.
        \end{equation}

The aim of this chapter is to obtain these equations in two alternative ways. Firstly, in the classical way, describing a variational principle which provides the Euler-Lagrange field equations. The second way to obtain these equations is using the $k$-symplectic formalism introduced in chapter \ref{chapter: k-symplectic formalism}.

Firstly, we shall give a detail description of $\tkq$, i.e. the tangent bundle of $k^1$-velocities and we  introduce some canonical geometric elements defined on this manifold. Finally we discuss the equivalence between the Hamiltonian and Lagrangian approaches when the Lagrangian function is regular or hyper-regular.

\section{The tangent bundle of $k^1$-velocities}\label{section k-tangent}

\index{Tangent bundle of $k^1$-velocities}
    In this section we consider again (this space was introduced in section \ref{section k-vector field}) the space $T^1_kQ$ associated to a differentiable manifold $Q$ and we shall give a complete description.
           %  Consider the tangent bundle $\tau\colon TQ\to Q$ of and arbitrary $n$-dimensional smooth manifold $Q$ and we define the space
%                \[
%                    T^{1}_{k}Q=TQ\oplus\stackrel{k}{\ldots}\oplus TQ\,,
%                \]
%            as the Whitney sum of $k$ copies of the tangent bundle. Let us observe that an element ${\rm v}_q$ of $T^1_kQ$ is a family of $k$ tangent vectors $({v_1}_q,\ldots, {v_k}_q)$ over the same point $q\in Q$. Thus one can consider the canonical projection
%                \begin{equation}\label{tkq: tauk}
%                    \begin{array}{rccl}
%                        \tau^k\colon & \tkq& \to & Q\\\noalign{\medskip}
%                         & {\rm v}_q=({v_1}_q,\ldots, {v_k}_q) & \mapsto & q\,.
%                    \end{array}
%                \end{equation}
\index{Tangent bundle of $k^1$-velocities!Canonical coordinates}
            Each coordinate system $(q^1, \ldots, q^n)$ defined on an open neighborhood $U\subset Q$, induces a local bundle coordinate system $(q^i,v^i_\alpha) $ on $T^1_kU\equiv(\tau^k)^{-1}(U)\subset \tkq$ defined as follows
                \begin{equation}\label{tkq: natural coord}
                    q^i({\rm v}_q)=q^i(q),\quad                 v^i_\alpha({\rm v}_q)={v_\alpha}_q(q^i)=(dq^i)_q({v_\alpha}_q)\,,
                \end{equation}
            where   ${\rm v}_q=({v_1}_q,\ldots,{v_k}_q)\in \tkq$, $\n$ and $\ak$.

            These coordinates are called \emph{canonical coordinates} on $\tkq$ and they endow to $\tkq$ of a structure of differentiable manifold of dimension  $n(k+1)$.

            The following diagram shows the notation which we shall use along this book:
                \[
                    \xymatrix@=15mm{
                        \tkq\ar[r]^-{\tau^{k,\alpha}}\ar[rd]_-{ \tau^k}& TQ\ar[d]^-{\tau}\\ & Q
                    }
                \]
            where $\tau^{k,\alpha}:\tkq\to TQ$ is the canonical projection defined as follows
                \begin{equation}\label{taukalpha}
                    \tau^{k,\alpha}({\rm v}_q)= \tau^{k,\alpha}({v_1}_q,\ldots, {v_k}_q)={v_\alpha}_q\,,
                \end{equation}
            with $\ak$.
                \begin{remark}\label{j1rkq}
                {\rm
                        The manifold $\tkq$ can be described as  a manifold of jets, (see \cite{lr1,{Saunders-89}}).

                        Let $\phi:U_0\subset \rk \to Q$  and $\psi:V_0\subset \rk \to Q$ be two maps defined in an open neighborhood of $0\in\rk$, such that $\phi(0)=\psi(0)=p$.  We say that $\phi$ and $\psi$ are related   on $0\in  \rk$ if $\phi_*(0)=\psi_*(0)$, which means that the partial derivatives of $\phi$ and $\psi$ coincide up to order one.

                         The equivalence classes determined by this relationship are called {\it jets of order 1}, or, simply, $1$-jets with source $0\in \rk$ and the same target.

                         The $1$-jet of a map $\phi:U_0\subset \rk \to Q$ is denoted by $j^1_{0,q}\phi$ where $\phi(0)=q$. The set of all $1$-jets at $0$ is denoted by
                         $$J_0^1(\rk,Q)= \ds\bigcup_{q\in Q}J^1_{0,\,q}(\rk,Q)= \ds\bigcup_{q\in Q}\{j^1_{0,q}\phi\, \vert\, \phi\colon\rk\to Q\;\makebox{smooth},\, \phi(0)=q\}$$

                         The canonical projection $\beta:J_0^1(\rk,Q)\to Q$ is defined by $\beta( j^1_0\phi )=\phi(0)$ and $J_0^1(\rk,Q)$                  is called the {\bf  tangent bundle of $k^1$-velocities}, (see Ereshmann \cite{Ehresmann}).
                         Let observe that for $k=1$, $J_0^1(\r,Q)$ is diffeomorphic to $TQ$.

                         We shall now describe the local coordinates on $J_0^1(\rk,Q)$. Let $U$ be a chart of $Q$ with local coordinates $(q^i)$,
                         $1\leq i\leq n$, $\phi:U_0\subset \rk \to Q$ a mapping such that $\phi(0)\in U$ and $\phi^i=q^i\circ \phi$. Then the $1$-jet $j^1_{0,q}\phi$ is uniquely represented in $\beta^{-1}(U)$ by
                         $$
                         (q^i,v^i_1, \ldots , v_k^i) \; , \quad 1\leq i\leq n
                         $$
                         where
                         \begin{equation}\label{jetcoorl}
                         q^i(j^1_{0,q}\phi)=q^i(\phi(0))=\phi^i(0) \, ,  \; v^i_\alpha(j^1_{0,q}\phi)=
                         \phi_*(0)\left(
                          \derpar{}{x^\alpha}\Big\vert_{0}\right)(q^i)\; ,
                          \end{equation}

                        %For each $q\in Q$ we consider
%                            \[
%                                J^1_{0,\,q}(\rk,Q)=\{j^1_{0,q}\phi\, \vert\, \phi\colon\rk\to Q\;\makebox{smooth},\, \phi(0)=q\} \,,
%                            \]
%                        where $j^1_{0,q}\phi$ denotes the $1$-jet of the map $\phi$ at  $0\in \rk$ such that $\phi(0)=q$. Thus,
%                            \[
%                                \tkq=\ds\bigcup_{q\in Q}(T^1_k)_{q}Q\equiv\ds\bigcup_{q\in Q}J^1_{0,\,q}(\rk,Q)=J^1_0(\rk,Q)\,,
                            %\]

                        The manifolds $\tkq$ and $J^1_0(\rk,Q)$ can be identified, via the diffeomorphism
                            \[
                                \begin{array}{ccc}
                                    J^1_0(\r^k,Q) & \equiv & TQ \oplus \stackrel{k}{\dots} \oplus TQ \\
                                    j^1_{0,q}\phi  & \equiv & ({v_1}_{q},\ldots, {v_k}_{q})
                                \end{array}
                            \]
                             defined by $$ {v_\alpha}_{q}=\phi_*(0)\Big(\frac{\partial}{\partial                 x^\alpha}\Big\vert_0\Big),\quad \alpha=1,\ldots, k\,\, ,$$
                        being $\phi (0)=q$.\rqed
                        }
                \end{remark}

\subsection{Geometric elements}\label{section tkq:geometry}

In this section we introduce some geometric constructions which are necessary to describe Lagrangian Classical Field Theories using the $k$-symplectic approach.

  %\subsection{Vertical lifts}\label{section tkq: vert lift}
  \paragraph{ Vertical lifts}\

        Given a tangent vector $u_q$ on an arbitrary manifold $Q$, one can consider the vertical lift to the tangent bundle of $TQ$. In a similar way, we can define the vertical lift to the tangent bundle of $k^1$-velocities by considering the lift on each copy of the tangent bundle.

\index{Tangent bundle of $k^1$-velocities!vertical lifts}
\index{Vertical lifts!Tangent vectors}
        \begin{definition}
            Let   $u_q\in T_{q}Q$ be  a tangent vector at $q\in Q$. For each  $\ak$, we define the \emph{vertical $\alpha$-lift}, $(u_q)^{V_\alpha}$, as the  vector field at the fibre $(\tau^k)^{-1}(q)\subset T_k^1Q$ given by
                \begin{equation}\label{tkq: alpha vert lift tangent vector}
                    (u_q)^{V_\alpha}_{{\rm v}_q} = \displaystyle\frac{d}{ds}({v_1}_q,\ldots,{v_{\alpha-1}}_q, {v_{\alpha}}_q+s  u_q,{v_{\alpha+1}}_q, \ldots,{v_k}_q)\Big\vert_{s=0}
                  \end{equation}
         for any point ${\rm v}_q=({v_1}_q,\ldots, {v_k}_q) \in (\tau^k)^{-1}(q)\subset T^1_kQ$.
        \end{definition}

        In local canonical coordinates (\ref{tkq: natural coord}), if   $u_q = u^i \,\ds\frac{\partial}{\partial q^i}\Big\vert_{q}$ then
            \begin{equation}\label{tkq: local alpha vert lift tangent vector}
                (u_q)^{V_\alpha}_{{\rm v}_q} =  u^i \displaystyle\frac{\partial}{\partial
                v^i_\alpha}\Big\vert_{{\rm v}_q}\, .
            \end{equation}

        The vertical lifts  of tangent vectors allows us to define the vertical lift of
    vector fields.
\index{Vertical lifts!Vector field}
        \begin{definition}
            Let $X$  be a    vector field on $Q$. For each $\ak$ we call \emph{$\alpha$-vertical lift} of $X$ to $T^1_kQ$, to the    vector field $X^{V_\alpha}\in \vf(\tkq)$  defined by
                \begin{equation}\label{tkq: alpha vert lift vector field}
                    X^{V_\alpha}({\rm v}_q)=(X(q))^{V_\alpha}_{{\rm v}_q}\,,
                \end{equation}
            for all point ${\rm v}_q=({v_1}_q,\ldots, {v_k}_q) \in T^1_kQ$.
        \end{definition}

        If  $X=X^i \derpar{}{q^i}$ then, from  (\ref{tkq: local alpha vert lift tangent vector})  and (\ref{tkq: alpha vert lift vector field}) we deduce that
            \begin{equation}\label{tkq: local alpha vert lift vector field}
                X^{V_\alpha}=( X^i\circ \tau^k)\displaystyle\frac{\partial}{\partial
                v^i_\alpha}\,,
            \end{equation}
        since
            \[
                (X(q))^{V_\alpha}_{{\rm v}_q} =\Big(X^i(q)\displaystyle\frac{\partial}{\partial q^i}\Big\vert_{q}\Big)^{V_\alpha}_{{\rm v}_q}= X^i(q)\displaystyle\frac{\partial}{\partial v^i_\alpha}\Big\vert_{{\rm v}_q}=( X^i\circ \tau^k)({\rm v}_q)\displaystyle\frac{\partial}{\partial v^i_\alpha}\Big\vert_{{\rm v}_q}\,.
            \]

%\subsection{Canonical  $k$-tangent structure.}\label{section tkq: k-tangent struct}\
\paragraph{ Canonical  $k$-tangent structure}\

In a similar way that in the tangent bundle, the vertical lifts of tangent vectors allows us to introduce a family   $\{J^1,\ldots,$ $ J^k\}$ of $k$ tensor fields of type $(1,1)$ on  $\tkq$. This family is the model of the so called \emph{$k$-tangent structures}  introduced by M. de Le\'{o}n {\it et al.} in \cite{LMS-88,LMS-91}.  In the case $k=1$, $J=J^1$ is the canonical tangent structure or vertical endomorphism (\ref{jtang})  (see \cite{Crampin-1983, cram2,grif1, grif2, grif3,klein}).

\index{Tangent bundle of $k^1$-velocities!$k$-tangent structure}
\index{$k$-tangent structure}
\begin{definition}\label{kte}
For each $\ak$  we define the tensor field $J^\alpha$ of type $(1,1)$
   on $\tkq$  as follows
\begin{equation}\label{ktetkq}
\begin{array}{ccccl}
J^\alpha({\rm v}_q) & : & T_{{\rm v}_q}(T^1_kQ) & \to & T_{{\rm v}_q}(T^1_kQ) \\
\noalign{\medskip}
          & &Z_{{\rm v}_q} & \to &
J^\alpha({\rm v}_q)(Z_{{\rm v}_q})=
\Big((\tau^k)_*({\rm v}_q)(Z_{{\rm v}_q})\Big)^{V_\alpha}_{{\rm v}_q}\end{array}\end{equation}
 where   $ {\rm v}_q
    \in T^1_kQ\,.$
 \end{definition}

From (\ref{tkq: local alpha vert lift tangent vector}) and (\ref{ktetkq}) we deduce that, for each $\ak$,  $J^\alpha$ is
locally given by
\begin{equation}\label{localJA}
{J^\alpha}=\displaystyle\frac
{\displaystyle\partial}{\displaystyle\partial v^i_\alpha} \otimes
dq^i\,.
\end{equation}

\begin{remark}
{\rm
    The family $\{J^1,\ldots, J^k\}$ can be obtained using the vertical lifts of the identity tensor field of  $Q$ to $T^1_kQ$
 defined  by Morimoto (see \cite{mor3,mor}).
     \rqed
     }
\end{remark}

%\subsection{Canonical vector fields.}\label{section tkq: canonical vf}\

\paragraph{ Canonical vector fields}\

An important geometric object on $\tkq$ is the generalized Liouville vector field.

\index{Liouville vector field}
\index{Tangent bundle of $k^1$-velocities!Liouville vector field}
\begin{definition} The {\bf   Liouville vector field} $\triangle$ on
$\tkq$ is the infinitesimal generator of the flow
\begin{equation}\label{flujoC}
\begin{array}{llcl}
\psi:&\r \times \tkq & \longrightarrow &  T^1_kQ  \\
\noalign{\medskip}  & (s,({v_1}_{q},\ldots, {v_k}_{q})) & \mapsto &
( e^s{v_1}_{q}, \ldots,e^s{v_k}_{q})\, .
\end{array}
\end{equation} and in local coordinates it has the form
\begin{equation}\label{loclvf}
 \triangle =   \displaystyle\sum_{i=1}^n\sum_{\alpha=1}^k v^i_\alpha
\frac{\displaystyle\partial}{\displaystyle\partial v^i_\alpha}\, .
\end{equation}
\end{definition}
\begin{definition}\label{cvcano}
For each $\ak$   we define the {\it canonical vector field}
  $\triangle_\alpha$ as the infinitesimal generator  of the
following flow
\begin{equation}\label{flujo CA}
\begin{array}{rcl}
\psi^\alpha:\r \times T^1_kQ & \longrightarrow &  T^1_kQ
\\ \noalign{\medskip}  (s,({v_1}_{q},\ldots, {v_k}_{q} )) & \mapsto &
( {v_1}_{q}, \ldots, {v_{\alpha-1}}_{q}, e^s {v_\alpha}_{q},
{v_{\alpha+1}}_{q}, \ldots, {v_k}_{q})\,,
\end{array}
\end{equation}  and in local coordinates it has the form
\begin{equation}\label{loccanvf}
 \triangle_\alpha =   \displaystyle\sum_{i= 1}^n v^i_\alpha
\frac{\displaystyle\partial}{\displaystyle\partial
v^i_\alpha}\,,\quad 1\leq \alpha\leq k\,,
\end{equation}
for each $\ak$.
\end{definition}

From  (\ref{loclvf}) and (\ref{loccanvf}) we deduce that $\triangle=\triangle_1+\ldots + \triangle_k\,.$

\begin{remark}\label{remark1.32}
{\rm
The vector fields $\triangle$ and $\triangle_\alpha$  can be also
defined using the vertical lifts. From  (\ref{tkq: local alpha vert lift tangent vector}), (\ref{loclvf})
and (\ref{loccanvf}) one obtains that
$$
\triangle({\rm v}_{q})= \ds\sum_{\alpha=1}^k
({v_\alpha}_{q})^{V_\alpha}_{{\rm v}_q},\quad \triangle_\alpha ({\rm v}_q
)=({v_\alpha}_{q})^{V_\alpha}_{{\rm v}_q}\,,
$$ where ${\rm v}_{q}=({v_1}_q,\ldots , {v_k}_{q})\in \tkq$ and $\ak$.
  \rqed}\end{remark}

\subsection{Prolongation of  vector fields}\label{section tkq: prolongation}

    In a similar way as in section \ref{section tkqh: prolongation} one can define the canonical prolongation of maps between manifolds to the corresponding tangent bundles of $k^1$-velocities (see (\ref{tkq: prolongation expr})).

    %Using the definition of the tangent map we introduce the following definition
%\index{Tangent bundle of $k^1$-velocities!Prolongation of maps}
%\index{Prolongation!Maps}
    %\begin{definition}\label{tkq: prolongation expr}
%        Let $f\colon M\to N$ be a diffeomorphism. The \emph{natural or canonical prolongation} of $f$ to the corresponding bundles of $k^1$-velocities is the map
%        \[
%            T^1_kf\colon T^1_kM\to T^1_kN
%        \]
%        defined by
%        \[
%            T^1_kf({v_1}_m,\ldots, {v_k}_m) = (f_*(m)({v_1}_m), \ldots, f_*(m)({v_k}_m))\,,
%        \]
%        where $({v_1}_m,\ldots, {v_k}_m)\in T^1_kM$ and $m\in M$.
%    \end{definition}

%This notion of prolongation allows us to introduce the complete lift of vector fields from a manifold $Q$ to $\tkq$.

\index{Tangent bundle of $k^1$-velocities! Prolongation of vector fields}
\index{Prolongation!Vector fields}
\begin{definition}\label{tkq: complete lift}
Let $Z\in \vf(Q)$ be a  vector field
 on $Q$ with local $1$-parametric group of diffeomorphisms  $h_s:Q\to Q$. The
{\bf complete or natural lift} of $Z$ to $\tkq$
 is the   vector field $Z^C$ on $\tkq$ whose   local $1$-parameter group of diffeomorphisms is   $T^1_k(h_s): T^1_kQ \to \tkq$.
\end{definition}

\begin{remark}
{\rm
The definition of $T^1_k(h_s)$ is just the one gives in (\ref{tkq: prolongation expr})
\rqed}
\end{remark}
In local canonical coordinates (\ref{tkq: natural coord}), if $Z=Z^i\ds\frac{\partial}{\partial q^i}$
then the local expression is
\begin{equation}\label{tkq: local complete lift}
Z^C=Z^i\displaystyle\frac{\partial}{\partial q^i} \, + \, v^j_\alpha
\ds \frac{\partial Z^i} {\partial q^j}
\displaystyle\frac{\partial}{\partial v^i_\alpha} \; .
\end{equation}

The following lemma shows that the canonical prolongations of maps to the tangent bundle of $k^1$-velocities leave invariant the canonical structures of $\tkq$.

\begin{lemma}
\label{lema.2}
 Let $\Phi=T^1_k\varphi:T^1_kQ \to T^1_kQ$ be  the canonical prolongation of a diffeomorphism $\varphi:Q\to Q$, then for each $\ak$, we have
$$
(a) \quad \Phi_* \circ J^\alpha = J^\alpha \circ \Phi_* \;, \quad (b)
\quad \Phi_*\triangle_\alpha=\triangle_\alpha \;,\quad (c)\quad
\Phi_*\triangle= \triangle\, .
$$
\end{lemma}
\proof (a) It is a direct consequence of the local expression
(\ref{localJA}) of $J^\alpha$ and the local expression of
$T^1_k\varphi$ given by $$
T^1_k\varphi(q^i,v^i_\alpha)=\Big(\varphi^j(q^i),v_\alpha^i\frac{\partial
\varphi^j}{\partial q^i}\Big) \, $$ where the functions $\varphi^j$
denote the components of the diffeomorphism $\varphi:Q\to Q$.

 (b) It is a consequence of $T^1_k\varphi\circ\psi_{{\rm v}_q}^\alpha=\psi_{{\rm v}_q}^\alpha\circ T^1_k\varphi$,
 where $\psi_{{\rm v}_q}^\alpha$ are  $1$-parameter group of diffeomorphisms
(\ref{flujo CA}) generated by $\triangle_\alpha$.

(c) It is a direct consequence of  (b) and of the identity
$\triangle=\triangle_1+\ldots + \triangle_k$.\qed

\subsection{First prolongation of maps}

\index{Tangent bundle of $k^1$-velocities!First prolongation}
\index{First prolongation}

Here we shall introduce the notion of first prolongation, which will be very important along this chapter and generalize the lift of a curve on $Q$ to the tangent bundle $TQ$ of $Q$.

        \begin{definition}\label{first_prol}
            We define the \emph{first prolongation} $\phi^{(1)}$ of a map $\phi\colon \rk\to Q$ as the map
            \begin{equation}\label{phi(1)}
                 \begin{array}{rccl}
                    \phi^{(1)} \colon & U_0 \subseteq \rk & \longrightarrow & T^1_kQ \\
                      & x & \longmapsto &  \Big( \phi_* (x) \Big( \displaystyle\frac{\partial}{\partial x^1}\Big\vert_{x}\Big),\ldots , \phi_* (x) \Big(\displaystyle\frac{\partial}{\partial x^k}\Big\vert_{x} \Big) \Big) \, ,
                 \end{array}
           \end{equation}
           where $(x^1,\ldots, x^k)$ denotes the coordinates on $\rk$ and $\tkq$ the tangent bundle of $k^1$-velocities introduced at the beginning of section \ref{section k-tangent}.
        \end{definition}

        If we consider canonical coordinates $(q^i, v^i_\alpha)$ on $\tkq$ (see (\ref{tkq: natural coord}) for the definition), then the first prolongation is locally given by
            \begin{equation}\label{localphi1}
                \begin{array}{rccl}
                    \phi^{(1)} \colon & U_0 \subseteq \rk & \longrightarrow & T^1_kQ \\
                      & x & \longmapsto & \phi^{(1)} \left( x \right) \, = \Big (\phi^i(x), \derpar{\phi^i}{x^\alpha}\Big\vert_{x}\Big) ,
                 \end{array}
            \end{equation}
        where $\phi^i=q^i\circ \phi$, and we are using that
            \[
                \phi_*(x)\Big(\derpar{}{x^\alpha}\Big\vert_{x}\Big) = \derpar{\phi^i}{x^\alpha}\Big\vert_{x}\derpar{}{q^i}\Big\vert_{\phi(x)}\,.
            \]

%\section{Euler-Lagrange field equations.}

\section{Variational principle for the Euler-La\-gran\-ge equations.}\label{section: lag variational principle}

    In this section we describe the problem in the calculus of variations for multiple integrals, which allows us to obtain the Euler-Lagrange field equations.

   % A reader not interested in technical details and who is prepared to accept the expression of the Euler-Lagrange equation (\ref{EcEL}) may omit this section.
\index{Lagrangian}
    Along this section we  consider a given Lagrangian function $L$ on the tangent bundle of $k^1$-velocities, i.e. $L\colon \tkq\to \r$. Thus we can evaluate $L$ in the first prolongation (\ref{phi(1)}) of a field $\phi\colon \rk\to Q$. Given $L$ we can construct the following operator:

        \begin{definition}\label{Alksim}
            Let us denote by $\mathcal{C}^\infty_C(\rk,Q)$ the set of maps $\phi:U_0\subset\rk\to Q,$ with compact support,         defined on an open set $U_0$. We define the action associated to $L$ by
                \[
                    \begin{array}{lccl}
                        {\mathcal J}:&\mathcal{C}^\infty_C(\rk,Q) & \to
                        &\r\\\noalign{\medskip}  & \phi &\mapsto &{\mathcal
                        J}(\phi)=\ds\int_{\rk} (L\circ \phi^{(1)})(x)\,  d^kx\,,
                    \end{array}
                \]
          where $d^kx=dx^1\wedge\ldots\wedge dx^k$ is a volume form on
          $\rk$ and $\phi^{(1)}:U_0\subset\rk\to T^1_kQ$
         denotes the first prolongation of $\phi$ defined in (\ref{phi(1)}). \end{definition}

            \begin{definition}\label{extremales}
            A  map  $\phi\in\mathcal{C}^\infty_C(\rk,Q)$, is an \textit{extremal} of $\mathcal{J}$ if $$
            \ds\frac{d}{ds}\Big\vert_{s=0}\mathcal{J}(\tau_s\circ\phi)=0,$$ for
            each flow  ${\tau_s}$ on $Q$ such that
             $\tau_s(q)=q$ for every  $q$ at the boundary of $\phi(U_0)\subset Q$.\end{definition}

Let us observe that the flow  $\tau_s:Q\to Q$, considered in this
definition, are generated by a    vector field on $Q$ which vanishes
at the  boundary of $\phi(U_0)$.

The variational problem associated to a Lagrangian $L$, is to find
the  extremals of the integral action    $\mathcal{J}$. In the following proposition we characterize these extremals.

\begin{prop}\label{varprinL}
Let $L:\tkq\to \r$ be a Lagrangian and
$\phi\in\mathcal{C}^\infty_C(\rk,Q)$. The following  assertions are
  equivalent :
\begin{enumerate}
\item $\phi:U_0\subset\rk\to Q$ is an extremal of $\mathcal{J}$.
\item For each  vector field $Z$ on $Q$, vanishing at all points on the boundary  of $\phi(U_0)$, one has  $$\ds\int_{U_0}
\left((\mathcal{L}_{Z^c} L) \circ \phi^{(1)}\right)(x) d^kx=0\,, $$
where $Z^C$ is the complete lift of  $Z$ to $\tkq$ (see \ref{tkq: complete lift}).
\item $\phi$ is solution of the   Euler-Lagrange field equations (\ref{EcEL}).
\end{enumerate}
\end{prop}
\begin{proof}
First we prove the equivalence between  $(1)$ and $(2)$.

Let $\phi\colon U_0\subset \rk \to Q$ be a map and $Z\in \vf(Q)$ be a   vector field
on  $Q$, with local $1$-parameter  group of diffeomorphism $\{\tau_s\}$, and vanishing at the boundary of $\phi(U_0)$,  then $T^1_k\tau_s$ is the
local $1$-parameter  group of diffeomorphism of  $Z^C$.

 A simple computation shows  $T^1_k\tau_s\circ\phi^{(1)}=(\tau_s\circ\phi)^{(1)}$, and thus
 we deduce
$$\begin{array}{lcl}
&&\ds\frac{d}{ds}\Big\vert_{s=0}\mathcal{J}(\tau_s\circ\phi) =
\ds\frac{d}{ds}\Big\vert_{s=0}\int_{\rk} (L\circ
(\tau_s\circ\phi)^{(1)})(x) d^kx\\\noalign{\medskip}

 &=&
\ds\lim_{s\to 0}\ds\frac{1}{s}\left(\int_{\rk} (L\circ
(\tau_s\circ\phi)^{(1)})(x) d^kx - \int_{\rk} (L\circ \tau_0\circ
\phi^{(1)})(x) d^kx\right)
\\\noalign{\medskip}

&=& \ds\lim_{s\to 0}\ds\frac{1}{s}\left(\int_{\rk} (L\circ
T^1_k\tau_s\circ\phi^{(1)})(x) d^kx - \int_{\rk} (L\circ
\phi^{(1)})(x) d^kx\right)
\\\noalign{\medskip}
%\end{array}$$
%$$\begin{array}{lcl}
&=& \ds\lim_{s\to 0}\ds\frac{1}{s}\left(\int_{\rk} \Big((L(
T^1_k\tau_s\circ\phi^{(1)})(x)) - L(
\phi^{(1)}(x)) \Big) d^kx\right)
\\\noalign{\medskip}

&=& \ds\int_{\rk} \left(\ds\lim_{s\to 0}\ds\frac{1}{s}\Big(L(
T^1_k\tau_s\circ\phi^{(1)}(x))  - L( \phi^{(1)}(x))\Big)\right) d^kx
\\\noalign{\medskip}

&=&\ds\int_{\rk} \left((\mathcal{L}_{Z^c} L) \circ
\phi^{(1)}\right)(x) d^kx\,,
\end{array}$$
so, we have done.

 We have proved that    $\phi:U_0\subset\rk\to
Q$ is an extremal of $\mathcal{J}$ if and only if for each vector
field $Z\in\vf(Q)$
 vanishing at the boundary  of $\phi(U_0)$ one has  \begin{equation}\label{extcampo} \ds\int_{U_0}
\left((\mathcal{L}_{Z^c} L) \circ \phi^{(1)}\right)(x) d^kx=0\,.
\end{equation}

We now prove that it is equivalent to say that $\phi$ is a solution of the Euler-Lagrange field equation.

 Let us suppose that
$Z=Z^i\derpar{}{q^i}$;   from the local expression (\ref{tkq: local complete lift}) of
$Z^C$ and the expression of integration by parts in multiple integrals and since $\phi$ has compact support, we deduce that:
\[
    \begin{array}{lcl}
    & &
        \ds\int_{\rk} \left((\mathcal{L}_{Z^c} L) \circ \phi^{(1)}\right)(x) d^kx\\\noalign{\medskip}
    &=&
        \ds\int_{\rk}\left( Z^i(\phi(x)) \derpar{L}{q^i}\Big\vert_{\phi^{(1)}(x)} + \derpar{\phi^j}{x^\alpha}\Big\vert_{x}\derpar{Z^i}{q^j}\Big\vert_{\phi(x)}\derpar{L}{v^i_\alpha}\Big\vert_{\phi^{(1)}(x)}\right) d^kx\\\noalign{\medskip}
    &=&
        \ds\int_{\rk}\left( Z^i(\phi(x)) \derpar{L}{q^i}\Big\vert_{\phi^{(1)}(x)} + \derpar{(Z^i\circ \phi)}{x^\alpha}\Big\vert_{x} \derpar{L}{q^i}\Big\vert_{\phi^{(1)}(x)}\right) d^kx\\\noalign{\medskip}
%        \end{array}
%        \]
%        \[
%    \begin{array}{lcl}
    &=&
        \ds\int_{\rk}\left( Z^i(\phi(x)) \derpar{L}{q^i}\Big\vert_{\phi^{(1)}(x)} - Z^i(\phi(x))\derpar{}{x^\alpha}\left(\derpar{L}{v^i_\alpha}\Big\vert_{\phi^{(1)}(x)}\right)\right)d^kx\\\noalign{\medskip}
    &=&
        \ds\int_{\rk}(Z^i\circ \phi)(x)\left(\derpar{L}{q^i}\Big\vert_{\phi^{(1)}(x)}-\derpar{}{x^\alpha}\left(\derpar{L}{v^i_\alpha}\Big\vert_{\phi^{(1)}(x)}\right) \right) d^kx\,.
    \end{array}
\]

Therefore we obtain that $\phi$ is an extremal of $\mathcal{J}$ if and only if
\[
    0= \ds\int_{\rk}(Z^i\circ \phi)(x)\left(\derpar{L}{q^i}\Big\vert_{\phi^{(1)}(x)}-\derpar{}{x^\alpha}\left(\derpar{L}{v^i_\alpha}\Big\vert_{\phi^{(1)}(x)}\right) \right) d^kx\,.
\]

Since this identity holds for all $Z^i$, applying lemma \ref{variations} we obtain that  $\phi$ is an extremal of $\mathcal{J}$ if and only if
\begin{equation}\label{ELe}
\displaystyle \sum_{\alpha=1}^k\ds\frac{\partial}{\partial x^\alpha}\Big\vert_{x} \left(\frac{\displaystyle\partial
L}{\displaystyle
\partial v^i_\alpha}\Big\vert_{\phi^{(1)}(x)} \right)= \frac{\displaystyle \partial
L}{\displaystyle
\partial q^i}\Big\vert_{\phi^{(1)}(x)} \;.
\end{equation} Equations (\ref{ELe}) are called    {\bf
 Euler-Lagrange field equations} for the Lagran\-gian function $L$.\qed\end{proof}

\index{Euler-Lagrange field equations}

\section{Euler-Lagran\-ge field equations: $k$-sym\-plec\-tic version}\label{section: Geometric EL equations k-symp}

In this section we give the geometric description of the Euler-Lagrange field equations (\ref{EcEL}) or (\ref{ELe}). In order to accomplish this task it is necessary to introduce some geometric elements associated to a Lagrangian function $L\colon \tkq\to \r$, (see for instance \cite{lr1}).

\subsection{Poincar\'{e}-Cartan forms on the tangent bundle of $k^1$-velocities}\label{section Poincare cartan forms}

\index{Tangent bundle of $k^1$-velocities!Poincar\'{e}-Cartan forms}
\index{Poincar\'{e}-Cartan forms}
    In a similar manner as in the case of Lagrangian Mechanics, the $k$-tangent structure on $\tkq$, allows us to define a family of $1$-forms, $\theta^1_L,\ldots, \theta^k_L$ on $\tkq$ as follows:
        \begin{equation}\label{thetaAtkq}
            \theta^\alpha_L=dL\circ J^\alpha \,,
        \end{equation}
    where $\ak$. Next we define the family $\omega^1_L,\ldots, \omega^k_L$ of presymplectic forms on $\tkq$ by
        \begin{equation}\label{omegaAtkq}
            \omega_L^\alpha=-\,d\theta_L^\alpha \,,
        \end{equation}
    which will be called \emph{Poincar\'{e}-Cartan forms} on $\tkq$.

    If we consider canonical coordinates $(q^i,v^i_\alpha)$ on $\tkq$, from (\ref{localJA}) and (\ref{thetaAtkq}) we deduce that for $\ak$,
        \begin{equation}\label{LocthetaAtkq}
            \theta_L^\alpha= \ds\frac{\partial L}{\partial v^i_\alpha}\, dq^i  \,,
        \end{equation}
    and so, from  (\ref{omegaAtkq}) and (\ref{LocthetaAtkq}), we obtain
        \begin{equation}\label{LocomegaAtkq}
            \omega_L^\alpha=dq^i \wedge d\left(\ds\frac{\partial
            L}{\partial v^i_\alpha}\right)= \ds\frac{\partial ^2 L}{\partial
            q^j\partial v^i_\alpha}dq^i\wedge dq^j + \ds\frac{\partial ^2 L}{\partial
            v^j_\beta\partial v^i_\alpha}dq^i\wedge dv^j_\beta \; .
        \end{equation}

    An important property of the family of presymplectic forms $\omega^1_L,\ldots, \omega^k_L$ occurs when the Lagrangian is regular.
    \begin{definition}
        A Lagrangian function $L\colon \tkq\to \r$ is said to be \emph{regular} if the matrix
            \[
                \left(\derpars{L}{v^i_\alpha}{v^j_\beta}\right)
            \]
        is regular.
    \end{definition}
    \index{Lagrangian!regular}
\index{Tangent bundle of $k^1$-velocities!regular Lagrangian}
    The regularity condition let us prove the following proposition, see  \cite{MRS-2004}.
        \begin{prop}\label{lag ksymp manifold}
         Given a Lagrangian function on $\tkq$, the following conditions are  equivalent:
        \begin{enumerate}
         \item $L$ is regular.
          \item $(\omega^1_L,\ldots, \omega^k_L, V)$ is a  $k$-symplectic structure on $\tkq$, where
        $$V=\ker (\tau^k)_*=span \left\{\derpar{}{v^i_1},\ldots,
        \derpar{}{v^i_k}\right\}$$  with $1\leq i\leq n$, is the vertical distribution
        of the vector bundle $\tau^k:\tkq\to Q$.
        \end{enumerate}
        \end{prop}

\subsection{Second order partial differential equations on  $T^1_kQ$.} \label{section sopdes tkq}
\index{\textsc{{sopde}}}
    The second geometric notion which we need in our description of the Euler-Lagrange equations is the notion of second order partial differential equation (or {\sc sopde}) on $\tkq$. Roughly speaking, a {\sc sopde} is a $k$-vector field on $\tkq$ whose integral sections are first prolongations of maps $\phi\colon \rk\to Q$.

    In this section is fundamental to recall the notion of $k$-vector field and integral section introduced in section \ref{section k-vector field}. Now, we only consider $k$-vector fields on $M=\tkq$. Thus using local coordinates $(q^i,v^i_\alpha)$ on an open set $T^1_kU$ the local expression of a $k$-vector field $X=(X_1,\ldots, X_k)$ on $\tkq$ is given by
        \begin{equation}\label{tkq: k-vf local}
            X_\alpha=(X_\alpha)^i\derpar{}{q^i}+ (X_\alpha)^i_\beta\derpar{}{v^i_\beta},\quad (\ak)\,.
        \end{equation}

    Let $$\varphi\colon U_0\subset\rk \to \tkq$$    be an integral section of $(X_1,\ldots, X_k)$   with components
    $$\varphi(x)=(\psi^i(x),\psi^i_\alpha(x))\,   .$$

    Then since
        \[
            \varphi_*(x)\Big(\derpar{}{x^\alpha}\Big\vert_{x}\Big)=\derpar{\psi^i}{x^\alpha}\Big\vert_{x}\derpar{}{q^i}\Big\vert_{\varphi(x)} + \derpar{\psi^i_\beta}{x^\alpha}\Big\vert_{x}\derpar{}{v^i_\beta}\Big\vert_{\varphi(x)}
        \]
    the condition of integral section (\ref{integral section expr}) for this case is locally equivalent to the following system of partial differential equations (condition (\ref{integral section equivalence cond}))
        \begin{equation}\label{tkq integral section equivalence cond}
            \derpar{\psi^i}{x^\alpha}\Big\vert_{x}=(X_\alpha)^i(\varphi(x))\,,\quad \derpar{\psi^i_ \beta}{x^\alpha}\Big\vert_{x}=(X_\alpha)^i_\beta(\varphi(x)),
        \end{equation}
    with $\n$ and $1\leq \alpha,\beta\leq k$.

\index{Tangent bundle of $k^1$-velocities!Poincar\'{e}-Cartan forms!sopde}

        \begin{definition} \label{sopde def}
            A  {\bf  second order partial differential equation} (or {\sc sopde} to short)  is  a $k$-vector field  $\mathbf{X}=(X_1,\ldots,X_k)$ on $T^1_kQ$, which is a section of the projection $T^1_k\tau^k:T^1_k(T^1_kQ)\rightarrow T^1_kQ$, i.e.
                \[
                    \tau^k_{\tkq}\circ \mathbf{X}= id_{\tkq} \makebox{ and } T^1_k\tau^k\circ \mathbf{X}= id_{\tkq}\,,
                \]
            where $\tau^k\colon \tkq\to Q$ and $\tau^k_{\tkq}\colon T^1_k(\tkq) \to \tkq$ are the canonical projections.
        \end{definition}

        Let us observe that when  $k=1$ this definition  coincides with the definition of {\sc sode} (second order differential equation), see for instance \cite{lr1}.

        Taking into account the definition of $T^1_k\tau^k$ (see definition \ref{tkq: prolongation expr}), the above definition is equivalent to say that a $k$-vector field $(X_1,\ldots, X_k)$ on $\tkq$ is a {\sc sopde} if and only if
            \[
                (\tau^k)_*({\rm v}_q)(X_\alpha({\rm v}_q))=         v_{\alpha_q}\,,
            \]
        for $\ak$, where ${\rm v}_q=(v_{1_q},\ldots, v_{k_q})\in T^1_kQ$.

        If we now consider the canonical coordinate system $(q^i,v^i_\alpha)$, from (\ref{tkq: k-vf local}) and the definition \ref{sopde def}, the local expression of  a  {\sc sopde} ${\bf X}=(X_1,\ldots, X_k)$ is the following:
            \begin{equation}\label{Locsopdetkq}
                X_\alpha(q^i,v^i_\alpha)= v^i_\alpha\frac{\displaystyle
                \partial} {\displaystyle
                \partial q^i}+
                (X_\alpha)^i_\beta \frac{\displaystyle\partial} {\displaystyle
                \partial v^i_\beta},
            \end{equation}
        where $\ak$ and $(X_\alpha)^i_\beta$ are functions on $T^1_kQ$.

        In the case $k=1$, the integral curves of as {\sc sode} on $TQ$ are lifts to $TQ$ of curves on $Q$. In our case, in order to characterize the integral sections of a {\sc sopde} we consider the definition \ref{phi(1)} of the first prolongation $\phi^{(1)}$ of a map $\phi\colon \rk \to Q$ to $\tkq$.

        Consider a {\sc sopde} $\mathbf{X}=(X_1,\ldots,X_k)$ and a map
            \[
                \begin{array}{rccl}
                \varphi\colon & \r^k & \to &T^1_kQ\\ \noalign{\medskip}
                      &     x  &  \to &    \varphi(x)=(\psi^i(x),\psi^i_\beta(x))
                \end{array}
            \]
         Since  a {\sc sopde} $\mathbf{X}$ is, in particular, a $k$-vector field on $\tkq$, from (\ref{tkq integral section equivalence cond}) and (\ref{Locsopdetkq}) one obtains that $\varphi$ is an integral section of $\mathbf{X}$ if and only if $\varphi$ is a solution of the following system of partial differential equations:
            \begin{equation}\label{solsopde}
                \frac{\displaystyle\partial\psi^i} {\displaystyle\partial x^\alpha}\Big\vert_{x}=v^i_\alpha(\varphi(x))=\psi^i_\alpha(x)\,
                ,\qquad \frac{\displaystyle\partial\psi^i_\beta} {\displaystyle\partial x^\alpha}\Big\vert_{x}=(X_\alpha)^i_\beta(\varphi(x))\, ,
            \end{equation}
         with $\n$ and $\ak$.

            Thus, from (\ref{localphi1}) and (\ref{solsopde}) it is easy to prove the following proposition.
                \begin{prop} \label{sope1}
                Let $\mathbf{X}=(X_1,\ldots,X_k)$ be an integrable  {\sc sopde}.
                \begin{enumerate}
                \item If $\varphi$ is an integral section of
                 ${\bf X}$ then $\varphi=\phi^{(1)}\,,$  where
                 $\phi^{(1)}\colon\rk\to\tkq$ is the first
                 prolongation of the map  $$\xymatrix{\phi:=\tau^k\circ\varphi:\rk\ar[r]^-{\varphi}
                & \tkq\ar[r]^-{\tau^k} &Q}\,.$$ Moreover,  $\phi(x)=(\psi^i(x))$ is
                solution of the
                 system of second order partial differential equations
                \begin{equation}\label{condsisopde}
                \frac{\displaystyle\partial^2 \psi^i} {\displaystyle\partial x^\alpha
                \partial x^\beta       }\Big\vert_{x}= (X_\alpha)^i_\beta(\psi^i(x),
                \ds\frac{\partial\psi^i}{\partial x^\gamma }(x))\,,
                \end{equation}
                with $1\leq i\leq n\, ;
                1\leq \alpha,\beta,\gamma\leq k.$
                \item
                Conversely, if $\phi:\r^k \to Q$, locally given by
                $\phi(x)=(\psi^i(x))$, is a map  satisfying (\ref{condsisopde})
                then $\phi^{(1)}$ is an  integral section of
                $\mathbf{X}=(X_1,\ldots,X_k)$.\end{enumerate} \qed
                \end{prop}

\begin{remark}{\rm
From equation (\ref{condsisopde}) we deduce that, when
the {\sc sopde} ${\bf X}$ is integrable (as a $k$-vector field), we have $(X_\alpha)^i_\beta=(X_\beta)^i_\alpha$
for all  $\alpha, \beta=1,\ldots, k$ and $1\leq i\leq n$. \rqed}
\end{remark}

 The following characterization
 of  {\sc sopde}s on $\tkq$ can be given using the canonical
   $k$-tangent structure  $J^1,\ldots, J^k$ and the canonical  vector fields
  $\Delta_1,\ldots,$ $ \Delta_k$, (these object were introduced in section \ref{section tkq:geometry})
\begin{prop}
Let  ${\bf X}=(X_1,\ldots, X_k)$ be a $k$-vector field  on $\tkq$.
The following conditions are equivalent
\begin{enumerate}
\item ${\bf X}$ is a {\sc sopde}.
\item $J^\alpha(X_\alpha)=\Delta_\alpha,$ for all $1\leq \alpha\leq k$.
\end{enumerate}
\end{prop}
\proof It is an immediate consequence of (\ref{loccanvf}) and
(\ref{Locsopdetkq}). \qed

\subsection{Euler-Lagrange field equations}\label{section EL k-symp eq}

\index{Euler-Lagrange field equations}
In this subsection we describe the Lagrangian formulation of Classical Field Theories using the geometrical elements introduced in previous sections of this book.

In a similar way as in the Hamiltonian case, given a Lagrangian function $L\colon\tkq\to \r$, we now consider the manifold $\tkq$ equipped with the Poincar\'{e}-Cartan forms $(\omega^1_L, \ldots, \omega^k_L)$  defined in section \ref{section Poincare cartan forms}, which allows us to define a $k$-symplectic structure on $\tkq$ when the Lagrangian function is regular.

Denote by $\vf^k_L(T^1_kQ)$ the
set of $k$-vector fields ${\bf X}=(X_1,\dots,X_k)$ in
$T^1_kQ$, which are solutions of the equation
 \begin{equation}
\label{EL k-symp eq} \sum_{\alpha=1}^k
\iota_{X_\alpha}\omega_L^\alpha=\d E_L\, .
 \end{equation}
where $E_L$ is the function on $\tkq$ defined by $E_L= \Delta(L)-L$.

Consider canonical coordinates $(q^i,v^i_\alpha)$ on $\tkq$, then each $ X_\alpha$ is locally given by the expression (\ref{tkq: k-vf local}).

Now, from (\ref{loclvf}) we obtain that the function $E_L$ is locally given   \[
        E_L = v^i_\alpha\derpar{L}{v^i_\alpha} - L
    \]
and then
    \begin{equation}\label{del}
        dE_L=\Big(v^i_\alpha\derpars{L}{q^j}{v^i_\alpha} - \derpar{L}{q^j}\Big) dq^j + v^i_\alpha\derpars{L}{v^i_\alpha}{v^j_\beta}dv^j_\beta\,.
    \end{equation}

    Therefore, from (\ref{LocomegaAtkq}), (\ref{tkq: k-vf local}) and (\ref{del}) one obtains that a $k$-vector field ${\bf X}=( X _1, \ldots , X_k)$  on $\tkq$ is a solution of (\ref{EL k-symp eq})
 if, and only if, the functions $( X_\alpha)^i$ and $( X_\alpha)^i_\beta$ satisfy the following local system of equations
\begin{eqnarray}\label{local EL1}
  \left( \frac{\partial^2 L}{\partial q^i \partial v^j_\alpha} -
   \frac{\partial^2 L}{\partial q^j \partial v^i_\alpha}
\right) ( X_\alpha)^j - \frac{\partial^2 L}{\partial
v_\alpha^i
\partial v^j_\beta} ( X_\alpha)^j_\beta &=&
 v_\alpha^j \frac{\partial^2 L}{\partial q^i\partial v^j_\alpha} - \frac{\partial  L}{\partial q^i } \, ,
\\\label{Local EL2}
\frac{\partial^2 L}{\partial v^j_\beta\partial v^i_\alpha} \, (
X_\alpha)^i
 &=& \frac{\partial^2 L}{\partial v^j_\beta\partial v^i_\alpha} \, v_\alpha^i ,
\end{eqnarray}
where $1\leq \alpha, \beta \leq k$ and $ 1\leq i,j\leq n$.

If the Lagrangian is regular, the above equations  are equivalent to
the equations \begin{eqnarray}\label{locel4} \frac{\partial^2 L}{\partial q^j
\partial v^i_\alpha} v^j_\alpha + \frac{\partial^2 L}{\partial v_\alpha^i\partial
v^j_\beta}( X_\alpha)^j_\beta = \frac{\partial  L}{\partial q^i}\,,
\\
\label{locel3}
 ( X_\alpha)^i= v_\alpha^i\, ,
\end{eqnarray}
where $1\leq \alpha, \beta \leq k$ and $ 1\leq i,j\leq n$.

Thus,  we can state the following theorem.
\begin{theorem}\label{flks}
Let $L:\tkq\to \r$ a Lagrangian  and ${\bf
X}=(X_1,\ldots, X_k)\in \vf^k_L(\tkq)$. Then,
\begin{enumerate}
\item If $L$ is regular then ${\bf
X}=(X_1,\ldots, X_k)$ is a {\sc sopde}.

Moreover if  $\varphi:\rk\to \tkq$ is an integral section of ${\bf X}$,
then the  map  $\phi=\tau^k\circ \varphi\colon \rk \to Q$ is a solution of the    Euler-Lagrange  field equations (\ref{ELe}).

\item If ${\bf
X}=(X_1,\ldots, X_k)$ is  integrable and $\phi^{(1)}:\rk\to\tkq$
is an integral section of ${\bf X}$ then $\phi:\rk\to Q$ is
a solution of the Euler-Lagrange field equations (\ref{ELe}).

\end{enumerate}
\end{theorem}

\proof
    \begin{enumerate}
        \item
        Let $L$ be a regular Lagrangian, then  ${\bf
X}=(X_1,\ldots, X_k)\in \vf^k_L(\tkq)$ if the coefficients of $X$ satisfy (\ref{locel4}) and (\ref{locel3}). The expression (\ref{locel3}) is locally equivalent to say that $\mathbf{X}$ is a {\sc sopde}.

Since in this case $\mathbf{X}$ is a {\sc sopde}, we can apply  Proposition \ref{sope1}, therefore, if $\varphi\colon \rk\to \tkq$ is an integral section of $X$, then $\varphi=\phi^{(1)}$.

Finally, from (\ref{solsopde}) and (\ref{locel4}) one obtains that
$\phi$ is a solution of the Euler-Lagrange equations (\ref{EcEL}).
    \item In this case we suppose that $\phi^{(1)}$ is an integral section of $\mathbf{X}$, then in a similar way that in proposition \ref{sope1}(1), one can prove that the components $\phi^i$, with $1\leq i\leq n$, of $\phi$ satisfy (\ref{condsisopde}).

        Thus from (\ref{solsopde}), (\ref{condsisopde}),  (\ref{local EL1}) and (\ref{Local EL2}) one obtains that $\phi$ is a solution of the Euler-Lagrange equations (\ref{EcEL}).
    \end{enumerate}

\qed

\begin{remark}
{\rm If we write a equation  (\ref{EL k-symp eq})
for the case  $k=1$, we obtain $$\iota_{X} \omega_L = dE_L$$ which is the  equation  of the geometric formulation of the Lagrangian Mechanics in symplectic terms.
\rqed}\end{remark}

\begin{remark}
{\rm
One important difference with the case $k=1$ on the tangent bundle $TQ$ is that for an arbitrary $k$  we can not ensure the unicity of solutions of the equation (\ref{EL k-symp eq}).

When the Lagrangian $L$ is regular, Proposition  \ref{lag ksymp manifold} implies that $(\tkq,\omega^1_L,$ $\ldots,\omega^k_L, V)$ is a $k$-symplectic manifold and the equation (\ref{EL k-symp eq}) is the same that the equation (\ref{ecHksym}) with $M=\tkq$ and $H=E_L$. Thus  from the discussion about existence of solutions of the equation (\ref{ecHksym}) (see section \ref{Section HDW eq: k-symp approach}), we obtain that in this particular case, the set $\vf^k_L(\tkq)$ is nonempty.
\rqed}
\end{remark}

\section[$k$-symplectic Legendre transformation]{The Legendre transformation and the equivalence between $k$-symplectic Hamiltonian and  Lagrangian formulations.}\protect\label{Sec 1.3.}

    In this section we shall describe the connection between the Hamiltonian and Lagrangian formulations of Classical Field Theories in the $k$-symplectic setting. %If one observes the examples of the chapters \ref{chapter Hamiltonian k-symp} and \ref{chapter Lagrangian k-symp} one can see that many examples are the same. This leads us to ask whether there is a relationship between Hamiltonian and Lagrangian formalism.

%    If we think in Classical Mechanics, the Hamiltonian Mechanics and the Lagrangian Mechanics are equivalents when the Lagrangian is a hyperregular function and are transformed one into the other by the Legendre transformation.
%
%    In this section we state a similar result for the $k$-symplectic approach. Firstly we introduce the definition of Legendre transformation in our context.

%    We recall that a $k$-symplectic Lagrangian is a function $L\colon \tkq\to \r$ defined on the tangent bundle of $k^1$-velocities of a manifold $Q$.

\index{Legendre transform}
    \begin{definition}
        Let $L\in \mathcal{C}^\infty(\tkq)$ be a Lagrangian. The \emph{Legendre transformation} for $L$ is the map  $FL:\tkq\to\tkqh$ defined as follows:
            \[
                FL({\rm v}_q)= ( [FL({\rm v}_q)]^1,\ldots ,[FL({\rm v}_q)]^k )
            \]
        where
            \[ [FL({\rm v}_q)]^\alpha(u_{q})=
            \displaystyle\frac{d}{ds}\Big\vert_{s=0}\displaystyle L\left(
            {v_1}_{q}, \dots,{v_\alpha}_{q}+su_{q}, \ldots, {v_k}_{q}
            \right)\,,
            \]
        for  $\ak$ and $u_{q}\in T_{q}Q$, ${\rm v}_q=({v_1}_{q},\ldots, {v_k}_{q})\in \tkq$.
    \end{definition}

    Using natural coordinates $(q^i, v^i_\alpha)$ on $\tkq$ and $(q^i, p^\alpha_i)$ on $\tkqh$, the local expression of the Legendre map is
        \begin{equation}\label{LocLegtkq}
            \begin{array}{rccl}
            FL\colon & \tkq & \to &\tkqh\\\noalign{\medskip}
            &(q^i,v^i_\alpha) & \longrightarrow  &\Big(q^i, \frac{\displaystyle\partial L}{\displaystyle\partial v^i_\alpha } \Big)\, .
            \end{array}
        \end{equation}

    %A direct consequence of the local expression (\ref{LocLegtkq}) of $FL$ is the following: taking into account the identification $\tkqh=T^*Q\oplus \stackrel{k}{\ldots}\oplus T^*Q$ and $\tkq=TQ\oplus\stackrel{k}{\ldots}\oplus TQ$ described in sections \ref{section k-cotangent} and section \ref{section k-tangent}, respectively, the Legendre map can be thought of as the usual Legendre transform between the tangent and cotangent bundle of a manifold on each copy of $TQ$.

    The Jacobian matrix of $FL$ is the following matrix of order $n(k+1)$,
    \[
    \left(
      \begin{array}{cccc}
        I_n & 0 & \cdots & 0 \\
        \derpars{L}{q^i}{v^j_1} & \derpars{L}{v^i_1}{v^j_1} & \cdots & \derpars{L}{v^i_k}{v^j_1} \\
        \vdots & \vdots &  & \vdots \\
        \derpars{L}{q^i}{v^j_k} & \derpars{L}{v^i_1}{v^j_k} & \cdots & \derpars{L}{v^i_k}{v^j_k} \\
      \end{array}
    \right)
    \]
    where $I_n$ is the identity matrix of order $n$ and $1\leq i,j\leq n$. Thus we deduce that $FL$ is a local diffeomorphism if and only if
    \[
    \det\left(
        \begin{array}{ccc}
         \derpars{L}{v^i_1}{v^j_1} & \cdots & \derpars{L}{v^i_k}{v^j_1} \\
         \vdots &  & \vdots \\
         \derpars{L}{v^i_1}{v^j_k} & \cdots & \derpars{L}{v^i_k}{v^j_k} \\
      \end{array}
    \right)\neq 0
    \]
    with $1\leq i,j\leq n$.

\index{Lagrangian!Regular}
\index{Lagrangian!hyperregular}
\index{Lagrangian!Singular}
    \begin{definition}
        A Lagrangian function  $L: T^1_kQ\longrightarrow \r $ is said to be  \emph{regular} (resp. \emph{hyperregular}) if the
        Legendre map $FL$ is a local diffeomorphism (resp. global). In other case  $L$ is said to be  \emph{singular}.
    \end{definition}

    The Poincar\'{e}-Cartan forms $\theta^\alpha_L, \omega^\alpha_L$, with $\ak$ (defined in section \ref{section Poincare cartan forms}) are related with the canonical forms $\theta^\alpha, \omega^\alpha$ of $\tkqh$ (defined in section \ref{section tkqh: canonical forms}), using the Legendre map $FL$.
        \begin{lemma}
            For all $\ak$ one obtains
                \begin{equation}\label{rel tkq-tkqh}
                    \theta^\alpha_L=FL^*\theta^\alpha\,, \quad\omega_L^\alpha=FL^*\omega^\alpha\,.
                \end{equation}
        \end{lemma}
        \proof
        It is a direct consequence of the local expressions (\ref{Locformcan}), (\ref{LocthetaAtkq}) and (\ref{LocomegaAtkq})
  of $\theta^\alpha,\;\omega^\alpha$ and $\omega_L^\alpha$ and the local expression of the Legendre map (\ref{LocLegtkq}).\qed

Consider $V=\ker (\tau^k)_*$ the vertical distribution of the bundle $\tau^k\colon \tkq \to Q$, then we obtain the following characterization of a regular Lagrangian (the proof of this result can be found in \cite{Tesis merino})
    \begin{prop}
        Let $L\in \mathcal{C}^\infty(\tkq)$ be a Lagrangian function. $L$ is regular if and only if $(\omega_L^1,\ldots, \omega_L^k, V)$ is a $k$-symplectic structure on $\tkq$.
    \end{prop}
 %\textcolor[rgb]{1.00,0.00,0.00}{Escribimos la demostraci\'{o}n????}

    Therefore one can state the following theorem:
    \begin{theorem}
        Given a Lagrangian function $L\colon \tkq\to \r$, the following conditions are equivalents:
        \begin{enumerate}
            \item $L$ is regular.
            \item $\det \left(\derpars{L}{v^i_\alpha}{v^j_\beta}\right)\neq 0$ with $1\leq i,j\leq n$ and $1\leq \alpha,\beta\leq k$.
            \item $FL$ is a local $k$-symplectomorphism.
        \end{enumerate}
    \end{theorem}

    Now we restrict ourselves to the case of hyperregular Lagrangians. In this case the  Legendre map $FL$ is a global  diffeomorphism and thus we can define a Hamiltonian  function    $H: (T^1_k)^*Q \to \r$ by $$H=(FL^{-1})^*E_L=E_L \circ FL^{-1}$$ where $FL^{-1}$ is the inverse  diffeomorphism of $FL$.

    In these conditions, we can state the  equivalence between both Hamiltonian and Lagrangian formalisms.

    \begin{theorem}\label{equivalence k-symp}
        Let $L\colon \tkq\to \r$ be a hyperregular Lagrangian then:
        \begin{enumerate}
            \item $\mathbf{X}=(X_1,\ldots, X_k)\in \vf^k_L(\tkq)$ if and only if $(T^1_kFL)(\mathbf{X})=(FL_*X_1,\ldots,$ $ FL_*X_k)\in \vf^k_H(\tkqh)$ where $H=E_L\circ FL^{-1}$.
            \item There exists a bijective correspondence between the set of maps $\phi\colon \rk\to Q$ such that $\phi^{(1)}$ is an integral section of some $(X_1,\ldots, X_k)\in \vf^k_L(\tkq)$ and the set of maps $\psi\colon \rk\to \tkqh$, which are integral section of some  $(Y_1,\ldots, Y_k)\in\vf^k_H(\tkqh)$, being $H=(FL^{-1})^*E_L$.

            %If ${\bf X}=(X^1,\dots,X^k)$ is integrable and $\phi^{(1)}$ is an integral section of ${\bf X}$, then $\varphi=FL \circ \phi^{(1)}$ is an integral section of $(FL_*X_1,\ldots, FL_*X_k)$ and thus $\varphi=FL \circ \phi^{(1)}$ is a solution of the Hamilton-De Donder- Weyl equations.
        \end{enumerate}
    \end{theorem}
\proof

\begin{enumerate}
    \item  Given $FL$ therefore we can consider the canonical prolongation $T^1_kFL$ following the definition of the section \ref{section tkq: prolongation}. Thus given a $k$-vector field $\mathbf{X}=(X_1,\ldots, X_k)\in\vf_L^k(\tkq)$, one can define a $k$-vector field on $\tkqh$ using the following diagram
            \[
                \xymatrix{
                \tkq\ar[r]^-{FL}\ar[d]_-{\mathbf{X}} & \tkqh\ar[d]^-{(T^1_kFL)(\mathbf{X})}\\
                T^1_k(\tkq) \ar[r]^-{T^1_kFL} & T^1_k(\tkqh)
                }
            \]
            that is, for each $\ak$, we consider the vector field on $\tkqh$, $FL_*(X_\alpha)$.

            We now consider the function $H=E_L\circ FL^{-1}= (FL^{-1})^*E_L$, then
            \[
                (T^1_kFL)(\mathbf{X})=(FL_*(X_1),\ldots, FL_*(X_k))\in \vf^k_H(\tkqh) \,
             \]
             if and only if
             \[
                \ds\sum_{\alpha=1}^k\iota_{FL_*(X_\alpha)}\omega^\alpha - d\Big((FL^{-1})^*E_L\Big) =0\,
            \]

            Since $FL$ is a diffeomorphism this is equivalent to
            \[
                0= FL^*\Big(\ds\sum_{\alpha=1}^k\iota_{FL_*X_\alpha}\omega^\alpha - d(FL^{-1})^*E_L\Big) = \ds\sum_{\alpha=1}^k\iota_{X_\alpha}(FL)^*\omega^\alpha-dE_L\,,
            \]
            and from (\ref{rel tkq-tkqh}), this fact occurs if and only if  $\mathbf{X}\in\vf^k_L(\tkq)$.

 %           \textcolor[rgb]{0.00,0.50,0.00}{Explicar m\'{a}s.}

Finally, observe that since $FL$ is a diffeomorphism, $T^1_kFL$ is also a diffeomorphism, and then any $k$-vector field on $\tkqh$ is of the type $T^1_kFL(\mathbf{X})$ for some $\mathbf{X}\in \vf^k(\tkq)$.

        \item Let $\phi\colon \rk\to Q$ be a map such that its first prolongation $\phi^{(1)}$ is an integral section of some $\mathbf{X}=(X_1,\ldots, X_k)\in \vf^k_L(\tkq)$, then the map $\psi=FL\circ \phi^{(1)}$ is an integral section of $T^1_kFL(\mathbf{X})=(FL_*(X_1),\ldots, FL_*(X_k))$. Since we have proved that  $T^1_kFL(\mathbf{X})\in \vf^k_H(\tkqh)$, we obtain the first part of the item $2$.

%            \textcolor[rgb]{0.00,0.50,0.00}{Explicar m\'{a}s.}

The converse is similar, if we consider that any $k$-vector field on $\tkqh$ is of the type $T^1_kFL(\mathbf{X})$ for some $\mathbf{X}\in \vf^k(\tkq)$. Thus given $\psi\colon \rk\to \tkqh$ integral section of any $(Y_1,\ldots, Y_k)\in\vf^k_H(\tkqh)$, there exists a $k$-vector field $\mathbf{X}\in \vf^k_L(\tkq)$ such that $T^1_kFL(\mathbf{X})=(Y_1,\ldots, Y_k)$. Finally, the map $\psi$ corresponds with $\phi^{(1)}$ where $\phi=\pi^k\circ \psi$.
\end{enumerate}\qed

 \newpage
\mbox{}
\thispagestyle{empty} % para que no se numere esta p\'{a}gina

\chapter{Examples}\label{chapter: K-sympExamples}
In this section we describe several physical examples using the $k$-symplectic formulation developed in this part of the book. In \cite{MSV-2009} one can find several of these examples.
 Previously, we  recall the geometric version of the  Hamiltonian and
 Lagrangian approaches for Classical Field Theories and its correspondence with the case $k=1$.
\begin{table}[h]
\centering
\scalebox{1}{
\begin{tabular}{|c|c|c|} \hline
 & \begin{tabular}{c}
 \\   $k$-symplectic formalism \\ \quad
 \end{tabular}
  & \begin{tabular}{c}
 Symplectic formalism ($k=1$)\\
  (Classical Mechanics)
 \end{tabular}\\\hline
 \begin{tabular}{c}
  Hamiltonian  \\ formalism \end{tabular} & $\begin{array}{c}\\
   \ds\sum_{\alpha=1}^k i_{X_\alpha}\omega^\alpha
= dH\\  \noalign{\bigskip}
  {\bf X} \in \vf^k(M)\\   \noalign{\medskip}M \mbox{   $k$-symplectic manifold} \\
  \quad
 \end{array}$ & $\begin{array}{c}\\
    i_{X}\omega
= dH \\  \noalign{\bigskip}
  X \in \vf(M)\\   \noalign{\medskip}M \mbox{ symplectic manifold} \\
  \quad  \\
  \quad
 \end{array}$\\\hline
% \end{tabular}
% \end{center}
%
% \begin{center}
%\begin{tabular}{|c|c|c|} \hline
% & \begin{tabular}{c}
% \\   $k$-symplectic formalism \\ \quad
% \end{tabular}
%  & \begin{tabular}{c}
% Symplectic formalism ($k=1$)\\
%  (Classical Mechanics)
% \end{tabular}\\\hline
\begin{tabular}{c}
Lagrangian \\ formalism\end{tabular} & $\begin{array}{c}\\
   \ds\sum_{\alpha=1}^k \, i_{X_\alpha} \omega_L^\alpha =\,
dE_L  \\  \noalign{\bigskip} {\bf X}  \in \vf^k(\tkq)\\\quad
 \end{array}$ & $\begin{array}{c}\\
   i_{X} \omega_L =\,
dE_L   \\  \noalign{\bigskip} X \in \vf(TQ)\\\quad
 \end{array}$\\\hline
\end{tabular}}
\caption{$k$-symplectic approach vs symplectic approach.}
\end{table}

As before, the canonical coordinates in $\rk$ are denoted by $(x^1,\ldots, x^k)$. Moreover we shall use the following notation for the partial derivatives of a map $\phi\colon \rk\to Q$:
        \begin{equation}\label{notation partials}
            \partial_\alpha\phi^i=\derpar{\phi^i}{x^\alpha},\quad \partial_{\alpha\beta}\phi^i=\derpars{\phi^i}{x^\alpha}{x^\beta}\,,
        \end{equation}
    where $1\leq \alpha,\beta\leq k$ and $1\leq i\leq n$.

%\textcolor[rgb]{0.00,0.50,0.00}{En cada ejemplo discutir la integrabilidad de $X$.}

\section{Electrostatic equations}\label{example k-symp: hamiltonian electrostatic}

\index{Electrostatic}

            We now consider the study of electrostatic in a $3$-dimensional manifold $M$  with coordinates $(x^1,x^2,x^3)$, for instance $M=\r^3$. We assume that $M$ is a  Riemannian manifold with a  metric $g$ with components $g_{\alpha\beta}(x)$ where $1\leq \alpha,\beta\leq 3$.

\index{Electrostatic!equations}

            The equations of electrostatics are (see \cite{du,KT-79}):
                \begin{equation}\label{electrostatic eq}
                    \begin{array}{l}
                        E=\star d\psi,\\
                        dE=-4\pi\rho\,,
                    \end{array}
                \end{equation}
            where $\star$ is the Hodge operator\footnote{In general, on an orientable $n$-manifold with a Riemannian metric $g$, the {\it Hodge operator}\index{Hodge operator}
$\star\colon \Omega^k(M)\to \Omega^{n-k}(M)$ is a linear operator that for every $\nu , \eta\in
\Omega^k(M)$
$$\nu\wedge \star \eta= g(\nu,\eta) \d vol_g,$$
where $\d vol_g$ is the Riemann volume.
In local coordinates we have $$\nu=\nu_{i_1\ldots i_k}\d x^{i_1}\wedge\d
x^{i_k},\eta=\eta_{j_1 \ldots j_k}\d x^{j_1}\wedge\d x^{j_k},\,
g(\nu,\eta)=\nu_{i_1\ldots i_k}\eta_{j_1\ldots j_k}g^{i_1j_1} \ldots
g^{i_kj_k}\,,$$ and $\d vol_g=\sqrt{|det(g_{ij})|}dx^1\wedge\ldots\wedge dx^n$, $(g^{ij})$
being the inverse of the metric matrix $(g_{ij})$. For more details see, for instance,
\cite{SCYYQLY-2008}.} associated with the metric $g$, $\psi$ is a scalar field $\psi\colon \r^3\to \r$ given the electric potential on $\r^3$ and $E=(\psi^1,\psi^2,\psi^3)\colon \r^3\to \r^3$ is a vector field which gives the electric field on $\r^3$ and such that can be interpreted by the $2$-form on $\r^3$ given by
            \[
                E=\psi^1dx^2\wedge dx^3 + \psi^2dx^3\wedge dx^1 + \psi^3dx^1\wedge dx^2\,,
            \]
            and $\rho$ is the $3$-form on $\r^3$ representing a fixed charge density
                \begin{equation}\label{density charge}
                    \rho(x)=\sqrt{g}r(x) dx^1\wedge dx^2\wedge dx^3\,,
                \end{equation}
            being $g=|\det{g_{\alpha\beta}}|$.

            In terms of local coordinates the above system of equations (\ref{electrostatic eq}) reads:
                \begin{equation}\label{local electrostatic eq}
                    \begin{array}{l}
                        \psi^{\alpha} =\sqrt{g}g^{\alpha\beta}\derpar{\psi}{x^\beta}\,,\\\noalign{\medskip}
                        \ds\sum_{\alpha=1}^k\derpar{\psi^\alpha}{x^\alpha}= -4\pi \sqrt{g}r\,,
                    \end{array}
                \end{equation}
            where $r$ is the scalar function defined by the equation $r=\star\rho$, or equivalently, by (\ref{density charge}).

                Suppose that $g$ is the euclidean metric on $\r^3$, thus the above equations can be written as follows:
                    \begin{equation}\label{electrosEC}
                        \begin{array}{rcl}
                        \psi^\alpha &=&
                            \frac{\displaystyle\partial\psi}{\displaystyle\partial
                            x^\alpha},
                            \\ \noalign{\medskip}
                            -\Big(\frac{\displaystyle\partial\psi^1}{\displaystyle\partial
                            x^1}+
                            \frac{\displaystyle\partial\psi^2}{\displaystyle\partial x^2}+
                            \frac{\displaystyle\partial\psi^3}{\displaystyle\partial x^3}\Big)& =
                            &4{\pi} r\,.
                        \end{array}
                   \end{equation}

                As we have seen in section \ref{section: electrostatic}, these equations can be obtained from the $3$-symplectic equation
                \[
                    \iota_{X_1}\omega^1 + \iota_{X_2}\omega^2+\iota_{X_3}\omega^3=dH
                \]
                being $H\colon (T^1_3)^*\r\to \r$ the Hamiltonian defined in (\ref{Helectros}).
\section{Wave equation}\label{example k-symp: hamiltonian wave}

\index{Wave equation}
            Consider the $(n+1)$-symplectic Hamiltonian  equation
                \begin{equation}\label{wave n+1}
                    \displaystyle\sum_{\alpha=1}^{n+1}\iota_{X_\alpha}\omega^\alpha =\d H\,,
                \end{equation}
            associated to the Hamiltonian function
                \begin{equation}\label{wave k-symp eq}
                    \begin{array}{rccl}
                        H\colon & (T^1_{n+1})^*\r & \to & \r\\\noalign{\medskip}
                            & (q,p^1,\ldots, p^{n+1}) & \mapsto & \displaystyle\frac{1}{2}\Big((p^{n+1})^2 - \displaystyle\frac{1}{c^2}\sum_{\alpha=1}^n(p^\alpha)^2\Big)\,.
                    \end{array}
                \end{equation}
           where $(q,p^1,\ldots, p^{n+1})$ are the canonical coordinates on $(T^1_{n+1})^*\r$  introduced in section \ref{section k-cotangent}.

            Let $\mathbf{X}= (X_1, \ldots, X_{n+1})$ be an  integrable $(n+1)$-vector field which is a solution of the equation (\ref{wave n+1}); then since
                \[
                    \ds\frac{\partial H}{\partial q}=0,\quad \ds\frac{\partial H}{\partial p^\alpha}=-\ds\frac{1}{c^2}p^\alpha,\, 1\leq \alpha\leq n\quad \text{and}\quad \ds \frac{\partial H}{\partial p^{n+1}}=p^{n+1}
                \]
             we deduce, from (\ref{ecHDWloc}),   that each $X_\alpha$ is locally given by
                \begin{equation}\label{vector field wave}
                    \begin{array}{ll}
                        X_\alpha& = -\ds\frac{1}{c^2}p^\alpha \frac{\partial}{\partial q} + \big(X_\alpha\big)^\beta\frac{\partial}{\partial p^\beta},\quad 1\leq \alpha \leq n,\\\noalign{\medskip}
                        X_{n+1}&= p^{n+1}\ds\frac{\partial}{\partial q} + \big(X_{n+1}\big)^\beta\ds\frac{\partial}{\partial p^\beta}\,,
                    \end{array}
                \end{equation}
            and the components $(X_\alpha)^\beta$ satisfy $\ds\sum_{\alpha=1}^{n+1}\big(X_\alpha\big)^\alpha=0$.

            \begin{remark}
                {\rm
                    In this particular case the integrability condition of $\mathbf{X}$ is equivalent to the following local conditions:
                    \begin{align*}
                        (X_\alpha)^\beta &= (X_\beta)^\alpha,\\
                            (X_\alpha)^{n+1}&= -\nicefrac{1}{c^2}(X_{n+1})^\alpha,\\
                            X_\alpha\Big((X_\beta)^\gamma\Big)&= X_\beta\Big((X_\alpha)^\gamma\Big),\\
                            X_\alpha\Big((X_{n+1})^\gamma\Big)&= X_{n+1}\Big((X_\alpha)^\gamma\Big),
                    \end{align*}
                    where $1\leq\alpha,\beta\leq n$ and $1\leq \gamma\leq n+1$.
                        %\[
%                            \begin{array}{l}
%                            (X_\alpha)^\beta= (X_\beta)^\alpha,\qquad 1\leq\alpha,\beta\leq n\,,\\\noalign{\medskip}
%                            (X_\alpha)^{n+1}= -\nicefrac{1}{c^2}(X_{n+1})^\alpha,\qquad \ak\,,\\\noalign{\medskip}
%                            X_\alpha\Big((X_\beta)^\gamma\Big)= X_\beta\Big((X_\alpha)^\gamma\Big),\qquad 1\leq\alpha,\beta\leq n,\,1\leq \gamma\leq n+1\,,\\\noalign{\medskip}
%                            X_\alpha\Big((X_{n+1})^\gamma\Big)= X_{n+1}\Big((X_\alpha)^\gamma\Big),\qquad \ak,\,1\leq \gamma \leq n+1\,.
%                            \end{array}
%                        \]}
                \rqed}
            \end{remark}

            We now consider  an integral section
                \[
                    (x^1,\ldots, x^n, t) \to (\psi(x^1,\ldots, x^n, t),\psi^1(x^1,\ldots, x^n, t), \ldots, \psi^{n+1}(x^1,\ldots, x^n, t))
                \]
            of the $(n+1)$-vector field $\mathbf{X}= (X_1, \ldots, X_{n+1})\in \vf^{n+1}_H(\r)$. From
            (\ref{vector field wave}) one deduce that that integral section satisfies
                \begin{eqnarray}
                    \label{wave 1}
                        \psi^\alpha &=&-c^2\ds\frac{\partial \psi}{\partial x^\alpha},\quad 1\leq \alpha \leq n,\\\noalign{\medskip}
                    \label{wave 2}
                        \psi^{n+1} &=& \ds\frac{\partial \psi}{\partial t}\,,\\\noalign{\medskip}
                    \label{wave 3}
                        0&=&\ds\sum_{\alpha=1}^n\frac{\partial \psi^\alpha}{\partial x^\alpha} + \ds\frac{\partial \psi^{n+1}}{\partial t}\; .
                \end{eqnarray}

            Finally, if we consider the identities (\ref{wave 1}) and (\ref{wave 2}) in (\ref{wave 3}) one deduces that $\psi$ is a solution of
                \begin{equation}\label{wave eq}
                    \ds\frac{\partial^2 \psi}{\partial t^2}=c^2 \nabla^2\psi\,,
                \end{equation}
            where $\nabla^2$ is the (spatial) Laplacian, i.e. $\psi$ is a solution of the $n$-dimensional wave equation. Let us recall that a solution of this equation is a scalar function $\psi=\psi(x^1,\ldots, x^n,t)$ whose values model the height of a wave at the position $(x^1,\ldots, x^n)$ and at the time $t$.

The Lagrangian counterpart of this example is the following.  Consider the Lagrangian $(n+1)$-symplectic equation
                \begin{equation}\label{n1sym}
                    \displaystyle\sum_{\alpha=1}^{n+1}\iota_{X_\alpha}\omega_L^\alpha =\d E_L\,,
                \end{equation}
            associated to the Lagrangian function
                \begin{equation}\label{wave k-symp lag}
                    \begin{array}{rccl}
                        L\colon & (T^1_{n+1})^*\r & \to & \r\\\noalign{\medskip}
                            & (q,v_1,\ldots, v_{n+1}) & \mapsto & \displaystyle\frac{1}{2}\Big((v_{n+1})^2 - c^2\sum_{\alpha=1}^nv_\alpha^2\Big)\,.
                    \end{array}
                \end{equation}
           where $(q,v_1,\ldots, v_{n+1})$ are the
           canonical coordinates on $T^1_{n+1}\r$.

           %Since
%                \[
%                    \ds\frac{\partial L}{\partial q}=0,\quad \ds\frac{\partial L}{\partial v_\alpha}=-\ds\frac{1}{c^2}v_\alpha,\, 1\leq \alpha\leq n,\quad \ds \frac{\partial L}{\partial v_{n+1}}=v_{n+1}
%                \]

            Let $\mathbf{X}= (X_1, \ldots, X_{n+1})$ be an integrable $(n+1)$-vector field solution of the equation (\ref{n1sym}), then
            \begin{equation}\label{xan11}
                        X_\alpha= v_\alpha \derpar{}{q^i} + (X_\alpha)_\beta\derpar{}{v_\beta}
                    \end{equation}
            and the components $(X_\alpha)_\beta$ satisfy the
             equations (\ref{locel4}), which in this case are
             \begin{equation}\label{xan110}\begin{array}{ccl}
             0&=&\ds\sum_{\alpha,\beta=1}^{n+1}  \frac{\partial^2 L}{\partial v_\alpha\partial
v_\beta}( X_\alpha)_\beta  \\ \noalign{\medskip}
&=&\ds\sum_{\alpha,\beta=1}^{n}  \frac{\partial^2 L}{\partial v_\alpha\partial
v_\beta}( X_\alpha)_\beta +  \derpars{L}{ v_{n+1}}{v_{n+1}}(X_{n+1})_{n+1}\\ \noalign{\medskip}
&=&- c^2\ds\sum_{\alpha=1}^n(X_\alpha)_\alpha+ (X_{n+1})_{n+1}
           \end{array}  \end{equation}
     since
                $$\frac{\partial^2 L}{\partial v_\alpha\partial
v_\beta}  = - c^2\delta^{\alpha\beta}, \quad 1\leq \alpha,\beta\leq n,
\quad  \frac{\partial^2 L}{\partial v_\alpha\partial
v_{n+1}}=0, \quad \frac{\partial^2 L}{\partial v_{n+1}\partial v_{n+1}}=1\; .
$$
Now, if
                $$ \begin{array} {cccl}\phi^{(1)}: & \r^{n+1}
                 & \longrightarrow & T^1_{n+1}\mathbb{R} \\ \noalign{\medskip}
                 & x & \to & \phi(x)=\Big(\phi(x),\derpar{\phi}{x^\alpha} (x))\Big)\;
                \end{array}$$ is an integral section of $\mathbf{X}$,  then we deduce from (\ref{condsisopde}) and (\ref{xan110}) that
                    $\phi:\r^{n+1} \to \mathbb{R}$  is a solution of the equations (\ref{wave eq}).

        \section{Laplace's equations} \label{example k-symp: hamiltonian laplace}

\index{Laplace's equation}
            On the $n$-symplectic manifold $(T^1_n)^*\r$ we define the Hamiltonian function
                \[
                    \begin{array}{rccl}
                        H\colon& (T^1_n)^*\r & \to & \r\\\noalign{\medskip}
                        & (q,p^1,\ldots, p^n) & \mapsto & \ds\frac{1}{2}\Big((p^1)^2+\ldots + (p^n)^2\Big)\,,
                    \end{array}
                \]
            where $(q, p^1,\ldots, p^n)$ are  canonical coordinates on $(T^1_n)^*\r$. Then
                \begin{equation}\label{partials Ham Laplace}
                    \derpar{H}{q}=0,\quad \derpar{H}{p^\alpha}= p^\alpha,
                \end{equation}
                with $1\leq \alpha \leq n$.

            The $n$-symplectic Hamiltonian  equation (\ref{ecHksym}) associated with $H$ is
                \begin{equation}\label{Laplace k-symp}
                    \iota_{X_1}\omega^1+\ldots + \iota_{X_n}\omega^n=dH\,.
                \end{equation}

             From  (\ref{ecHDWloc}) and (\ref{partials Ham Laplace}) we deduce that an integrable $n$-vector field solution of (\ref{Laplace k-symp}),      has the following local expression:
                \begin{equation}\label{local vf Laplacian}
                    X_\alpha=p^\alpha\derpar{}{q}+ (X_\alpha)^\beta\derpar{}{p^\beta}\, ,
                \end{equation}
           and its components satisfy the following equations
                \begin{eqnarray}\label{sum0}
            0 = \ds\sum_{\alpha=1}^n(X_\alpha)^\alpha\, ,\\\noalign{\medskip}
                (X_\alpha)^\beta=(X_\beta)^\alpha,\,\\\noalign{\medskip}
                X_\alpha\Big((X_\beta)^\gamma\Big)= X_\beta\Big((X_\alpha)^\gamma\Big)\,,
                \end{eqnarray}
        with $1\leq \alpha,\beta,\gamma\leq n$. Let us observe that the two last groups of equations of (\ref{sum0}) are the integrability condition of the $n$-vector field $\mathbf{X}=(X_1,\ldots, X_n)$.

             If \[
                    \begin{array}{ccccl}
                     \varphi &:& \r^3 &  \longrightarrow & (T^1_3)^*\r \\ \noalign{\medskip}
                                            & &   x  & \to &  \varphi(x)=(\psi(x),\psi^1(x),\psi^2(x),\psi^3(x))
                   \end{array} \]
                is an integral section of $(X_1,\ldots, X_n)$, then from (\ref{local vf Laplacian}) and (\ref{sum0})  we obtain that
                \[
                    \begin{array}{l}
                        \psi^\alpha=\derpar{\psi}{x^\alpha}\,,\\\noalign{\medskip}
                        \ds\sum_{\alpha=1}^n\derpar{\psi^\alpha}{x^\alpha}=0\,.
                    \end{array}
                \]
             Therefore, $\psi$ is a solution of
                \begin{equation}\label{Laplace}
                    \derpar{^2\psi}{(x^1)^2}+\ldots + \derpar{^2\psi}{(x^n)^2}=0\,,
                \end{equation}
             that is, $\psi$ is a solution of Laplace's equations \cite{Olver,olver}.

Let $(X_1, \ldots,X_n)$ be a $n$-vector field on $T^1_n\r$, with coordinates
$(q,v_1,\ldots,v_n)$, which is  a solution of
\begin{equation}\label{globfor1}
 \iota_{X_1}
          \omega_L^1 + \ldots +\iota_{X_n}
          \omega_L^n= dE_L\, ,
\end{equation}
where $L$ is the regular Lagrangian
$$\begin{array}{rccl}
L\colon  & T^1_n\r &\to &\r\\\noalign{\medskip}
 & (q, v_1,\ldots,v_n) &\mapsto & \ds\frac{1}{2}
 ( (v_1)^2+ \ldots +(v_n)^2)\,.\end{array}$$

\noindent   From (\ref{solsopde3}), and taking into account that
 \[\derpar{L}{q}=0\;, \quad \derpar{L}{v_\alpha}=v_\alpha\,,\] with $\ak$, we obtain that if    $\phi$ is a solution
 of the $n$-vector field  $(X_1,\ldots,X_n)$ on $T^1_n\r$,  then    $\phi$ satisfies
 \[\partial_{11}\phi + \ldots+ \partial_{nn}\phi=0\,,\]
 or equivalently
 \[
    \nabla^2\phi=0\,,
 \]
  which is the  {\it Laplace equation} (\ref{Laplace}).  Thus equations (\ref{globfor1})
  can be considered as the geometric version of Laplace's equations.

             \begin{remark}
             {\rm
                The solutions of the Laplace equations are important in many fields of science, for instance, electromagnetism, astronomy and fluid dynamics, because they describe the behavior of electric, gravitational and fluid potentials. The solutions of Laplace's equations are called harmonic functions.
               \rqed }
             \end{remark}

            %First we turn to the example with which we began this chapter: the wave equation, but now we consider this equation over any dimension $n$. Let us remember that the wave equation in dimension $n$ is the equation
%                \[
%                    \ds\frac{\partial^2 \phi}{\partial t^2}=c^2 \nabla^2\phi
%                \]
%            where $\nabla^2$ is the (spatial) Laplacian an $c$ is a constant related with the velocity of the displacement of the wave. A solution of this equation is a scalar function $\phi=\phi(t,x_1,\ldots, x_n)$ whose values model the height of a wave.

            \section{Sine-Gordon equation}\label{example k-symp: hamiltonian SG}
\index{Sine-Gordon equation}
                Define the  Hamiltonian function
                    \[
                        \begin{array}{rccl}
                            H\colon & (T^1_2)^*\r & \to & \r\\\noalign{\medskip}
                            &(q,p^1,p^2)&\mapsto & \ds\frac{1}{2}\Big ((p^1)^2 - \frac{1}{a^2}(p^2)^2\Big)-\Omega^2 \cos{q}
                        \end{array}\,,
                    \]
                $a^2$ and $\Omega^2$ being two positive constants.

                Consider the $2$-symplectic Hamiltonian equations associated to this Hamiltonian, i.e.,
                    \begin{equation}\label{Ham 2-symp}
                        \iota_{X_1}\omega^1 + \iota_{X_2}\omega^2 = dH\,,
                    \end{equation}
                and let $\mathbf{X}=(X_1,X_2)$ be a solution

                In  canonical coordinates $(q,p^1,p^2)$ on $(T^1_2)^*Q$, a solution $\mathbf{X}$  has the following local expression
                    \begin{equation}\label{local SG vf}
                        \begin{array}{l}
                            X_1= p^1\derpar{}{q} + (X_1)^1\derpar{}{p^1} + (X_1)^2\derpar{}{p^2},\\\noalign{\medskip}
                            X_2= -\ds\frac{1}{a^2}p^2\derpar{}{q} + (X_2)^1\derpar{}{p^1} + (X_2)^2\derpar{}{p^2},
                        \end{array}
                    \end{equation}
                where the functions $(X_\alpha)^\beta$ satisfy $(X_1)^1+(X_2)^2=-\Omega^2\sin{q}$.

                If $(X_1,X_2)$ is an integrable $2$-vector field, that is $[X_1,X_2]=0$, then the functions $(X_1)^2 $ and $(X_2)^1$ satisfy $(X_2)^1 = -\nicefrac{1}{a^2}(X_1)^2$.

                Let $\varphi\colon \r^2\to (T^1_2)^*\r,\;\varphi(x)=(\psi(x),\psi^1(x),\psi^2(x))$ be an integral section of the $2$-vector field $\mathbf{X}$. Then from (\ref{local SG vf}) one has that $\varphi$ satisfies
                    \begin{eqnarray}
                        \label{SG eq 1}
                            \psi^1 = \derpar{\psi}{x^1}\,,
                        \label{SG eq 2}
                            \psi^2= -a^2\derpar{\psi}{x^2}\,,
                        \label{SG eq 3}
                            \derpar{\psi^1}{x^1}+\derpar{\psi^2}{x^2}=-\Omega^2\sin{\psi}\,,
                    \end{eqnarray}
                and hence $\psi\colon\mathbb{R}^2\to\mathbb{R}$ is a solution of
                    \begin{equation}\label{sine-gordon}
                        \ds\frac{\partial^2\psi}{\partial (x^1)^2} - a^2 \ds\frac{\partial^2\psi}{\partial (x^2)^2} + \Omega^2\sin{\psi}=0\,,
                   \end{equation}
                that is, $\psi$ is a solution of the \textit{Sine-Gordon equation} (see \cite{saletan}).
                \begin{remark}
                    {\rm
                    The Sine-Gordon equation were know in the 19th century, but the equation grew greatly in importance when it was realized that it led to solutions ``kink'' and ``antikink'' with the collisional properties of solitons \cite{PS-1962}. This equation also appears in other physical applications \cite{BEMS-1971,BS-1981,Davydov-1985,GJM-1979,IR-2000}, including the motion of rigid pendula attached to a stretched wire, and dislocations in crystals.\rqed}
                \end{remark}

            This equation (\ref{sine-gordon}) can be obtained also from the Lagrangian approach if we consider the $2$-symplectic equation
\begin{equation}\label{globfor}
 \iota_{X_1}
          \omega_L^1 + \iota_{X_2}
          \omega_L^2= dE_L\, ,
\end{equation}where $(X_1,X_2)$ is a $2$-vector field on $T^1_2\r$ and the Lagrangian is the function $$L(q,v_1,v_2)=\ds\frac{1}{2}((v_1)^2-a^2(v_2)^2)-\Omega^2(1-\cos(q))\,$$  $a^2$ and $\Omega^2 $ being two
positive constants.

Thus we have
\begin{equation}\label{partiall SG}
\derpar{L}{q}= -\Omega^2\sin (q)\;, \quad \derpar{L}{v_1}=v_1 \;, \quad \derpar{L}{v_2}=-a^2v_2\,.
\end{equation}

From (\ref{solsopde3}) and (\ref{partiall SG}), we know that if $\phi$ is a solution of
$(X_1,X_2)$ then
\[ 0=\partial_{11}\phi-a^2\partial_{22} \phi+ \Omega^2\sin\phi\,,\] that is, $\phi$ is a solution of the
Sine-Gordon equation (\ref{sine-gordon}).

            \section{Ginzburg-Landau's equation}\label{example k-symp: hamiltonian GL}

\index{Ginzburg-Landau's equation}
                Let us consider the Hamiltonian function
                    \[
                        \begin{array}{rccl}
                            H\colon & (T^1_2)^*\r & \to & \r\\\noalign{\medskip}
                            &(q,p^1,p^2)&\mapsto & \ds\frac{1}{2}\Big ((p^1)^2 - \frac{1}{a^2}(p^2)^2\Big)-\lambda(q^2-1)^2
                        \end{array}
                    \]
                where $a$ and $\lambda$ are supposed to denote constant quantities. Then
                    \[
                        \derpar{H}{q}=-4\lambda q(q^2-1),\quad \derpar{H}{p^1}=p^1,\quad \derpar{H}{p^2}=-\ds\frac{1}{a^2}p^2\,.
                    \]

                Consider the $2$-symplectic Hamiltonian  equations associated to this Hamiltonian, i.e.,
                    \begin{equation}\label{Ham 2-symp GL}
                        \iota_{X_1}\omega^1 + \iota_{X_2}\omega^2 = dH\,,
                    \end{equation}
                and let $\mathbf{X}=(X_1,X_2)$ be a solution.

                In the canonical coordinates $(q,p^1,p^2)$ on $(T^1_2)^*\r$, a $2$-vector field $\mathbf{X}$ solution  of (\ref{Ham 2-symp GL}) has the following local expression
                    \begin{equation}\label{local GL vf}
                        \begin{array}{l}
                            X_1= p^1\derpar{}{q} + (X_1)^1\derpar{}{p^1} + (X_1)^2\derpar{}{p^2},\\\noalign{\medskip}
                            X_2= -\ds\frac{1}{a^2}p^2\derpar{}{q} + (X_2)^1\derpar{}{p^1} + (X_2)^2\derpar{}{p^2},
                        \end{array}
                    \end{equation}
                where the functions $(X_\alpha)^\beta$ satisfy $(X_1)^1+(X_2^2)=4\lambda q(q^2-1)$.

                A necessary condition for the integrability of the $2$-vector field $(X_1,X_2)$ is that $(X_2)^1=-\nicefrac{1}{a^2}(X_1)^2$.

                Let $\varphi\colon \r^2\to (T^1_2)^*\r$ be an integral section of the $2$-vector field $\mathbf{X}$ with components $\varphi(x)=(\psi(x),$ $\psi^1(x),\psi^2(x))$. Then from (\ref{local GL vf}) one obtains that $\varphi$ satisfies
                    \begin{eqnarray}
                        \label{GL eq 1}
                            \psi^1 = \derpar{\psi}{x^1}\;, \quad
                        \label{GL eq 2}
                            \psi^2= -a^2\derpar{\psi}{x^2}\,, \quad
                        \label{GL eq 3}
                            \derpar{\psi^1}{x^1}+\derpar{\psi^2}{x^2}=4\lambda \psi(\psi^2-1)\,.
                    \end{eqnarray}
               Hence $\psi$ is a solution of
                    \begin{equation}\label{Ginzburg}
                        \ds\frac{\partial^2\psi}{\partial (x^1)^2} - a^2 \ds\frac{\partial^2\psi}{\partial (x^2)^2} - 4\lambda \psi(\psi^2-1)=0\,,
                    \end{equation}
                that is, $\psi$ is a solution of \textit{Ginzburg-Landau's equation}.

Next, let us consider the Lagrangian
$$\begin{array}{rccl}
L\colon  & T^1_2\r\equiv T\r\oplus T\r&\to &\r\\\noalign{\medskip}
 & (q, v_1,v_2) &\mapsto & \ds\frac{1}{2}[(v_1)^2- a^2(v_2)^2]+\lambda (q^2-1)^2\,.\end{array}$$
Here $a$ and $\lambda$ are supposed to denote {\it constant} quantities. Then
\begin{equation}\label{partialL GL}
\derpar{L}{q}= 4\lambda q(q^2-1) \;,\quad \derpar{L}{v_1}= v_1 \;, \quad \derpar{L}{v_2}= -a^2 v_2\,.
\end{equation}

Let $(X_1,X_2)$ be a $2$-vector field    on $T^1_2\r$ solution of
\[
\iota_{X_1}
          \omega_L^1 + \iota_{X_2}
          \omega_L^2= dE_L\,.\]

           If
$\phi$ is a solution of $(X_1,X_2)$,
 then from (\ref{solsopde3}) and (\ref{partialL GL}) we obtain that $\phi$ satisfies the equation
 $$
 0=\partial_{11}\phi - a^2\partial_{22}\phi - 4\lambda\phi(\phi^2-1)\,,
 $$
 which is the Ginzburg-Landau equation (\ref{Ginzburg}).

                \begin{remark}{\rm
The phenomenological Ginzburg-Landau theory (1950) is a mathematical theory used for modeling superconductivity \cite{GL-1950}. \rqed}
\end{remark}

            \section{$k$-symplectic quadratic  systems}\label{example k-symp: hamiltonian quadratic hamiltonian}
\index{$k$-symplectic quadratic systems}
                Many Hamiltonian and Lagrangian systems in field theories are of ``quadratic'' type and they can be modeled as follows.

                Consider the canonical model of $k$-symplectic manifold $((T^1_k)^*Q,\omega^\alpha, V)$. Let $g_1,\ldots, g_k$ be $k$ semi-Riemannian metrics in $Q$. For each $q\in Q$ and for each $1\leq \alpha \leq k$ we have the following linear isomorphisms:
                    \[
                        \begin{array}{rccl}
                            g_\alpha^\flat\colon & T_qQ & \to &T^*_qQ\\\noalign{\medskip}
                             & v & \mapsto & \iota_{v}g_\alpha
                        \end{array}\,,
                    \]
                  and then we   introduce the dual metric $g_\alpha^*$ of $g_\alpha$,  defined by
                  as follows
                    \[
                        g^*_\alpha(\nu_q,\gamma_q)\colon = g_{\alpha}\big( (g^\flat_\alpha)^{-1}(\nu_q), (g^\flat_\alpha)^{-1}(\gamma_q)\big)\,,
                    \]
                for each $\nu_q,\gamma_q\in T^*_qQ$ and $\ak$.

                We can define a function $K\in \mathcal{C}^\infty((T^1_k)^*Q)$ as follows: for every $({\nu_1}_q,\ldots,$ $ {\nu_k}_q)\in (T^1_k)^*Q$,
                    \[
                        K({\nu_1}_q,\ldots, {\nu_k}_q)=\ds\frac{1}{2}\ds\sum_{\alpha=1}^k g^*_\alpha({\nu_\alpha}_q,{\nu_\alpha}_q)\,.
                    \]
                Then, if $V\in \mathcal{C}^\infty(Q)$ we define the   Hamiltonian function $H\in \mathcal{C}^\infty((T^1_k)^*Q)$ of ``quadratic'' type as follows
                    \[
                        H=K+(\pi^k)^*V\,.
                    \]

                Using canonical  coordinates $(q^i, p^\alpha_i)$ on $(T^1_k)^*Q$, the local expression of $H$ is
                    \[
                        H(q^i,p^\alpha_i)=\ds\frac{1}{2}\sum_{\alpha=1}^kg^{ij}_\alpha(q^m)p^\alpha_ip^\alpha_j+ V(q^m)\,,
                    \]
                where $g^{ij}_\alpha$ denote the coefficients of the matrix associated to $g^*_\alpha$. Then
                    \[
                        dH=\sum_{\alpha=1}^k\left[\left(\frac{1}{2}\derpar{g^{ij}_\alpha}{q^k}p^\alpha_ip^\alpha_j + \derpar{V}{q^k}\right)dq^k + (g^{ij}_\alpha p^\alpha_i)dp^\alpha_j\right]\,.
                    \]

                Consider now the $k$-symplectic Hamiltonian field equation (\ref{ecHksym}) associated with the above Hamiltonian function, i.e.
                    \[
                        \ds\sum_{\alpha=1}^k\iota_{X_\alpha}\omega^\alpha=dH\,.
                    \]
                    If  a $k$-vector field $\mathbf{X}=(X_1,\ldots, X_k)$ is
                solution of this equation then each $X_\alpha $ has      the following local expression (for each $\alpha$ fixed):
                    \begin{equation}\label{quadratic local vfh}
                        X_\alpha= g^{ij}_\alpha p^\alpha_j \derpar{}{q^i} + (X_\alpha)^\beta_i\derpar{}{p^\beta_i}
                    \end{equation}
                and its components $(X_\alpha)^\beta_i$ satisfy
                    \begin{equation}\label{quadratic vf condition}
                        \ds\sum_{\beta=1}^k(X_\beta)^\beta_k=-\left(\frac{1}{2}\derpar{g^{ij}_\beta}{q^k}p^\beta_ip^\beta_j + \derpar{V}{q^k}\right)\,.
                    \end{equation}

                We now assume that $\mathbf{X}$ is integrable and
                $$ \begin{array} {cccl}\varphi: & \r^k & \longrightarrow & (T^1_k)^*Q \\ \noalign{\medskip}
                 & x & \to & \varphi(x)=(\psi^i(x),\psi^\alpha_i(x))
                \end{array}$$ is an integral section of $\mathbf{X}$ then
                    \begin{equation}\label{quadratic is}
                        X_\alpha(\varphi(x))=\varphi_*(x)\left(\derpar{}{x^\alpha}\Big\vert_{x}\right) = \derpar{\psi^i}{x^\alpha}\Big\vert_{x}\derpar{}{q^i}\Big\vert_{\varphi(x)} + \derpar{\psi^\beta_i}{x^\alpha}\Big\vert_{x}\derpar{}{p^\beta_i}\Big\vert_{\varphi(x)}\,.
                    \end{equation}

                Thus, from (\ref{quadratic local vfh}), (\ref{quadratic vf condition}) and (\ref{quadratic is}) we obtain that $\varphi$ is a solution of the following  Hamilton-De Donder-Weyl equations
                    \begin{eqnarray*}
                        \derpar{\psi^i}{x^\alpha} = g^{ij}_\alpha\psi^\alpha_j\,,\quad (\alpha\makebox{ fixed})\\ \noalign{\medskip}
                        \ds\sum_{\beta=1}^k\derpar{\psi^\beta_l}{x^\beta}= -\left(\frac{1}{2}\derpar{g^{ij}_\beta}{q^l}\psi^\beta_i\psi^\beta_j + \derpar{V}{q^l}\right)
                    \end{eqnarray*}

In the Lagrangian approach we obtain a similar description. In fact, we consider the tangent bundle of $k^1$-velocities and let $g_1,\ldots, g_k$ be $k$ semi-Riemannian metrics in $Q$.

                We can define a function $K\in \mathcal{C}^\infty(\tkq)$ as follows: for every element ${\rm v}_q=({v_1}_q,\ldots, {v_k}_q)\in \tkq$,
                    \[
                        K({\rm v}_q)=\ds\frac{1}{2}\ds\sum_{\alpha=1}^k g_\alpha({v_\alpha}_q,{v_\alpha}_q)\,.
                    \]
                Then, if $V\in \mathcal{C}^\infty(Q)$ define the  $k$-symplectic Lagrangian function $L\in \mathcal{C}^\infty(\tkq)$ of ``quadratic'' type as follows
                    \[
                        L=K-(\tau^k)^*V\,.
                    \]

                Using canonical  coordinates $(q^i, v_\alpha^i)$ on $\tkq$, the local expression of $L$ is
                    \[
                        L(q^i,v^i_\alpha)=
                        \ds\frac{1}{2}
                        \sum_{\alpha=1}^kg^\alpha_{ij} (q^m)v_\alpha^i\, v_\alpha^j - V(q^m)\,,
                    \]
                where $g_{ij}^\alpha$ denote the coefficients of the matrix associated to $g_\alpha$.

                Consider now the $k$-symplectic Lagrangian field equation associated with the above Lagrangian function, i.e.
                    \[
                        \ds\sum_{\alpha=1}^k\iota_{X_\alpha}\omega_L^\alpha=dE_L\,.
                    \]
                    If  a $k$-vector field $\mathbf{X}=(X_1,\ldots, X_k)$ is
                solution of this equation, i.e. if $\mathbf{X}\in \vf^k_L(\tkq)$ then, since $L$ is regular,   each $X_\alpha $ has      the following local expression (for each $\alpha$ fixed):
                    \begin{equation}\label{quadratic local vf}
                        X_\alpha= v^i_\alpha \derpar{}{q^i} + (X_\alpha)_\beta^i\derpar{}{v_\beta^i}
                    \end{equation}
                and its components $(X_\alpha)_\beta^i$ satisfy equations (\ref{locel4}),
             that in this case are

             $$
             \derpar{g^\alpha_{il}}{q^j}v^l_\alpha v^j_\alpha +
             g^\alpha_{ij} \,  (X_\alpha)^j_\alpha= \ds\frac{1}{2}\derpar{g^\alpha_{lm}}{q^i}v^l_\alpha v^m_\alpha- \derpar{V}{q^i}
             $$
   Thus, if the components of the metrics $g^\alpha_{il}$ are constant then
   $$
   g^\alpha_{ij} \,  (X_\alpha)^j_\alpha= -\derpar{V}{q^i}\; .
   $$

                Now, if
                $$ \begin{array} {cccl}\phi^{(1)}: & \r^k & \longrightarrow & \tkq \\ \noalign{\medskip}
                 & x & \to & \phi(x)=(\phi^i(x),\derpar{\phi}{x^\alpha} (x)))
                \end{array}$$ is an integral section of $\mathbf{X}\in \vf^k_L(\tkq)$ then
                    $\phi:\rk \to Q$ is a solution the following
                Euler-Lagrange equations
                    $$
                    g^\alpha_{ij} \,  \derpars{\phi^j}{x^\alpha}{x^\beta} = -\derpar{V}{q^i}
                    $$
                \begin{remark}\label{remxan11}
                {\rm
                    The   examples of the previous subsections can be considered a particular case of this situation.
                    \begin{itemize}
                        \item \emph{The electrostatic equations} correspond with the case $Q=\r \,(n=1)$, $k=3$, the function $V\in \mathcal{C}^\infty(\r)$ is $V(q)=4\pi r$ and the semi-Riemannian metrics in $\r$.,
                                \[
                                    g_\alpha=dq^2\,\quad 1\leq \alpha\leq 3\,,
                                \]
                            $q$ being the standard coordinate in $\r$.

                        \item \emph{The wave equation} corresponds to the case $Q=\r\, (n=1)$, $k=n+1$, the function $V=0$ and the semi-Riemannian metrics in $\r$.
                                \[
                                    g_\alpha=-c^2dq^2,\quad 1\leq \alpha \leq n \makebox{ and } g_{n+1}= dq^2\,,
                                \]
                            $q$ being the standard coordinate in $\r$.
                        \item \emph{Laplace's equations} corresponds with the case $Q=\r,\, k=n, V(q)=0$ and the semi-Riemaniann metrics $g_\alpha=dq^2$.
                        \item \emph{The Sine-Gordon equation} corresponds with the case $Q=\r,\, k=2,\, V(q)=-\Omega^2\cos{q}$, and the semi-Riemannian metrics in $\r$.
                                \[
                                    g_1=dq^2 \makebox{ and } g_2= -a^2dq^2\,,
                                \]
                            $q$ being the standard coordinate in $\r$..
                        \item In the case of \emph{Ginzburg-Landau's equation}, $Q=\r,\, k=2,\, V(q)= -\lambda(q^2-1)^2$ and the semi-Riemannian metrics in $\r$.,
                                \[
                                    g_1=dq^2 \makebox{ and } g_2= -a^2dq^2\,,
                                \]
                            $q$ being the standard coordinate in $\r$.
                    \end{itemize}
                    }
                \end{remark}

                \section{Navier's equations} \label{section k-symp lag examples: Navier}

\index{Navier's equations}
We consider the  equation (\ref{globfor}) but with $Q=\r^2$ and Lagrangian $L\colon T\r^2\oplus T\r^2 \to \r$ given by
\[L(q^1,q^2, v^1_1,v^1_2,v^2_1,v^2_2)=
(\ds\frac{1}{2}\lambda + \mu)[(v_1^1)^2+ (v_2^2)^2]+
\ds\frac{1}{2}\mu[(v^1_2)^2 + (v^2_1)^2] +
(\lambda+\mu)v^1_1v^2_2\,.\]

This Lagrangian is regular if $\mu\neq 0$ and $\lambda \neq -(3/2)\mu$. In this case we obtain:
\begin{equation}\label{partialL Navier}
\begin{array}{lcl}
  \derpar{L}{q}=0 \;, & & \\\noalign{\medskip}  \derpar{L}{v^1_1}= ( \lambda+2\mu)v^1_1 + (\lambda+\mu)v^2_2 & \;,
  \quad & \derpar{L}{v^1_2}=\mu v^1_2 \\\noalign{\medskip}
        \derpar{L}{v^2_1}=\mu v^2_1 & \;, \quad & \derpar{L}{v^2_2}=(\lambda+2\mu) v^2_2 + (\lambda+\mu)v^1_1
\end{array}
\end{equation}

Let $(X_1, X_2)$ be an integrable solution of  (\ref{globfor}) for this particular case. From (\ref{solsopde3}) and (\ref{partialL Navier}),
we have that, if  $$\begin{array}{rccl} \phi\,\colon  &  \r^2 &\to &
 \r^2 \\\noalign{\medskip}
     & (x^1,x^2) & \mapsto  &  (\phi^1(x),\phi^2(x))\end{array}$$ is a solution of  $(X_1,X_2)$,  then
     $\phi$ satisfies
\[\begin{array}{lcl}
(\lambda+2\mu) \partial_{11}\phi^1   +
(\lambda+\mu)\partial_{12}\phi^2  +
\mu \partial_{22}\phi^1  &=&0\,,\\\noalign{\medskip}
\mu \partial_{11}\phi^2 +(\lambda+\mu)\partial_{12}\phi^1
+ (\lambda+2\mu) \partial_{22}\phi^2   &=&0\,,
\end{array}\]
which are  {\it Navier's equations}, see \cite{Olver,olver}. These  are the equations of motion for a
viscous fluid in which one consider the effects of attraction and repulsion between neighboring molecules
(Navier, 1822). Here $\lambda$ and $\mu$ are coefficients of viscosity.

                \section{Equation of minimal surfaces}\label{section k-symp lag examples: minimal}

\index{Equation of minimal surfaces}
We consider again  $Q=\r$ and $(X_1, X_2)$   a solution of
(\ref{globfor})
where $L$ is now the regular Lagrangian
$$\begin{array}{rccl} L\colon  & T\r\oplus T\r &
\to &\r\\\noalign{\medskip} & (q,v_1,v_2) & \mapsto & \ds\sqrt{
1+v_1^2 + v_2^2}\end{array}$$

Then one obtains,
\begin{equation}\label{partialL minimal}
\derpar{L}{q}=0\;, \;\; \derpar{L}{v_1}=\ds\frac{v_1}{\sqrt{1+(v_1)^2+(v_2)^2}} \;, \;\; \derpar{L}{v_2}=
\ds\frac{v_2}{\sqrt{1+(v_1)^2+(v_2)^2}}\,.
\end{equation}
From (\ref{solsopde3}) and (\ref{partialL minimal}), we deduce that  if  $\phi$ is solution of the $2$-vector field
$(X_1,X_2)$, then $\phi$ satisfies
\[0=(1+(\partial_2\phi)^2)\partial_{11}\phi - 2\partial_1\phi\,\partial_2\phi\,\partial_{12}\phi + (1+(\partial_1
\phi)^2)\partial_{22}\phi\,,
\]
 which is the {\it equation of minimal surfaces,} (see for instance \cite{EM-92, olver}).
 \begin{remark}{\rm An alternative Lagrangian for the equation of minimal surfaces is given by $L(q,v_1,v_2)= 1+ v_1^2
 + v_2^2$.}
 \end{remark}
            \section{The massive scalar field}\label{example k-symp: hamiltonian scalar field}

\index{Massive scalar field}

                The equation of a scalar field $\phi$ (for instance the gravitational field) which acts on the $4$-dimensional space-time is (see \cite{Goldstein,KT-79}):
                    \begin{equation}\label{scalar}
                        (\square + m^2)\phi =F'(\phi)\,,
                    \end{equation}
               where $m$ is the mass of the particle over which the field acts, $F$ is a scalar function such that $F(\phi)-\ds\frac{1}{2}m^2\phi^2$ is the potential energy of the particle of mass $m$, and $\square$ is the Laplace-Beltrami operator given by
                    \[
                        \square\phi\colon = div\, grad \phi = \ds\frac{1}{\sqrt{-g}}\derpar{}{x^\alpha}\left(\sqrt{-g}g^{\alpha\beta}\derpar{\phi}{t^\beta}\right)\,,
                    \]
               $(g_{\alpha\beta})$ being a pseudo-Riemannian metric tensor in the $4$-dimensional space-time of signature $(-+++)$, and $\sqrt{-g}=\sqrt{-\det{g_{\alpha\beta}}}$. In this case we suppose that the metric $(g_{\alpha\beta})$ is the Minkowski metric on $\r^4$, i.e.
                \[
                    d(x^2)^2 + d(x^3)^2+d(x^4)^2- d(x^1)^2\,.
                \]

               Consider the Hamiltonian function
                    \[
                        \begin{array}{lccl}
                            H\colon & (T^1_4)^*\r & \to & \r\\\noalign{\medskip}
                             & (q, p^1, p^2, p^3, p^4) & \mapsto &  \ds\frac{1}{2}g_{\alpha\beta}p^\alpha p^\beta-\left( F(q)-\ds\frac{1}{2}m^2q^2\right)\,,
                        \end{array}
                    \]
               where $q$ denotes the scalar field $\phi$ and $(q,p^1,p^2,p^3,p^4)$ the natural coordinates on $(T^1_4)^*\r$. Then
                \begin{equation}\label{partial Ham scalar}
                    \derpar{H}{q}=-\Big(F'(q)-m^2q\Big),\quad \derpar{H}{p^\alpha}=g_{\alpha\beta}p^\beta\,.
                \end{equation}

                Consider the  $4$-symplectic Hamiltonian equation
                    \[
                        \iota_{X_1}\omega^1 +\iota_{X_2}\omega^2+ \iota_{X_3}\omega^3+ \iota_{X_4}\omega^4 = dH,
                    \]
                associated to the above Hamiltonian function. From (\ref{partial Ham scalar}) one obtains that, in natural coordinates, a $4$-vector field solution of this equation has the following local expression
                    \begin{equation}\label{vf scalar}
                        X_\alpha=g_{\alpha\beta}p^\beta\derpar{}{q} + (X_\alpha)^\beta\derpar{}{p^\beta}\,,
                    \end{equation}
                where the functions $(X_\alpha)^\beta\in \mathcal{C}^\infty((T^1_4)^*\r)$ satisfies
                    \begin{equation}\label{condition vf scalar}
                        F'(q)-m^2q= (X_1)^1+(X_2)^2+(X_3)^3+(X_4)^4\,.
                    \end{equation}

                Let $\varphi\colon \r^4\to (T^1_4)^*\r,\, \varphi(x)=(\psi(x),\psi^1(x),\psi^2(x),\psi^3(x),\psi^4(x))$ be an integral section of a $4$-vector field solution of the $4$-symplectic Hamiltonian  equation. Then from (\ref{vf scalar}) and (\ref{condition vf scalar}) one obtains
                    \[
                        \begin{array}{l}
                            \derpar{\psi}{x^\alpha}= g_{\alpha\beta}\psi^\beta\\
                            F'(\psi(x))-m^2\psi(x)=\derpar{\psi^1}{x^1}+ \derpar{\psi^2}{x^2}+ \derpar{\psi^3}{x^3}+\derpar{\psi^4}{x^4}
                        \end{array}
                    \]

                Therefore, $\psi\colon \r^4\to \r$ is a solution of the equation
                    \[
                        F'(\psi)-m^2\psi= \derpar{}{x^\alpha}\left( g^{\alpha\beta}\derpar{\psi}{x^\beta} \right)\,,
                    \]
                that is, $\psi$ is a solution of the scalar field equation.
                \begin{remark}
                {\rm
                    The scalar equation can be described using the Lagrangian approach with Lagrangian function
                    \begin{equation}\label{scalar lagrangian}
L(x^1,\ldots, x^4, q, v_1,\ldots, v_4)= \sqrt{-g}\Big(F(q)- \ds\frac{1}{2}m^2q^2\Big)+\ds\frac{1}{2}g^{\alpha\beta}v_\alpha v_\beta\,,
\end{equation}
where $q$ denotes the scalar field $\phi$ and $v_\alpha$ the partial derivative $\nicefrac{\partial \phi}{\partial x^\alpha}$. Then the equation (\ref{scalar}) is the Euler-Lagrange equation associated to $L$.
                \rqed}
                \end{remark}
                \begin{remark}
                {\rm
                    Some particular cases of the scalar field equation are the following:
                        \begin{enumerate}
                            \item If $F=0$ we obtain the linear scalar field equation.
                            \index{Linear scalar field equation}
                            \index{Klein-Gordon equation}
                            \item If $F(q)=m^2q^2$, we obtain the Klein-Gordon equation \cite{saletan}
                                    \[
                                        (\square + m^2)\psi=0\,.
                                    \]
                        \end{enumerate}
                        }\rqed
                \end{remark}
%
%\newpage
%\mbox{}
%\thispagestyle{empty} % para que no se numere esta p\'{a}gina

\part{$k$-cosymplectic formulation of Classical Field Theories}\label{part k-cosymp}

The Part II of this book has been devoted to give a geometric description of certain kinds of Classical Field Theories. The purpose of Part  III is to extend the above study to Classical Field Theories involving the independent parameters, i.e. the ``space-time'' coordinates $(x^1,\ldots, x^k)$ in an explicit way. In others words, in this part we shall give a geometrical description of Classical Field Theories whose Lagrangian and Hamiltonian functions are of the form $L=L(x^\alpha, q^i, v^i_\alpha)$ and $H=H(x^\alpha, q^i, p^\alpha_i)$.

The model of the convenient geometrical structure for our approach is extracted of the so-called stable cotangent bundle of $k^1$-covelocities $\rktkqh$. These structures are called $k$-cosymplectic manifolds and they were introduced by M. de Le\'{o}n {\it et al.} \cite{LMORS-1998,LMeS-2001}.

In chapter \ref{chapter: k-cosympManifolds} we shall recall the notion of $k$-cosymplectic manifold using as model $\rktkqh$. Later, in chapter \ref{k-cosymp eq} we shall describe the $k$-cosymplectic formalism. This formulation can be applied to give a geometric version of the Hamilton-De Donder-Weyl and Euler-Lagrange equations for field theories. We also present several physical examples which can be described using this approach and the relationship between the Hamiltonian and Lagrange approaches.
\newpage
\mbox{}
\thispagestyle{empty} % para que no se numere esta p\'{a}gina

\chapter{$k$-cosymplectic geometry}\label{chapter: k-cosympManifolds}

    The $k$-cosymplectic formulation is based in the so-called $k$-cosymplectic geometry. In this chapter we introduce these structures which are a generalization of the notion of cosymplectic forms.

    Firstly, we describe the model of the called $k$-cosymplectic manifolds,  that is the stable cotangent bundle of $k^1$-covelocities $\rktkqh$ and  introduce the canonical structures living there.  Using this model we introduce the notion of $k$-cosymplectic manifold. A complete description of these structures can be found in \cite{LMORS-1998,LMeS-2001}.

\section{The stable cotangent bundle of $k^1$-cove\-lo\-cities $\rktkqh$}\label{section: rktkqh}

\index{Stable cotangent bundle of $k^1$-covelocities}
    Let $J^1(Q,\rk)_0$ be the manifold of $1$-jets of maps from $Q$ to $\rk$ with target at $0\in \rk$, which we described in Remark \ref{j1qrk}. Let us remember that this manifold is diffeomorphic to the cotangent bundle of $k^1$-covelocities $\tkqh$ via the diffeomorphism described in (\ref{difeo jiqrk-tkq}).

    Indeed, for each $x\in \rk$ we can consider the manifold $J^1(Q,\rk)_x$ of $1$-jets of maps from $Q$ to $\rk$ with target at $x\in \rk$, i.e.,
        \[
            J^1(Q,\rk)_x= \ds\bigcup_{q\in Q}J^1_{q,\,x}(Q,\rk)=\ds\bigcup_{q\in Q}\{j^1_{ q,x}\sigma\, \vert\,\sigma:Q\to \rk\makebox{ smooth, }\sigma( q)=x\}\,.
        \]

    If we consider the collection of all these space,  we obtain the set $J^1(Q,\rk)$ of $1$-jets of maps from $Q$ to $\rk$, i.e.
        \[
            J^1(Q,\rk)=\bigcup_{x\in\rk}J^1(Q,\rk)_x\,.
        \]

    The set $J^1(Q,\rk)$ can be identified with $\rktkqh$ via
        \begin{equation}\label{difeo j1qrk-rktkqh}
        \begin{array}{ccccc}
            J^1(Q,\rk) & \to &\rk\times J^1(Q,\rk)_ 0&\to & \rktkqh\\\noalign{\medskip}
            j^1_{q,x}\sigma & \to & (x,j^1_{q,0}\sigma_q)& \to & (x,d\sigma^1_q(q),\ldots, d\sigma^k_q(q))
        \end{array}\,,
        \end{equation}
    where the last identification was described in (\ref{difeo jiqrk-tkq}), being $\sigma_q\colon Q\to \rk$ the map defined by $\sigma_q(\tilde{q})=\sigma(\tilde{q})- \sigma(q)$ for any $\tilde{q}\in Q$.

    \begin{remark}\label{difeo jipiQ}
      {\rm  We recall that the manifold of $1$-jets of mappings from $Q$ to $\rk$, can be identified with the manifold $J^1\pi_Q$ of $1$-jets of sections of the trivial bundle $\pi_Q\colon \rk\times Q\to Q$, (a full description of the first-order jet bundle associated to an arbitrary bundle $E\to M$ can be found in \cite{Saunders-89}). The diffeomorphism which establishes this relation is given by
            \[
                \begin{array}{ccccc}
                    J^1\pi_{Q} & \to & J^1(Q,\rk) & \to & \rk\times J^1(Q,\rk)_ 0\\\noalign{\medskip}
                    j^1_q\phi & \to & j^1_{q,\sigma(q)}\sigma & \to & (\sigma(q),j^1_{q,0}\sigma_{q})
                \end{array}
            \]
    where $\phi\colon Q\to \rkq$ is a section of $\pi_Q$, $\sigma\colon Q\to \rk$ is given by $\sigma=\pi_{\rk}\circ \phi$ being $\pi_{\rk}\colon \rkq\to \rk$ the canonical projection and $\sigma_q\colon Q\to \rk$ is defined by $\sigma_q(\tilde{q})=\sigma(\tilde{q})- \sigma(q)$  for any $\tilde{q}\in Q$.}\rqed
    \end{remark}

    From the above comments we know that an element of  $J^1(Q,\rk)\equiv \rktkqh$ is a $(q+1)$-tuple  $(x,{\nu_1}_q,\ldots, {\nu_k}_q)$ where $x\in \rk$ and $({\nu_1}_q,\ldots, {\nu_k}_q)\in\tkqh$. Thus we can consider the following canonical projections:
        \[
            \xymatrix@C=20mm{ & \rktkqh \ar[dr]^-{(\pi_Q)_1} \ar[d]_-{(\pi_Q)_{1,0}}& \\
            \rk & \rkq \ar[l]_-{\pi_{\rk}}\ar[r]^-{\pi_Q}& Q}
        \]
    defined by
        \begin{equation}\label{projections rktkqh}
            \begin{array}{ll}
                    (\pi_Q)_{1,0}(x,\nu_{1_q}, \ldots ,\nu_{k_q})=(x,q), &  \pi_{\rk}(x,q)=x,
                \\\noalign{\medskip}
                    (\pi_Q)_1(x,{\nu_1}_q,\ldots, {\nu_k}_q)=q, &  \pi_Q(x,q)= q,
            \end{array}
        \end{equation}
  with $x\in \rk $, $q\in Q$ and $({\nu_1}_q,\ldots, {\nu_k}_q)\in (T^1_k)^*Q$.

In the following diagram we collect the notation used for the projections in this part of the book:
  \begin{figure}[h]
\centering
\scalebox{1.3}{
%\begin{sideways}
%\begin{tabular}{c}
$
\xymatrix @R=4pc {
 & & \rktkqh \ar[r]^-{\bar{\pi}_2}\ar@/^2pc/[rr]^-{\pi^\alpha_2} \ar[d]^-{(\pi_Q)_{1,0}}\ar[dr]^-{(\pi_Q)_1} \ar[dl]_-{\bar{\pi}_1} \ar@/^-2pc/[lld]_-{\pi^\alpha_1}& \tkqh \ar[d]^-{\pi^k}\ar[r]^{\pi^{k,\alpha}}& T^*Q\ar[dl]_-{\pi} \\
 \r & \rk \ar[l]_-{\pi^\alpha}& \rk\times Q \ar[l]_-{\pi_{\rk}}\ar[r]^-{\pi_Q}& Q &
}$
%\end{tabular}
%\end{sideways}
}
\caption{Canonical projections associated to $\rktkqh$}
\label{K-cosymp-Maps}
\end{figure}

\index{Stable cotangent bundle of $k^1$-covelocities!Coordinates}
    If $(q^i)$ with $\n$, is a local coordinate system defined on an open set $U\subset Q$, the induced \emph{local coordinates} $(x^\alpha, q^i,p^\alpha_i)$, $\n,\,\ak$ on $\rk\times (T^1_k)^*U=\Big((\pi_Q)_1\Big)^{-1}(U)$ are given by
        \begin{equation}\label{canonical coordinates rktkqh}
            \begin{array}{lcl}
                    x^\alpha(x,{\nu_1}_q,\ldots, {\nu_k}_q) = x^\alpha(x)=x^\alpha\, , \\\noalign{\medskip}
                    q^i(x,{\nu_1}_q,\ldots, {\nu_k}_q) = q^i(q)\, ,
                \\\noalign{\medskip}
                    p^\alpha_i(x,{\nu_1}_q,\ldots, {\nu_k}_q) = {\nu_\alpha}_q\left(\ds\frac{\partial}{\partial q^i}\Big\vert_q\right)\,.
            \end{array}
        \end{equation}
\index{Cotangent bundle of $k^1$-covelocities!canonical coordinates}
        Thus, $\rktkqh$ is endowed with a  structure of differentiable manifold of dimension $k+n(k+1)$, and the manifold $\rktkqh$ with the projection $(\pi_Q)_1$ has the  structure of a vector bundle over $Q$.

        If we consider the identification (\ref{difeo j1qrk-rktkqh}), the above coordinates can be defined in terms of $1$-jets of map from $Q$ to $\rk$  in the following way
            \[
                x^\alpha(j^1_{q,x}\sigma)=x^\alpha(x)=x^\alpha\,, \quad q^i(j^1_{q,x}\sigma)=q^i(q)\,,\quad p_\alpha^i(j^1_{q,x}\sigma)=\derpar{\sigma^\alpha}{ q^i}\Big\vert_{q}\,.
            \]

\index{Cotangent bundle of $k^1$-covelocities!canonical forms}
    On $\rktkqh$ we can define a family of canonical forms as follows
    \begin{equation}\label{canonical forms on rktkqh}
        \eta^\alpha=(\pi^\alpha_1)^*dx,\quad \Theta^\alpha=
(\pi^\alpha_2)^{\;*}\theta \quad \makebox{and}\quad \Omega^\alpha=
(\pi^\alpha_2)^{\;*}\omega\,,
    \end{equation}
    with $\ak$, being $\pi^\alpha_1:\rk \times (T^1_k)^{\;*}Q \rightarrow \r$ and $\pi^\alpha_2:\rk \times (T^1_k)^{\;*}Q \rightarrow T^{\;*}Q$ the projections defined by
        \[
            \pi^\alpha_1(x,{\nu_1}_q,\ldots, {\nu_k}_q)=x^\alpha \,,\quad \pi^\alpha_2(x,{\nu_1}_q,\ldots, {\nu_k}_q)=\nu_{\alpha_q}
        \]
    and $\theta$ and $\omega$ the canonical Liouville and symplectic forms on $T^*Q$, respectively. Let us observe that, since $\omega=-d\theta$, then
    $\Omega^\alpha=-d\Theta^\alpha$.

\index{Stable cotangent bundle of $k^1$-covelocities!canonical structure}
    If we consider a local coordinate system $(x^\alpha, q^i,p^\alpha_i)$ on $\rktkqh$ (see (\ref{canonical coordinates rktkqh})), the \emph{canonical forms} $\eta^\alpha,\, \Theta^\alpha$ and $\Omega^\alpha$ have the following local expressions:
    \begin{equation}\label{local canonical forms rktkqh}
        \eta^\alpha=dx^\alpha,\, \quad \Theta^\alpha = \displaystyle \, p^\alpha_i dq^i \quad  ,  \quad \Omega^\alpha = \displaystyle  dq^i \wedge dp^\alpha_i\, .
    \end{equation}

    Moreover, let be $\mathcal{V}^*=\ker \big((\pi_Q)_{1,0}\big)_*$; then a simple inspection in local coordinates shows that the forms $\eta^\alpha$ and $\Omega^\alpha$, with $\ak$ are closed and the following relations hold:
    \begin{enumerate}
        \item  $dx^1\wedge\dots\wedge dx^k\neq 0$,\quad $dx^\alpha\vert_{\mathcal{V}^{\;*}}=0,\quad \Omega^\alpha\vert_{\mathcal{V}^{\;*}\times \mathcal{V}^{\;*}}=0,$
        \item  $(\ds {\cap_{\alpha=1}^{k}} \ker dx^\alpha) \cap (\ds {\cap_{\alpha=1}^{k}} \ker \Omega^\alpha)=\{0\}$, \quad ${\rm dim}\,(\ds {\cap_{\alpha=1}^{k}} \ker \Omega^\alpha)=k.$
    \end{enumerate}

    \begin{remark}
      {\rm   Let us observe that the canonical forms on $\tkqh$ and on $\rktkqh$ are related by the expressions
        \begin{equation}\label{symcosym}
            \theta^\alpha=(\bar{\pi}_2)^*\theta^\alpha \quad \makebox{ and} \quad \Omega^\alpha=(\bar{\pi}_2)^*\omega^\alpha\,,
        \end{equation}
        with $\ak$.}\rqed
    \end{remark}

\section{$k$-cosymplectic geometry}

From the above model, that is, the stable cotangent bundle of $k^1$-covelocities  with the canonical forms (\ref{canonical forms on rktkqh}), M. de Le\'{o}n and collaborators have introduced the notion  of $k$-cosymplectic structures  in \cite{LMORS-1998,LMeS-2001}.

Let us recall that the $k$-cosymplectic manifolds provide a natural arena to develop Classical Field Theories as an alternative to other geometrical settings as the polysymplectic geometry \cite{Sarda2, Sd-95, Sarda, Sd-95b} or multisymplectic geometry.

Before of introducing the formal definition of $k$-cosymplectic manifold we consider the linear case.

\subsection{$k$-cosymplectic vector spaces}
Inspired in the above geometrical model one can define $k$-cosymplectic structures on a vector space in the following way (see \cite{Tesis merino}).
\index{$k$-cosymplectic vector space}
\begin{definition}
    Let $E$ a $k+n(k+1)$-dimensional vector space.
    A family $(\eta^\alpha,\Omega^\alpha,$ $V;1\leq \alpha\leq k)$ where  $\eta^1,\ldots, \eta^k$ are $1$-forms, $\Omega^1,\ldots, \Omega^k$ are $2$-forms and $V$ is a vector subspace of $E$ of dimension $nk$, defines a \emph{$k$-cosymplectic structure} on the vector space $E$ if the following conditions hold:
    \begin{enumerate}
        \item $\eta^1\wedge\ldots\wedge \eta^k \neq 0$,
        \item ${\rm dim}\, (\ker \Omega^1\cap \ldots\cap \ker\Omega^k)=k$,
        \item $(\ds {\cap_{\alpha=1}^{k}} \ker \eta^\alpha) \cap (\ds {\cap_{\alpha=1}^{k}} \ker \Omega^\alpha)=\{0\}$,
        \item $\eta^\alpha\vert_{V}=0,\; \Omega^\alpha\vert_{V\times V}=0,\, 1\leq \alpha\leq k\,.$
    \end{enumerate}

    $(E,\eta^\alpha, \Omega^\alpha, V)$ is called \emph{$k$-cosymplectic vector space}.
\end{definition}

\index{$k$-cosymplectic structure!on vector spaces}

\begin{remark}
    {\rm
        If $k=1$, then $E$ is a vector space of dimension $2n+1$ and we have a family $(\eta,\Omega,V)$ given by a $1$-form $\eta$, a $2$-form $\Omega$ and a subspace $V\subset E$ of dimension $n$.

        From conditions $(2)$ and $(3)$ of the above definition one deduces that $\eta\wedge\Omega^n\neq 0$ since $\dim\ker\omega=1$, and then ${\rm rank}\, \Omega=2n$, moreover $\ker \eta\cap \ker \Omega=0$.

        The pair $(\eta,\Omega)$ define a cosymplectic structure on $E$. Moreover, from condition $(4)$ one deduce that $(\eta,\Omega,V)$ define a stable almost cotangent structure on $E$.
    \rqed}
\end{remark}

Given a $k$-cosymplectic structure on a vector space one can prove the following results (see \cite{Tesis merino}):
\begin{theorem}[Darboux coordinates]
    If $(\eta^\alpha,\Omega^\alpha,V;1\leq \alpha\leq k)$ is a $k$-cosymplectic structure on $E$ then there exists a basis $(\eta^1,\ldots, \eta^k,\gamma^i, \gamma_i^\alpha; 1\leq i\leq n, 1\leq \alpha\leq k)$ of $E^*$ such that
    \[
        \Omega^\alpha = \gamma^i\wedge \gamma^\alpha_i\,.
    \]
\end{theorem}

\index{Reeb vectors}
For every $k$-cosymplectic structure $(\eta^\alpha,\Omega^\alpha,V;1\leq \alpha\leq k)$ on a vector space $E$, there exists a family of $k$ vectors $R_1,\ldots, R_k$, which are called \emph{Reeb vectors}, characterized by the conditions
\[
    \iota_{R_\alpha}\eta^\beta =\delta^\beta_\alpha,\; \iota_{R_\alpha}\omega^\beta=0\,.
\]

\subsection{$k$-cosymplectic manifolds}

We turn now to the globalization of the ideas of the previous section to manifolds. The following definition was introduced  in \cite{LMORS-1998}:

\index{$k$-cosymplectic structure}
\begin{definition}\label{deest}
Let $M$ be a differentiable manifold  of dimension $k(n+1)+n$. A
\emph{$k$-cosymplectic structure} on $M$ is a family
$(\eta^\alpha,\Omega^\alpha,V;1\leq \alpha\leq k)$, where each $\eta^\alpha$ is a closed
$1$-form, each $\Omega^\alpha$ is a closed $2$-form and $V$ is an
integrable  $nk$-dimensional distribution on $M$ satisfying
\begin{enumerate}
        \item  $\eta^1\wedge\dots\wedge \eta^k\neq 0$,\quad $\eta^\alpha\vert_{V}=0,\quad \Omega^\alpha\vert_{V\times V}=0,$
        \item  $(\ds {\cap_{\alpha=1}^{k}} \ker \eta^\alpha) \cap (\ds {\cap_{\alpha=1}^{k}} \ker \Omega^\alpha)=\{0\}$, \quad ${\rm dim}\,(\ds {\cap_{\alpha=1}^{k}} \ker \Omega^\alpha)=k.$
    \end{enumerate}

$M$ is said to be a \emph{$k$-cosymplectic manifold}.
\end{definition}
\index{$k$-cosymplectic manifold}

In particular, if $k=1$, then $dim \, M=2n+1$ and $(\eta,\Omega)$ is a cosymplectic structure on $M$.

\index{Reeb vector fields}
For every $k$-cosymplectic structure  $(\eta^\alpha ,\Omega^\alpha,{\cal
V})$ on $M$, there exists a family of $k$ vector fields
$\{R_\alpha\}$, which are called \emph{Reeb vector fields},
 characterized by the following conditions
$$
\iota_{R_\alpha}\eta^\beta=\delta^\beta_\alpha \quad ,\quad \iota_{R_\alpha}\Omega^\beta=0 \ .
$$

In the canonical model $R_\alpha=\derpar{}{x^\alpha}$.

The following theorem has been proved in \cite{LMORS-1998}:
\index{$k$-cosymplectic manifold!Darboux theorem}
\begin{theorem}[Darboux Theorem]\label{Darbox k-cosymp}
If $M$ is   an  $k$--cosymplectic manifold, then around each
  point of $M$ there exist local coordinates
$(x^\alpha,q^i,p^\alpha_i;1\leq A\leq k, 1\leq i \leq n)$ such that
$$
\eta^\alpha=dx^\alpha,\quad \Omega^\alpha=dq^i\wedge dp^\alpha_i, \quad
V=\left\langle\frac{\displaystyle\partial} {\displaystyle\partial
p^1_i}, \dots, \frac{\displaystyle\partial}{\displaystyle\partial
p^k_i}\right\rangle_{i=1,\ldots , n}.
$$
\end{theorem}

The canonical model for these geometrical structures is  $$(\rk
\times (T^1_k)^{*}Q,\eta^\alpha,\Omega^\alpha,V^{*}).$$

\begin{example}
{\rm
    Let $(N,\omega^\alpha, V)$ be an arbitrary $k$-symplectic manifold. Then, denoting by
    \[
        \pi_{\rk}\colon \rk\times N\to \r^k,\qquad \pi_N\colon \rk\times N\to N
    \]
    the canonical projections, we consider the differential forms
    \[
        \eta^\alpha = \pi_{\rk}^*(dx^\alpha),\quad \Omega^\alpha=\pi_N^*\omega^\alpha\,,
    \]
    and the distribution $V$ in $N$ defines a distribution $\mathcal{V}$ in $M=\rk\times N$ in a natural way. All conditions given in Definition \ref{deest} are verified, and hence $\rk\times N$ is a $k$-cosymplectic manifold.
}\end{example}

 \newpage
\mbox{}
\thispagestyle{empty} % para que no se numere esta p\'{a}gina

\chapter{$k$-cosymplectic formalism} \label{k-cosymp eq}

In this chapter we describe the $k$-cosymplectic formalism. As we shall see in the following chapters, using this formalism we can study Classical Field Theories that explicitly involve the space-time coordinates on the Hamiltonian and Lagrangian. This is the principal difference with the $k$-symplectic approach. As in previous case, in this formalism is fundamental the notion of $k$-vector field; let us recall that this notion was introduced in Section \ref{section k-vector field} for an arbitrary manifold $M$.

    \section{$k$-cosymplectic   Hamiltonian equations}\label{Section HDW eq: k-cosymp approach}

%In this section we introduce the $k$-cosymplectic description of the Hamilton-De Donder-Weyl equations (\ref{HDW field eq}). This approach was firstly introduced by  M.de Le\'{o}n {\it et al.} \cite{LMORS-1998}.

Let $(M,\eta^\alpha,\Omega^\alpha,V)$ be a $k$-cosymplectic manifold and $H$   a Hamiltonian on $M$, that is, a function $H\colon M\to \r$  defined on $M$.
\index{Hamiltonian function}

\begin{definition}
   The family $(M,\eta^\alpha, \Omega^\alpha, H)$ is called \emph{$k$-cosym\-plectic Hamiltonian system}.
\end{definition}

\index{$k$-cosymplectic Hamiltonian system}
Let $(M,\eta^\alpha, \Omega^\alpha, H)$ be a $k$-cosymplectic Hamiltonian system and ${\bf X}=(X_1,\ldots,$ $ X_k)$ a
$k$-vector field on  $M$ solution of the system of equations
\begin{equation}\label{geonah}
\begin{array}{l}
\eta^\alpha(X_\beta)=\delta^\alpha_\beta, \quad 1\leq \alpha,\beta\leq k\\
\noalign{\medskip}\displaystyle \sum_{\alpha=1}^k \,
 \iota_{X_\alpha}\Omega^\alpha =
dH-\displaystyle\sum_{\alpha=1}^k R_\alpha(H)\eta^\alpha \, \, ,
\end{array}
\end{equation}
where $R_1,\ldots, R_k$ are the Reeb vector fields associated with the $k$-cosymplectic structure on $M$.

Given a local coordinate system $(x^\alpha,q^i,p^\alpha_i)$, each $X_\alpha$, $1\leq \alpha\leq k$
is locally given by
$$
X_\alpha=(X_\alpha)_\beta\derpar{}{x^\alpha}+ (X_\alpha)^i\derpar{}{q^i} +(X_\alpha)^\beta_i \derpar{}{v^i_\beta}
$$

Now, since $$dH=  \derpar{H}{x^\alpha } dx^\alpha + \derpar{H}{q^i} dq^i +\derpar{H}{p^\alpha_i} dp^\alpha_i$$ and
$$\eta^\alpha=dx^\alpha, \quad \omega^\alpha=dq^i\wedge dq^i,\quad R_\alpha=\nicefrac{\partial}{\partial x^\alpha}$$
we deduce that the equation (\ref{geonah}) is locally equivalent to  the following conditions
\begin{equation}
(X_\alpha)_\beta=\delta^\beta_\alpha\, , \quad \derpar{ H}{q^i}=-\ds\sum_{\beta=1}^k(X_\beta)^\beta_i\, \quad \derpar{ H}{p^\alpha_i}=(X_\alpha)^i\
\label{k-cosymp condvf}
 \end{equation}
 with $1\leq i\leq n$ and $  1\leq \alpha\leq k$.

Let us suppose that ${\bf X}=(X_1,\ldots, X_k)$ is integrable, and
          $$\begin{array}{cccl}\varphi: & \colon   \rk &  \to  & M \\ \noalign{\medskip}
                                  &             x &   \to &     \varphi(x)=(\psi_\alpha(x),\psi^i(x),\psi^\alpha_i(x)) \end{array}$$ is an integral
                section of ${\bf X}$, then
        \begin{equation}
            \varphi_*(x)\Big(\derpar{}{x^\alpha}
            \Big\vert_{x}\Big)=
        \derpar{\psi_\beta}{x^\alpha}\Big\vert_{x}\derpar{}{x^\beta}\Big\vert_{\varphi(x)}     +
            \derpar{\psi^i}{x^\alpha}\Big\vert_{x}\derpar{}{q^i}\Big\vert_{\varphi(x)} + \derpar{\psi^\beta_i}{x^\alpha}\Big\vert_{x}\derpar{}{p^\beta_i}\Big\vert_{\varphi(x)}\,.
       \label{equation1}
 \end{equation}

        From  (\ref{equation1}) we obtain that $\varphi$ is given by $\varphi(x)=(x,\psi^i(x),\psi^\alpha_i(x)) $ and the following equations
        \begin{equation}\label{tkqh integral section equivalence cond1}
        \derpar{\psi_\beta}{x^\alpha}\Big\vert_{x}=\delta^\alpha_\beta  \,,\quad   \derpar{\psi^i}{x^\alpha}\Big\vert_{x}=(X_\alpha)^i(\varphi(x))\,,\quad \derpar{\psi^\beta_i}{x^\alpha}\Big\vert_{x}=(X_\alpha)^\beta_i(\varphi(x))
        \, ,\end{equation}
        hold.

This theory can be summarized in the following

\begin{theorem}\label{fhkc}Let $(M,\eta^\alpha, \Omega^\alpha, H)$ a $k$-cosymplectic Hamiltonian system and ${\bf X}=(X_1,\ldots, X_k)$ a $k$-vector field on  $M$ solution of the system of equations
\[
\begin{array}{l}
\eta^\alpha(X_\beta)=\delta^\alpha_\beta, \quad 1\leq \alpha,\beta\leq k\\
\noalign{\medskip}\displaystyle \sum_{\alpha=1}^k \,
 \iota_{X_\alpha}\Omega^\alpha =
dH-\displaystyle\sum_{\alpha=1}^k R_\alpha(H)\eta^\alpha \, \, ,
\end{array}
\]
where $R_1,\ldots, R_k$ are the Reeb vector fields associated with the $k$-cosymplectic structure on $M$.

 If $\mathbf{X}$ is integrable and $\varphi:\rk\to M,\,\varphi(x)=(x^\alpha,\psi^i(x),\psi^\alpha_i(x))$ is an integral section of the $k$-vector field ${\bf X}$, then $\varphi$ is a solution of the following system of partial differential equations
 \[
 \derpar{H}{q^i}\Big\vert_{\varphi(x)} = -\displaystyle\sum_{\alpha=1}^k\derpar{\psi^\alpha_i}{x^\alpha}\Big\vert_t
                \;, \quad \derpar{H}{p^\alpha_i}\Big\vert_{\varphi(x)}=\derpar{\psi^i}{x^\alpha}\Big\vert_{x}.
                \]
\end{theorem}
\qed

 From now, we shall call these equations (\ref{geonah}) as \emph{$k$-cosymplectic Hamiltonian equations}.
\index{$k$-cosymplectic Hamiltonian equations}
\index{$k$-cosymplectic Hamiltonian $k$-vector field}
\begin{definition}
    A $k$-vector field $\mathbf{X}=(X_1,\ldots, X_k)\in \vf^k(M)$ is called a \emph{$k$-cosymplectic Hamiltonian $k$-vector field} for a $k$-cosymplectic Hamiltonian system $(M,\eta^\alpha,\Omega^\alpha, H)$ if $\mathbf{X}$ is a solution of (\ref{geonah}).
    We denote by $\vf^k_H(M)$ the set of all $k$-cosymplectic Hamiltonian $k$-vector fields.
\end{definition}

%\begin{remark} \label{existence kcos sol}
It should be noticed that  equations (\ref{geonah}) have always a solution but  this is not   unique. In fact, if $(M,\eta^\alpha,\Omega^\alpha,V)$ is a
$k$-cosymplectic manifold we can define two vector bundle morphism $\Omega^\flat\colon TM\to (T^1_k)^*M$ and $\Omega^\sharp\colon T^1_kM\to T^*M$ as follows:
    \begin{align*}
        &\Omega^\flat(X)=(\iota_X\Omega^1+\eta^1(X)\eta^1,\ldots, \iota_X\Omega^k+\eta^k(X)\eta^k)\\
        \intertext{and $\Omega^\sharp(X_1,\ldots, X_k)$ such that}
        &\Omega^\sharp(X_1,\ldots, X_k)(Y)={\rm trace} ((\Omega^\flat(X_\beta))_\alpha(Y))=\ds\sum_{\alpha=1}^k(\Omega^\flat(X_\alpha))_\alpha(Y))\\
        \intertext{for all $Y\in TM$, i.e.}
        &\Omega^\sharp(X_1,\ldots, X_k)=\ds\sum_{\alpha=1}^k (\iota_{X_\alpha}\Omega^\alpha + \eta^\alpha(X_\alpha)\eta^\alpha)
    \end{align*}

The above morphisms induce two morphisms of $\mathcal{C}^\infty(M)$-module between the corresponding spaces of sections.
Let us observe that the equations (\ref{geonah}) are equivalent to
\begin{align*}
    \eta^\alpha(X_\beta)&=\delta^\alpha_\beta,\quad \forall \alpha,\beta,\\
    \Omega^\sharp(X_1,\ldots, X_k) &= dH + \ds\sum_{\alpha=1}^k(1-R_\alpha(H))\eta^\alpha\,,
\end{align*}
where $R_1,\ldots, R_k$ are the Reeb vector fields of the $k$-cosymplectic structure $(\eta^\alpha, \Omega^\alpha, V)$.
    \begin{remark}
    {\rm
        If $k=1$ then $\Omega^\flat=\Omega^\sharp$ is defined from $TM$ to $T^*M$, and it is the morphism $\chi_{\eta,\Omega}$ associated to the cosymplectic manifold $(M,\eta,\Omega)$ and defined by
        \[
            \chi_{\eta,\Omega}(X)=\iota_X\Omega+\eta(X)\eta\,,
        \]
        (for more details see  \cite{Albert-1994, CLL-1992} and Appendix \ref{cosymma}).
    \rqed}
    \end{remark}

Next we shall discuss the existence of solution of the above equations. From the local conditions (\ref{k-cosymp condvf}) we can define a $k$-vector field that satisfies (\ref{k-cosymp condvf}), on a neighborhood of each point $x\in M$. For example, we can put
$$(X_\alpha)^\beta= \delta^\beta_\alpha \;, \quad (X_1)^1_i=\ds\frac{\partial H}{\partial q^i}\;, \quad
(X_\alpha)^\beta_i=0 \,\,\,\mbox{for $\alpha\neq 1\neq \beta$} \;, \quad (X_\alpha)^i=
\ds\frac{\partial H}{\partial p^A_i}\, .$$  Now by
using a partition of unity in the manifold $M$, one can construct a
global $k$-vector field which is a solution of (\ref{geonah}), (see \cite{LMORS-1998}.)

Equations (\ref{geonah}) have not, in general, a unique solution. In fact, denoting by $ {\cal M}_k(C^\infty(M))$ the space of matrices of
order $k$ whose entries are functions on $M$, we define the vector bundle
morphism
\begin{equation}\label{etasharp}
\begin{array}{rccl}
\eta^{\sharp}: & T^1_kM & \longrightarrow & {\cal M}_k(C^\infty(M))  \\
\noalign{\medskip}
 & (X_1,\dots,X_k) & \mapsto &
\eta^{\sharp}(X_1,\dots,X_k) =  (\eta_\alpha(X_\beta))\, .
\end{array}
\end{equation}

Then the solutions of (\ref{geonah}) are given by
$(X_1,\dots,X_k)+(\ker\Omega^{\sharp}\cap\ker\eta^{\sharp})$, where
$(X_1,\dots,X_k)$ is a particular solution.

Let us observe that given a $k$-vector field $Y=(Y_1,\ldots, Y_k)$ the condition $Y\in \ker\Omega^{\sharp}\cap\ker\eta^{\sharp}$ is locally equivalent to the conditions
\begin{equation}\label{kerkco}
(Y_\beta)_\alpha=0,\quad Y^i_\beta= 0, \quad \ds\sum_{\alpha=1}^k(Y_\alpha)^\alpha_i=0\,.
\end{equation}

Finally, in the proof of the theorem \ref{fhkc} it is necessary assume the integrability of the $k$-vector field $(X_1,\ldots, X_k)$, and since the $k$ vector fields $X_1,\ldots, X_k$ on $M$ are linearly independent, the integrability condition is equivalent to require that $[X_\alpha,X_\beta]=0$, for all $1\leq \alpha, \beta \leq k$.

\begin{remark}
{\rm
Sometimes the Hamiltonian (or Lagrangian) functions are not defined on a $k$-cosymplectic manifold, for instance, in the reduction theory where the reduced ``phase spaces'' are not, in general, $k$-cosymplectic manifolds, even when the original phase space is a $k$-cosymplectic manifold. In mechanics this problem is solved using Lie algebroids (see \cite{LMM-2005,{MAR-01}, Mart-2001,we-1996}). In \cite{MV-2010} we introduce a geometric description of classical field theories on Lie algebroids in the frameworks of $k$-cosymplectic geometry. Classical field theories on Lie algebroids have already been studied in the literature. For instance, the multisymplectic formalism on Lie algebroids was presented in \cite{MAR-04,Mart-2005}, the $k$-symplectic formalism on Lie algebroids in \cite{LMSV-2009}, in \cite{VC-06} a geometric framework for discrete field theories on Lie groupoids has been discussed.
\rqed}
\end{remark}

\section{Example: massive scalar field}\label{section k-cosymp Examples ham}
\index{Massive Scalar Field}
Consider the   $4$-cosymplectic Hamiltonian equation
                    \begin{equation}\label{exee}
                        \begin{array}{l}
dx^\alpha(X_\beta)=\delta_{\alpha\beta}, \quad 1\leq \alpha,\beta\leq 4\\
\noalign{\medskip}\displaystyle \sum_{\alpha=1}^4 \,
 \iota_{X_\alpha}\Omega^\alpha =
dH-\displaystyle\sum_{\alpha=1}^4 R_\alpha(H)dx^\alpha \, \, .
\end{array}
                    \end{equation}
                associated to the   Hamiltonian function $H\in\mathcal{C}^\infty(\r^4\times (T^1_4)^*\r)$,
                \[
                       H (x^1,x^2,x^3,x^4,q, p^1, p^2, p^3, p^4) =  \ds\frac{1}{2\sqrt{-g}}g_{\alpha\beta}p^\alpha p^\beta-\sqrt{-g}\left( F(q)-\ds\frac{1}{2}m^2q^2\right)\,.
                    \]

               If $(X_1,X_2,X_3,X_4)$ is a solution of (\ref{exee}), then from the following identities
\begin{equation}\label{partial Ham scalar k-co}
                    \derpar{H}{q}=-\sqrt{-g}\Big(F'(q)-m^2q\Big),\quad \derpar{H}{p^\alpha}=\frac{1}{\sqrt{-g}}g_{\alpha\beta}p^\beta\,.
                \end{equation}
                and from  (\ref{k-cosymp condvf}) we obtain,
               in natural coordinates, the local expression of each $X_\alpha$:
                    \begin{equation}\label{vf scalar k-co}
                        X_\alpha=\derpar{}{x^\alpha}+\frac{1}{\sqrt{-g}}g_{\alpha\beta}p^\beta\derpar{}{q} + (X_\alpha)^\beta\derpar{}{p^\beta}\,,
                    \end{equation}
                where the functions $(X_\alpha)^\beta\in \mathcal{C}^\infty(\r^4\times (T^1_4)^*\r)$ satisfies
                    \begin{equation}\label{condition vf scalar k-co}
                        \sqrt{-g}\Big(F'(q)-m^2q\Big)= (X_1)^1+(X_2)^2+(X_3)^3+(X_4)^4\,.
                    \end{equation}

                Let us suppose that ${\bf X}=(X_1,X_2,X_3,X_4)$ is integrable and $\varphi\colon \r^4\to \r^4\times(T^1_4)^*\r,\,$ with \[ \varphi(x)=(x,\psi(x),\psi^1(x),\psi^2(x),\psi^3(x),\psi^4(x))\] is an integral section of ${\bf X}$, then
                  one obtains that  $(\psi(x),\psi^1(x),\psi^2(x),\psi^3(x),\psi^4(x))$  are solution of the following equations
                    \[
                        \begin{array}{l}
                            \derpar{\psi}{x^\alpha}= \frac{1}{\sqrt{-g}}g_{\alpha\beta}\psi^\beta\\
                            \sqrt{-g}\Big(F'(\psi)-m^2\psi\Big)=\derpar{\psi^1}{x^1}+ \derpar{\psi^2}{x^2}+ \derpar{\psi^3}{x^3}+\derpar{\psi^4}{x^4}\,.
                        \end{array}
                    \]

                Thus $\psi\colon \r^4\to \r$ is a solution of the equation
                    \[
                        \sqrt{-g}\Big(F'(\psi)-m^2\psi\Big)= \sqrt{-g}\derpar{}{x^\alpha}\left( g^{\alpha\beta}\derpar{\psi}{t^\beta} \right)\,,
                    \]
                that is, $\psi$ is a solution of the scalar field equation (for more details about this equation see Sections \ref{example k-symp: hamiltonian scalar field} and  \ref{example k-cosymp: hamiltonian scalar field}.)

\chapter{Hamiltonian Classical Field Theories} \label{chapter: k-cosympCFT}

\index{Hamilton-De Donder-Weyl equations}
In this chapter we shall study Hamiltonian Classical Field Theories when the Hamiltonian function involves the space-time coordinates,  that is, $H$ is a function defined on $\rktkqh$. Therefore, we shall discuss the Hamilton-De Donder-Weyl equations (HDW equations for short)  which have the following local expression
\begin{equation}\label{HDW field eq}
                \derpar{H}{q^i}\Big\vert_{\varphi(x)} = -\displaystyle\sum_{\alpha=1}^k\derpar{\psi^\alpha_i}{x^\alpha}\Big\vert_t
                \;, \quad \derpar{H}{p^\alpha_i}\Big\vert_{\varphi(x)}=\derpar{\psi^i}{x^\alpha}\Big\vert_{x}\,.
            \end{equation}

         A solution of these equations is a map
        \[
                    \begin{array}{ccccl}
                     \varphi &:& \r^k &  \longrightarrow & \rktkqh\\ \noalign{\medskip}
                                            & &   x  & \to &  \varphi(x)=(x^\alpha,\psi^i(x),\psi_i^\alpha(x))
                   \end{array} \]
        where $1 \leq i\leq n,\, 1\leq \alpha\leq k$.

In the classical approach these equations can be obtained from a multiple integral variational problem. In this chapter we shall describe this variational approach and, then,  we shall give a new geometric way of obtaining the HDW equations using the $k$-cosymplectic formalism described in chapter \ref{k-cosymp eq} when the $k$-cosymplectic manifolds is just the geometrical model, i.e.  $(M=\rktkqh,\eta^1,\ldots,\eta^k,\Omega^1,\ldots, \Omega^k, V)$ as it has been described in section \ref{section: rktkqh}.

    \section{Variational approach}\label{section k-cosymp Ham variational}

In this subsection we shall see that the HDW field equations (\ref{HDW field eq})
                                are obtained from a variational principle on the space of smooth maps  with compact support $\mathcal{C}^\infty_C(\rk,\rktkqh)$. %A more detail description of this variational principle is given in this section.

             To describe this variational principle we need the notion of prolongation  of diffeomorphism and vector fields from $Q$ to $\rktkqh$, which we shall introduce now. First, we shall describe some properties of the $\pi_Q$-projectable vector fields.
      \subsection{Prolongation of vector fields.}\label{section rktkqh prolongation}
\index{Vector field}
    On the manifold $\rktkqh$ there exist two families or groups of  vector fields that are relevant for our purposes. The first of these families is the set of vector fields which are obtained as  canonical prolongations of  vectors field on $\rkq$ to $\rktkqh$. The second interesting family is the set of $\pi_Q$-projectable vector fields defined on  $\rkq$. In this paragraph we briefly describe these two sets of vector fields.
\index{Vector field!projectable}
    \begin{definition}\label{projectable vf}
        Let $Z$ be a vector field on $\rkq$. $Z$ is say to be \emph{$\pi_Q$-projectable} if there exists a vector field $\bar{
        Z}$ on $Q$, such that
            \[
                (\pi_Q)_*\circ Z= \bar{Z}\circ \pi_Q\,.
            \]
    \end{definition}

    To find the coordinate representation of the vector field $Z$ we use coordinates   $(x^\alpha, q^i,\dot{x}^\alpha, \dot{q}^i)$ on $T(\rkq)$ and $(x^\alpha, q^i)$ on $\rkq$. Since $Z$ is a section of $T(\rkq)\to \rkq$, the $x^\alpha$ and $q^i$ components of the coordinate representation are fixed, so that $Z$  is determined by the functions $Z_\alpha=\dot{x}^\alpha\circ Z$ and $Z^i=\dot{q}^i\circ Z$, i.e,
    \[
        Z(x,q)=Z_\alpha(x,q)\derpar{}{x^\alpha}\Big\vert_{(x,q)} + Z^i(x,q)\derpar{}{q^i}\Big\vert_{(x,q)}\,.
    \] On the other hand,   $\bar{Z}\in \vf(Q)$,   is locally expressed by
    \[
        \bar{Z}(q)=\bar{Z}^i(q)\derpar{}{q^i}\Big \vert_{q}\,.
    \] where $\bar{Z}^i\in C^\infty(Q)$.

    Now the condition of the definition \ref{projectable vf} implies that
    \[
        Z^i(x,q)=(\bar{Z}^i\circ \pi_Q)(x,q)=\bar{Z}^i(q)\,.
    \]

    We usually write $Z^i$ instead of $\bar{Z}^i$. With this notation we have
    \[
        \begin{array}{lcl}
        Z(x,q)&=&Z_\alpha(x,q)\derpar{}{x^\alpha}\Big\vert_{(x,q)} + Z^i(q)\derpar{}{q^i}\Big\vert_{(x,q)}\,,\\\noalign{\medskip}
        \bar{Z}(q)&=&Z^i(q)\derpar{}{q^i}\Big\vert_{q}\,.
        \end{array}
    \]

    As a consequence,   we deduce that if $\{\sigma_s\}$ is the one-parameter group of diffeomorphism associated to a $\pi_Q$-projectable vector field $Z\in\vf(\rkq)$ and  $\{\bar{\sigma}_s\}$ is the one-parameter group of diffeomorphism associate to $\bar{Z}$, then
    \[
        \bar{\sigma}_s\circ\pi_Q =\pi_Q\circ \sigma_s\,.
    \]
    Given a $\pi_Q$-projectable vector field $Z\in \vf(\rkq)$, we can define a vector field $Z^{1*}$ on $\rktkqh$ such that it is $(\pi_Q)_{1,0}$-projectable and its projection is $Z$. Here we give the idea of the definition. A complete description of this notion can be found in \cite{Saunders-89} where the author define the prolongation of vector fields to the first-order jet bundle of an arbitrary bundle $E\to M$.

    Before to construct the  prolongation of a vector field it is necessary the following definition:
\index{First prolongation}
    \begin{definition}
        Let $f\colon \rkq\to \rkq$ be a map and $\bar{f}\colon Q\to Q$ be a diffeomorphism, such that $\pi_Q\circ f=\bar{f}\circ \pi_Q$. The \emph{first prolongation} of $f$ is a map
            \[
                j^{1\,*}f\colon J^1(Q,\rk)\equiv\rktkqh\to J^1(Q,\rk)\equiv\rktkqh
            \]
        defined by
            \begin{equation}\label{j1*f}
                (j^{1\,*}f)(j_{q,\sigma(q)}\sigma)= j^1_{(\bar{f}(q),\tilde{\sigma}(\bar{f}(q)))}\tilde{\sigma}
            \end{equation}
        where $\sigma\colon Q\to \rk,\, j_{q,\sigma(q)}\sigma\in J^1(Q,\rk)$ and $\tilde{\sigma}\colon Q\to \rk$ is the map given by the composition
            \[
                \xymatrix@C=17mm{
                        Q \ar@/^{16mm}/[rrrr]^-{\tilde{\sigma}}\ar[r]^-{\bar{f}^{\,-1}}& Q\ar[r]^-{(\sigma,id_Q)} & \rkq \ar[r]^-{f}& \rkq\ar[r]^-{\pi_{\rk}} & \rk
                        }
            \]
            i.e. $\tilde{\sigma}= \pi_{\rk}\circ f\circ (\sigma, id_Q)\circ \bar{f}^{\,-1}$.
    \end{definition}

    \begin{remark}
    {\rm
        In a general bundle setting \cite{Saunders-89},  the conditions of the above definition are equivalent to say that the pair $(f,\bar{f})$ is a bundle automorphism of the bundle $\rkq\to Q$.\rqed
    }
    \end{remark}
    \begin{remark}
    {\rm
        If we consider the identification between $J^1(Q,\rk)$ and $J^1\pi_Q$ given in Remark \ref{difeo jipiQ}, the above definition coincides with the definition 4.2.5 in \cite{Saunders-89} of the first prolongation of $f$ to the jet bundles.\rqed}
    \end{remark}

    In a local coordinates system $(x^\alpha, q^i, p^\alpha_i)$ on $J^1(Q,\rk)\equiv \rktkqh$, if $f(x,q)=(f^\alpha(x,q), \bar{f}^i(q))$, then
    \begin{equation}\label{Kcos local j1*f}
    {\small
        j^{1\,*}f=\Big( f^\alpha(x^\beta,q^j), \bar{f}^i(q^j), \Big(\derpar{f^\alpha}{q^k}+p^\beta_k\derpar{f^\alpha}{x^\beta} \Big)\Big(\derpar{(\bar{f}^{-1})^k}{\bar{q}^i}\circ \bar{f}\Big)(q^j)
        \Big)\,,
        }
    \end{equation}
    where $\bar{q}^i$ are the coordinates on $Q=\bar{f}(Q)$.

    Now we are in conditions to give the definition of prolongation of $\pi_Q$-projectable vector field.
    \begin{definition}\label{lift vf rktkqh}
        Let $Z\in\mathfrak{X}(\rkq)$ be a $\pi_Q$-projectable vector field, with local $1$-parameter group of diffeomorphisms $\sigma_s:\rkq\to \rkq$. Then the local $1$-parameter group of diffeomorphisms $j^{1*}\sigma_s: \rktkqh \to \rktkqh $ is  generated by  a vector field $Z^{1*}\in\mathfrak{X}( \rktkqh )$, called \emph{the natural prolongation (or complete lift) of $Z$ to $ \rktkqh $.}
    \end{definition}
\index{Vector field!complete lift}

    In local coordinates, if $Z\in\vf(\rkq)$ is a $\pi_Q$-projectable vector field with local expression,
        \[
            Z=Z_\alpha\ds\frac{\partial }{\partial x^\alpha}+Z^i\ds\frac{\partial}{\partial q^i}\,,
        \]
    then  from (\ref{Kcos local j1*f}) and Definition \ref{lift vf rktkqh} we deduce that the natural prolongation  $Z^{1*}$  has the following local expression
        \begin{equation}\label{local z1*}
            Z^{1*}=Z_\alpha\ds\frac{\partial }{\partial         x^\alpha}+Z^i\ds\frac{\partial}{\partial q^i}+\left(\ds\frac{d         Z_\alpha}{dq^i}-p^\beta_j\ds\frac{dZ^j}{dq^i}\right)\ds\frac{\partial         }{\partial p^\beta_i}\,,
        \end{equation}
    where $d/dq^i$ denoted the total derivative, that is,
        \[
            \ds\frac{d}{dq^i}=\ds\frac{\partial }{\partial q^i}+p^\beta_i\ds\frac{\partial }{\partial x^\beta}\,.
        \]
    \begin{remark}
    {\rm
        In the general framework of first order jet bundles, there exists a notion of prolongation of vector field which reduces to  definition \ref{lift vf rktkqh} when one considers the trivial bundle $\pi_Q\colon \rkq\to Q$. For a full description in the general case, see \cite{Saunders-89}.\rqed}
    \end{remark}

   \subsection{Variational principle}

        Now we are in conditions to describe the multiple integral problem from which one obtains the Hamilton-De Donder-Weyl equations.

 \begin{definition} Denote by $Sec_c (\rk, \rk \times (T^1_k)^*Q)$ the set of sections with compact support of the bundle $$\pi_{\rk}
        \circ(\pi_Q)_{1,0}
         :\rk\times(T^1_k)^*Q \to \rk.$$

  Let  $H\colon\rktkqh\to\r$ be a Hamiltonian function: then we define the integral action associated to $H$ by
        \[\begin{array}{lrcl}
        \mathbb{H}: & Sec_c(\rk,\rk \times (T^1_k)^*Q)&\to
        &\r\\\noalign{\medskip} & \varphi &\mapsto &
         \mathbb{H}(\varphi) \, =\,\ds
        \int_{\ds \rk}\, \varphi^* \Theta\,,
        \end{array}\]
where
\begin{equation}\label{The0} \Theta=\ds\sum_{\alpha=1}^k\theta^\alpha\wedge d^{k-1}x_\alpha-Hd^kx
,   \end{equation}    is a $k$-form on $\rk \times (T^1_k)^*Q$
being
 $d^{k-1}x_\alpha=\iota_{\frac{\partial}{\partial x^\alpha}} d^kx$ and
 $d^kx=dx^1\wedge\cdots\wedge dx^k$ as in section \ref{section HDW eq variational}.
\end{definition}

\begin{remark}{\rm
In the above definition we consider   the following  commutative diagram

 \[\xymatrix{  \rktkqh\ar[r]_-{(\pi_Q)_{1,0}}& \rkq\ar[d]^-{\pi_{\mathbb{R}_k}}\\
  \rk\ar[u]^-{\varphi}\ar[r]^-{Id_{\rk}}& \rk}
 \] \rqed}\end{remark}

With the aim to describe the extremals of $\mathbb{H}$ we first prove the following
\begin{lemma} Let $\varphi\colon\rk\to\rktkqh$ be an element of $Sec_c(\rk,\rk \times
(T^1_k)^*Q)$. For each  $\pi_{\rk}$-vertical vector field $Z\in\vf(\rkq)$ with one-parameter group of diffeomorphism
$\{\sigma_s\}$ one has that
 $$\varphi_s:=j^{1*}\sigma_s\circ\varphi$$ is a section of the canonical projection $\pi_{\rk}
\circ(\pi_Q)_{1,0}:\rk\times (T^1_k)^*Q\to\rk$.\end{lemma}

\proof  If $Z$ is  $\pi_{\rk}$-vertical vector field,
then one has the following local expression
\begin{equation}\label{Zloc*}
Z(x,q)=Z^i(x,q)\ds\frac{\partial}{\partial
q^i}\Big\vert_{(x,q)}\,,
\end{equation} for some $Z^i\in
\mathcal{C}^\infty(\rkq)$.

Since  $\{\sigma_s\}$ is the one-parameter group of diffeomorphism of $Z$ one obtains
$$
\begin{array}{lcl}
Z(x,q)&=&(\sigma_{(x,q)})_*(0)\Big(\ds\frac{d}{ds}\Big\vert_{0}\Big)\\\noalign{\medskip}&=&\ds\frac{d
(x^\alpha\circ \sigma_{(x,q)})}{ds}\Big\vert_{0}\ds\frac{\partial
}{\partial x^\alpha}\Big\vert_{(x,q)}+\ds\frac{d (q^i\circ
\sigma_{(x,q)})}{ds}\Big\vert_{0}\ds\frac{\partial }{\partial
q^i}\Big\vert_{(x,q)}\,.
\end{array}
$$

Comparing  (\ref{Zloc*}) and the above expression of $Z$ at an arbitrary point $(x,q)\in\rkq$, we have
$$
\ds\frac{d (x^\alpha\circ \sigma_{(x,q)})}{ds}\Big\vert_{0}=0\,,
$$
and then we deduce that  $$(x^\alpha\circ \sigma_{(x,q)})(s)=\makebox{
constant}\,,$$ but as $\sigma_{(x,q)}(0)=(x,q)$ we know that
$(x^\alpha\circ \sigma_{(x,q)})(0)=x^\alpha$ and, thus,
$$(x^\alpha\circ \sigma_{(x,q)})(s)= x^\alpha.$$ Then
$$\sigma_s(x,q)=(x,q^i\circ \sigma_s(x,q)),$$ which implies $\pi_{\rk}\circ\sigma_s=\pi_{\rk}\,.$

Now, from (\ref{Kcos local j1*f}) one obtains
$$\begin{array}{lcl}
&&\pi_{\rk}\circ(\pi_Q)_{1,0}\circ\varphi_s(x)=\pi_{\rk}\circ(\pi_Q)_{1,0}\circ
j^{1\,*}\sigma_s\circ\varphi(x)\\\noalign{\medskip}&=&
\pi_{\rk}\circ(\pi_Q)_{1,0}(x,(\sigma_s)_{Q}^i(q ),p^\alpha_k\ds\frac{\partial
((\sigma_s)_{Q}^{-1})^k}{\partial q^i}\circ (\sigma_s)_{Q})=x
\end{array}$$
that is, $\varphi_s$ is a section of
$\pi_{\rk}\circ(\pi_Q)_{1,0}$.\qed
 \begin{definition} A section $\varphi:\rk\to\rktkqh\in Sec_c (\rk, \rk \times (T^1_k)^*Q)$, is an {\bf extremal} of
 $\mathbb{H}$ if
 $$\ds\frac{d}{ds}\Big\vert_{s=0}\mathbb{H}(j^{1*}\sigma_s\circ\varphi)=0$$
 where $\{\sigma_s\}$ is the one-parameter group of diffeomorphism of some $\pi_{\rk}$-vertical and $\pi_Q$-projectable vector field  $Z\in \vf(\rkq)$.\end{definition}

\begin{remark}
{\rm
 In the above definition it is a necessary that $Z$ is a $\pi_{\rk}$-vertical vector field
  to guarantee that each
 $$\varphi_s:=j^{1*}\sigma_s\circ\varphi$$ is a section of the canonical projection $\pi_{\rk}
\circ(\pi_Q)_{1,0}:\rk\times (T^1_k)^*Q\to\rk$, as we have proved in the above lemma.\rqed}\end{remark}

 The multiple integral variational problem associated to a Hamiltonian $H$ consists to obtain the extremals of the integral action $\mathbb{H}$.

\begin{theorem} Let    be $\varphi \in Sec_c (\rk, \rk \times
(T^1_k)^*Q)$  and $H\colon\rktkq\to\r$  a Hamiltonian function. The following statements are equivalents:
\begin{enumerate}
    \item $\varphi$ is an extremal of the variational problem associated to $H$

    \item $\ds\int_{ \rk}\, \varphi^* \mathcal{L}_{\ds Z^{1\,*}}\,\Theta
    = 0$, for each vector field $Z\in\vf(\rkq)$ which is  $\,\pi_{\rk}$-vertical and $\pi_Q$-projectable.
    \item $\varphi^* \iota_{\ds Z^{1\,*}}\, d\Theta
    = 0$, for each $\,\pi_{\rk}$-vertical and $\pi_Q$-projectable vector field $Z$.
        \item  If $(U;x^\alpha,q^i,p^\alpha_i)$ is a canonical system of coordinates  on $\rk \times (T^1_k)^*Q$, then $\varphi$ satisfies the Hamilton-De Donder-Weyl equations (\ref{HDW field eq}).
\end{enumerate}
\end{theorem}

\dem

$(1\Leftrightarrow 2)$ Let $Z\in\mathfrak{X}(\rkq)$ be a   $\pi_{\rk}$-vertical and  $\pi_Q$-projectable vector field.
Denote by $\{\sigma_s\}$ the one-parameter group of diffeomorphism associated to $Z$.

Under these conditions we have

\[\begin{array}{lcl}
\ds\frac{d}{ds}\Big\vert_{s=0}\mathbb{H}(j^{1*}\sigma_s\circ\varphi)&=&\ds\frac{d}{ds}\Big\vert_{s=0}\ds\int_{\rk}
(j^{1*}\sigma_s\circ\varphi)^*\Theta\\\noalign{\medskip}&=&
\ds\lim_{s\to
0}\ds\frac{1}{s}\left(\ds\int_{\rk}(j^{1*}\sigma_s\circ\varphi)^*\Theta
-\ds\int_{\rk}\varphi^*\Theta\right)\\\noalign{\medskip}&=&
\ds\lim_{s\to
0}\ds\frac{1}{s}\left(\ds\int_{\rk}\varphi^*\Big((j^{1*}\sigma_s)^*\Theta\Big)
-\ds\int_{\rk}\varphi^*\Theta\right)
\\\noalign{\medskip}&=&
\ds\lim_{s\to
0}\ds\frac{1}{s}\ds\int_{\rk}\varphi^*[(j^{1*}\sigma_s)^*\Theta-\Theta
]
\\\noalign{\medskip}&=&\ds\int_{\rk}\varphi^*\mathcal{L}_{Z^{1\,*}}\Theta\,.
\end{array}\]

Therefore we obtain the equivalence between the items $(1)$ and $(2)$.

$(2\Leftrightarrow 3)$ Taking into account that
\[\mathcal{L}_{Z^{1\,*}}\Theta=d\iota_{Z^{1\,*}}\Theta+\iota_{Z^{1\,*}}d\Theta
\,,\] one obtains
\[\ds\int_{\rk}\varphi^*\mathcal{L}_{Z^{1\,*}}\Theta=\ds\int_{\rk}\varphi^*
d\iota_{Z^{1\,*}}\Theta+\ds\int_{\rk}\varphi^*\iota_{Z^{1\,*}}d\Theta
\]
and since $\varphi$ has compact support, using Stokes's theorem we deduce
\[
\ds\int_{\rk}\varphi^*d\iota_{Z^{1\,*}}\Theta=\ds\int_{\rk}
d\varphi^*\iota_{Z^{1\,*}}\Theta\,=0\,,\] and  then
\[\ds\int_{\rk}\varphi^*\mathcal{L}_{Z^{1\,*}}\Theta=0\] (for each $Z$
vector field $\pi_{\rk}$-vertical) if and only if,
$$\ds\int_{\rk}\varphi^*\iota_{Z^{1\,*}}d\Theta=0\, .$$
But by the fundamental theorem of the variational calculus, the last equality is  equivalent to  $$\varphi^*\iota_{Z^{1\,*}}d\Theta=0.$$

\bigskip

$(3\Leftrightarrow 4)$ Suppose that $$\varphi:\rk\to\rk\times
(T^1_k)^*Q$$ is a section of $\pi_{\rk}\circ
(\pi_{Q})_{1,0}$ such that
$$\varphi^*\iota_{Z^{1\,*}}d\Theta=0\,,$$ for each
$Z\in\mathfrak{X}(\rkq)$
$\pi_{\rk}$-vertical and $\pi_Q$-projectable vector field.

In canonical coordinates we have $$Z=Z^i\ds\frac{\partial }{\partial
q^i},$$ for some functions $Z^i\in\mathcal{C}^\infty(Q)$
then; from (\ref{local z1*}) one has
\[Z^{1\,*}=Z^i\ds\frac{\partial}{\partial
q^i}-p^\alpha_j \ds\frac{\partial Z^j}{\partial
q^i}\ds\frac{\partial}{\partial p^\alpha_i}\,.\]

Therefore,

\begin{equation}\label{phj}
\begin{array}{lcl} \ds\iota_{\ds
Z^{1\,*}}d\Theta&=&\ds\iota_{\ds
Z^{1\,*}}\left(\ds\sum_{\alpha=1}^kdp^\alpha_i\wedge dq^i\wedge
d^{k-1}x_\alpha-dH\wedge d^kx \right)\\\noalign{\medskip}&=&
-Z^i\left(\ds\sum_{\alpha=1}^kdp^\alpha_i\wedge d^{k-1}x_\alpha+\ds\frac{\partial
H}{\partial q^i}d^kx\right)=\ds\sum_{\alpha=1}^k p^\alpha_j \ds\frac{\partial Z^j}{\partial q^i}\left(dq^i\wedge d^{k-1}x_\alpha -\ds\frac{\partial
H}{\partial p^\alpha_i}d^kx\right)\,.
\end{array}
\end{equation}

 So, if $$\varphi(x)=(x^\alpha,\psi^i(x),\psi^\alpha_i(x))$$ then       $q^i=\psi^i(x)$    and
 $p^\alpha_ i= \psi^\alpha_ i(x)$ along   the image of $\varphi$. Using  (\ref{phj}) and taking into account that $Z^i(x):=(Z^i\circ\varphi)(x)$ one has,
\[\begin{array}{rl}
0=&[\varphi^*\ds\iota_{\ds
Z^{1\,*}}d\Theta](x)=-Z^i(x)\left(\ds\sum_{\alpha=1}^k\ds\frac{\partial
\psi^\alpha_i}{\partial x^\alpha}\Big\vert_{x}+\ds\frac{\partial
H}{\partial q^i}\Big\vert_{\varphi(x)}\right)d^kx
\\\noalign{\medskip}-&
\ds\sum_{\alpha=1}^k \psi^\alpha_j(x) \ds\frac{\partial Z^j}{\partial
q^i}\Big\vert_{x}\left(\ds\frac{\partial \psi^i}{\partial
x^\alpha}\Big\vert_{x}-\ds\frac{\partial H}{\partial
p^\alpha_i}\Big\vert_{\varphi(x)}\right)d^kx
\\\noalign{\medskip}=&\left[-Z^i(x)\left(\ds\sum_{\alpha=1}^k\ds\frac{\partial
\psi^\alpha_i}{\partial x^\alpha}\Big\vert_{x}+\ds\frac{\partial
H}{\partial q^i}\Big\vert_{\varphi(x)}\right)-\ds\sum_{\alpha=1}^k
\psi^\alpha_j(x) \ds\frac{\partial Z^j}{\partial
q^i}\Big\vert_{x}\left(\ds\frac{\partial \psi^i}{\partial
x^\alpha}\Big\vert_{x}-\ds\frac{\partial H}{\partial
p^\alpha_i}\Big\vert_{\varphi(x)}\right)\right]d^kx
\end{array}\]
for any $\pi_{\rk}$-ver\-ti\-cal and
$\pi_Q$-projectable vector field $Z$.

The above identity is equivalent to the following expression:
\[Z^i(x)\left(\ds\sum_{\alpha=1}^k\ds\frac{\partial
\psi^\alpha_i}{\partial x^\alpha}\Big\vert_{x}+\ds\frac{\partial
H}{\partial q^i}\Big\vert_{\varphi(x)}\right)+\ds\sum_{\alpha=1}^k
\psi^\alpha_j(x) \ds\frac{\partial Z^j}{\partial
q^i}\Big\vert_{x}\left(\ds\frac{\partial \psi^i}{\partial
x^\alpha}\Big\vert_{x}-\ds\frac{\partial H}{\partial
p^\alpha_i}\Big\vert_{\varphi(x)}\right)=0\,,\] for each
$Z^i(x^\alpha,q^j)$. Therefore,
\begin{equation}\label{Principio de Hamilton: ecuacion 1}
\begin{array}{lcl}
  Z^i(x)\left(\ds\sum_{\alpha=1}^k\ds\frac{\partial
\psi^\alpha_i}{\partial x^\alpha}\Big\vert_{x}+\ds\frac{\partial
H}{\partial q^i}\Big\vert_{\varphi(x)}\right) &=& 0
\\\noalign{\medskip}
 \ds\sum_{\alpha=1}^k \psi^\alpha_j(x) \ds\frac{\partial Z^j}{\partial q^i}\Big\vert_{x}\left(\ds\frac{\partial
\psi^i}{\partial x^\alpha}\Big\vert_{x}-\ds\frac{\partial
H}{\partial p^\alpha_i}\Big\vert_{\varphi(x)}\right) &=&
0\,.\end{array}
\end{equation}

From the first of the identities of (\ref{Principio de Hamilton:
ecuacion 1}) one obtain the first set of the Hamilton-De Donder-Weyl field equations, that is,
\[\ds\sum_{\alpha=1}^k\ds\frac{\partial
\psi^\alpha_i}{\partial x^\alpha}\Big\vert_{x}=-\ds\frac{\partial
H}{\partial q^i}\Big\vert_{\varphi(x)}\,.\]

Consider now a coordinate neighborhood $(U; x^\alpha, q^i,p^\alpha_i)$. Since there exists a critical section for each point on $U$,  one obtains that
\[ \ds\frac{\partial Z^j}{\partial q^i}\Big\vert_{x}\left(\ds\frac{\partial \psi^i}{\partial
x^\alpha}\Big\vert_{x}-\ds\frac{\partial H}{\partial
p^\alpha_i}\Big\vert_{\varphi(x)}\right)=0\,,
\]
using the second identity of  (\ref{Principio de Hamilton: ecuacion 1}).

Finally, as the $Z^i$ can be arbitrarily choose, then $ \ds\frac{\partial Z^j}{\partial
q^i}\Big\vert_{x}$ can take arbitrary values, and thus we have,
\[\ds\frac{\partial \psi^i}{\partial
x^\alpha}\Big\vert_{x}-\ds\frac{\partial H}{\partial
p^\alpha_i}\Big\vert_{\varphi(x)}=0\,,
\] which is the second set of the Hamilton-De Donder-Weyl equations.

The converse can be proved by reversing the above arguments. \qed

\begin{remark} {\rm A.
Echeverr\'{\i}a-Enr\'{\i}quez \textit{et al.}  have obtained in
\cite{EMR-00} a similar result but considering a variational principle in the multisymplectic setting.}
\rqed\end{remark}

    \section{Hamilton-De Donder-Weyl equations: k-cosymplectic approach}\label{Section HDW eq: k-cosymp approach}

The above variational principle allows us to obtain the HDW equations but as in the case of the Hamiltonian functions independent of the space-time coordinates, there exist  another methods to obtain these equations. In this section we describe as these equations can be obtained  using the $k$-cosymplectic Hamiltonian system  when we consider the  $k$-cosymplectic manifold $M=\rktkqh$ with the canonical $k$-cosymplectic structure.

So, we now consider a $k$-cosymplectic Hamiltonian system $$(\rktkqh, \eta^\alpha,\Omega^\alpha, H),$$ where the Hamiltonian function $H$ is now a function defined on $\rktkqh$. From the Theorem \ref{fhkc} one obtains that if $X=(X_1,\ldots, X_k)\in \vf^k_H(\rktkqh)$ is an integrable $k$-vector field and  $\varphi\colon U\subset\rk\to \tkqh$ is an integral section of $X$, then $\varphi$ is a solution of the following systems of partial differential equations
                    \[
                            \frac{\partial H}{\partial q^i}\Big\vert_{\varphi(x)} \, = \, -
                            \displaystyle\sum_{\beta=1}^k \frac{\partial \psi^\beta_i} {\partial
                            x^\beta}\Big\vert_{x} \,, \quad \frac{\partial H}{\partial
                            p^{\alpha}_i}\Big\vert_{\varphi(x)} \, = \, \frac{\partial \psi^i}{\partial
                            x^\alpha}\Big\vert_{x}\,,
                        \]
                        that is $\varphi$ is a solution of the HDW equations (\ref{HDW field eq}).

Therefore, the Hamilton-De Donder-Weyl equations are a particular case of the system of partial differential equations which one can obtain from the $k$-cosymplectic equation.

\begin{remark}
{\rm
  In the case $k=1$, with $M=\r\times T^*Q$, the equations
(\ref{geonah}) reduces to the equations of the non-autonomous Hamiltonian Mechanics. Therefore the formalism described here includes as a particular case the Hamiltonian formalism for non-autonomous Mechanics.%\rqed
\rqed}
\end{remark}

%\newpage
%\mbox{}
%\thispagestyle{empty} % para que no se numere esta p\'{a}gina

\chapter{Hamilton-Jacobi equation}

There are several attempts to extend the Hamilton-Jacobi theory  for classical field theories. In \cite{LMMSV-2010} we have described this theory in the framework of the so-called $k$-symplectic formalism \cite{Awane-1992,Gu-1987,LMS-88,LMS-91}. In this section we consider the $k$-cosymplectic framework. Another attempts  in the framework of the multisymplectic formalism \cite{CIL-1999,{KT-79}} have been discussed in \cite{LMM-2009,{PR-2002-b},{PR-2002}}.

In Classical Field Theory the Hamilton-Jacobi equation is (see \cite{Rund})
\begin{equation}\label{HJCFT}
    \derpar{W^\alpha}{x^\alpha} + H\Big(x^\beta, q^i, \derpar{W^\alpha}{q^i}\Big)=0
\end{equation}
where $W^1,\ldots, W^k\colon \rkq\to \r$, $1\leq\alpha\leq k$.

The classical statement of time-dependent Hamilton-Jacobi equation is the following \cite{AM-1978}:
\begin{theorem}
    Let $H\colon \r\times T^*Q\to \r$ be a Hamiltonian and $T^*Q$ the symplectic manifold with the canonical symplectic structure $\omega=-d\theta$. Let $X_{H_t}$  be a Hamiltonian vector field on $T^*Q$ associated to the Hamiltonian  $H_t\colon T^*Q\to \r,\, H_t(\nu_q)=H(t,\nu_q)$, and $W\colon \r\times Q\to \r$ be a smooth function. The following two conditions are equivalent:
    \begin{enumerate}
        \item for every curve $c$ in $Q$ satisfying
            \[
                c'(t)=\pi_*\Big( X_{H_t}(dW_t(c(t)))\Big)
            \]
            the curve $t\mapsto dW_t(c(t))$ is an integral curve of $X_{H_t}$, where  $W_t\colon Q\to \r,\, W_t(q)=W(t,q)$.
        \item $W$ satisfies the Hamilton-Jacobi equation
            \[
                H(x,q^i,\derpar{W}{q^i}) + \derpar{W}{t}=\makebox{ constant on } T^*Q
            \]
            that is,
            \[
                H_t\circ dW_t+ \derpar{W}{t}= K(t)\,.
            \]
    \end{enumerate}
\end{theorem}

Now we will extend this result to  Classical Field Theories.
\section{The Hamilton-Jacobi equation}\label{HJsection}
In this section we introduce a geometric version of the Hamilton-Jacobi theory based in the $k$-co\-sym\-plec\-tic formalism. In the particular case $k=1$ we recover the above theorem for time-dependent Classical Mechanics.

 For each $x=(x^1,\ldots, x^k)\in \rk$ we consider the following mappings
\[
    \begin{array}{lcl}
        \begin{array}{rccl}
         i_x\colon & Q & \to& \rk\times Q\\\noalign{\medskip}
          & q & \mapsto & (x,q)
        \end{array} & \quad and \quad & \begin{array}{rccl}
         j_x\colon & \tkqh & \to & \rktkqh\\\noalign{\medskip}
          & ({\nu_1}_q,\ldots, {\nu_k}_q) & \mapsto & (x,{\nu_1}_q,\ldots, {\nu_k}_q)
        \end{array}
    \end{array}
\]

Let $\gamma\colon \rkq\to \rktkqh$ be a section of $(\pi_Q)_{1,0}$. Let us observe that given a section $\gamma$ is equivalent to giving a section $\bar{\gamma}\colon \rkq\to \tkqh$ of $\pi^k\colon \tkqh\to Q$ along the map $\pi_Q\colon \rkq\to Q$.% \footnote{Let $\pi\colon E\to M$ be a fibre bundle and $\phi\colon N\to M$ a differentiable map. A section of $\pi$ along $\phi$ is a map $\sigma\colon N\to E$ such that the following diagram is commutative
%\[
%    \xymatrix{ & E\ar[d]^-{\pi}\\ N\ar[ru]^-{\sigma}\ar[r]^-{\phi} & M}
%\]
%that is, $\pi\circ \sigma=\phi$.}.

If fact, given $\gamma$ we define $\bar{\gamma}=\bar{\pi}_2\circ \gamma$ where $\bar{\pi}_2$ is the canonical projection $\bar{\pi}_2\colon \rktkqh\to \tkqh$. Conversely, given $\bar{\gamma}$ we define $\gamma$ as the composition $\gamma(x,q)=(j_x\circ \bar{\gamma})(x,q)$. Now, since $\tkqh$ is the Whitney sum of $k$ copies of the cotangent bundle, giving $\gamma$ is equivalent to giving a family $(\bar{\gamma}^1,\ldots, \bar{\gamma}^k) $ of $1$-forms along the map $\pi_Q$, that is $\pi\circ\gamma^\alpha=\pi_Q$.

If we consider local coordinates $(x^\alpha, q^i, p^\alpha_i)$ we have the following local expressions:
\begin{equation}\label{localgamma}
    \begin{array}{l}
    \gamma(x^\alpha, q^i)= (x^\alpha, q^i, \gamma^\beta_j(x^\alpha, q^i))\,,\\\noalign{\medskip}
    \bar{\gamma}(x^\alpha, q^i) = (q^i, \gamma^\beta_j(x^\alpha, q^i))\,,\\\noalign{\medskip}
    \bar{\gamma}^\alpha(x,q)=\gamma^\alpha_j(x,q)dq^j(q)\,.
    \end{array}
\end{equation}

%Moreover, along this section we suppose that each $\bar{\gamma^\alpha}$ is a closed form. In local coordinates, using the local expressions (\ref{localgamma}), this condition implies that
%\begin{equation}\label{gammacond}
%    \derpar{\gamma^\alpha_i}{x^\beta}=0\quad \makebox{and}\quad \derpar{\gamma^\alpha_i}{q^j}= \derpar{\gamma^\alpha_j}{q^i}\,.
%\end{equation}

Moreover, along this section we suppose that each $\bar{\gamma}^\alpha$ satisfies that its exterior differential $d\bar{\gamma}^\alpha$ vanishes over two $\pi_{\rk}$-vertical vector fields. In local coordinates, using the local expressions (\ref{localgamma}), this condition implies that
\begin{equation}\label{gammacond}
    \derpar{\gamma^\alpha_i}{q^j}= \derpar{\gamma^\alpha_j}{q^i}\,.
\end{equation}

Now, let $Z=(Z_1,\ldots, Z_k)$ be a $k$-vector field on $\rktkqh$. Using $\gamma$ we can construct a $k$-vector field $Z^\gamma=(Z^\gamma_1,\ldots, Z^\gamma_k)$  on $\rkq$ such that the following diagram is commutative
\[
 \xymatrix{ \rktkqh
\ar@/^1pc/[dd]^{(\pi_Q)_{1,0}} \ar[rrr]^{Z}&   & &T_k^1(\rktkqh)\ar[dd]^{T^1_k(\pi_Q)_{1,0}}\\
  &  & &\\
 \rkq\ar@/^1pc/[uu]^{\gamma}\ar[rrr]^{Z^{\gamma}}&  & & T_k^1(\rkq) }
\]
that is,
$$Z^\gamma:= T^1_k(\pi_Q)_{1,0}\circ Z\circ \gamma\;.$$

Let us recall that for an arbitrary differentiable map $f:M\to N$, the induced map $T^1_kf:T^1_kM\to T^1_kN$ of $f$ is defined by (\ref{tkq: prolongation expr}).

Let us observe that if $Z$ is integrable then $Z^\gamma$ is also integrable.

In local coordinates, if each $Z_\alpha$ is locally given by
\[
    Z_\alpha= (Z_\alpha)_\beta\derpar{}{x^\beta} + Z^i_\alpha\derpar{}{q^i} + (Z_\alpha)^\beta_i\derpar{}{p^\beta_i}\,
\]
then $Z^\gamma_\alpha$ has the following local expression:
\begin{equation}\label{local z gamma kcos}
    Z^\gamma_\alpha= \big((Z_\alpha)_\beta\circ \gamma\big)\derpar{}{x^\beta} + (Z^i_\alpha\circ \gamma)\derpar{}{q^i}\,.
\end{equation}

In particular, if we consider the $k$-vector field $R=(R_1,\ldots, R_k)$ given by the Reeb vector fields, we obtain, by a similar procedure, a $k$-vector field $(R_1^\gamma,\ldots, R_k^\gamma)$ on $\rk\times Q$. In local coordinates,  since $$R_\alpha=\frac{\partial}{\partial x^\alpha}$$ we have
\[
    R_\alpha^\gamma =\derpar{}{x^\alpha}\,.
\]

Next, we consider a Hamiltonian function $H\colon\rktkq \to \mathbb{R}$, and the corresponding Hamiltonian system on $\rktkq$. Notice that if $Z$ satisfies the Hamilton-De Donder-Weyl equations (\ref{HDW field eq}), then we have
\[
    (Z_\alpha)_\beta=\delta_{\alpha\beta}\,,
\]
for all $\alpha,\beta$.
\begin{theorem}[\bf Hamilton-Jacobi theorem]\label{HJtheorem}
 Let $Z\in \vf^k_H(\rktkqh)$ be a $k$-vector field solution of the $k$-cosymplectic Hamiltonian equation (\ref{geonah}) and $\gamma\colon \rkq\to \rktkqh$ be a  section of $(\pi_Q)_{1,0}$ satisfying the property described above. If $Z$ is integrable then the following statements are equivalent:
    \begin{enumerate}
        \item If  a section $\psi\colon U\subset \rk\to \rkq$ of $\pi_{\rk}\colon \rkq\to \rk$ is an integral section of $Z^\gamma$, then $\gamma\circ \psi$ is a solution of the Hamilton-De Donder-Weyl equations (\ref{HDW field eq});
        \item $(\pi_Q)^*[d(H\circ \gamma\circ i_x)] + \sum_{\alpha} \iota_{R^\gamma_\alpha}d\bar{\gamma}^\alpha=0$ for all $x\in \rk$.
    \end{enumerate}
\end{theorem}
\proof

Let us suppose that a section $\psi\colon U\subset\rk\to \rk\times Q$ is an integral section of $Z^\gamma$. In local coordinates that means that if $\psi(x)=(x^\alpha,\psi^i(x))$, then
\[
    [(Z^\gamma_\alpha)^\beta\circ \gamma](\psi(x))=\delta_\alpha\beta,\quad (Z^i_\alpha\circ \gamma)(\psi(x))=\derpar{\psi^i}{x^\alpha}\,.
\]

Now, by hypothesis, $\gamma\circ\psi\colon U\subset\rk\to \rktkqh$ is a solution of the Hamilton-De Donder-Weyl equation for $H$.  In local coordinates, if $\psi(x)=(x,\psi^i(x))$, then
 $\gamma\circ \psi(x)=(x,\psi^i(x),\gamma^\alpha_i(\psi(x)))$ and, since it is a solution of the Hamilton-De Donder-Weyl equations for $H$, we have
\begin{equation}\label{HJkcos aux1}
\derpar{\psi^i}{x^\alpha}\Big\vert_{x} = \derpar{H}{p^\alpha_i}\Big\vert_{\gamma(\psi(x))} \makebox{ and } \ds\sum_{\alpha=1}^k\derpar{(\gamma^\alpha_i\circ \psi)}{x^\alpha}\Big\vert_{x} = -\derpar{H}{q^i}\Big\vert_{\gamma(\psi(x))}\,.
\end{equation}

Next, if we compute the differential of the function $H\circ \gamma\circ i_x\colon Q\to \r$, we obtain that:
\begin{equation}\label{HJ dif aux 2}
\begin{array}{ll}
    &(\pi_Q)^*[d(H\circ \gamma\circ i_x)] + \sum_\alpha\iota_{R^\gamma_\alpha}d\bar{\gamma}^\alpha\\\noalign{\medskip}
    =& \left(\derpar{H}{q^i}\circ \gamma\circ i_x + \left(\derpar{H}{p^\alpha_j}\circ \gamma\circ i_x\right)\left(\derpar{\gamma^\alpha_j}{q^i}\circ i_x\right) + \left(\derpar{\gamma^\alpha_i}{x^\alpha}\circ i_x\right)  \right)dq^i\,.
    \end{array}
\end{equation}

Therefore, from (\ref{gammacond}), (\ref{HJkcos aux1}) and (\ref{HJ dif aux 2}) and taking into account that one can write $\psi(x)= (i_x\circ \pi_Q\circ \psi)(x)$, where $\pi_Q\colon \rkq \to Q$ is the canonical projection, we obtain
\[
    \begin{array}{rl}
       & ((\pi_Q)^*[d(H\circ \gamma\circ i_x)] + \sum_\alpha\iota_{R^\gamma_\alpha}d\bar{\gamma}^\alpha )(\pi_Q\circ\psi(x)) \\\noalign{\medskip}=& \left(\derpar{H}{q^i}\Big\vert_{\gamma(\psi(x))} + \derpar{H}{p^\alpha_j}\Big\vert_{\gamma(\psi(x))}\derpar{\gamma^\alpha_j}{q^i}\Big\vert_{\psi(x)} + \derpar{\gamma^\alpha_i}{x^\alpha}\Big\vert_{\psi(x)}
        \right)dq^i (\pi_Q\circ \psi(x))\\\noalign{\medskip}
        =& \left(- \ds\sum_{\alpha=1}^k\derpar{(\gamma^\alpha_i\circ \psi)}{x^\alpha}\Big\vert_{x} + \derpar{\psi^j}{x^\alpha}\Big\vert_{x} \derpar{\gamma^\alpha_j}{q^i}\Big\vert_{\psi(x)}  + \derpar{\gamma^\alpha_i}{x^\alpha}\Big\vert_{\psi(x)} \right) dq^i (\pi_Q\circ \psi(x))\\\noalign{\medskip}
         =& \left(- \ds\sum_{\alpha=1}^k\derpar{(\gamma^\alpha_i\circ \psi)}{x^\alpha}\Big\vert_{x} + \derpar{\psi^j}{x^\alpha}\Big\vert_{x} \derpar{\gamma^\alpha_i}{q^j}\Big\vert_{\psi(x)}    + \derpar{\gamma^\alpha_i}{x^\alpha}\Big\vert_{\psi(x)}\right) dq^i (\pi_Q\circ \psi(x))
         \\\noalign{\medskip}
         =&0\,.
    \end{array}
\]

As we have mentioned above, since $Z$ is integrable, the $k$-vector field $Z^\gamma$ is also integrable, and then for each point $(x,q)\in \rkq$ we have an integral section $\psi\colon U\subset \rk\to \rk\times Q$ of $Z^\gamma$ passing trough that point. Therefore, for any $x\in \rk$, we get
\[
    (\pi_Q)^*[d(H\circ \gamma\circ i_x)]+ \sum_\alpha\iota_{R^\gamma_\alpha}d\bar{\gamma}^\alpha=0\,.
\]

Conversely, let us suppose that $(\pi_Q)^*[d(H\circ \gamma\circ i_x)]+ \sum_\alpha\iota_{R^\gamma_\alpha}d\bar{\gamma}^\alpha=0$ and take $\psi$ an integral section of $Z^\gamma$. We now will prove that $\gamma\circ \psi$ is a solution of the Hamilton-De Donder-Weyl field equations for $H$.

Since $(\pi_Q)^*[d(H\circ \gamma\circ i_x)]+ \sum_\alpha\iota_{R^\gamma_\alpha}d\bar{\gamma}^\alpha=0$ then from (\ref{HJ dif aux 2}) we obtain
\begin{equation}\label{HJ dif aux 3}
    \derpar{H}{q^i}\circ \gamma\circ i_x + \left(\derpar{H}{p^\alpha_j}\circ \gamma\circ i_x\right)\left(\derpar{\gamma^\alpha_j}{q^i}\circ i_x\right) + \left(\derpar{\gamma^\alpha_i}{x^\alpha}\circ i_x\right)  =0\,.
\end{equation}

From (\ref{k-cosymp condvf}) and (\ref{local z gamma kcos}), we know that
\begin{equation}\label{zgammaH}
    Z^\gamma_\alpha= \derpar{}{x^\alpha} + \Big(\derpar{H}{p^\alpha_i}\circ \gamma\Big)\derpar{}{q^i}\,;
\end{equation}
and then, since $\psi(x,q)=(x,\psi^i(x,q))$ is an integral section of $Z^\gamma$, we obtain
\begin{equation}\label{HJ kcos aux 4}
    \derpar{\psi^i}{x^\alpha}= \derpar{H}{p^\alpha_i}\circ \gamma\circ \psi\,.
\end{equation}

On the other hand, from (\ref{gammacond}), (\ref{HJ dif aux 3}) and (\ref{HJ kcos aux 4}) we obtain
\[
    \begin{array}{ll}
       & \ds\sum_{\alpha=1}^k\derpar{(\gamma^\alpha_i\circ \psi)}{x^\alpha}\Big\vert_{x}= \ds\sum_{\alpha=1}^k\left(\derpar{\gamma^\alpha_i}{x^\alpha}\Big\vert_{\psi(x)} +\derpar{\gamma^\alpha_i}{ q^j}\Big\vert_{\psi(x)}\derpar{\psi^j}{x^\alpha}\Big\vert_{x}\right) = \\\noalign{\medskip} =& \ds\sum_{\alpha=1}^k\left(\derpar{\gamma^\alpha_i}{x^\alpha}\Big\vert_{\psi(x)} +\derpar{\gamma^\alpha_i}{ q^j}\Big\vert_{\psi(x)}\derpar{H}{p^\alpha_j}\Big\vert_{\gamma(\psi(x))}\right)
       \\\noalign{\medskip} =&  \ds\sum_{\alpha=1}^k\left(\derpar{\gamma^\alpha_i}{x^\alpha}\Big\vert_{\psi(x)} +\derpar{\gamma^\alpha_j}{ q^i}\Big\vert_{\psi(x)}\derpar{H}{p^\alpha_j}\Big\vert_{\gamma(\psi(x))}\right) = -\derpar{H}{q^i}\Big\vert_{\gamma(\psi(x))}
    \end{array}
\]
and thus we have proved that $\gamma\circ \psi$ is a solution of the Hamilton-de Donder-Weyl equations.
\qed

\begin{theorem}
    Let $Z\in \vf^k_H(\rktkqh)$ be a $k$-vector field solution of the $k$-cosymplectic Hamiltonian equation (\ref{geonah}) and $\gamma\colon \rkq\to \rktkqh$ be a  section of $(\pi_Q)_{1,0}$ satisfying the same conditions of the above theorem. Then, the following statements are equivalent:
    \begin{enumerate}
        \item $Z\vert_{Im\,\gamma} - T^1_k\gamma(Z^\gamma)\in \ker \Omega^\sharp\cap \ker \eta^\sharp$, being $\Omega^\sharp$ and $\eta^\sharp$ the vector bundle morphism defined in Section \ref{Section HDW eq: k-cosymp approach}.
        \item $(\pi_Q)^*[d(H\circ \gamma\circ i_x)] + \sum_\alpha\iota_{R^\gamma_\alpha}d\bar{\gamma}^\alpha=0$.
    \end{enumerate}
\end{theorem}

\proof A direct computation shows that $Z_\alpha\vert_{Im\,\gamma} - T\gamma(Z_\alpha^\gamma)$ has the following local expression
\[
    \left((Z_\alpha)^\beta_j\circ \gamma - \derpar{\gamma^\beta_j}{x^\alpha} - (Z^i_\alpha\circ \gamma)\derpar{\gamma^\beta_j}{q^i} \right) \derpar{}{p^\beta_j}\circ \gamma\,.
\]

Thus from (\ref{kerkco}) we know that $Z\vert_{Im\,\gamma} - T^1_k\gamma(Z^\gamma)\in \ker \Omega^\sharp\cap \ker \eta^\sharp$ if and only if
\begin{equation}\label{condker}
 \ds\sum_{\alpha=1}^k\left((Z_\alpha)^\alpha_j\circ \gamma - \derpar{\gamma^\alpha_j}{x^\alpha} - (Z^i_\alpha\circ \gamma)\derpar{\gamma^\alpha_j}{q^i} \right)=0\,.
\end{equation}

Now we are ready to prove the result.

Assume that $(i)$ holds, then from  (\ref{gammacond}) and (\ref{condker}) we obtain
\[
\begin{array}{ll}
    0=& \ds\sum_{\alpha=1}^k\left((Z_\alpha)^\alpha_j\circ \gamma - \derpar{\gamma^\alpha_j}{x^\alpha} - (Z^i_\alpha\circ \gamma)\derpar{\gamma^\alpha_j}{q^i} \right)\\\noalign{\medskip}
    =& -\left(\left(\derpar{H}{q^j}\circ \gamma\right) + \derpar{\gamma^\alpha_j}{x^\alpha} + \left(\derpar{H}{p^\alpha_i}\circ \gamma\right)\derpar{\gamma^\alpha_j}{q^i} \right)\\\noalign{\medskip}
    =& -\left(\left(\derpar{H}{q^j}\circ \gamma\right) + \derpar{\gamma^\alpha_j}{x^\alpha} + \left(\derpar{H}{p^\alpha_i}\circ \gamma\right)\derpar{\gamma^\alpha_i}{q^j} \right)\,.
\end{array}
\]
Therefore $(\pi_Q)^*[d(H\circ \gamma\circ i_x)] + \sum_\alpha\iota_{R^\gamma_\alpha}d\bar{\gamma}^\alpha=0$ (see (\ref{HJ dif aux 2})).

The converse is proved in a similar way by reversing the arguments.
\qed

\begin{corol}
    Let $Z\in \vf^k_H(\rktkqh)$ be a  solution of (\ref{geonah}) and $\gamma\colon \rkq\to \rktkqh$ be a  section of $(\pi_Q)_{1,0}$ as in the above theorem. If $Z$ is integrable then the following statements are equivalent:
    \begin{enumerate}
        \item $Z\vert_{Im\gamma} - T^1_k\gamma(Z^\gamma)\in \ker \Omega^\sharp\cap \ker \eta^\sharp$;
        \item $(\pi_Q)^*[d(H\circ \gamma\circ i_x)] + \sum_\alpha\iota_{R^\gamma_\alpha}d\bar{\gamma}^\alpha=0$;
        \item If  a section $\psi\colon U\subset \rk\to \rkq$ of $\pi_{\rk}\colon \rkq\to \rk$ is an integral section of $Z^\gamma$ then $\gamma\circ \psi$ is a solution of the Hamilton-De Donder-Weyl equations (\ref{HDW field eq}).
    \end{enumerate}
\end{corol}

Let us observe that there exist $k$ local functions $W^\alpha\in \mathcal{C}^\infty(U) $ such that $\bar{\gamma}^\alpha=dW^\alpha_x$ where the function $W^\alpha_x$ id defined by $ W^\alpha_x(q)=W^\alpha(x,q)$. Thus $\gamma^\alpha_i=\nicefrac{\partial W^\alpha}{\partial q^i}$ (see \cite{KN}). Therefore, the condition
\[
(\pi_Q)^*[d(H\circ \gamma\circ i_x)]+ \sum_\alpha\iota_{R^\gamma_\alpha}d\bar{\gamma}^\alpha=0
\]
can be equivalently written as
\[
\derpar{}{q^i}\left(\derpar{W^\alpha}{x^\alpha} + H(x^\beta,q^i,\derpar{W^\alpha}{q^i})\right)=0.
\]

The above expressions mean that
\[
    \derpar{W^\alpha}{x^\alpha} + H(x^\beta,q^i,\derpar{W^\alpha}{q^i})=K(x^\beta)
\]
so that if we put $\tilde{H}=H-K$ we deduce the standard form of the Hamilton-Jacobi equation (since $H$ and $\tilde{H}$ give the same Hamilton-De Donder Weyl equations).
\begin{equation}\label{HJeq W}
    \derpar{W^\alpha}{x^\alpha} + \tilde{H}(x^\beta,q^i,\derpar{W^\alpha}{q^i})=0\,.
\end{equation}

Therefore the equation
\begin{equation}\label{HJkcoseq}
(\pi_Q)^*[d(H\circ \gamma\circ i_x)]+ \sum_\alpha\iota_{R^\gamma_\alpha}d\bar{\gamma}^\alpha=0
\end{equation}
can be considered as a geometric version of the \emph{Hamiton-Jacobi equation for $k$-cosymplectic field theories}.
\index{Hamilton-Jacobi equation}

\section{Examples}\label{example}

In this section we shall apply our method to a particular example in classical field theories.

%Now we describe two examples of the Hamilton-Jacobi equation following the techniques of \cite{PR-2002-b, ZS}.
%\subsection{The scalar field}
\index{Scalar field}
We consider again the equation of a scalar field $\phi$ (for instance the gravitational field) which acts on the $4$-dimensional space-time. Let us recall that its equation is given by (\ref{scalar}).

We consider the Lagrangian
\begin{equation}\label{scalar lagrangian}
L(x^1,x^2,x^3,x^4, q, v_1,v_2,v_3, v_4)= \sqrt{-g}\Big(F(q)- \displaystyle\frac{1}{2}m^2q^2\Big)+\displaystyle\frac{1}{2}g^{\alpha\beta}v_\alpha v_\beta\,,
\end{equation}
where $q$ denotes the scalar field $\phi$ and $v_\alpha$ the partial derivative $\nicefrac{\partial \phi}{\partial x^\alpha}$. Then the equation (\ref{scalar}) is just the Euler-Lagrange equation associated to $L$.

               Consider the Hamiltonian function $H\in\mathcal{C}^\infty(\mathbb{R}^4\times (T^1_4)^*\mathbb{R} )$ given by
                    \[
                        H(x^1,x^2,x^3,x^4,q, p^1, p^2, p^3, p^4)=  \displaystyle\frac{1}{2\sqrt{-g}}g_{\alpha\beta}p^\alpha p^\beta-\sqrt{-g}\left( F(q)-\displaystyle\frac{1}{2}m^2q^2\right)\,,
                    \]
               where $(x^1,x^2,x^3,x^4)$ are the coordinates on $\mathbb{R}^4$, $q$ denotes the scalar field $\phi$ and $(x^1,x^2,x^3,x^4,q,p^1,p^2,p^3,$ $p^4)$ the canonical coordinates on $\mathbb{R}^4\times (T^1_4)^*\mathbb{R} $. Let us recall that this Hamiltonian function can be obtained from the Lagrangian $L$ just using the Legendre transformation.

Then
                \begin{equation}\label{partial Ham scalar k-co}
                    \derpar{H}{q}=-\sqrt{-g}\Big(F'(q)-m^2q\Big),\quad \derpar{H}{p^\alpha}=\frac{1}{\sqrt{-g}}g_{\alpha\beta}p^\beta\,.
                \end{equation}

The Hamilton-Jacobi equation becomes
\begin{equation}\label{campo escalar HJ}
    -\sqrt{-g}\Big(F'(q)-m^2q\Big)+ \displaystyle\frac{1}{\sqrt{-g}}g_{\alpha\beta}\gamma^\beta\derpar{\gamma^\alpha}{q} + \displaystyle\derpar{\gamma^\alpha}{x^\alpha}=0\,.
\end{equation}

Since our main goal is to show how the method developed in this chapter works, we will consider, for simplicity, the following particular case:
$$
F(q)=\frac{1}{2}m^2q^2,
$$
being $(g_{\alpha\beta})$ the Minkowski metric on $\mathbb{R}^4$, i.e. $(g_{\alpha\beta})=diag (-1, 1,1,1)$.

Let $\gamma\colon \mathbb{R}^4\to \mathbb{R}^4\times (T^1_k)^*\mathbb{R} $ be the section of $(\pi_\mathbb{R} )_{1,0}$ defined by the family of four $1$-forms along of $\pi_\mathbb{R} \colon\mathbb{R}^4\times \mathbb{R} \to \mathbb{R} $
\[
    \bar{\gamma}^\alpha= \frac{1}{2}C_\alpha q^2dq\,
\]
with $1\leq \alpha\leq 4$ and where $C_\alpha$ are four constants such that $C_1^2=C_2^2+C_3^3+C_4^2$.  This section $\gamma$ satisfies the Hamilton-Jacobi equation (\ref{campo escalar HJ}) that in this particular case is given by
\[
 -\frac{1}{2} C_1^2q^3 + \displaystyle\frac{1}{2}\sum_{a=2}^4C_a^2q^3=0\,;
\]
therefore, the condition $(2)$ of the Theorem \ref{HJtheorem} holds.

The $4$-vector field $Z^\gamma=(Z^\gamma_1,Z^\gamma_2,Z^\gamma_3,Z^\gamma_4)$ is locally given by
\[
    Z^\gamma_1=\derpar{}{x^1}-\frac{1}{2}C_1 q^2\derpar{}{q}\,,\quad Z^\gamma_a=\derpar{}{x^a}+\frac{1}{2}C_a q^2\derpar{}{q}\,,
\]
with $a=2,3,4$. The map $\psi\colon\mathbb{R}^4\to\mathbb{R}^4\times \mathbb{R} $ defined by
\[
    \psi(x^1,x^2,x^3,x^4)=\displaystyle\frac{2}{C_1x^1-C_2x^2-C_3x^3-C_4x^4+C}\,,\quad C\in \mathbb{R} \,,
\]
 is an integral section of the $4$-vector field $Z^\gamma$.

 By Theorem \ref{HJtheorem} one obtains that the map $\varphi=\gamma\circ \psi$, locally given by
 \[
    (x^\alpha)\to (x^\alpha,\psi(x^\alpha), \frac{1}{2}C_\beta(\psi(x^\alpha))^2)\,,
 \]
 is a solution of the Hamilton-De Donder-Weyl equations associated to $H$, that is,
 \[
 \begin{array}{rcl}
    0&=&\displaystyle \sum_{\alpha=1}^4\derpar{}{x^\alpha}\Big(\frac{1}{2} C_\alpha\psi^2\Big)\,,\\\noalign{\medskip}
    -\frac{1}{2}C_1\psi^2 &=&\derpar{\psi}{x^1}\,,\\\noalign{\medskip}
    \frac{1}{2}C_a\psi^2&=&\derpar{\psi}{x^a}\,,\quad a=2,3,4\,.
     \end{array}
 \]
Let us observe that  these equations imply that the scalar field $\psi$ is a solution to the $3$-dimensional wave equation.

In this particular example the functions $W^\alpha$ are given by
\[
W^\alpha(x,q)=\displaystyle\frac{1}{6}C_\alpha q^3 + h(x)\,,
\]
where $h\in \mathcal{C}^\infty (\mathbb{R}^4)$.

\bigskip

In \cite{PR-2002-b, ZS}, the authors describe an alternative method that can be compared with the one above.

First, we consider the set of functions $W^\alpha\colon \mathbb{R}^4\times \mathbb{R} \to \mathbb{R} , \, 1\leq \alpha\leq 4$ defined by
\[
    W^\alpha(x,q)= (q-\displaystyle\frac{1}{2}\phi(x))\sqrt{-g}g^{\alpha\beta}\derpar{\phi}{x^\beta}\,,
\]
where $\phi$ is a solution to the Euler-Lagrange equation (\ref{scalar}).
Using these functions we can consider a section $\gamma$ of $(\pi_\mathbb{R} )_{1,0}\colon \mathbb{R}^4\times (T^1_4)^*\mathbb{R} \to \mathbb{R}^4\times\mathbb{R} $  with components
\[
    \gamma^\alpha= \derpar{W^\alpha}{q} = \sqrt{-g}g^{\alpha\beta}\derpar{\phi}{x^\beta}\,.
\]

By a direct computation we obtain that this section $\gamma$ is a solution to the Hamilton-Jacobi equation (\ref{HJkcoseq}).

Now from (\ref{zgammaH}) and (\ref{partial Ham scalar k-co})  we obtain the $4$-vector field $Z^\gamma$ is given by
\begin{equation}\label{examplezgamma}
           Z^\gamma_\alpha = \derpar{}{x^\alpha} +\derpar{\phi}{x^\alpha}\derpar{}{q}\,.
\end{equation}

Let us observe that $Z^\gamma$ is an integrable $4$-vector field on $\mathbb{R}^4\times \mathbb{R} $. Using the Hamilton-Jacobi theorem we obtain that if $\sigma=(id_{\mathbb{R}^4},\phi)\colon \mathbb{R}^4\to \mathbb{R}^4\times \mathbb{R} $ is an integral section of the $4$-vector field $Z^\gamma$ defined by (\ref{examplezgamma}), then $\gamma\circ \sigma$ is  a solution of the Hamilton-De Donder Weyl equation associated with the Hamiltonian of the massive scalar field.

If we now consider the particular case $F(q)=m^2q^2$, we obtain the Klein-Gordon equation; this is just the case discussed in \cite{PR-2002-b}.

\newpage
\mbox{}
\thispagestyle{empty} % para que no se numere esta p\'{a}gina

\chapter{Lagrangian
 Classical Field Theories}

In a similar way to that developed in Chapter \ref{chapter: K-SympLagCFT}, we now give a  description of the Lagrangian Classical Field Theories using two different approaches: a variational principle and a $k$-cosymplectic approach.

\index{Lagrangian function}
\index{Euler-Lagrange field equations}
Given a Lagrangian $L\in\mathcal{C}^\infty(\rktkq)$, we shall obtain the local Euler-Lagrange field equations
\begin{equation}\label{EL field eq}
    \ds\sum_{\alpha=1}^k\ds\frac{\partial}{\partial
x^\alpha}\Big\vert_{x}\left(\ds\frac{\partial L}{\partial
v^i_\alpha}\Big\vert_{\varphi(x)}\right)=\ds\frac{\partial
L}{\partial q^i}\Big\vert_{\varphi(x)}\;,\quad
v^i_\alpha(\varphi(x))=\ds\frac{\partial (q^i\circ \varphi)}{\partial
x^\alpha}\Big\vert_{x}\,,
\end{equation}
 with $\varphi\colon \rk\to \rktkq$. First, we shall use a multiple integral variational problem approach, later we shall give a geometric version of these equations.

Finally, we shall define a Legendre transformation on this new setting which shall allows to prove the equivalence between both Hamiltonian and Lagrangian formalisms when the Lagrangian satisfies certain regularity condition.
We shall use the notation introduced in (\ref{notation partials}) and the  notion of  prolongations to $\rktkq$.

\section{The stable tangent bundle of $k^1$-veloci\-ties $\rktkq$}\label{section: rktkq}

\index{Stable tangent bundle of $k^1$-veloci\-ties}
\index{$\rktkq$}
    In Chapter \ref{chapter: k-cosympManifolds} we have introduced the model of the so-called $k$-cosymplectic manifolds that we have used to develop the geometric description of the Hamilton-De Donder-Weyl field equations when the Hamiltonian function depends on the  coordinates $(x^\alpha)$ on the base manifold. In this section we introduce its Lagrangian counterpart, i.e., a manifold where we shall develop the $k$-cosymplec\-tic Lagrangian formalism. Roughly speaking, this manifold is the Cartesian product of the $k$-dimensional euclidean space and the tangent bundle of $k^1$-velocities of a $n$-dimensional smooth manifold $Q$. In this section we shall introduce formally the manifold $\rktkq$ and some canonical geometric elements defined on it.

    Let us recall that in Remark \ref{j1rkq} we have introduced the manifold $J^1_0(\rk,Q)$ of $1$-jets of maps from $\rk$ to $Q$ with source $0\in \rk$. In an analogous way fixed a point $x\in \rk$, we can consider the manifold $J^1_x(\rk,Q)$ of $1$-jets of maps from $\rk$ to $Q$ with source $x\in \rk$, i.e.,
        \[
            J_x^1(\rk,Q)= \ds\bigcup_{q\in Q}J^1_{x,\,q}(\rk,Q)= \ds\bigcup_{q\in Q}\{j^1_{x,q}\phi\, \vert\, \phi\colon\rk\to Q\;\makebox{smooth},\, \phi(x)=q\}\,.
        \]

    Let $J^1(\rk,Q)$ be the set of $1$-jets from $\rk$ to $Q$, that is,
        \[
            J^1(\rk,Q)=\ds\bigcup_{x\in \rk}J^1_x(\rk,Q)\,.
        \]

    This space can be identified with $\rktkq$ as follows
        \begin{equation}\label{difeo j1rkq-rktkq}
        \begin{array}{ccccc}
            J^1(\rk,Q) & \to &\rk\times J^1_0(\rk,Q)&\to & \rktkq\\\noalign{\medskip}
            j^1_x\phi=j^1_{x,q}\phi & \to & (x,j^1_{0,x}\phi_x)& \to & (x,{v_1}_q,\ldots, {v_k}_q)
        \end{array}\,,
        \end{equation}
    where $\phi_x(\tilde{x})=\phi(x+\tilde{x})$, with $\tilde{x}\in \rk$ and
        \[
            {v_\alpha}_{q}=(\phi_x)_*(0)\Big(\frac{\partial}{\partial                 x^\alpha}\Big\vert_0\Big)= \phi_*(x)\Big(\frac{\partial}{\partial                 x^\alpha}\Big\vert_x\Big),
        \]
     being $ q=\phi_x (0)=\phi(x)$ and with $\ak$.

     Therefore, an element in $J^1(\rk,Q)$ can be thought as a family $$(x,{v_1}_q,\ldots, {v_k}_q)\in\rktkq$$ where $x\in \rk$ and $({v_1}_q,\ldots, {v_k}_q)\in\tkq$. Thus, we can consider the following canonical projections defined by
        \begin{equation}\label{projections rktkq}
            \begin{array}{ll}
                    (\pi_{\rk})_{1,0}(x,{v_1}_q,\ldots, {v_k}_q)=(x,q), &  \pi_{\rk}(x,q)=x,
                \\\noalign{\medskip}
                    (\pi_{\rk})_1(x,{v_1}_q,\ldots, {v_k}_q)=x, &  \pi_Q(x,q)= q,
                    \\\noalign{\medskip}
                    p_Q(x,{v_1}_q,\ldots, {v_k}_q)=q\, &
            \end{array}
        \end{equation}
  with $x\in \rk $, $q\in Q$ and $({v_1}_q,\ldots, {v_k}_q)\in\tkq$. The following diagram illustrates the situation:
  \begin{center}
     \begin{figure}[h]
\centering
\scalebox{1.3}{
%\begin{sideways}
%\begin{tabular}{c}
$
            \xymatrix@C=20mm{ & \rktkq \ar[dl]_-{(\pi_{\rk})_1} \ar[dr]^-{p_Q} \ar[d]_-{(\pi_{\rk})_{1,0}}& \\
            \rk & \rkq \ar[l]_-{\pi_{\rk}}\ar[r]^-{\pi_Q}& Q}
$
%\end{tabular}
%\end{sideways}
}
\caption{Canonical projections associated to $\rktkq$.}
\label{K-cosymp-Maps Lag}
\end{figure}
\end{center}

\index{Stable tangent bundle of $k^1$-veloci\-ties!coordinates}
    If $(q^i)$ with $\n$, is a local \emph{coordinate system} defined on an open set $U\subset Q$, then the induced local coordinates $(x^\alpha, q^i,v^i_\alpha)$, $\n,\,\ak$ on $\rk\times T^1_kU= p_Q^{-1}(U)$ are given by
        \begin{equation}\label{canonical coordinates rktkq}
            \begin{array}{lcl}
                    x^\alpha(x,{v_1}_q,\ldots, {v_k}_q) = x^\alpha(x)=x^\alpha\, , \\\noalign{\medskip}
                    q^i(x,{v_1}_q,\ldots, {v_k}_q) = q^i(q)\, ,
                \\\noalign{\medskip}
                    v^i_\alpha(x,{v_1}_q,\ldots, {v_k}_q) ={v_\alpha}_q(q^i)\,.
            \end{array}
        \end{equation}
\index{$\rktkq$!canonical coordinates}
    Thus, $\rktkq$ is endowed with a   structure of differentiable manifold of  dimension $k+n(k+1)$, and the manifold $\rktkq$ with the projection $(\pi_{\rk})_1$ has a structure of vector bundle over $\rk$.

    Considering the identification (\ref{difeo j1rkq-rktkq}), the above coordinates can be defined in terms of $1$-jets of maps from $\rk$ to $Q$ with source in $0\in \rk$ as follows
            \[
                \begin{array}{lcl}
                x^\alpha(j^1_x\phi)&=&x^\alpha(x)=x^\alpha\,, \\\noalign{\medskip} q^i(j^1_x\phi)&=&q^i(\phi(x)))\,,\\\noalign{\medskip} v^i_\alpha(j^1_x\phi)&=&\derpar{\phi^i}{ x^\alpha}\Big\vert_{x}=\phi_*(x)\Big(\derpar{}{x^\alpha}\Big\vert_{x}
                \Big)(q^i)\,,
                \end{array}
            \]
    being $\phi\colon \rk\to Q$.

    \begin{remark}\label{difeo jipi}
    {\rm
    Let us observe that each map $\phi\colon \rk\to Q$ can be identified with a section $\tilde{\phi}$ of the trivial bundle $\pi_{\rk}\colon \rkq\to\rk$. Thus the manifold $J^1(\rk,Q)$ is diffeomorphic to $J^1\pi_{\rk}$ (see  \cite{Saunders-89} for a full description of the first-order jet bundle associated to an arbitrary bundle $E\to M$). The diffeomorphism between these two manifolds is given by
    \[
        \begin{array}{ccc}
            J^1\pi_{\rk} & \to & J^1(\rk,Q)\\\noalign{\medskip}
            j^1_x\tilde{\phi} & \to & j^1_{x,\phi(x)}\phi
        \end{array}
    \]
    being $\tilde{\phi}\colon \rk \to \rkq$ a section of $\pi_{\rk}$ and $\phi=\pi_Q\circ \tilde{\phi}\colon \rk\to Q$.
    \rqed}
    \end{remark}

    On $\rktkq$ there exist several canonical  structures which will allow us to introduce the necessary objets for develop a $k$-cosymplectic description of the Euler-Lagrange field equations. In the following subsections we introduce these geometric elements.
    \subsection{Canonical tensor fields}\label{section rktkq: canonical tensor fields}
\index{$\rktkq$!canonical tensor fields}

    We first define a family $(J^1,\ldots, J^k)$ of $k$ tensor fields of type $(1,1)$ on $\rktkq$. These tensors fields  allow us to define the Poincar\'{e}-Cartan  forms, in a similar way that in the $k$-symplectic setting.

    To introduce this family we will use the canonical $k$-tangent structure $\{J^1,$ $\ldots, J^k\} $ which we have introduced in section \ref{section tkq:geometry}.

    For each $\ak$ we consider the {\it natural extension}  of the tensor fields  $J^\alpha$ on $\tkq$ to $\rktkq$,  (we denote this tensor field also by $J^\alpha$) whose local expression is
   \begin{equation}\label{localSA} J^\alpha= \ds\frac{\partial}{\partial v^i_\alpha}\otimes dq^i .\end{equation}

  % On $\rktkq$ there exist another family of $k$ tensor fields $\tilde{S}_\alpha$ of type $(1,1)$ defined by
%    \begin{equation}\label{tildesa}
%        \tilde{S}^\alpha=\derpar{}{x^\alpha}\otimes dx^\alpha+ J^\alpha\,,
%    \end{equation}
%    for each $\ak$.

\index{$\rktkq$!Liouville vector field}
\index{Liouville vector field}
    Another interesting group of canonical tensors defined on $\rktkq$ is the set of canonical vector fields on $\rktkq$ defined as follows:

    \begin{definition}
        The \emph{Liouville vector field} $\Delta$ on $\rktkq$ is the infinitesimal generator   of the flow
        \[
            \begin{array}{ccl}
            \r \times (\r^{k}\times T^1_kQ) & \longrightarrow & \r^{k}\times
            T^1_kQ  \\ \noalign{\medskip} (s,(x,{v_1}_{q },\ldots,
            {v_k}_{q })) & \mapsto & (x, e^s{v_1}_{q },
            \ldots,e^s{v_k}_{q })\,,
            \end{array}
        \]
        and its local expression is
        \begin{equation}\label{locci}
            \Delta =   \displaystyle\sum_{i,A} v^i_\alpha \frac{\displaystyle\partial}{\displaystyle\partial v^i_\alpha}\, .
        \end{equation}
        \end{definition}
        \begin{definition}
        The \emph{canonical vector fields} $\Delta_1,\ldots, \Delta_k$
        on $\rktkq$ are generators infinitesimals of the flows
            \[
                \begin{array}{ccl}
                \r \times (\r^{k}\times T^1_kQ) & \longrightarrow & \r^{k}\times
                T^1_kQ  \\ \noalign{\medskip} (s,(x,{v_1}_{q },\ldots,
                {v_k}_{q } )) & \mapsto & (x, {v_1}_{q }, \ldots,
                {v_{\alpha-1}}_{q }, e^s {v_\alpha}_{q }, {v_{\alpha+1}}_{q }, \ldots,
                {v_k}_{q })\,,
                \end{array}
            \]
        for each $\alpha=1,\ldots, k\,$, respectively. Locally
                \begin{equation}\label{LoccanCA}
                 \Delta_\alpha =   \displaystyle\sum_{i= 1}^n v^i_\alpha
                \frac{\displaystyle\partial}{\displaystyle\partial v^i_\alpha}\,,\quad
                1\leq \alpha\leq k\, .
                \end{equation}
        \end{definition}

From (\ref{locci}) y (\ref{LoccanCA})
we see that  $\Delta=\Delta_1+\ldots + \Delta_k\,.$

    \subsection{Prolongation of diffeomorphism and vector fields}\label{prolon-kco-l}

    In this section we shall describe how to lift a diffeomorphism of $\rkq$ to  $\rktkq$ and, as a consequence, we shall introduce the prolongation of $\pi_{\rk}$-projectable vector fields  on $\rkq$ to $\rktkq$.

    Firstly we introduce the following definition of first prolongation of a map $\phi\colon\rk\to Q$ to $\rktkq$.
\index{$\rktkq$!first prolongation}
\index{First prolongation}
    \begin{definition}\label{de652}
        Let $\phi:\r^k \rightarrow Q$  be a map, we define the \emph{first prolongation $\phi^{[1]}$ of  $\phi$ to $\rktkq$}
          as the map
            $$
            \begin{array}{rcl}
            \phi^{[1]}:\r^k & \longrightarrow &   \r^k \times T^1_kQ \\ x &
            \longrightarrow &
            (x,j^1_0\phi_{x})\equiv\left(x,\phi_*(x)\Big(\ds\frac{\partial}{\partial
            x^1}\Big\vert_{x}\Big), \ldots,
            \phi_*(x)\Big(\ds\frac{\partial}{\partial
            x^k}\Big\vert_{x}\Big) \right)
            \end{array}
            $$ where $\phi_{x}(y)=\phi(x+y)$.
    \end{definition}

    In local coordinates one has
    \begin{equation}\label{localfi2}
    \phi^{[1]}(x^1, \dots, x^k)=(x^1, \dots, x^k, \phi^i (x^1, \dots,
    x^k), \frac{\displaystyle\partial\phi^i}{\displaystyle\partial x^\alpha}
    (x^1, \dots, x^k))\,.
    \end{equation}
\begin{remark}
    {\rm
        Let us observe that $\phi^{[1]}$ can be defined as the pair $(id_{\rk},\phi^{(1)})$, where $\phi^{(1)}$ is the first prolongation of $\phi$ to $\tkq$ introduced in Definition \ref{first_prol}.
    }
\end{remark}

Comparing the local expression (\ref{localfi2}) with the second set of the equations (\ref{EL field eq}), one observes that a solution $\varphi\colon \rk\to \rktkq$ of the Euler-Lagrange equations (\ref{EL field eq}) is of the form $\varphi=\phi^{[1]}$, being $\phi$ the map given by the composition
    \[
        \xymatrix@C=29mm{
                            \rk \ar[r]^-{\varphi} \ar@/^{10mm}/[rr]^-{\phi}& \rktkq \ar[r]^-{p_Q}& Q
                        }
    \]

    Therefore, the equations (\ref{EL field eq}) can be written as follows:
    \begin{equation}\label{EL eq}
        \ds\sum_{\alpha=1}^k\derpar{}{x^\alpha}\left(\derpar{L}{x^\alpha}\Big\vert_{\phi^{[1]}(x)}\right)=\derpar{L}{q}\Big\vert_{\phi^{[1]}(x)}\,,
    \end{equation}
    where $\n$ and a solution is a map $\phi\colon \rk\to Q$.

    The equations (\ref{EL field eq}) are equivalent to
    \[{\tiny
         \ds\frac{\partial^2 L}{\partial x^\alpha \partial v^i_\alpha}\big\vert_{\phi^{[1]}(x)}+ \ds\frac{\partial^2 L}{\partial q^j \partial v^i_\alpha}\big\vert_{\phi^{[1]}(x)} \ds\frac{\partial \phi^j}{\partial x^\alpha}\big\vert_{x}+ \ds\frac{\partial^2 L}{\partial v^j_\beta \partial v^i_\alpha}\big\vert_{\phi^{[1]}(x)} \ds\frac{\partial^2 \phi^j}{\partial x^\alpha \partial x^\beta}\big\vert_{x}=\ds\frac{\partial L}{\partial q^i}\Big\vert_{\phi^{[1]}(x)}\,.}
    \]

    Let us observe that an element in $\rktkq$ is of the form $\phi^{[1]}(x)$ for some $\phi\colon\rk\to Q$ and some $x\in \rk$. We introduce the prolongation of diffeomorphisms using the first prolongation of maps from $\rk$ to $Q$.
\index{$\rktkq$!first prolongation of diffeomorphism}
\index{Prolongation of diffeomorphism}
    \begin{definition}
        Let $f\colon \rkq\to \rkq$ be a map and $f_{\rk}\colon \rk\to \rk$ be a diffeomorphism, such that $\pi_{\rk}\circ f=f_{\rk}\circ \pi_{\rk}$
         \footnote{These conditions are equivalent to say that the pair $(f,f_{\rk})$ is a bundle automorphism of the bundle $\rkq\to \rk$.}. The \emph{first prolongation} of $f$ is a map
            \[
                j^1f\colon J^1(\rk,Q)\equiv\rktkq\to J^1(\rk,Q)\equiv\rktkq
            \]
        defined by
            \begin{equation}\label{j1f}
                (j^1f)(\phi^{[1]}(x))= (\pi_Q\circ f\circ\tilde{\phi}\circ f_{\rk}^{-1})^{[1]}(f_{\rk}(x))
            \end{equation}
        where $\tilde{\phi}$ is the section of $\pi_{\rk}$ induced by $\phi$, that is, $\tilde{\phi}=(id_{\rk},\phi)$ and we are considering the first prolongation of the map given by the following composition:
    \end{definition}
\vspace{2cm}
\bigskip \makebox{}
\bigskip \makebox{}
\bigskip \makebox{}
\[
                \xymatrix@C=16mm{
                        \rk \ar@/^{16mm}/[rrrr]\ar[r]^-{f_{\rk}^{\,-1}}& \rk\ar[r]^-{\tilde{\phi} }& \rkq \ar[r]^-{f}& \rkq\ar[r]^-{\pi_Q} & Q
                        }
            \]
            \vspace{-0.5cm}
    \begin{remark}
    {\rm
        If we consider the identification between $J^1(\rk,Q)$ and $J^1\pi_{\rk}$ given in remark \ref{difeo jipi}, the above definition coincides with the definition 4.2.5 in \cite{Saunders-89} of the first prolongation of $f$ to the jet bundles.\rqed
        }
    \end{remark}

Locally, if $f(x^\alpha,q^i)=(f_{\rk}^\alpha(x^\beta)=f^\alpha(x^\beta),f^i(x^\beta,q^j))$ then
\begin{equation}\label{local j1f}
j^1f(x^\alpha,q^i,v^i_\alpha)=(f^\alpha(x^\beta),f^i(x^\beta,q^j)
,\ds\frac{df^i}{dx^\beta}\Big(\ds\frac{\partial
(f_{\rk}^{-1})^\beta}{\partial \bar{x}^\alpha}\circ f_{\rk}(x^\gamma)\Big))\;,
\end{equation}
where $(\bar{x}^1,\ldots, \bar{x}^k)$ are the coordinates on $\rk=f_{\rk}(\rk)$ and $df^i/dx^\beta$ is the total derivative defined by
$$
\ds\frac{df^i}{dx^\beta}=\ds\frac{\partial f^i}{\partial
x^\beta}+v^j_\beta\ds\frac{\partial f^i}{\partial q^j}\,.
$$

\index{$\rktkq$!first prolongation of vector fields}
\index{Prolongation of vector fields}
    Using the prolongation of diffeomorphism we can define the prolongation of vector field to $\rktkq$ in an analogous way that we did  in section \ref{section rktkqh prolongation}.  Given a $\pi_{\rk}$-projectable vector field $Z\in\vf(\rkq)$, we can define the canonical prolongation $Z^1\in\vf(\rktkq)$ using the prolongation of the diffeomorphism of the set $\{\sigma_s\}$, being this set the one-parameter group of diffeomorphism of $Z$.

    Locally if $Z\in\vf(\rkq)$ is a $\pi_{\rk}$-projectable vector field
with local expression
$$Z=Z^\alpha\ds\frac{\partial }{\partial
x^\alpha}+Z^i\ds\frac{\partial}{\partial q^i}\,,$$ then from (\ref{local j1f}) we deduce that the natural prolongation  $Z^{1}$  has the following local expression
 $$
Z^1= Z^\alpha\ds\frac{\partial}{\partial
x^\alpha}+Z^i\ds\frac{\partial}{\partial q^i} + \Big(
 \ds\frac{dZ^i}{d x^\alpha}- v_\beta^i  \ds\frac{dZ^\beta}{d
x^\alpha}\Big)  \ds\frac{\partial}{\partial v_\alpha^i} \,,
$$ where $d/dx^\alpha$ denotes the  total derivative, that is
$$
\ds\frac{d}{dx^\alpha}=\ds\frac{\partial}{\partial x^\alpha}+
v^j_\alpha\ds\frac{\partial }{\partial q^j} \; .
$$

\subsection{$k$-vector fields and {\sc sopdes}}\label{section rktkq: k-vf and sopde}

\index{$\rktkq$!sopdes}
In this section we shall consider again the notion of $k$-vector field introduced in section \ref{section k-vector field}, but in this case,  $M=\rktkq$. Moreover we describe a particular type of vector fields which are very important in the $k$-cosymplectic Lagrangian description of the field equations.

We consider $M=\rk\times \tkq$, with local coordinates $(x^\alpha,q^i,v^i_\alpha)$ on an open set $U$.

\index{$k$-vector field}
    A $k$-vector field $\mathbf{X}$ on $\rk\times\tkq$ is a family of $k$ vector fields $(X_1,\ldots, X_k)$ where each $X_\alpha\in \vf(\rk\times\tkq)$. The local expression of a $k$-vector field on $\rk\times\tkq$ is given by ($\ak$)
        \begin{equation}\label{tkq: k-vf local1}
            X_\alpha=(X_\alpha)_\beta\derpar{}{x^\beta}+(X_\alpha)^i\derpar{}{q^i}+ (X_\alpha)^i_\beta\derpar{}{v^i_\beta},
        \end{equation}

    Let $$\varphi\colon U_0\subset\rk \to  \rk\times \tkq$$    be an integral section of $(X_1,\ldots, X_k)$   with components
    $$\varphi(x)=(\psi_\alpha(x),\psi^i(x),\psi^i_\alpha(x))\quad   .$$

    Then, since
        \[
            \varphi_*(x)\Big(\derpar{}{x^\alpha}\Big\vert_{x}\Big)=
         \derpar{\psi^\beta}{x^\alpha}\Big\vert_{x}\derpar{}{x^\beta}\Big\vert_{\varphi(x)}+
            \derpar{\psi^i}{x^\alpha}\Big\vert_{x}\derpar{}{q^i}\Big\vert_{\varphi(x)} + \derpar{\psi^i_\beta}{x^\alpha}\Big\vert_{x}\derpar{}{v^i_\beta}\Big\vert_{\varphi(x)}
        \]
    the condition (\ref{integral section expr})  is locally equivalent to the following system of partial differential equations (condition (\ref{integral section equivalence cond}))
        \begin{equation}\label{rktkq integral section equivalence cond}
        \begin{array}{c}
           \derpar{\psi_\beta}{x^\alpha}\Big\vert_{x}= (X_\alpha)_\beta(\varphi(x))  \,,\quad
            \derpar{\psi^i}{x^\alpha}\Big\vert_{x}=(X_\alpha)^i(\varphi(x))\,,\quad \derpar{\psi^i_ \beta}{x^\alpha}\Big\vert_{x}=(X_\alpha)^i_\beta(\varphi(x)),
        \end{array}
        \end{equation}
    with $\n$ and $1\leq \alpha,\beta\leq k$.
Next, we shall characterize the integrable
$k$-vector fields on $\r^k \times T^1_kQ$ whose integral
sections are canonical prolongations of maps from $\r^k$ to $Q$.

\index{Second order partial differential equation}
\begin{definition}\label{sode2}
A $k$-vector field  ${\bf X}=(X_1,\dots,X_k)$ on $\r^k \times
T^1_kQ$ is a \emph{second order partial differential equation} ({\sc
sopde} for short)  if
$$
\eta^\alpha(X_\beta)= \delta_\beta^\alpha
$$
and
$$
J^\alpha(X_\alpha)=\Delta_\alpha \, ,
$$
for all $1\leq \alpha,\beta\leq k$.
\end{definition}

Let $(q^i)$ be a coordinate system on $Q$ and $(x^\alpha,q^i,v^i_\alpha)$  the
induced coordinate system on $\r^k \times T^1_kQ$.  From (\ref{localSA}) and (\ref{LoccanCA}) we deduce that   that the
local expression of a {\sc sopde} $(X_1 ,\ldots,X_k) $ is
\begin{equation}\label{localsode2}
X_\alpha=\frac{\partial}{\partial
x^\alpha}+v^i_\alpha\frac{\displaystyle
\partial} {\displaystyle
\partial q^i}+
(X_\alpha)^i_\beta \frac{\displaystyle\partial} {\displaystyle \partial v^i_\beta}\ ,
\end{equation}
where $(X_\alpha)^i_\beta $ are functions on $\r^k \times T^1_kQ$. As a
direct consequence of the above local expressions, we deduce that
the family of vector fields $\{X_1, \ldots , X_k\}$ are linearly
independent.

\begin{lemma}\label{lem0}
Let $(X_1 ,\ldots,X_k) $ be a {\sc sopde}. A map $\varphi:\rk \to \rk
\times T^1_kQ$, given by
$$\varphi(x)=(\psi_\alpha(x),\psi^i(x),\psi^i_\alpha(x))$$  is an integral section
of $(X_1 ,\ldots,X_k) $ if, and only if,
\begin{equation}\label{nn}
\psi_\alpha(x)=x^\alpha+c^\alpha \, , \quad \psi^i_\alpha(x)=\frac{\ds\partial \psi^i}
{\ds\partial x^\alpha}\Big\vert_{x}\, , \quad \frac{\ds\partial^2 \psi^i}
{\ds\partial x^\alpha \ds\partial x^\beta}\Big\vert_{x}=(X_\alpha)^i_\beta(\varphi(x)) \, ,
\end{equation}
where $c^\alpha$ is a constant.
\end{lemma}
\proof  Equations (\ref{nn})  follow directly from (\ref{rktkq integral section equivalence cond}) and  (\ref{localsode2}).
\qed

\begin{remark}\label{rem1}
{\rm
The integral sections of a {\sc sopde} are given by
$$\varphi(x)=\left(x^\alpha+c^\alpha,\psi^i(x), \ds\frac{\partial \psi^i}{\partial
x^\alpha}(x)\right),$$ where the functions $(\psi^i(x))$ satisfy the
equation
$$\frac{\ds\partial^2 \psi^i}
{\ds\partial x^\alpha \ds\partial x^\beta}\Big\vert_{x}=(X_\alpha)^i_\beta(\psi(x)) $$ in (\ref{nn}), and the $c^\alpha$'s are constants.

In the particular
case when $c=0$, we have that $\varphi= \phi^{[1]}$  where $$\phi= p_Q\circ
\varphi:\rk \stackrel{\varphi}{\to}\r^k \times T^1_kQ
\stackrel{p_Q}{\to}Q$$  that is, $\phi(x)=(\psi^i(x))$.\rqed}
\end{remark}

\begin{lemma}\label{lem1} Let ${\bf X}=(X_1,\ldots , X_k)$ be
an integrable $k$-vector field on $\r^k \times T^1_kQ$. If  every
integral section of ${\bf X}$ is the  first prolongation
$\phi^{[1]}$ of map $\phi:\rk \to Q$, then ${\bf X}$ is a {\sc
sopde}.
\end{lemma}
\proof Let us suppose that each $X_\alpha$ is locally given by
\begin{equation}\label{localsode22}
X_\alpha=(X_\alpha)_\beta\frac{\partial}{\partial
x^\beta}+(X_\alpha)^i\frac{\displaystyle
\partial} {\displaystyle
\partial q^i}+
(X_\alpha)^i_\beta \frac{\displaystyle\partial} {\displaystyle
\partial v^i_\beta}\ .
\end{equation}

Let $\psi=\phi^{[1]}:\rk \to \r^k \times T^1_kQ$ be an integral
section of ${\bf X}$, then from (\ref{de652}), (\ref{rktkq integral section equivalence cond}),  (\ref {localsode2}) and (\ref
{localsode22}),   we
obtain
$$\begin{array}{c}
(X_\alpha)_\beta(\phi^{[1]}(x))=\delta_\alpha^\beta, \quad
(X_\alpha)^i(\phi^{[1]}(x))=\ds\frac{\partial \phi^i}{\partial
x^\alpha}\Big\vert_{x}=v^i_\alpha(\phi^{[1]}(x)),\quad \text{and}\quad
  (X_\alpha)^i_\beta(\phi^{[1]}(x))=
\ds\frac{\partial^2 \phi^i}{\partial x^\alpha
\partial x^\beta }\Big\vert_{x}\end{array}$$
thus $X_\alpha$ is locally given as in (\ref{localsode2}) and then it is a {\sc sopde}.
\qed

\section{Variational principle }

In this section we describe the problem in the setting of the calculus of variations for multiple integrals, which allows us to obtain the Euler-Lagrange field equations. The procedure is similar to section \ref{section: lag variational principle} but in this case the Lagrangian function depends on the coordinates of the basis space, that is, $L$ is defined on $\rktkq$. In particular, if $L$ does not depend on the space-time coordinates we obtain again the results in Section \ref{section: lag variational principle}.

Let us observe that given a section $\phi$ of $\pi_{\rk}:U_0\subset\rkq\to\rk$, it can be identified with the pair
\[
    \bar{\phi}=(id_{\rk},\phi)
\]
where $\bar{\phi}=\pi_Q\circ \phi$. Therefore, any section $\bar{\phi}$ of $\pi_{\rk}$ can be identified with a map $\phi\colon \rk\to Q$. Along this section we consider this identification.

\begin{definition} Let $L\colon \rktkq\to \r$ be a Lagrangian. Denote by $Sec_C(\rk,$ $\rkq)$ the set of sections of
$$\pi_{\rk}:U_0\subset\rkq\to\rk$$ with compact support. We define the action associated to $L$ by:
\[\begin{array}{lccl}
\mathbb{S}: & Sec_C(\rk,\rkq)&\to &\r\\\noalign{\medskip} & \bar{\phi}
&\mapsto & \mathbb{S}(\bar{\phi}) = \ds\int_{\rk} (\phi^{[1]})^*(Ld^kx)
\end{array}\]
\end{definition}

\begin{lemma}
Let $\bar{\phi}\in Sec_C(\rk,\rkq)$ be a section with compact support. If   $Z\in \vf(\rkq)$ is $\pi_{\rk}$-vertical then $$\bar{\phi}_s\colon= \tau_s\circ\bar{\phi}$$ is a  section of  $\pi_{\rk}:\rkq\to\rk$.\end{lemma}

\proof Since $Z\in\vf(\rkq)$ is $\pi_{\rk}$-vertical, then it has the following local expression
\begin{equation}\label{Zloc}
Z(x,q)=Z^i(x,q)\ds\frac{\partial}{\partial
q^i}\Big\vert_{(x,q)}\,.
\end{equation}
Now, if  $\{\tau_s\}$ is the one-parameter group of diffeomorphisms generated by $Z$, then,  one has
$$
Z(x,q)=(\tau_{(x,q)})_*(0)\Big(\ds\frac{d}{ds}\Big\vert_{0}\Big)=\ds\frac{d
(x^\alpha\circ \tau_{(x,q)})}{ds}\Big\vert_{0}\ds\frac{\partial
}{\partial x^\alpha}\Big\vert_{(x,q)}+\ds\frac{d (q^i\circ
\tau_{(x,q)})}{ds}\Big\vert_{0}\ds\frac{\partial }{\partial
q^i}\Big\vert_{(x,q)}\,.
$$

Comparing  (\ref{Zloc}) and the above expression of $Z$, and taking into account that  it is valid for any point $(x,q)\in \rkq$, one has
$$
\ds\frac{d (x^\alpha\circ \tau_{(x,q)})}{ds}=0\,.
$$
Then $$(x^\alpha\circ \tau_{(x,q)})(s)=\makebox{constant}.$$
Moreover $\tau_{(x,q)}(0)=(x,q)$, and we obtain that $$(x^\alpha\circ
\tau_{(x,q)})(0)=x^\alpha\, .$$ Thus $(x^\alpha\circ \tau_{(x,q)})(s)=
x^\alpha$ or $(x^\alpha\circ\tau_{s})(x,q)=x^\alpha\, ,$ and hence
 $$\pi_{\rk}\circ\tau_{s}=\pi_{\rk}.$$

 Therefore, taking into account this  identity we deduce that $\bar{\phi}_s$ is a section of $\pi_{\rk}$. In fact,
\[\pi_{\rk}\circ\bar{\phi}_s=\pi_{\rk}\circ\tau_s\circ\bar{\phi}=\pi_{\rk}\circ\bar{\phi}=id_{\rk}\,,\]
where in the last identity we use that $\bar{\phi}$ is a section of $\pi_{\rk}$. \qed
\index{Extremal}
\begin{definition} A section
$\bar{\phi}=(id_{\rk},\phi)\colon\rk\to\rkq$, such that $\bar{\phi}\in Sec_C(\rk,\rkq)$, is an {\bf extremal} of $\mathbb{S}$ if \[\ds\frac{d}{ds}\Big\vert_{s=0}\mathbb{S}(\tau_s\circ\bar{\phi})=0\]
where $\{\tau_s\}$ is the one-parameter group of diffeomorphism for some $\pi_{\rk}$-vertical vector field $Z\in\vf(\rkq)$.\end{definition}

The variational problem associated with $L$ consists in calculate the extremals of the action $\mathbb{S}$.

\begin{theorem} Let
$\bar{\phi}=(id_{\rk},\phi)\in Sec_C(\rk,\rkq)$ and $L:\rk\times T^1_kQ\to \r$ be a Lagrangian. The following statements are equivalent:\begin{enumerate}
    \item $\bar{\phi}$ is an extremal of $\mathbb{S}$.
    \item $\ds\int_{\rk}(\phi^{[1]})^*(\mathcal{L}_{Z^1}(Ld^kx))=0$,
    for each $\pi_{\rk}$-vertical $Z\in\vf(\rkq)$.
    \item   $\bar{\phi}$  is a solution of the Euler-Lagrange field equations (\ref{EL eq}).
\end{enumerate}
\end{theorem}
\dem

($1\Leftrightarrow 2$) Let $Z\in\mathfrak{X}(\rkq)$ a $\pi_{\rk}$-vertical vector field and $\{\tau_s\}$ the one-parameter group of diffeomorphism associated to $Z$.

Along this prove we denote by $\phi_s$ the composition $\phi_s=\pi_Q\circ \tau_s\circ \bar{\phi}$. Let us observe that $\bar{\phi}_s= \tau_s\circ \bar{\phi} = (id_{\rk},\phi_s)$.

Since $\phi_s^{[1]}=(\pi_Q\circ\tau_s\circ
\bar{\phi})_Q^{[1]}=j^1\tau_s\circ\phi^{[1]}$, then\[\begin{array}{lcl}
&&
\ds\frac{d}{ds}\Big\vert_{s=0}\mathbb{S}(\tau_s\circ\phi)=
\ds\frac{d}{ds}\Big\vert_{s=0}\ds\int_{\rk}
((\phi_s) ^{[1]})^*(Ld^kx)\\\noalign{\medskip}&=&\ds\lim_{s\to
0}\ds\frac{1}{h}\left(\ds\int_{\rk}((\phi_s) ^{[1]})^*(Ld^kx)
-\ds\int_{\rk}((\phi_0) ^{[1]})^*(Ld^kx)\right)\\\noalign{\medskip}
&=&
\ds\lim_{s\to
0}\ds\frac{1}{h}\left(\ds\int_{\rk}((\pi_Q\circ \tau_s\circ\phi) ^{[1]})^*(Ld^kx)
-\ds\int_{\rk}(\phi^{[1]})^*(Ld^kx)\right)\\\noalign{\medskip}
&=&
\ds\lim_{s\to
0}\ds\frac{1}{h}\left(\ds\int_{\rk}(\phi^{[1]})^*(j^1\tau_s)^*(Ld^kx)
-\ds\int_{\rk}(\phi^{[1]})^*(Ld^kx)\right)
\\\noalign{\medskip}
&=&
\ds\lim_{s\to
0}\ds\frac{1}{h}\ds\int_{\rk}(\phi^{[1]})^*[(j^1\tau_s)^*(Ld^kx)-(Ld^kx)]
\\\noalign{\medskip}&=&\ds\int_{\rk}(\phi^{[1]})^*\mathcal{L}_{Z^1}(Ld^kx)\,,
\end{array}\]
which implies the equivalence between items $(1)$ and $(2)$.

$(2\Leftrightarrow 3)\,$ We know that $\phi$ is an extremal or critical section of $\mathbb{S}$ if and only if for each $\pi_{\rk}$-vertical vector field
$Z$ one has
\[\ds\int_{\rk}(\phi^{[1]})^*(\mathcal{L}_{Z^1}(Ld^kx))=0\,.\]

Taking into account the identity
\begin{equation}\label{princ Hamilton:igualdad 1} \mathcal{L}_{Z^1}(Ld^kx)=
\iota_{Z^1}(dL\wedge d^kx)+d\iota_{Z^1}(Ld^kx)\,
\end{equation} and since $\phi$ has compact support, from Stokes' theorem one deduces that \begin{equation}\label{consecuencia stokes}\ds\int_{\rk}
(\phi^{[1]})^*d\iota_{Z^1}(Ld^kx)=
\ds\int_{\rk}d\big((\phi^{[1]})^*\iota_{Z^1}(Ld^kx)\Big)=0\,.
\end{equation}

Thus, from (\ref{princ Hamilton:igualdad 1}) and (\ref{consecuencia
stokes}) we obtain that $ \bar{\phi}$ is an extremal if and only if \[\ds\int_{\rk}(\phi^{[1]})^*\iota_{Z^1}(dL\wedge d^kx)=0\,.\]

 If $Z_{(x,q)}=Z^i(x,q)\ds\frac{\partial}{\partial
q^i}\Big\vert_{(x,q)}$ then the local expression of $Z^1$ is
\[Z^1=Z^i\ds\frac{\partial}{\partial
q^i}+\left(\ds\frac{\partial Z^i}{\partial x^\alpha}+\ds\frac{\partial
Z^i}{\partial q^j}v^j_\alpha\right)\ds\frac{\partial}{\partial
v^i_\alpha}\,,\] therefore
\begin{equation}\label{princ Hamilton:igualdad 2}
\iota_{Z^1}(dL\wedge d^kx)=\left(Z^i\ds\frac{\partial L}{\partial
q^i}+\left(\ds\frac{\partial Z^i}{\partial x^\alpha}+\ds\frac{\partial
Z^i}{\partial q^j}v^j_\alpha\right)\ds\frac{\partial L}{\partial
v^i_\alpha}\,\right )d^kx\,.
\end{equation}

Then, from (\ref{princ Hamilton:igualdad 2}) we obtain
\begin{equation}\label{princ Hamilton:igualdad 3}
\begin{array}{lcl} &&[(\phi^{[1]})^*\iota_{Z^1}(dL\wedge
d^kx)](q)=\\\noalign{\medskip} &=&\left((Z^i\circ
\bar{\phi})(x)\ds\frac{\partial L}{\partial
q^i}\Big\vert_{\phi^{[1]}(x)}+\left(\ds\frac{\partial
Z^i}{\partial x^\alpha}\Big\vert_{\bar{\phi}(x)}+\ds\frac{\partial
Z^i}{\partial q^j}\Big\vert_{\bar{\phi}(x)}\ds\frac{\partial
\phi^j}{\partial x^\alpha}\Big\vert_{x}\right)\ds\frac{\partial
L}{\partial v^i_\alpha}\Big\vert_{\phi^{[1]}(x)}\,\right
)d^kx\,.\end{array}
\end{equation}

Let us observe that the last term of  (\ref{princ Hamilton:igualdad
3}) satisfies
\[\ds\frac{\partial Z^i}{\partial
q^j}\Big\vert_{\bar{\phi}(x)}\ds\frac{\partial \phi^j}{\partial
x^\alpha}\Big\vert_{x}\ds\frac{\partial L}{\partial
v^i_\alpha}\Big\vert_{\phi^{[1]}(x)}d^kx=
\left(\ds\frac{\partial (Z^i\circ\bar{\phi})}{\partial
x^\alpha}\Big\vert_{x}-\ds\frac{\partial Z^i}{\partial
x^\alpha}\Big\vert_{\bar{\phi}(x)}\right)\ds\frac{\partial L}{\partial
v^i_\alpha}\Big\vert_{\phi^{[1]}(x)}d^kx\,.\]

After a easy computation we obtain
\[
{\footnotesize\ds\int_{\rk}(\phi^{[1]})^*\iota_{Z^1}(dL\wedge
d^kx)=\ds\int_{\rk}(Z^i\circ \phi)\ds\frac{\partial L}{\partial
q^i}\Big\vert_{\phi^{[1]}(x)}d^kx+\ds\int_{\rk}
\ds\frac{\partial (Z^i\circ\phi)}{\partial
x^\alpha}\Big\vert_{x}\ds\frac{\partial L}{\partial
v^i_\alpha}\Big\vert_{\phi^{[1]}(x)}d^kx\,.}
\]

Since $\bar{\phi}$ has compact support and using integration by parts, we have
{
\[\ds\int_{\rk} \ds\frac{\partial
(Z^i\circ\phi)}{\partial x^\alpha}\Big\vert_{x}\ds\frac{\partial
L}{\partial
v^i_\alpha}\Big\vert_{\phi^{[1]}(x)}d^kx=-\ds\int_{\rk}
(Z^i\circ\phi)(x)\ds\frac{\partial }{\partial
x^\alpha}\left(\ds\frac{\partial L}{\partial
v^i_\alpha}\Big\vert_{\phi^{[1]}(x)}\right)d^kx
\]}
and thus,
{
\[\ds\int_{\rk}(\phi^{[1]})^*\iota_{Z^1}(dL\wedge
d^kx)=\ds\int_{\rk}(Z^i\circ \phi)(x)\left(\ds\frac{\partial
L}{\partial
q^i}\Big\vert_{\phi^{[1]}(x)}-\ds\frac{\partial
}{\partial x^\alpha}\left(\ds\frac{\partial L}{\partial
v^i_\alpha}\Big\vert_{\phi^{[1]}(x)}\right)\right)d^kx=0\,.
\]}

Since the functions $Z^i$  are arbitrary, from the last identity we obtain the  Euler-Lagrange field equations,
\[\ds\frac{\partial
L}{\partial
q^i}\Big\vert_{\phi^{[1]}(x)}-\ds\frac{\partial
}{\partial x^\alpha}\left(\ds\frac{\partial L}{\partial
v^i_\alpha}\Big\vert_{\phi^{[1]}(x)}\right)=0\,, \quad 1\leq
i\leq n\,.
\]

\begin{remark}
{\rm

    In \cite{EMR-1996} the authors have considered a more general situation; instead of the bundle $\rktkq\to \rk$, they consider an arbitrary fiber-bundle $E\to M$.
    \rqed}
\end{remark}

\section{$k$-cosymplectic version of Euler-Lagran\-ge field equations}

In this section we give the $k$-cosymplectic description of the Euler-Lagrange field equations (\ref{EL field eq}). With this purpose, we introduce   some geometric elements associated to a Lagrangian function $L\colon \rktkq\to \r$

\subsection{Poincar\'{e}-Cartan forms on $\rktkq$}\label{kcolag}

\index{Poincar\'{e}-Cartan forms}
In a similar way that in the $k$-symplectic  approach,  one can define a family of $1$-forms $\Theta_L^1,\ldots, \Theta_L^k$ on $\rktkq$ associated with the Lagrangian function $L:\rktkq\to \r$, using the canonical tensor fields $J^1,\ldots, J^k$ defined in (\ref{localSA}). Indeed, we put
\begin{equation}\label{thetaLA cosim}
\Theta_L^\alpha =dL\circ J^\alpha\,,
\end{equation} with $\ak$. The exterior differential of these $1$-forms allows us to consider the family of $2$-forms  on  $\rktkq$ by
\begin{equation}\label{omegaLA cosim}
\Omega_L^\alpha =-\,d\Theta_L^\alpha \,.
\end{equation}

From (\ref{localSA}) and  (\ref{thetaLA cosim})  we obtain that  $\Theta_L^\alpha$ is locally given by
\begin{equation}\label{Loc thetala}
\Theta_L^\alpha =   \ds\frac{\displaystyle\partial
L}{\displaystyle\partial v^i_\alpha} \,  dq^i\,,\quad 1\leq  \alpha\leq k\,\end{equation} and from  (\ref{omegaLA cosim}) and  (\ref{Loc
thetala}) we obtain that  $\Omega_L^\alpha$ is locally given by
\begin{equation}\label{Loc omegaLa}
\small\Omega_L^\alpha =\derpars{L}{x^\beta}{v^i_\alpha}\,dq^i\wedge dx^\beta +
\derpars{L}{q^j}{v^i_\alpha}\,dq^i\wedge dq^j +
\derpars{L}{v^i_\beta}{v^i_\alpha}\,dq^i\wedge dv^i_\beta \,%, \quad 1\leq A\leq k
\;.
\end{equation}

An important case
is when the Lagrangian is regular, i.e, when \[\det\left(\derpars{L}{v^i_\alpha}{v^j_\beta}\right)\neq 0\,.\]
The following proposition gives a characterization of the regular Lagrangians.
        \begin{prop}\cite{LMeS-2001}\label{lag kcosymp manifold}
         Given a Lagrangian function on $\rktkq$, the following conditions are  equivalent:
        \begin{enumerate}
         \item $L$ is regular.
          \item $(dx^\alpha, \Omega^1_L,\ldots, \Omega^k_L, V)$ is a  $k$-cosymplectic structure on $\rktkq$, where
        $$V=\ker ( (\pi_{\rk})_{1,0})_*=span \left\{\derpar{}{v^i_1},\ldots,
        \derpar{}{v^i_k}\right\}$$  with $1\leq i\leq n$, is the vertical distribution
        of the vector bundle $(\pi_{\rk})_{1,0}:\rktkq\to \rkq$.
        \end{enumerate}
    \end{prop}

%%%%%%%%%%%%%%%%%%%%%%%%%%%%%%%%%%%%%%%%%%%%%%%%%%%%%%%%%%%%%%%%
\subsection{$k$-cosymplectic Euler-Lagrange  equation.}\protect\label{Sec 6.2.3.}

%%%%%%%%%%%%%%%%%%%%%%%%%%%%%%%%%%%%%%%%%%%%%%%%%%%%%%%%%%%%%%%%%%

We recall the   $k$-cosymplectic formulation of the   Eu\-ler-\-La\-gran\-ge  equations (\ref{EL eq}) developed by
M. de Le\'{o}n \textit{et al.} in \cite{LMeS-2001}.

 Let us consider the equations
\begin{equation}\label{lageq0}
\begin{array}{l}
dx^\alpha (X_\beta)  =   \delta^\alpha_\beta, \quad 1 \leq \alpha, \beta \leq k\,,\\
\noalign{\medskip} \ds\sum_{\alpha=1}^k \, \iota_{X_\alpha} \Omega_L^\alpha  = dE_L
+ \,\ds\sum_{\alpha=1}^k\ds\frac{\partial L}{\partial x^\alpha}dx^\alpha
\end{array}
\end{equation}
where $E_L=\Delta(L)-L$ and  denote by $\vf^k_L(\rk\times T^1_kQ)$ the
set of $k$-vector fields ${\bf X}=(X_1,\dots,X_k)$ on
$\rk\times T^1_kQ$ that are solutions of (\ref{lageq0}).

 Let us suppose that $(X_1,\ldots X_k)\in \vf^k_L(\rk\times T^1_kQ)$ and that each
$X_\alpha $ is locally given by
$$
X_\alpha=(X_\alpha)_\beta\frac{\partial}{\ds\partial x^\beta} +
(X_\alpha)^i\frac{\partial}{\ds\partial q^i}
  +(X_\alpha)^i_\beta\frac{\partial}{\ds\partial v_\beta^i}, \quad 1\leq  \alpha\leq k\,.
$$
Equations   (\ref{lageq0}) are  locally equivalent  to the equations
{\footnotesize
\begin{equation}\label{lform}\hspace{-0,5cm}\begin{array}{rcl}
  (X_\alpha)_\beta &=& \delta_\alpha^\beta\,, \\\noalign{\medskip}
  (X_\beta)^i \,
\ds\frac{\partial^2 L}{\partial x^\alpha \partial v^i_\beta}&=& v^i_\beta \,
\ds\frac{\partial^2 L}{\partial x^\alpha
\partial v^i_\beta}\,,  \\\noalign{\medskip}
  (X_\gamma)^j \, \ds\frac{\partial^2 L}{\partial v^i_\beta
\partial v^j_\gamma} &=& v^j_\gamma\,  \ds\frac{\partial^2 L}{\partial v^i_\beta
\partial v^j_\gamma}\,, \\\noalign{\medskip}
  \ds\frac{\partial^2 L}{\partial q^j \partial
v^i_\beta}\left(v^j_\beta -(X_\beta)^j\right)+ \ds\frac{\partial^2 L}{\partial
x^\beta \partial v^i_\beta}&+&v^k_\beta \ds\frac{\partial^2 L}{\partial q^k
\partial v^i_\beta} + (X_\beta)_\gamma^k \ds\frac{\partial^2 L}{\partial
v^k_\gamma \partial v^i_\beta} = \ds\frac{\partial L}{\partial q^i}\,.
\end{array}\end{equation}}

If $L$ is regular  then these equations are transformed in the following ones
\begin{equation}\label{xal}
 (X_\alpha)_\beta=\delta^\beta_\alpha,\quad (X_\alpha)^i = v^i_\alpha,\quad  \ds\sum_{\alpha=1}^kX_\alpha\Big(
 \ds\frac{\partial L}{\partial v^i_\alpha}\Big)=\ds\frac{\partial L}{\partial
 q^i} \; ,
\end{equation}
so that
$$
X_\alpha= \ds\frac{\partial}{\partial x^\alpha}+ v^i_\alpha
\ds\frac{\partial}{\partial q^i}+ (X_\alpha)_\beta^i \ds\frac{\partial
 }{\partial v^i_\beta}\, ,
$$
that is    $(X_1, \ldots, X_k)$ is a {\sc sopde}.

\begin{theorem}\label{rel} Let $L$ be a Lagrangian and
 ${\bf X}=(X_1,\dots,X_k)$ a  $k$-vector field such that
\[
dx^\alpha (X_\beta) = \delta^\alpha_\beta, \;, \quad
  \ds\sum_{\alpha=1}^k \, \iota_{X_\alpha} \Omega_L^\alpha  =\,
dE_L + \,\ds\sum_{\alpha=1}^k\ds\frac{\partial L}{\partial x^\alpha}dx^\alpha
\]
where $E_L=\Delta(L)-L$ and $1 \leq \alpha, \beta \leq k\,$. Then
\begin{enumerate}
\item If $L$ is regular,   ${\bf
X}=(X_1,\dots,X_k)$ is a  {\sc sopde}.

Moreover, if $\psi:\rk \to \rk\times T^1_kQ$ is   integral section of  ${\bf X}$, then $$\phi:\rk
\stackrel{\psi}{\longrightarrow}\rk \times T^1_kQ
 \stackrel{p_{Q}}{\longrightarrow}Q$$ is a solution of the
    Euler-Lagrange  equations (\ref{EL eq}).
\item If $(X_1,\dots,X_k)$ is integrable and $\phi^{[1]}:\rk
\to \rk\times T^1_kQ$ is an integral section,
then $\phi:\rk \to Q$ is   solution of
 the  Euler-Lagrange equations (\ref{EL eq}).
\end{enumerate}
\end{theorem}
\dem

${\bf (1)}$ It is a direct consequence of  the third equation in
   (\ref{lform})  and the third equation in  (\ref{xal}).

 ${\bf (2)}$   If
$\phi^{[1]}$ is an integral section  of ${\bf X}$ then from the last equation in
   (\ref{lform})  and the local expression (\ref{localfi2}) of
$\phi^{[1]}$, we deduce that  $\phi$ is solution of the
Euler-Lagrange equations (\ref{EL eq}).\qed

Therefore, Equations (\ref{lageq0}) can be considered as a geometric version of the Euler-Lagrange field equations. From now, we will refer these equations (\ref{lageq0}) as \emph{$k$-cosymplectic Lagrangian equations}.
\index{k-cosymplectic Lagrangian equations}

\begin{remark}
{\rm
If $L:\rk\times T^1_kQ\longrightarrow \r$ is   regular, then $(dx^\alpha,\Omega_L^\alpha ,V)$ is     a $k$-cosymplectic structure on $\rk
\times T^1_kQ$. The Reeb vector fields    $(R_L)_\alpha$ corresponding to this structure
are characterized by the conditions
$$
\iota_{(R_L)_\alpha} \, dx^\beta=\delta_\alpha^\beta \;, \quad \iota_{(R_L)_\alpha} \,
\Omega^\beta_L=0 \, ,
$$
and   they satisfy  $$(R_L)_\alpha(E_L)=-\derpar{L}{x^\alpha}\,.$$

Hence, if we write the $k$-cosymplectic Hamiltonian system  (\ref{geonah})
for  $H=E_L$ and the  $k$-cosymplectic manifold
$$(M=\rk\times T^1_kQ,dx^\alpha,\Omega_L^\alpha ,V)$$ we obtain
$$
dx^\alpha(X_\beta)=\delta_\beta^\alpha, \quad  \displaystyle \sum_{\alpha=1}^k \,
 \iota_{X_\alpha}\Omega^\alpha =
dE_L-\displaystyle\sum_{\alpha=1}^k (R_L)_\alpha(E_L)dx^\alpha \, \, .
$$
which are the equations  (\ref{lageq0}). Therefore, the $k$-cosymplectic Lagrangian formalism developed in this section is a particular case of the $k$-cosymplectic formalism described in chapter \ref{k-cosymp eq}. As in the Hamiltonian case, when the Lagrangian is regular one can prove that there exists a solution $(X_1,\ldots, X_k)$ of the system (\ref{lageq0}) but this solution is not unique.\rqed}\end{remark}

\index{$k$-cosymplectic Lagrangian $k$-vector fields}
\begin{definition}
    A $k$-vector field $\mathbf{X}=(X_1,\ldots, X_k)\in \vf^k(\rktkq)$ is called a \emph{$k$-cosymplectic Lagrangian $k$-vector field} for a $k$-cosymplectic Hamiltonian system $(\rktkq,dx^\alpha,\Omega_L^\alpha, E_L)$ if $\mathbf{X}$ is a solution of (\ref{lageq0}).
    We denote by $\vf^k_L(\rktkq)$ the set of all $k$-cosymplectic Lagrangian $k$-vector fields.
\end{definition}

\begin{remark}
{\rm
 If we write the equations   (\ref{lageq0}) for the case $k=1$,
 we obtain
\begin{equation}\label{k1d}
dt(X)=1 \;, \quad \iota_{X_L}\Omega_L = dE_L + \ds\frac{\partial
L}{\partial t }dt\,,\end{equation} which are equivalent to the dynamical equations
 $$dt(X)=1 \;, \quad \iota_{X_L}\Omega_L =0\,,$$
where $\Omega_L=\Omega_L+dE_L\wedge dt$  is   Poincar\'{e}-Cartan  $2$-form
Poincar\'{e}-Cartan, see
\cite{EMR-1991}.

It is well know that these equations give the dynamics  of the    non-autonomous mechanics.\rqed}
\end{remark}

\section[$k$-cosymplectic Legendre transformation]{The Legendre transformation and the equivalence between $k$-cosymplectic Ha\-miltonian and  Lagrangian formulations of Classical Field Theories}

As in the $k$-symplectic case, the $k$-cosymplectic Hamiltonian and Lagrangian description of Classical Field Theories are two equivalent formulations when the Lagrangian function satisfies some  regularity condition. The $k$-cosym\-plectic Legendre transformation transforms one of these formalisms into the other. In this section we shall define the Legendre transformation in the $k$-cosymplectic approach and prove the equivalence between both Hamiltonian and Lagrangian settings.
  Recall that in the $k$-cosymplectic approach a Lagrangian is a function defined on $\rktkq$, i.e. $L\colon \rktkq\to \r$.

%As in the $k$-symplectic case the Poincar\'{e}-Cartan forms can be obtained from the canonical forms
%  $\Theta^\alpha,\,\Omega^\alpha,\,1\leq
%\alpha\leq k$ on  $\rktkqh$ (defined on  (\ref{defform})) using the corresponding Legendre transformation associated to the Lagrangian function
%$L:\rktkq\to \r$.

\index{Legendre transformation}
\index{$k$-cosymplectic approach!Legendre transformation}
\begin{definition}\label{leg trans k-co}
Let $L:\rk \times T^1_kQ \to \r$ be  a Lagrangian, then the
\emph{Legendre transformation} associated to $L$, $$FL: \rk\times T^1_kQ \longrightarrow \rk\times (T^1_k)^*Q$$
is defined as follows
$$FL(x,{\rm v}_q)=(x, [FL(x,{\rm v}_q)]^1,\ldots,[FL(x,{\rm v}_q)]^k )$$
where
$$ [FL(x,{\rm v}_q)]^\alpha(u_{q})=
\displaystyle\frac{d}{ds}\Big\vert_{s=0}\displaystyle L\left(
x,{v_1}_{q}, \dots,{v_\alpha}_{q}+su_{q}, \ldots, {v_k}_{q}
\right)\,,
$$
for $1\leq \alpha \leq    k $, being $u_q\in T_qQ$ and $(x,{\rm v}_q)=(x,{v_1}_{q},\ldots, {v_k}_{q})\in \rktkq$.\end{definition}

 Using canonical coordinates $(x^\alpha, q^i,v^i_\alpha)$ on $\rktkq$ and $(x^\alpha, q^i,p^\alpha_i)$ on $\rktkqh$, we deduce that $FL$ is locally given by
  \begin{equation}\label{locfl1}
    \begin{array}{rccl}
FL\colon & \rktkq & \to & \rktkqh\\\noalign{\medskip}
&(x^\alpha,q^i,v^i_\alpha) & \mapsto &  \Big(x^\alpha,q^i,
\frac{\displaystyle\partial L}{\displaystyle\partial v^i_\alpha } \Big)\, .
\end{array}
\end{equation}
 The Jacobian matrix of $FL$ is the following matrix of order $n(k+1)$,
    \[
    \left(
      \begin{array}{cccccc}
      I_k &  0  & 0 & \cdots & 0 \\
      0 & I_n & 0 & \cdots & 0 \\
 \derpars{L}{x^\alpha}{v^j_1}   &\derpars{L}{q^i}{v^j_1} & \derpars{L}{v^i_1}{v^j_1} & \cdots & \derpars{L}{v^i_k}{v^j_1} \\
   \vdots     & \vdots & \vdots &  & \vdots \\
       \derpars{L}{x^\alpha}{v^j_k}   & \derpars{L}{q^i}{v^j_k} & \derpars{L}{v^i_1}{v^j_k} & \cdots & \derpars{L}{v^i_k}{v^j_k} \\
      \end{array}
    \right)
    \]
    where $I_k$ and $I_n$ are the identity matrix of order $k$ and $n$ respectively and $1\leq i,j\leq n$. Thus we deduce that $FL$ is a local diffeomorphism if and only if \[det\Big(\derpars{L}{v^i_\alpha}{v^j_\beta} \Big)\neq 0\,.\]

\index{$k$-cosymplectic approach!regular Lagrangian}
\index{$k$-cosymplectic approach!hyperregular Lagrangian}
\index{$k$-cosymplectic approach!sigular Lagrangian}
     \begin{definition}
        A Lagrangian function  $L: \rk\times T^1_kQ\longrightarrow \r $ is said to be  \emph{regular} (resp. \emph{hyperregular}) if the
        Legendre transformation $FL$ is a local diffeomorphism (resp. global). Otherwise,  $L$ is said to be  \emph{singular}.
    \end{definition}

 From the local expressions   (\ref{Loc thetala}), (\ref{Loc omegaLa}) and
 (\ref{locfl1}) of $\Theta^\alpha,\;\Omega^\alpha,\;\Theta_L^\alpha $ y $\Omega_L^\alpha $
we deduce that the relationship between the canonical and Poincar\'{e}-Cartan forms  is given by ($\ak$)
  \begin{equation}\label{rfh-fl}
\Theta_L^\alpha =FL^*\Theta^\alpha\,, \quad\Omega_L^\alpha =FL^*\Omega^\alpha\,.
\end{equation}

     Consider $V=\ker ((\pi_{\rk})_{1,0})_*$ the vertical distribution of the bundle $(\pi_{\rk})_{1,0}\colon $ $\rktkq \to \rkq$, then one easily obtains the following  characterization of a regular Lagrangian (the proof of this result can be found in \cite{Tesis merino}).
    \begin{prop}
        Let $L\in \mathcal{C}^\infty(\rktkq)$ be a Lagrangian function. $L$ is regular if and only if $(dx^1,\ldots, dx^k,\Omega_L^1,\ldots, \Omega_L^k, V)$ is a $k$-cosymplectic structure on $\rktkq$.
    \end{prop}
 %\textcolor[rgb]{1.00,0.00,0.00}{Escribimos la demostraci\'{o}n????}

    Therefore one can state the following theorem:
    \begin{theorem}
        Given a Lagrangian function $L\colon \rktkq\to \r$, the following conditions are equivalents:
        \begin{enumerate}
            \item $L$ is regular.
            \item $\det \left(\derpars{L}{v^i_\alpha}{v^j_\beta}\right)\neq 0$ with $1\leq i,j\leq n$ and $1\leq \alpha,\beta\leq k$.
            \item $FL$ is a local diffeomorphism.
        \end{enumerate}
    \end{theorem}

    Now we restrict ourselves to the case of hyperregular Lagrangian. In this case the  Legendre transformation $FL$ is a global  diffeomorphism and thus we can define a Hamiltonian  function    $H: \rk\times(T^1_k)^*Q \to \r$ by $$H=(FL^{-1})^*E_L=E_L \circ FL^{-1}$$ where $FL^{-1}$ is the inverse  diffeomorphism of $FL$.

    Under these conditions, we can state the  equivalence between both Hamiltonian and Lagrangian formalisms.

    \begin{theorem}\label{equivalence k-symp}
        Let $L\colon \rktkq\to \r$ be a hyperregular Lagrangian then:
        \begin{enumerate}
            \item $\mathbf{X}=(X_1,\ldots, X_k)\in \vf^k_L(\rktkq)$ if and only if $(T^1_kFL)(\mathbf{X})=(FL_*(X_1),$ $\ldots, FL_*(X_k))\in \vf^k_H(\rktkqh)$ where $H=E_L\circ FL^{-1}$.
            \item There exists a one to one correspondence between the set of maps $\phi\colon \rk\to Q$ such that $\phi^{[1]}$ is an integral section of some $(X_1,\ldots, X_k)\in \vf^k_L(\rktkq)$ and the set of maps $\psi\colon \rk\to \rktkqh$, which are integral section of some  $(Y_1,\ldots, Y_k)\in\vf^k_H(\rktkqh)$, being $H=(FL^{-1})^*E_L$.

            %If ${\bf X}=(X^1,\dots,X^k)$ is integrable and $\phi^{[1]}$ is an integral section of ${\bf X}$, then $\varphi=FL \circ \phi^{[1]}$ is an integral section of $(FL_*X_1,\ldots, FL_*X_k)$ and thus $\varphi=FL \circ \phi^{[1]}$ is a solution of the Hamilton-De Donder- Weyl equations.
        \end{enumerate}
    \end{theorem}
\proof

\begin{enumerate}
    \item  Given $FL$  we can consider the canonical prolongation $T^1_kFL$ following the definition given in (\ref{tkq: prolongation expr}). Thus given a $k$-vector field $\mathbf{X}=(X_1,\ldots, X_k)\in\vf_L^k(\rktkq)$, one can define a $k$-vector field on $\rktkqh$ by means of  the following diagram
            \[
                \xymatrix{
                \rktkq\ar[r]^-{FL}\ar[d]_-{\mathbf{X}} & \rktkqh\ar[d]^-{(T^1_kFL)(\mathbf{X})}\\
                T^1_k(\rktkq) \ar[r]^-{T^1_kFL} & T^1_k(\rktkqh)
                }
            \]
            that is, for each $\ak$, we consider the vector field on $\rktkqh$, $FL_*(X_\alpha)$.

            We now consider the function $H=E_L\circ FL^{-1}= (FL^{-1})^*E_L$; then
            \[
                (T^1_kFL)(\mathbf{X})=(FL_*(X_1),\ldots, FL_*(X_k))\in \vf^k_H(\rktkqh) \,
             \]

 %$$
%\begin{array}{l}
%\eta^\alpha(X_\beta)=\delta^\alpha_\beta, \quad 1\leq \alpha,\beta\leq k\\
%\noalign{\medskip}\displaystyle \sum_{\alpha=1}^k \,
% \iota_{X_\alpha}\Omega^\alpha =
%dH-\displaystyle\sum_{\alpha=1}^k R_\alpha(H)\eta^\alpha \, \, .
%\end{array}
%$$
             if and only if
             \[
             \begin{array}{l}
              dx^\alpha( FL_*(X_\beta))=\delta^\alpha_\beta\,, \\\noalign{\medskip}
                \ds\sum_{\alpha=1}^k\iota_{FL_*(X_\alpha)}\Omega^\alpha - d\Big((FL^{-1})^*E_L\Big)+\displaystyle\sum_{\alpha=1}^k R_\alpha\Big((FL^{-1})^*E_L\Big)dx^\alpha   =0\,.
              \end{array}
            \]

            Since $FL$ is a diffeomorphism  the above condition  is equivalent to the condition
            \[
                dx^\alpha(X_\beta)=\delta^\alpha_\beta
            \]
             and
             \[
             \begin{array}{rl}
                0=& FL^*\Big(\ds\sum_{\alpha=1}^k\iota_{FL_*(X_\alpha)}\Omega^\alpha - d(FL^{-1})^*E_L +\displaystyle\sum_{\alpha=1}^k R_\alpha\Big((FL^{-1})^*E_L\Big)dx^\alpha  \Big) \\\noalign{\medskip}
                =& \ds\sum_{\alpha=1}^k\iota_{X_\alpha}(FL)^*\Omega^\alpha-dE_L + \displaystyle\sum_{\alpha=1}^k R_\alpha(E_L)dx^\alpha=\ds\sum_{\alpha=1}^k\iota_{X_\alpha}(FL)^*\Omega^\alpha-dE_L - \displaystyle\sum_{\alpha=1}^k\derpar{L}{x^\alpha}dx^\alpha\,.               \end{array}
            \]
            But from (\ref{rfh-fl}) this  occurs if and only if  $\mathbf{X}\in\vf^k_L(\rktkq)$.

            Finally, observe that since $FL$ is a diffeomorphism, $T^1_kFL$ is so also, and then all $k$-vector field on $\rktkqh$ is of the type $T^1_kFL(\mathbf{X})$ for some $\mathbf{X}\in \vf^k(\rktkq)$.

        \item Let $\phi\colon \rk\to Q$ be a map such that its first prolongation $\phi^{[1]}$ is an integral section of some $\mathbf{X}=(X_1,\ldots, X_k)\in \vf^k_L(\rktkq)$, then the map $\psi=FL\circ \phi^{[1]}$ is an integral section of $$T^1_kFL(\mathbf{X})=(FL_*(X_1),\ldots, FL_*(X_k))\, .$$
        Since we have proved in $(1)$ that  $T^1_kFL(\mathbf{X})\in \vf^k_H(\rktkqh)$, we obtain the first part of the  item $(2)$.

            The converse is proved in a similar way. Notice that any $k$-vector field on $\rktkqh$ is of the form  $T^1_k\mathbf{X}$ for some $\mathbf{X}\in \vf^k(\rktkq)$. Thus given $\psi\colon \rk\to \rktkqh$ integral section of any $(Y_1,\ldots, Y_k)\in\vf^k_H(\rktkqh)$, there exists a $k$-vector field $\mathbf{X}\in \vf^k_L(\rktkq)$ such that $T^1_kFL(\mathbf{X})=(Y_1,\ldots, Y_k)$. Finally, the map $\psi$ corresponds with $\phi^{[1]}$, where $\phi=(\pi_Q)_1\circ \psi$. \qed
\end{enumerate}

 \begin{remark}{\rm
 Throughout this chapter we have developed the $k$-cosymplectic Lagrangian formalism on the trivial bundle $\rktkq\colon \rk$. In \cite{MSV-2005} we study the consequences on this theory when we consider a nonstandard flat connection on the bundle $\rktkq\colon \rk$. This paper, \cite{MSV-2005}, is devoted to the analysis of the deformed dynamical equations and solutions, both in Hamiltonian and Lagrangian settings and we establish a characterization of the energy $E_L$ based on variational principles. We conclude that the energy function is the only function that performs the equivalence between the Hamiltonian and Lagrangian variational principles when a nonstandard flat connection is considering. As a particular case, when $k=1$ we obtain the results of the paper \cite{EMR-1995}.
 \rqed
 }\end{remark}

 \begin{remark}
 {\rm
 The $k$-cosymplectic Lagrangian and Hamiltonian formalism of first-order classical field theories are reviewed and completed in \cite{RRSV-2012}, where several alternative formulations are developed. First, generalizing the construction of Tulczyjew for mechanics \cite{T1,T2}, we give a new interpretation of the classical field equations (in the multisymplectic approach this study can be see, for instance, in \cite{LMS-03}). Second, the Lagrangian and Hamiltonian formalisms are unified by giving an extension of the Skinner-Rusk formulation on classical mechanics \cite{skinner2}.
 \rqed}
 \end{remark}
\newpage
\mbox{}
\thispagestyle{empty} % para que no se numere esta p\'{a}gina

\chapter{Examples}
In this chapter we shall present some physical examples which can be described using the $k$-cosymplectic formalism (see \cite{MSV-2009} for more details).

%\section{Examples}\label{section k-cosymp Examples ham}

\section{Electrostatic equations}\label{example k-cosymp: hamiltonian electrostatic}

                   Consider the $3$-cosymplectic Hamiltonian  equations (\ref{geonah})
\begin{equation}\label{electroestatic_k-cosymp}
                        \begin{array}{l}
dx^\alpha(X_\beta)=\delta_{\alpha\beta}, \quad 1\leq \alpha,\beta\leq 3\\
\noalign{\medskip}\displaystyle \sum_{\alpha=1}^3 \,
 \iota_{X_\alpha}\Omega^\alpha =
dH-\displaystyle\sum_{\alpha=1}^3 R_\alpha(H)dx^\alpha \, \, .
\end{array}
                    \end{equation}
                    where   $H$ is the Hamiltonian function given  by           \begin{equation}\label{k-cosymp Hamil electro}
                \begin{array}{rccl}
                    H\colon & \r^3\times (T^1_3)^*\r & \to & \r\\\noalign{\medskip}
                     &(x^\alpha, q,p^\alpha) & \mapsto & 4\pi r(x)\sqrt{g}q+\ds\frac{1}{2\sqrt{g}}g_{\alpha\beta}p^\alpha p^\beta
                \end{array}\,,
            \end{equation}
        with $1\leq \alpha,\beta \leq 3$ and $r(x)$ is the scalar function on $\r^3$ determined by (\ref{density charge}), and $(X_1,X_2,X_3)$
                  is  a $3$-vector field   on $\r^3\times(T^1_3)^*\r$.

                 If $(X_1,X_2,X_3)$ is solution of (\ref{electroestatic_k-cosymp}) then,  from (\ref{k-cosymp condvf}), we deduce that each $X_\alpha$, with $1\leq \alpha \leq 3$  has the local expression
                    \[
                        X_\alpha= \derpar{}{x^\alpha} + \ds\frac{1}{\sqrt{g}}g_{\alpha\beta}p^\beta \derpar{}{q} + (X_\alpha)^\beta\derpar{}{p^\beta},
                    \]
                and  the  components $(X_\alpha)^\beta$, $1\leq \alpha,\beta \leq 3 $, satisfy the identity
                    \[
                         (X_1)^1+ (X_2)^2+ (X_3)^3=-4\pi r(x) \sqrt{g} \, .
                    \]

                Assume that $(X_1,X_2,X_3)$ is an integrable $3$-vector field; then, if \[
                    \begin{array}{ccccl}
                     \varphi &:& \r^3 &  \longrightarrow & \r^3\times (T^1_3)^*\r \\ \noalign{\medskip}
                                            & &   x  & \to &  \varphi(x)=(\psi(x),\psi^1(x),\psi^2(x),\psi^3(x))
                   \end{array} \]
                 is  an integral section of a $3$-vector field $(X_1,X_2,X_3)$ solution of (\ref{electroestatic_k-cosymp}),  we obtain that $\varphi$  is a solution of the electrostatic equations (\ref{local electrostatic eq}).

 \section{The massive scalar field}\label{example k-cosymp: hamiltonian scalar field}

               Consider the Hamiltonian function $H\colon \r^4\times (T^1_4)^*\r\to \r$ given by
                    \[
                        H(x^1,x^2,x^3,x^4,q, p^1, p^2, p^3, p^4)=\ds\frac{1}{2\sqrt{-g}}g_{\alpha\beta}p^\alpha p^\beta-\sqrt{-g}\left( F(q)-\ds\frac{1}{2}m^2q^2\right)\,,
                    \]
               where $(x^1,x^2,x^3,x^4)$ are the coordinates on $\r^4$, $q$ denotes the scalar field $\phi$ and $(x^1,x^2,x^3,x^4,q,p^1,p^2,$ $p^3,p^4)$ are the canonical coordinates on $\r^4\times (T^1_4)^*\r$.

                Consider the $4$-cosymplectic Hamiltonian  equation
                    \[
                        \begin{array}{l}
dx^\alpha(X_\beta)=\delta_{\alpha\beta}, \quad 1\leq \alpha,\beta\leq 4\\
\noalign{\medskip}\displaystyle \sum_{\alpha=1}^4 \,
 \iota_{X_\alpha}\Omega^\alpha =
dH-\displaystyle\sum_{\alpha=1}^4 R_\alpha(H)dx^\alpha \, \, .
\end{array}
                    \]
                associated to the above Hamiltonian function. From (\ref{partial Ham scalar k-co}) one obtains that, in natural coordinates, a $4$-vector field solution of this system of equations has the following local expression (with $1\leq \alpha\leq 4$)
                    \begin{equation}\label{vf scalar k-co}
                        X_\alpha=\derpar{}{x^\alpha}+\frac{1}{\sqrt{-g}}g_{\alpha\beta}p^\beta\derpar{}{q} + (X_\alpha)^\beta\derpar{}{p^\beta}\,,
                    \end{equation}
                where the functions $(X_\alpha)^\beta\in \mathcal{C}^\infty(\r^4\times (T^1_4)^*\r)$ satisfy
                    \begin{equation}\label{condition vf scalar k-co}
                        \sqrt{-g}\Big(F'(q)-m^2q\Big)= (X_1)^1+(X_2)^2+(X_3)^3+(X_4)^4\,.
                    \end{equation}

                Assume that $(X_1,X_2,X_3,X_4)$ is an integrable $4$-vector field. Let $\varphi\colon \r^4\to \r^4\times(T^1_4)^*\r,\, \varphi(x)=(x,\psi(x),\psi^1(x),\psi^2(x),\psi^3(x),\psi^4(x))$ be an integral section of a $4$-vector field solution of the $4$-cosymplectic Hamiltonian  equation. Then from (\ref{vf scalar k-co}) and (\ref{condition vf scalar k-co}) one obtains
                    \[
                        \begin{array}{l}
                            \derpar{\psi}{x^\alpha}= \frac{1}{\sqrt{-g}}g_{\alpha\beta}\psi^\beta\\
                            \sqrt{-g}\Big(F'(\psi)-m^2\psi\Big)=\derpar{\psi^1}{x^1}+ \derpar{\psi^2}{x^2}+ \derpar{\psi^3}{x^3}+\derpar{\psi^4}{x^4}\,.
                        \end{array}
                    \]

                Therefore, $\psi\colon \r^4\to \r$ is a solution of the equation
                    \[
                        \sqrt{-g}\Big(F'(\psi)-m^2\psi\Big)= \sqrt{-g}\derpar{}{x^\alpha}\left( g^{\alpha\beta}\derpar{\psi}{t^\beta} \right)\,,
                    \]
                that is, $\psi$ is a solution of the scalar field equation.
                \begin{remark}
                {\rm
                    Some particular case of the scalar field equation are the following:
                        \begin{enumerate}
                            \item If $F=0$ we obtain the linear scalar field equation.
                            \item If $F(q)=m^2q^2$, we obtain the Klein-Gordon equation \cite{saletan},
                                    \[
                                        (\square + m^2)\psi=0\,.
                                    \]
                        \end{enumerate}
                        \rqed}
                \end{remark}

For the Lagrangian counterpart, we consider again the Lagrangian (\ref{scalar lagrangian}).

Let ${\bf X}=(X_1,X_2,X_3,X_4)$ be an integrable solution of the equation (\ref{lageq0}) for $L$ and $k=4$, then if $\phi\colon\r^4\to\r$ is a solution of
${\bf X}$, then
we obtain that $\phi$ is a
solution of the equations:
\[\begin{array}{rl}0=&\derpars{L}{x^\alpha}{v_\alpha}\Big\vert_{\phi^{[1]}(t)} - \derpars{L}{q}{v_\alpha}\Big\vert_
{\phi^{[1]}(t)}\derpar{\phi}{x^\alpha}+
\derpars{L}{v_\alpha}{v_\beta}\Big\vert_{\phi^{[1]}(t)}\derpars{\phi}{x^\alpha}{x^\beta}-
\derpar{L}{q}\Big\vert_{\phi^{[1]}(t)} \\\noalign{\medskip} =& \sqrt{-g}\derpar{}{x^\alpha}\left(g^{\alpha\beta}
\derpar{\phi}{x^\beta}\right)-\sqrt{-g}(F'(\phi)-m^2\phi) \end{array}\] and thus, $\phi$
is a solution of the scalar field equation (\ref{scalar}).

%\section{Examples}

\section{Harmonic maps}

 Let us recall that a smooth map $\varphi\colon M\to N$ between two Riemannian
manifolds $(M,g)$ and $(N,h)$ is called \textit{harmonic} if it is a critical point of the energy
functional $E$, which, when $M$ is compact, is defined as
$$E(\varphi)=\int_{M}\frac{1}{2}trace_g\varphi^*h\,dv_g,$$
where $dv_g$ denotes the measure on $M$ induced by its metric and, in local coordinates,
the expression $\frac{1}{2}{\rm trace}_g\varphi^*h$ reads
$$\frac{1}{2}{\rm trace}_g\varphi^*h=\frac{1}{2}
g^{ij}h_{\alpha\beta}\derpar{\varphi^\alpha}{x^i}\derpar{\varphi^\beta}{x^j},$$
$(g^{ij})$ being the inverse of the metric matrix $(g_{ij})$.

 This definition can be extended to the case where $M$ is not compact by requiring that the restriction
 of $\varphi$ to every compact domain be harmonic, (for more details see \cite{{CGR-2001}, CM-2008, ,EL-1978})

Now we consider the particular case $M=\rk$, with coordinates
$(x^\alpha)$. In this case, taking the Lagrangian
$$\begin{array}{rccl} L\colon  & \rk\times T^1_kN & \to &
\r\\\noalign{\medskip} &(x^\alpha,q^i,v^i_\alpha) & \mapsto &
\frac{1}{2}g^{\alpha\beta}(x)h_{ij}(q)v^i_\alpha v^j_\beta\end{array}$$ and the
$k$-cosymplectic Euler-Lagrange equations (\ref{lageq0}) associated to it,
we obtain the following result:  if $\varphi\colon \rk\to N$ is such that $\varphi^{[1]}$
is an integral section of ${\bf X}=(X_1,\ldots, X_k)$, being ${\bf
X}=(X_1,\ldots, X_k)$ a solution of the geometric equation
(\ref{lageq0}), then, $\varphi$ is a solution of the Euler-Lagrange
equations

\begin{equation}\label{harmonic}
 \derpars{\varphi^i}{x^\alpha}{x^\beta}-\Gamma^\gamma_{AB}\derpar{\varphi^i}{x^\gamma}+\widetilde{\Gamma}^i
 _{jk}\derpar{\varphi^j}{x^\alpha}\derpar{\varphi^k}{x^\beta}=0\,\qquad 1\leq i\leq n\,,
 \end{equation}
 where $\Gamma^\gamma_{\alpha\beta}$ and $\widetilde{\Gamma}^i_{jk}$ denote the Christoffel symbols
 of the Levi-Civita connections of $g$ and $h$, respectively.

Let us observe that these equations are the Euler-Lagrange equations associated to the
energy functional $E$, and (\ref{harmonic}) can be written  as $${\rm trace}_g\nabla d\varphi^*h=0,$$
where $\nabla$ is the connection on the vector bundle $T^*\r^k\otimes \varphi^*(TN)$
induced by the Levi-Civita  connections on $\rk$ and $N$ (see, for example, \cite{EL-1978}). Therefore, if
$\varphi\colon\rk\to N$ is a solution of (\ref{harmonic}), then $\varphi$ is harmonic.

\begin{remark}
{\rm Some examples of harmonics maps are the following ones:
\begin{itemize}
\item Identity and constant maps are harmonic.

\item In the case $k=1$, that is, when $\varphi\colon\r\to N$ is a curve on $N$, we deduce that $\varphi$
is a harmonic map if and only if it is a geodesic.

\item Now,   consider the case $N=\r$ (with the standard metric). Then $\varphi\colon\rk\to\r$ is a
harmonic map if and only if it is a harmonic function, that is, is a solution of the Laplace equation.
\end{itemize}\rqed}
\end{remark}

\section{Electromagnetic Field in vacuum: \newline Maxwell's equations.}

\index{Maxwell's equations}
As it is well know (see \cite{Frankel}),  Maxwell's equations in $\r^3$, are
\begin{eqnarray}
  \label{Gauss's LAwane-1994} \makebox{(Gauss's Law)}   &
    \nabla \cdot\mathbf{E} =  \rho  \\\noalign{\medskip}
\label{Ampere's LAwane-1994} \makebox{(Ampere's Law)}  &  \qquad  \nabla \times \mathbf{B} = \mathbf{j} +
\ds\frac{\partial\mathbf{E}}{\partial t} \\\noalign{\medskip}
\label{Faraday's LAwane-1994} \makebox{(Faraday's Law)} & \nabla\times \mathbf{E} + \ds\frac{\partial\mathbf{B}}
{\partial t}=0\\\noalign{\medskip}
\label{(Absence of Free Magnetic Poles2)} \makebox{(Absence of Free Magnetic Poles)} & \nabla\cdot
\mathbf{B} =0\,.
  \end{eqnarray}

Here, the symbols in bold represent vector quantities in
$\r^3$, whereas symbols in italics represent scalar quantities.

  The first two equations are inhomogeneous, while the other two are homogeneous. Here, $\rho$ is
  the charge density, $\mathbf{E}$ is the electric field vector, $\mathbf{B}$ is the magnetic field
  an,d   $\mathbf{j}$ is the   current density vector, which satisfies the continuity equation
  $$\derpar{ \rho}{ t} + \nabla\cdot \mathbf{j}=0\,.$$

In what follows, we consider a four-dimensional formulation of Maxwell's equations. To do that, one considers the Minkowski Space of Special Relativity. Therefore, the space-time is a
$4$-dimensional manifold $M^4$ that is topologically just $\r^4$. A point in space-time
has coordinates $(x,y,z,t)$ which we shall write as $(x^1,x^2,x^3,x^4)$ instead. In this
space we consider the Minkowski metric ($ds^2=\d r^2 - \d x^2$ where $\d r^2$ denotes the
euclidean metric of $\r^3$), that is, (for simplicity we shall assume the velocity of
light $c=1$):
$$ds^2=  \d (x^1)^2+\d (x^2)^2+\d (x^3)^2-\d(x^4)^2\,.$$

In the four-dimensional Minkowski space, Maxwell's equations assume an extremely compact
form, which we recall now, (see \cite{Frankel,SCYYQLY-2008,WR-2006} for more details).

First, we consider the Faraday $2$-form
\begin{equation}\label{faraday form}
\begin{array}{lcl}
\mathcal{F}&=& E_1\d x^1\wedge\d x^4 + E_2\d x^2\wedge\d x^4 + E_3\d x^3\wedge\d x^4 \\& +& B_1\d x^2\wedge
\d x^3 + B_2\d x^3\wedge \d x^1 + B_3\d x^1\wedge \d x^2\,.
\end{array}
\end{equation}

If we compute $\d \mathcal{F}$, we obtain that the homogeneous Maxwell equations
 (\ref{Faraday's LAwane-1994}-\ref{(Absence of Free Magnetic Poles2)}) are equivalent to $\d\mathcal{F}=0$,
 that is, the Faraday form is closed.

Since $\d\mathcal{F}=0$ in $\r^4$, we must have
\begin{equation}\label{Faraday exact}
\mathcal{F}=d\mathcal{A}
\end{equation}
  $\mathcal{A}$ being the ``potential" $1$-form,  which is written  as
\begin{equation}\label{potencial form}
\mathcal{A}= A_1\d x^1 +  A_2 \d x^2 +  A_3 \d x^3+\Phi\d x^4\in \Lambda^1(\r^4)\;,
\end{equation} where $A_1,A_2,A_3$ are the components of the magnetic vector potential
and $\Phi$ is the scalar electric potential.

To develop a four-dimensional formulation of  the divergence law for the electric flux
density (\ref{Gauss's LAwane-1994}) and Ampere's law (\ref{Ampere's LAwane-1994}), we introduce the
four-current differential form
\begin{equation}\label{four-current}
\mathcal{J}= j_1\d x^1 + j_2 \d x^2 + j_3\d x^3-\rho\d x^4\in \Lambda^1(\r^4)
\end{equation} where $j_1,j_2,j_3$ are the components of the electric current and $\rho$ is the
density of electric charge.

The four-dimensional formulation of  the divergence law (\ref{Gauss's LAwane-1994}) and Ampere's
law (\ref{Ampere's LAwane-1994}), is
\begin{equation}\label{inhomogeneous}
\delta \mathcal{M}= \mathcal{J}
\end{equation}
where $\mathcal{M}\in \Lambda^2(\r^4)$ is the {\it Maxwell form} defined by $\mathcal{M}=\star
\mathcal{F}$ and $\delta\colon=\star \d \star$ is the coderivative; here $\star\colon \Omega^k
(\r^4)\to \Omega^{4-k}(\r^4)$ denotes the four-dimensional Hodge operator for Minkowski's space.

%\begin{remark}
%{\rm
%In general, on an orientable $n$-manifold with a Riemannian metric $g$, the {\it Hodge operator}
%$\star\colon \Omega^k(M)\to \Omega^{n-k}(M)$ is a linear operator such that for every $\alpha , \beta\in
%\Omega^k(M)$, we have
%$$\alpha\wedge \star \beta= g(\alpha,\beta) \d vol_g.$$
%
%In local coordinates we have $$\alpha=\alpha_{i_1\ldots i_k}\d x^{i_1}\wedge\d
%x^{i_k},\beta=\beta_{j_1 \ldots j_k}\d x^{j_1}\wedge\d x^{j_k},\,
%g(\alpha,\beta)=\alpha_{i_1\ldots i_k}\beta_{j_1\ldots j_k}g^{i_1j_1} \ldots
%g^{i_kj_k}\,,$$ and $\d vol_g=\sqrt{|det(g_{ij}|}dx^1\wedge\ldots\wedge dx^n$, $(g^{ij})$
%being the inverse of the metric matrix $(g_{ij})$. For more details see, for instance,
%\cite{SCYYQLY-2008}.}
%\end{remark}

In conclusion, in a four-dimensional Minkowski's space, Maxwell's equations can be written
as follows
\begin{equation}\label{4-maxwell eq}
\begin{array}{rcl}
\d\mathcal{F} &=& 0\,,\\\noalign{\medskip}
\delta \mathcal{M} &=& \mathcal{J}\,.
\end{array}
\end{equation}

Now we show that, since $\mathcal{F}= \d\mathcal{A}$, then the inhomogeneous equation $\delta \mathcal{M}=
\mathcal{J}$ is equivalent to the Euler-Lagrange equations for some Lagrangian $L$.

In that case, a solution of Maxwell's equations is a $1$-form $\mathcal{A}$ on the  Minkow\-ski's space, that
is, $\mathcal{A}$ is a section of the canonical projection $\pi_{\r^4}\colon $ $ T^*\r^4\cong \r^4\times \r^4\to
\r^4$. Here $Q= \r^4$. Moreover, see \cite{{EM-92},s1,s2}, $\r^4\times T^1_4\r^4$ is canonically
isomorphic to $(\pi_{\r^4})^*T^*\r^4\otimes (\pi_{\r^4})^*T^*\r^4$ via the identifications
$$\begin{array}{ccc}
\r^4\times x^1_4\r^4 & \equiv& (\pi_{\r^4})^*T^*\r^4\otimes (\pi_{\r^4})^*T^*\r^4\,\\\noalign{\medskip}
A^{[1]}({\bf t})=(x^j,A_i({\bf t}),\derpar{A_i}{x^j}({\bf t})) & \equiv & \derpar{A_j}{x^i}({\bf t})(dx^i\otimes
dx^j)\end{array}$$ where $1\leq i,j\leq 4$ and $A_4= \Phi$, and $\mathcal{A}^{[1]}\colon \r^4\to \r^4
\times T^1_4\r^4$ is the first prolongation of a section $\mathcal{A}\in \Lambda^1(\r^4)$ of $\pi_{\r^4}$.

  Then the Lagrangian  $L\colon \r^4\times T^1_4\r^4=(\pi_{\r^4})^*T^*\r^4\otimes (\pi_{\r^4})^*T^*\r^4 \to \r$ is given by
$$ L(\mathcal{A}^{[1]})=\ds\frac{1}{2}||\mathfrak{A}(\mathcal{A}^{[1]})|| - <\mathcal{J},\mathcal{A}> =
\ds\frac{1}{2}||\d \mathcal{A} ||  - <\mathcal{J},\mathcal{A}> ,$$ where $\mathfrak{A}$
is the alternating operator, and we have used the induced metric on
$(\pi_{\r^4})^*T^*\r^4\otimes (\pi_{\r^4})^*T^*\r^4$ by the metric on $\r^4$, see
\cite{Poor}. Here, $<\mathcal{J},\mathcal{A}>$ denotes the scalar product in
$(\r^4)^*$ given by the scalar product on $\r^4$, see \cite{Poor},
$$ <\mathcal{J},\mathcal{A}>=j_1A_1+j_2A_2+j_3A_3+\rho\Phi$$

As in the above section, if we take $(x^\alpha)=(x^1,x^2,x^3,x^4)$  coordinates on $\r^4$, $\, q^i$  are the
coordinates on the fibres of $T^*\r^4=\r^4\times \r^4$ and $v^i_\alpha$ are the induced coordinates on the fibres
of $\r^4\times T^1_4\r^4$, then $L$ is locally given by
\begin{equation}\label{Maxwell Lagrangian}\begin{array}{lcl}
L(x^\alpha,q^i,v^i_\alpha)&=&\ds\frac{1}{2}((v^2_1-v^1_2)^2+(v^3_1-v^1_3)^2+(v^3_2-v^2_3)^2
-(v^4_1-v^1_4)^2  \\ \noalign{\medskip}  &  &-(v^4_2-v^2_4)^2
 -(v^4_3-v^3_4)^2
)   \\ \noalign{\medskip}  &  & -j_1q^1-j_2q^2-j_3q^3-\rho q^4\,.\end{array}\end{equation}

\begin{remark}
{\rm
Let us observe that for a section $\mathcal{A}=A_1\d x^1+A_2\d x^2+A_3\d x^3+\Phi\d x^4$, if
$\mathcal{F}= \d \mathcal{A}$, we have:
$$||\mathfrak{A}(\mathcal{A}^{[1]})||=||\d \mathcal{A} ||=\ds\sum_{i<j<4}(\derpar{A_j}{x^i}-
\derpar{A_i}{x^j})^2-\ds\sum_{i<4}(\derpar{\Phi}{x^i}-\derpar{A_i}{x^4})^2 = ||B||^2-||E||^2$$\,.
\rqed}
\end{remark}

Now, we consider the $4$-cosymplectic equation \begin{equation}\label{geome4}\begin{array}{l}
\d x^\alpha (X_\beta) \, = \, \delta^\alpha_\beta, \quad 1 \leq A , B \leq 4\, ,\\
\noalign{\medskip} \ds\sum_{\alpha=1}^4 \, i_{X_\alpha} \Omega_L^\alpha =\,
dE_L + \,\ds\sum_{\alpha=1}^4\ds\frac{\partial L}{\partial x^\alpha}\d x^\alpha
\end{array}\end{equation}
 where the Lagrangian $L$ is given by (\ref{Maxwell Lagrangian}) and
  ${\bf X}=(X_1,X_2,X_3,X_4)$ is a $4$-vector field on $\r^4\times T^1_4\r^4$.

   Let $\mathcal{A}\in \Lambda^1(\r^4)$ be a section of $\pi_{\r^4}$ which is a solution of ${\bf X}$,
   then from  (\ref{Maxwell Lagrangian}) we obtain that $\mathcal{A}$ is a solution of
   the following system of equations:
   \begin{equation}\label{maxwell EL-eq}
   \begin{array}{rcl}
   \derpars{A_1}{x^2}{x^2}- \derpars{A_2}{x^1}{x^2} + \derpars{A_1}{x^3}{x^3}-\derpars{A_3}{x^1}{x^3} -
   \derpars{A_1}{x^4}{x^4} + \derpars{\Phi}{x^1}{x^4} &=& -j_1\\\noalign{\medskip}
   \derpars{A_2}{x^1}{x^1}- \derpars{A_1}{x^2}{x^1} + \derpars{A_2}{x^3}{x^3}-\derpars{A_3}{x^2}{x^3} -
   \derpars{A_2}{x^4}{x^4} + \derpars{\Phi}{x^2}{x^4} &=& -j_2\\\noalign{\medskip}
   \derpars{A_3}{x^1}{x^1}- \derpars{A_1}{x^3}{x^1} + \derpars{A_3}{x^2}{x^2}-\derpars{A_2}{x^3}{x^2} -
   \derpars{A_3}{x^4}{x^4} + \derpars{\Phi}{x^3}{x^4} &=& -j_3\\\noalign{\medskip}
   \derpars{A_1}{x^4}{x^1}- \derpars{\Phi}{x^1}{x^1} + \derpars{A_2}{x^4}{x^2}-\derpars{\Phi}{x^2}{x^2} +
   \derpars{A_3}{x^4}{x^3} - \derpars{\Phi}{x^3}{x^3} &=& -\rho
   \end{array}
   \end{equation}

   On the other hand, using  $\mathcal{F} =\d \mathcal{A}$, from (\ref{faraday form})
   one obtains that the equations (\ref{maxwell EL-eq}) can we written as follow
   $$\begin{array}{rcl}
   -\derpar{B_3}{x^2}+\derpar{B_2}{x^3}+\derpar{E_1}{x^4}&=& -j_1
   \\\noalign{\medskip}
   \derpar{B_3}{x^1}-\derpar{B_1}{x^3}+\derpar{E_2}{x^4}&=& -j_2
   \\\noalign{\medskip}
   -\derpar{B_2}{x^1}+\derpar{B_1}{x^2}+\derpar{E_3}{x^4}&=& -j_3
   \\\noalign{\medskip}
   \derpar{E_1}{x^1}+\derpar{E_2}{x^2}+\derpar{E_3}{x^3}&=& \rho
   \end{array}$$
   which is equivalent to the condition $\delta\mathcal{M}= \mathcal{J}$.

   In conclusion, the $4$-cosymplectic equation (\ref{geome4}) is a geometric version of
   the inhomogeneous Maxwell equation $\delta\mathcal{M}= \mathcal{J}$, and considering
   $\mathcal{F}\colon = \d \mathcal{A}$ we also recover the homogeneous Maxwell equation
   $\d \mathcal{F}= 0$.

\begin{remark}\
{\rm
\begin{enumerate}
\item In the particular case $\mathcal{J}=0$, that is when $\rho=0, \mathbf{j}=0$, the Lagrangian
(\ref{Maxwell Lagrangian}) is a function defined on $\mathcal{C}^\infty(T^1_4\r^4)$.
Therefore, it is another example of the $k$-symplectic Lagrangian formalism. This
Lagrangian corresponds to the electromagnetic field without currents.
\item The Lagrangian (\ref{Maxwell Lagrangian})  can also be written as follows:
$$L=-\ds\frac{1}{4}f_{ik}f^{ik} - <\mathcal{J},\mathcal{A}>\,,$$ where
$$f_{ik}=\derpar{A_k}{x^i}-\derpar{A_i}{x^k} \quad \makebox{and}\quad f_{ik}f^{ik}=g^{il}g^{km}f_{ik}f_{lm}\,.$$

This Lagrangian  can be extended, in presence of gravitation, as follows (see \cite{Carmeli}):
     \begin{equation}\label{maxwell grav}
     L=-\ds\frac{1}{4}\sqrt{-g}f_{ik}f^{ik} -  \sqrt{-g}<\mathcal{J},\mathcal{A}>\,,
     \end{equation} where now we have used the  space-time metric tensor $(g_{ij})$ to raise the indices
     of the Maxwell tensor,
     $$ f^{ik}=g^{il}g^{km} f_{lm}\,.$$ In this case, in a similar way to the above discussion, and using that
     $$\derpar{L}{v^4_\beta}= \sqrt{-g}f^{4\beta}\quad , \quad \derpar{L}{v^i_\beta}=\sqrt{-g}f^{i\beta}\;,\quad 1\leq \beta,
     i\leq 3,$$ we obtain that the equations (\ref{geome4}) for the Lagrangian  are the geometric version
     of the following equations
     \begin{equation}\label{maxwell grav eq}\begin{array}{l}
     \nabla_{k}f^{4k}= \rho\\\noalign{\medskip}
     \nabla_{k}f^{ik}=j_i,\quad i=1,2,3\\\noalign{\medskip}
     \nabla_{l}f_{ik}+ \nabla_{i}f_{kl} +\nabla_{k}f_{li}=0\,,
     \end{array}\end{equation} where
     $$\nabla_kf^{ik}:=\ds\frac{1}{\sqrt{-g}}\frac{\partial}{\partial x^k}\left( \sqrt{-g}f^{ik}
     \right)$$ is the covariant divergence of a skew-symmetric tensor in the curved spacetime. These equations (\ref{maxwell grav eq}) are called the  Maxwell equations in the
     presence of gravitation, see \cite{Carmeli}.
\end{enumerate}
\rqed}
\end{remark}

Finally it is important to observe that all these physical examples  that can be described using the $k$-symplectic formalism can be also described using the $k$-cosymplectic approach.

\newpage
\mbox{}
\thispagestyle{empty} % para que no se numere esta p\'{a}gina

\chapter{$k$-symplectic systems versus autonomous $k$-cosymplectic systems}

In this book we are presenting two different approaches to describe first-order Classical Field Theories: first, when the Lagrangian and Hamiltonian
do not  depend on the base coordinates, and, later, when the Lagrangian and Hamiltonian also depend on the ``space-time'' coordinates. However if we observe the corresponding   descriptions we see that in local coordinates they give a geometric description of the same system of partial differential equations. Therefore the natural question is: \textit{Is there any relationship between $k$-symplectic and $k$-cosymplectic systems?} In this section we give an affirmative answer to this question. Naturally, this relation will be establish only when the Lagrangian and Hamiltonian do not depend on the base coordinates.

Along this section we work over the geometrical models of $k$-symplectic and $k$-co\-sym\-plec\-tic manifolds, that is $\tkqh$ and $\rktkqh$ and the Lagrangian counterparts $\tkq$ and $\rktkq$. However the following results and comments can be extend to the case $\rk\times M$ and $M$, being $M$ an arbitrary $k$-symplectic manifold.

Following a similar terminology   to that in Mechanics, we introduce the following definition.
\index{Autonomous Hamiltonian}
\begin{definition}
A $k$-cosymplectic Hamiltonian system
 $(\rk\times(T^1_k)^*Q,{\mathcal H})$
is said to be \emph{autonomous} if  $\ds\Lie{R_\alpha}{\mathcal
H}=\nicefrac{\partial\mathcal H}{\partial x^\alpha}=0$ for all $\ak$.
\label{autonomous}
\end{definition}

Observe that the condition in definition \ref{autonomous} means
that ${\mathcal H}$ does not depend on the variables $x^\alpha$, and thus
${\mathcal H}=\bar\pi_2^*H$ for some $H\in\Cinfty((T^1_k)^*Q)$, being $\bar{\pi}_2\colon \rktkqh\to\tkqh$ the canonical projection.

For an autonomous $k$-cosymplectic Hamiltonian system, the equations
(\ref{geonah}) become
\begin{equation}
\label{geonahaut}
 \ds\sum_{\alpha=1}^k\iota_{\bar X_\alpha}\Omega^\alpha =\d {\mathcal H}, \quad
\eta^\alpha(\bar X_\beta)=\delta^\alpha_\beta\ .
\end{equation}

Therefore:

\begin{prop}
 Every autonomous $k$-cosymplectic Hamiltonian system
 $(\rk\times(T^1_k)^*Q,{\mathcal H})$
defines a $k$-symplectic Hamiltonian system $((T^1_k)^*Q,H)$, where
${\mathcal H}=\bar\pi_2^*H$, and conversely.
\end{prop}

We have the following result for the solutions of the HDW equations.
\begin{theorem}
Let $(\rk\times(T^1_k)^*Q,{\mathcal H})$ be an autonomous
$k$-cosymplectic Ha\-miltonian system and let $((T^1_k)^*Q,H)$ be its
associated $k$-symplectic Hamiltonian system. Then, every section
$\bar\psi\colon\rk\to\rk\times(T^1_k)^*Q$, that is a solution of the
HDW-equation (\ref{HDW field eq}) for the system $(\rk\times(T^1_k)^*Q,{\mathcal
H})$ defines a map $\psi\colon\rk\to (T^1_k)^*Q$ which is a solution of the HDW-equation (\ref{HDW_eq}) for the system $((T^1_k)^*Q,H)$;
and conversely. \label{onetoone}
\end{theorem}
\begin{proof} Since ${\mathcal H}=\bar\pi_2^*H$ we have
\begin{equation}\label{h}
\ds\frac{\partial{\mathcal H}}{\partial q^i}= \ds\frac{\partial
H}{\partial q^i}\,,\quad \ds\frac{\partial {\mathcal H}}{\partial p^\alpha_i}= \ds\frac{\partial H}{\partial p^\alpha_i}\;.
\end{equation}

Let $\bar\psi\colon\rk\to\rk\times(T^1_k)^*Q$ be a section of the
projection $\bar\pi_1$, which in coordinates is expressed by
$\bar\psi(x)=(x,\bar\psi^i(x),\bar\psi^\alpha_i(x))$. Then we construct
the map $\psi=\bar\pi_2\circ\bar\psi\colon\rk\to(T^1_k)^*Q$, which
in coordinates is expressed as
$\psi(x)=(\psi^i(x),\psi^\alpha_i(x))=(\bar\psi^i(x),\bar\psi^\alpha_i(x))$.
Thus  if $\bar\psi$ is a solution of the HDW-equations (\ref{HDW field eq}),
from (\ref{h}) we obtain that $\psi$ is a solution of the
HDW-equations (\ref{HDW_eq}), and the statement holds.

Conversely, consider a map $\psi\colon\rk\to(T^1_k)^*Q$. We define
$\bar\psi=(Id_{\r^k},\psi):\rk\to\rk\times(T^1_k)^*Q$.
Furthermore, if $\psi(x)=(\psi^i(x),\psi^\alpha_i(x))$, then
$\bar\psi(x)=(x,\bar\psi^i(x),\bar\psi^\alpha_i(x))$, with
$\bar\psi^i(x)=\psi^i(x)$ and $\bar\psi^\alpha_i(x)=\psi^\alpha_i(x)$ (observe
that, in fact,  ${\rm Im}\,\bar\psi={\rm graph}\,\psi$). Hence,  if
$\psi$ is a solution of the HDW-equations (\ref{HDW_eq}), from
(\ref{h}) we obtain that $\bar\psi$ is a solution of the
HDW-equations (\ref{HDW field eq}), and the statement holds.
 \end{proof}

For $k$-vector fields that are solutions of the geometric field
equations (\ref{ecHksym}) and (\ref{geonahaut}) we have:

\begin{prop}
Let $(\rk\times(T^1_k)^*Q,{\mathcal H})$ be an autonomous
$k$-cosymplectic Ha\-miltonian system and let $((T^1_k)^*Q,H)$ be its
associated $k$-symplectic Hamiltonian system. Then every $k$-vector
field ${\bf X}\in\vf^k_H(T^1_k)^*Q)$ defines a  $k$-vector field
${\bf \bar X}\in\vf^k_{\mathcal H}(\rk\times(T^1_k)^*Q)$.

Furthermore, ${\bf X}$ is integrable if, and only if, its associated
${\bf \bar X}$ is integrable too.

\end{prop}
\begin{proof}
 Let ${\bf X}=(X_1,\dots,X_k)\in\vf^k_H((T^1_k)^*Q)$. For
every $\alpha=1,\ldots,k$, let $\bar X_\alpha\in\vf(\rk\times(T^1_k)^*Q)$ be
the {\sl suspension} of the corresponding vector field
$X_\alpha\in\vf((T^1_k)^*Q)$, which is defined as follows (see \cite{AM-1978},
p. 374, for this construction in mechanics): for every ${\rm
p}\in(T^1_k)^*Q$, let $\gamma^\alpha_{\rm p}\colon\r\to(T^1_k)^*Q$ be the
integral curve of $X_\alpha$ passing through ${\rm p}$; then, if
$x_0=(x_0^1,\ldots,x_0^k)\in\rk$, we can construct the curve
$\bar\gamma^\alpha_{\bar{\rm p}}\colon\r\to\rk\times(T^1_k)^*Q$, passing
through the point $\bar{\rm p}\equiv(x_0,{\rm
p})\in\rk\times(T^1_k)^*Q$, given by $\bar\gamma^\alpha_{\bar{\rm
p}}(x^\alpha)=(x_0^1,\ldots,x^\alpha+x_0^\alpha,\ldots,x_0^k;\gamma_{\rm p}(x^\alpha))$.
Therefore, $\bar X_\alpha$ is the vector field tangent to
$\bar\gamma^\alpha_{\bar{\rm p}}$ at $(x_0,{\rm p})$. In natural
coordinates, if $X_\alpha$ is locally given by \[
X_\alpha= (X_\alpha)^i\frac{\partial}{\partial  q^i}+
(X_\alpha)_i^\beta\frac{\partial}{\partial p_i^\beta}
\] then $\bar
X_\alpha$ is locally given by
$$
\bar X_\alpha =\derpar{}{x^\alpha}+(\bar X_\alpha)^i\frac{\partial}{\partial q^i}+
(\bar X_\alpha)_i^\beta\frac{\partial}{\partial p_i^\beta}=
\derpar{}{x^\alpha}+\bar\pi_2^*(X_\alpha)^i\frac{\partial}{\partial  q^i}+
\bar\pi_2^*(X_\alpha)_i^\beta\frac{\partial}{\partial p_i^\beta}\, .
$$
Observe that the $\bar X_\alpha$ are $\bar\pi_2$-projectable vector fields,
and $(\bar\pi_2)_*\bar X_\alpha=X_\alpha$. In this way we have defined a
$k$-vector field ${\bf\bar  X}=(\bar X_1,\dots,\bar X_k)$ in
$\rk\times(T^1_k)^*Q$. Therefore, taking (\ref{symcosym}) into
account, we obtain
$$
\ds\sum_{\alpha=1}^k\iota_{\bar X_\alpha}\Omega^\alpha-\d{\mathcal
H}=\ds\sum_{\alpha=1}^k\iota_{\bar
X_\alpha}\bar\pi_2^*\omega^\alpha-\d(\bar\pi_2^*H)=
\bar\pi_2^*(\sum_{\alpha=1}^k\iota_{(\pi_2)_*\bar X_\alpha}\omega^\alpha-\d H)=0 \ ,
$$
since ${\bf X}=(X_1,\dots,X_k)\in\vf^k_H(T^1_k)^*Q)$, and therefore
${\bar  X}=(\bar X_1,\dots,\bar X_k)\in\vf^k_{\mathcal
H}(\rk\times(T^1_k)^*Q)$.

Furthermore, if $\psi\colon\rk\to (T^1_k)^*Q$ is an integral section
of ${\bf X}$, then  $\bar\psi\colon\rk\to\rk\times(T^1_k)^*Q$ such
that $\bar\psi=(Id_{\rk},\psi)$ (see Theorem \ref{onetoone}) is an
integral section of ${\bf \bar X}$.

Now, if $\bar\psi$ is an integral section of ${\bf \bar X}$, the
equations (\ref{HDW field eq}) hold for
$\bar\psi(x)=(x,\bar\psi^i(x),\bar\psi^\alpha_i(x))$ and, since $(\bar
X_\alpha)^i=\bar\pi_2^*(X_\alpha)^i$ and $(\bar
X_\alpha)_i^\beta=\bar\pi_2^*(X_\alpha)_i^\beta$, this is equivalent to say that
the equations (\ref{HDW_eq}) hold for
$\psi(x)=(\psi^i(x),\psi^\alpha_i(x))$; in other words, $\psi$ is an integral
section of ${\bf X}$. \end{proof}

\begin{remark}
{\rm The converse statement is not true. In fact, the
$k$-vector fields that are solution of the geometric field equations
(\ref{geonahaut}) are not completely determined, and then there are $k$-vector fields in
$\vf^k_{\mathcal H}(\rk\times(T^1_k)^*Q)$ that are not
$\bar\pi_2$-projectable (in fact, it suffices to take their
undetermined component functions to be not $\bar\pi_2$-projectable).
However, we have the following partial result:}
\end{remark}

\begin{prop}
Let $((T^1_k)^*Q,H)$ be an admissible $k$-symplectic Hamiltonian
system, and we consider $(\rk\times(T^1_k)^*Q,$  ${\mathcal H})$ its associated
autonomous $k$-cosymplectic Hamiltonian system. Then, every
integrable $k$-vector field ${\bf \bar X}\in\vf^k_{\mathcal
H}(\rk\times(T^1_k)^*Q)$ is associated with an integrable $k$-vector field
${\bf X}\in\vf^k_H((T^1_k)^*Q)$.
\end{prop}
\begin{proof}
 If ${\bf \bar X}\in\vf^k_{\mathcal H}(\rk\times(T^1_k)^*Q)$ is
an integrable $k$-vector field, denote by $\bar{\mathcal S}$ the set of
its integral sections (i.e., the solutions of the HDW-equation
(\ref{HDW field eq})). Let ${\mathcal S}$ be the set of maps
$\psi\colon\rk\to(T^1_k)^*Q$ associated with these sections by
Theorem \ref{onetoone}. But, since that
$((T^1_k)^*Q,\omega^\alpha,H)$ is an admissible $k$-symplectic
Hamiltonian system, we have that they are admissible solutions of the
HDW-equation (\ref{HDW_eq}),  and then they are integral sections of some
integrable $k$-vector field ${\bf  X}\in\vf^k((T^1_k)^*Q)$. Then,
by Proposition \ref{int}, ${\bf X}$ satisfies the field equation (\ref{ecHksym})
on the image of $\psi$, for every $\psi\in{\mathcal S}$,
and thus ${\bf X}\in\vf^k_H((T^1_k)^*Q)$
since every point of $(T^1_k)^*Q$ is on the image of one of these sections. \qed
 \end{proof}

We now consider the Lagrangian case. In this situation we define

\index{Autonomous Lagrangian system}
\begin{definition}
A $k$-cosymplectic (or  $k$-precosymplectic) Lagrangian system
is said to be \emph{autonomous} if \ $\ds \derpar{\mathcal L}{x^\alpha}=0$ or,
what is equivalent, $\ds \derpar{{\mathcal E}_{\Lag}}{x^\alpha}=0$.
 \label{autonomousL}
\end{definition}

Now, all the results obtained for the Hamiltonian case can be stated
and proved in the same way for Lagrangian approach, considering the systems
 $(\rk\times T^1_kQ,\Lag)$
and $(T^1_kQ,L)$ instead of $(\rk\times (T^1_k)^*Q, {\mathcal H})$ and
$((T^1_k)^*Q,H)$.

\newpage
The following table summarizes  the above discussion (we also include the particular case of Classical Mechanics).

\begin{center}
\begin{table}[h]
\scalebox{0.88}{\begin{tabular}{|c|c|c|}
\hline
 & \begin{tabular}{c}
 {\sc Hamiltonian formalism}\\
 Geometric Hamiltonian equations
 \end{tabular}
  & \begin{tabular}{c}
 {\sc Lagrangian formalism}\\
 Geometric Lagrangian equations
 \end{tabular}\\\hline
 \begin{tabular}{c}
  $k$-cosymplectic \\ formalism\end{tabular} & $\begin{array}{c}\\
   dt^A(X_B)=\delta^A_B\\  \noalign{\medskip}
 \ds\sum_{A=1}^k \iota_{X_A}\omega^A
= dH-\displaystyle\sum_{A=1}^k \ds\frac{\partial H}{\partial
t^A}dt^A\\  \noalign{\bigskip}
  (X_1,\dots,X_k)\,\, \mbox{k-vector field on} \,\, \rk\times (T^1_k)^{\;*}Q\\
  \quad
 \end{array}$ & $\begin{array}{c}\\
 dt^A(Y_B)=\delta^A_B \\  \noalign{\medskip}
 \ds\sum_{A=1}^k \, i_{Y_A} \omega_L^A =\,
dE_L + \,\ds\sum_{A=1}^k\ds\frac{\partial L}{\partial t^A}dt^A  \\  \noalign{\bigskip}
(Y_1,\dots,Y_k) \,\, \mbox{k-vector field on} \,\, \rk\times T^1_kQ\\\quad
 \end{array}$\\\hline
 \begin{tabular}{c}
  $k$-symplectic \\ formalism\end{tabular} & $\begin{array}{c}\\
   \ds\sum_{A=1}^k \iota_{X_A}\omega^A
= dH\\  \noalign{\bigskip}
  (X_1,\dots,X_k)\,\, \mbox{k-vector field on} \,\, (T^1_k)^{\;*}Q\\
  \quad
 \end{array}$ & $\begin{array}{c}\\
   \ds\sum_{A=1}^k \, i_{Y_A} \omega_L^A =\,
dE_L  \\  \noalign{\bigskip}
(Y_1,\dots,Y_k) \,\, \mbox{k-vector field on} \,\,   T^1_kQ\\\quad
 \end{array}$\\\hline
 \begin{tabular}{c}
Cosymplectic \\ formalism\\ $k=1$\\
(Non-autonomous \\Mechanics)\end{tabular} & $\begin{array}{c}\\
   dt(X)=1\\  \noalign{\medskip}
 \iota_{X}\omega
= dH- \ds\frac{\partial H}{\partial
t}dt \\  \noalign{\bigskip}
  X\,\, \mbox{vector field on} \,\, \r\times T^*Q \\\noalign{\bigskip} \mbox{or equivalently}\\\noalign{\bigskip} dt(X)=1\, , \, i_{X}\Omega =0\\\mbox{where} \;\Omega=\omega+dH\wedge dt
  \\
  \quad
 \end{array}$ & $\begin{array}{c}\\
dt(Y)=1 \\  \noalign{\medskip}
 i_{Y} \omega_L =\,
dE_L + \,\ds\frac{\partial L}{\partial t}dt  \\  \noalign{\bigskip}
Y \,\, \mbox{vector field on} \,\, \r\times TQ\\\noalign{\bigskip} \mbox{or equivalently}\\\noalign{\bigskip} dt(Y)=1 \, , \, i_{Y}\Omega_L =0\\\mbox{where} \;\Omega_L=\omega_L+dE_L\wedge dt\\\quad
 \end{array}$\\\hline
  \begin{tabular}{c}
Symplectic \\ formalism\\ $k=1$\\
(Autonomous \\Mechanics)\end{tabular} & $\begin{array}{c}\\
    \iota_{X}\omega
= dH \\  \noalign{\bigskip}
  X\,\, \mbox{vector field on} \,\,   T^*Q
  \\
  \quad
 \end{array}$ & $\begin{array}{c}\\
   i_{Y} \omega_L =\,
dE_L   \\  \noalign{\bigskip}
Y \,\, \mbox{vector field on} \,\,   TQ\\\\\quad
 \end{array}$\\\hline
\end{tabular}}
\caption{$k$-cosymplectic and $k$-symplectic formalisms}
\end{table}
\end{center}

\part{Relationship between $k$-symplectic and $k$-cosymplectic  approaches and the multisymplectic formalism}\label{relation_k-cosym_multi}

\chapter{Multisymplectic  formalism}\label{multy}

In this book, we have developed a framework for describing Classical Field Theories using $k$-symplectic and $k$-cosymplectic manifolds. An alternative geometric framework is  the multisymplectic formalism \cite{CCI-91,Gymmsy, Gymmsy2,MS-99,RR-2009},
first introduced in \cite{Ki-73,KS-75,KT-79,Snia},
 which is based on the use of multisymplectic manifolds.
In particular, jet bundles are the appropriate domain to develop the Lagrangian formalism \cite{Saunders-89},
and different kinds of multimomentum bundles are used for
developing the Hamiltonian description
\cite{EMR-00,HK-04,LMM-2009}.
In these models, the field equations can be  also obtained in terms of
multivector fields \cite{EMR-1998,EMR-1999,PR-2002}.

Multisymplectic models allow us to describe a higher variety of field theories
than the $k$-cosymplectic or $k$-symplectic models,
since for the latter the configuration bundle of the theory
must be a trivial bundle; which is not the case for the multisymplectic formalism.
The mail goal of this chapter is to show the equivalence between
 the multisymplectic and $k$-cosymplectic descriptions,
 when theories with trivial configuration bundles are considered,
for both the Hamiltonian and Lagrangian formalisms (for more details see \cite{RRSV-2011}).
In this way we complete the results obtained in \cite{LMcNRS-2002,{LMORS-1998}},
where an initial analysis about the relation between multisymplectic, $k$-cosymplectic and
$k$-symplectic structures was carried out.

\section{First order jet bundles.}

For a more detail discussion of the contents of this section, we refer to
\cite{EMR-1996,Saunders-89}.

Let $\pi:E\rightarrow M$ be a bundle where $E$ is an $(m+n)$-dimensional
manifold, which is fibered over an $m$-dimensional manifold $M$.

If $(y^i)$ are coordinates on $M$, where $1\leq i\leq m$, then we denote the fibered coordinates on $E$ by $(y^i,u^\alpha)$ where $1\leq \alpha \leq n$.

\index{local section}
\begin{definition}
If $(E,\pi,M)$ is a fiber bundle then a \emph{local section} of $\pi$ is  a map $\phi:W\subset M \to E$, where $W$ is an open set of $M$, satisfying the condition   $\pi\circ \phi=id_W$.  If $p\in M$ then the set of all sections of $\pi$ whose domains contain $p$ will be denoted by $\Gamma_p(\pi)$.
\end{definition}

 \begin{definition} Define the local sections $\phi,\psi\in \Gamma_p(\pi)$ to be equivalent if $\phi(p)=\psi(p)$ and if, in some fibered coordinate system $(y^i,u^\alpha)$ around $\phi(p)$
 $$
 \derpar{\phi^\alpha}{y^i}\Big\vert_{p}=\derpar{\psi^\alpha}{y^i}\Big\vert_{p}
 $$
 for $1\leq i \leq n$, $1\leq \alpha\leq n$. The equivalence class containing $\phi$ is called the $1$-jet of $\phi$ at $p$ and   is denoted $j^1_p\phi$.
\end{definition}

Let us observe that $j^1_p\phi=j^1_p\psi$ if, and only if, $\phi_*(p)=\psi_*(p)$.

The set of all $1$-jets of local sections of $\pi$ has a natural structure  as a differentiable manifold. The atlas which describe this structure  is constructed from an atlas of fibered coordinate charts on the total space $E$, in much the same way that the induced atlas on the tangent bundle of $TM$  (or on the $k$-tangent bundle $T^1_kM$)  is constructed from an atlas on $M$.

 The first jet manifold of  $\pi$ is the set
 $$
 \{ j^1_p\phi\, \vert \, p\in M\, , \, \phi\in \Gamma_p(\pi) \}
 $$
 and is denoted $J^1\pi$. The functions $\pi_1$ and $\pi_{1,0}$ called the source and target projections respectively, are defined by
 $$
 \begin{array}{ccccl}
 \pi_1 & : & J^1\pi & \longto & M \\ \noalign{\medskip}
       &  &      j^1_p\phi     & \to & p
 \end{array}
 $$and
 $$\begin{array}{ccccl}
 \pi_{1,0} & : & J^1\pi & \longto & E \\ \noalign{\medskip}
       &  &      j^1_p\phi     & \to & \phi(p)

  \end{array}$$

Let $(U,y^i,u^\alpha)$ be an adapted coordinate system on $E$. The induced coordinate system $(U^1,y^i,u^\alpha, u^\alpha_i)$ on $J^1\pi$ is defined on
 $$U^1=\{j^1_p\phi\, : \, \phi(p)\in U\}
  $$ where
 \begin{equation}\label{coorj1pi}
  y^i(j^1_p\phi)=y^i(p)  \; , \quad  u^\alpha(j^1_p\phi)=u^\alpha(\phi(p))    \; , \quad  u^\alpha_i(j^1_p\phi)=\derpar{u^\alpha \circ \phi}{y^i}\Big\vert_{p}
  \end{equation}
  and are known as {\it derivative coordinates}.

$J^1\pi $ is a manifold of dimension $m+n(1+m)$.
The canonical projections $\pi_{1,0}$ and $\pi_{1}$ are smooth surjective submersions.

\begin{remark}
{\rm
If we consider Remarks \ref{difeo jipiQ} and \ref{difeo jipi} one obtain that the manifolds $\rktkqh$ and $\rktkq$ are two examples of jet bundles.
\rqed}
\end{remark}

\section{Multisymplectic Hamiltonian formalism}

\subsection{Multimomentum bundles}
\protect\label{mmb}

A more completed description of the multisymplectic manifolds can be found in \cite{{CIL-1996},CIL-1999,CCI-91,EMR-00,{Sarda2},Go1,Go2,Go3}.

\index{multisymplectic manifold}
\begin{definition}
The couple $({\mathcal M},\Omega)$, with $\Omega\in\Omega^{k+1}({\mathcal
M})$ ($2\leq k+1\leq\dim\,{\mathcal M}$), is a \emph{multisymplectic
manifold}
 if $\Omega$ is closed and $1$-nondegenerate;
that is, for every $p\in{\mathcal M}$, and $X_p\in T_p{\mathcal M}$, we have
that $\iota_{X_p}\Omega_p=0$ if, and only if, $X_p=0$.
\end{definition}

\index{Multimomentum bundle}
\index{Extended multimomentum bundle}
A very important example of multisymplectic manifold is the {\sl
multicotangent bundle} $\Lambda^k Q$ of a manifold $Q$, which
is the bundle of $k$-forms in $Q$, and is endowed with a canonical
multisymplectic  $(k+1)$-form.
 Other examples of multisymplectic manifolds which are relevant in
field theory are the so-called {\sl multimomentum bundles}: let
$\pi\colon E\to M$ be a fiber bundle, ($\dim\, M=k$, $\dim\,
E=n+k$), where $M$ is an oriented manifold with volume form
 $\omega\in\Omega^k(M)$, and denote by $(x^\alpha,q^i)$
 the natural coordinates in $E$
 adapted to the bundle, such that
 $\omega=\d x^1\wedge\ldots\wedge\d x^k\equiv{\rm d}^kx$.
We denote by $\Lambda^k _2 E$ the
bundle of $k$-forms on $E$ vanishing by the action of two
$\pi$-vertical vector fields. This is called the {\sl extended
multimomentum bundle}, and its canonical submersions are denoted by
$$
\kappa\colon{\Lambda^k _2 E}\to E \quad ; \quad
 \bar{\kappa}=\pi\circ\kappa\colon{\Lambda^k _2 E}\to M
$$
We can introduce natural coordinates in ${\Lambda^k _2 E}$ adapted to the
bundle $\pi\colon E\to M$, which are denoted by $(x^\alpha,q^i,p^\alpha_i,p)$,
 such that $\omega=d^kx$. Then, denoting
$\ds\d^{k-1}x_\alpha=\iota_{\frac{\partial }{\partial x^\alpha}}{\rm
d}^kx$,  the elements of ${\Lambda^k _2 E}$ can be written as
   $p^\alpha_i \, \d q^i \wedge d^{k-1}x_\alpha \, + \,p\,\d^kx$.

${\Lambda^k _2 E}$ is a subbundle of $\Lambda^kE$, and hence ${\Lambda^k _2 E}$ is also endowed with canonical forms. First we have the
``tautological form'' $\Theta\in\Omega^k({\Lambda^k _2 E})$, which is
defined as follows: let $\nu_x\in\Lambda^k _2 E $, with $x\in
E$ then, for every
$X_1,\ldots,X_k\in T_{\nu_x}({\Lambda^k _2 E})$, we have
$$
\Theta \nu_x(X_1,\ldots,X_k):= \nu
(x)(T_{\nu_x}\kappa(X_1),\ldots ,T_{\nu_x}\kappa(X_k)).
$$
Thus we define the multisymplectic form
$$
\Omega:=-{\rm d}\Theta\in\Omega^{k+1}({\Lambda^k _2 E})
$$
and  the local expressions of the above forms are
\begin{equation}
 \Theta=p^\alpha_i{\rm d} q^i\wedge{\rm d}^{k-1}x_\alpha+p\,d^kx
 \ , \
 \Omega=
-{\rm d} p^\alpha_i\wedge{\rm d} q^i\wedge{\rm d}^{k-1}x_\alpha-{\rm d}p\wedge d^kx
 \label{coormult}
\end{equation}
\index{Restricted multimomentum bundle}
Consider $\pi^*\Lambda^kT^*M$, which is another bundle over $E$,
whose sections are the $\pi$-semibasic $k$-forms on $E$, and denote
by $J^1\pi^*$ the quotient $\ds\frac{\Lambda^k _2 E}{\pi^*\Lambda^kT^*M}$.
$J^1\pi^*$ is usually called the {\sl restricted multimomentum
bundle} associated with the bundle $\pi\colon E\to M$. Natural
coordinates in $J^1\pi^*$ (adapted to the bundle $\pi\colon E\to M$)
are denoted by $(x^\alpha,q^i,p^\alpha_i)$.
 We have the natural submersions specified in the following diagram
$$
\begin{array}{ccc}
{\Lambda^k _2 E} &
\begin{picture}(135,20)(0,0)
\put(65,8){\mbox{$\mu$}} \put(0,3){\vector(1,0){135}}
\end{picture}
& J^1\pi^*
\\ &
\begin{picture}(135,100)(0,0)
\put(34,84){\mbox{$\kappa$}} \put(93,82){\mbox{$\sigma$}}
\put(7,55){\mbox{$ \bar{\kappa}$}}
 \put(115,55){\mbox{$ \bar{\sigma}$}}
 \put(58,30){\mbox{$\pi$}}
\put(65,55){\mbox{$E$}}
 \put(65,0){\mbox{$M$}}
\put(0,102){\vector(3,-2){55}}
 \put(135,102){\vector(-3,-2){55}}
\put(0,98){\vector(2,-3){55}}
 \put(135,98){\vector(-2,-3){55}}
\put(70,48){\vector(0,-1){35}}
\end{picture} &
\end{array}
$$

\subsection{Hamiltonian systems}
\protect\label{hsjpistar}

The Hamiltonian formalism in $J^1\pi^*$
presented here is based on the construction made in \cite{CCI-91}
(see also \cite{ELMR-2005, ELMMR-2004,EMR-00,RR-2009}).

\index{Hamiltonian section}
\index{Hamiltonian Cartan forms}
\begin{definition}
 A section $h\colon J^1\pi^*\to{\Lambda^k _2 E}$ of the projection
 $\mu$ is called a \emph{Hamiltonian section}.
 The differentiable forms
 $\Theta_{h}:=h^*\Theta$ and $\Omega_{h}:=-{\rm d}\Theta_{h}=h^*\Omega$
 are called the \emph{Hamilton-Cartan $k$ and $(k+1)$ forms} of $J^1\pi^*$
 associated with the Hamiltonian section $h$.
$(J^1\pi^*,h)$ is said to be a {\sl Hamiltonian system} in $J^1\pi^*$.
\end{definition}

 In natural coordinates we have that
 $$h(x^\alpha,q^i,p^\alpha_i)= (x^\alpha,q^i,p^\alpha_i,p=-{\mathcal H}(x^\alpha,q^i,p^\alpha_i)),$$
and ${\mathcal H}\in C^\infty  (U)$, $U\subset J^1\pi^*$,
is a {\sl local Hamiltonian function}. Then we have
\begin{equation}\label{multiformh}
 \Theta_h = p^\alpha_i{\rm d} q^i\wedge{\rm d}^{k-1}x_\alpha -{\mathcal H}{\rm
 d}^kx
 \ , \
 \Omega_h = -{\rm d} p^\alpha_i\wedge{\rm d} q^i\wedge{\rm d}^{k-1}x_\alpha +
 {\rm d}{\mathcal H}\wedge d^kx \ .
\end{equation}

The Hamilton-De Donder-Weyl equations can also be derived from the corresponding  {\it Hamilton--Jacobi variational principle}.
In fact:

 \begin{definition}
 Let $(J^1\pi^*,h)$ be a Hamiltonian system.
 Let $\Gamma(M,J^1\pi^*)$ be
 the set of sections of $ \sigma$. Consider the map
  $$\begin{array}{cccl}
 {\bf H} \colon & \Gamma(M,J^1\pi^*) & \longrightarrow & \r, \\ \noalign{\medskip}
                 & \psi               &  \mapsto  & \int_M\psi^*\Theta_{h},
  \end{array}$$
 where the convergence of the integral is assumed.
 The {\rm variational problem} for this Hamiltonian system
 is the search for the critical (or
 stationary) sections of the functional ${\bf H}$,
 with respect to the variations of $\psi$ given
 by $\psi_t =\sigma_t\circ\psi$, where $\{\sigma_t\}$ is the
 local one-parameter group of any compact-supported
 $Z\in\vf^{{\rm V}( \sigma)}(J^1\pi^*)$
 (the module of $ \sigma$-vertical vector fields in $J^1\pi^*$), that is:
 \[
  \frac{\d}{\d t}\Big\vert_{t=0}\int_M\psi_t^*\Theta_{h} = 0   .
  \]
 \label{hjvp}
 \end{definition}

 The field equations for these multisymplectic Hamiltonian systems can be stated as follows

 \begin{theorem}\label{equics}
 The following assertions on a
 section $\psi\in\Gamma(M,J^1\pi^*)$ are equivalent:
 \begin{enumerate}\itemsep=0pt
 \item[$1.$]
 $\psi$ is a critical section for the variational problem posed by
the Hamilton--Jacobi principle.
 \item[$2.$]
 $\psi^*\iota_X\Omega_{h}= 0$, $\forall\, X\in\vf (J^1\pi^*)$.
\end{enumerate}
 \end{theorem}

If $(U;x^\alpha,q^i,p^\alpha_i )$ is a natural system of
 coordinates in $J^1\pi^*$, then  $\psi$
 satisfies the  Hamilton-De Donder-Weyl equations in $U$
 \begin{equation}
 \derpar{ \psi^i}{x^\alpha}\Big\vert_{x}=
 \derpar{\mathcal{H}}{p^\alpha_i}\Big\vert_{\psi(x)}
,\qquad
 \derpar{\psi^\alpha_i}{x^\alpha}=
 - \derpar{\mathcal{H}}{q^i}\Big\vert_{\psi(x)}\, ,
 \label{HDWeqs}
 \end{equation}
where $\psi(x)=(x,\psi^i(x),\psi^\alpha_i(x))$.

%The existence of $k$-vector fields that are solutions of (\ref{hameq1})
% is assured, and in a local system of coordinates they depend on
%$n(k^2-1)$ arbitrary functions, but the number of arbitrary functions
%for integrable solutions is, in general, less than $n(k^2-1)$.
%
%At this point, similar results as those stated in Theorems \ref{intsec1} and Proposition \ref{int1}
% for the $k$-cosymplectic Hamiltonian formalism hold for multisymplectic Hamiltonian systems.

 %Thus,  if $ {\bf X}\in\vf^k_h(J^1\pi^*)$ is integrable
%and $ \psi\colon M\to J^1\pi^*$ is an integral section
%of $ {\bf X}$ then $ \psi$ is a solution
%to the HDW equation (\ref{hem}).
%
%Conversely, If  a section $ \psi\colon M\to J^1\pi^*$ of $ \sigma$
%is a solution of the HDW equation (\ref{hem})
% and $ \psi$ is an integral section of
% an integrable $k$-vector field $ {\bf X}\in\vf^k(J^1\pi^*)$, then
% ${\bf X}=(  X_1,\dots,  X_k)$ is a solution of the equations
%(\ref{hameq1}) at the points of the image of $\psi$.

\subsection{Relation with the $k$-cosymplectic Hamiltonian formalism}
\protect\label{rkcf}

In order to compare the multisymplectic and the $k$-cosymplectic formalisms
for field theories,  we consider the case
when $\pi\colon E\to M$ is the trivial bundle
$\rk\times Q \rightarrow \rk$.
Then we can establish some relations between the canonical
multisymplectic form on ${\Lambda^k _2 E}\equiv\Lambda^k _2 (\rk\times Q)$
and the canonical $k$-cosymplectic structure on $\rk\times(T^1_k)^*Q$.

First recall that on $M=\rk$ we have the canonical volume form
 $\omega=\d x^1\wedge\ldots\wedge\d x^k\equiv d^kx$.
Then:

\begin{prop}If $\pi\colon E\to M$ is the trivial bundle $\rkq\to\rk$, we have the following diffeomorphisms:
\ben
\item
${\Lambda^k _2 E}\equiv\Lambda^k _2 (\rk\times Q)$
 is diffeomorphic to $\rk\times\r\times (T^1_k)^*Q$.
\item
$J^1\pi^*=J^{1*}(\rkq)$ is diffeomorphic to $\rk\times (T^1_k)^*Q$.
\een
\end{prop}
\begin{proof}

  \begin{enumerate}
\item
Consider the canonical embedding $\iota_x\colon Q\hookrightarrow
\rk\times Q$ given by $i_x(q)=(x,q)$, and the canonical submersion
$\pi_Q\colon\rk\times Q\rightarrow Q$.
 We can define the map
$$
\begin{array}{cccc}
 \Psi\colon& \Lambda^k _2 (\rk\times Q) & \longrightarrow & \rk\times\r\times (T^1_k)^*Q \\
 & \nu_{(x,q)} & \mapsto & (x,p,\nu^1_q, \dots ,\nu^k_q)
\end{array}
$$
where
{$$
 \begin{array}{rl}
p=&\nu_{(x,q)}\left(\derpar{}{x^1}\Big\vert_{(x,q)},\dots,\derpar{}{x^k}
\Big\vert_{(x,q)}\right)
\\
\nu^\alpha_q(X) =&
\nu_{(x,q)}\left(\derpar{}{x^1}\Big\vert_{(x,q)},\dots,
\derpar{}{x^{\alpha-1}}\Big\vert_{(x,q)},(\iota_x)_*X,
 \derpar{}{x^{\alpha+1}}\Big\vert_{(x,q)},\ldots,
\derpar{}{x^k}\Big\vert_{(x,q)}\right)
 \ ,  \\  X \in&  \vf(Q)
      \end{array}$$}
\noindent(note that $x^\alpha$ and $p$ are now global coordinates in the
corresponding fibres and the global coordinate $p$
can be identified also with the natural projection
 $p\colon\rk\times\r\times (T^1_k)^*Q\to\r$). The inverse of $\Psi$ is given by
$$
\nu_{(x,q)}= p \, \d^k
x\vert_{(x,q)}+(\pi_Q)_{(x,q)}^*\nu^\alpha_q\wedge\d^{k-1}x^\nu\alpha_{(x,q)}
\ .
$$
Thus, $\Psi$ is a diffeomorphism. Locally $\Psi$ is written
as the identity.
\item
It is a straightforward consequence of $(1)$, because
$$
J^1\pi^*=\Lambda^k _2 E/\pi^*\Lambda^kT^*M\simeq
\rk\times\r\times (T^1_k)^*Q/\r\simeq\rk\times (T^1_k)^*Q\,.
$$
\end{enumerate}
\end{proof}

%
%\bigskip
%
%{\bf Relationship between the multisymplectic form and the
%$k$-cosymplectic structure in
%$J^1\pi^*\simeq\rk\times  (T^1_k)^*Q$}
%
%\bigskip

It is important to point out that since the bundle
$$\mu\colon{\Lambda^k _2 E}\simeq\rk\times\r\times (T^1_k)^*Q\to
 J^1\pi^*\simeq\rk\times (T^1_k)^*Q$$
is trivial, then the Hamiltonian sections can be taken to be global
sections of the projection $\mu$ by giving a global Hamiltonian
function ${\rm H}\in C^\infty(\rk\times (T^1_k)^*Q)$.

Then we can
  relate the non-canonical multisymplectic form $\Omega_h$ with the
$k$-cosymplectic structure in $\rk\times (T^1_k)^*Q$ as follows:
\begin{itemize}

\item Starting from the forms $\Theta_h$ and $\Omega_h$ in
$\rk\times(T^1_k)^*Q$, we can define the forms $\Theta^\alpha$ and
$\Omega^\alpha$ on $\rk\times(T^1_k)^*Q$ as follows:
for $  X,  Y\in\vf(\rk\times(T^1_k)^*Q)$,
 \bea
 \Theta^\alpha(  X) &=& -
\left(\iota_{\ds\frac{\partial}{\partial x^k}}\ldots
\iota_{\ds\frac{\partial}{\partial x^1}}(\Theta_h\wedge\d x^\alpha)\right)(  X)
\nonumber \\
\Omega^\alpha(  X,  Y)&=&-\d\Theta^\alpha(  X,  Y)
= (-1)^{k+1}\left(\iota_{\ds\frac{\partial}{\partial x^k}}\ldots
\iota_{\ds\frac{\partial}{\partial x^1}}(\Omega_h\wedge\d x^\alpha)\right)(  X,  Y) \ ,
 \label{relatomega5}
\eea
and the $1$-forms $\eta^\alpha=\d x^\alpha$ are canonically defined.

\item Conversely, starting from the canonical $k$-cosymplectic structure on $\rk\times(T^1_k)^*Q$,
and from ${\mathcal H}$, we can construct
\begin{equation}
\begin{array}{c}
\Theta_h= -{\mathcal H} \d^k t+\Theta^\alpha\wedge\d^{k-1}x_\alpha \; , \\\noalign{\medskip}
\Omega=-\d\Theta=\d{\mathcal H}\wedge\d^k t+\Omega^\alpha\wedge\d^{k-1}x_\alpha \ .\end{array}
 \label{relatomega6}
\end{equation}
\end{itemize}

Let  $\vf^k_h(J^1\pi^*)$ be the set of
$k$-vector fields ${\bf   X}=(  X_1,\ldots,  X_k)$ in $J^1\pi^*$
 which are solution of the equations
 \begin{equation}
 \iota_{\bf  X}\Omega_h=\iota_{X_1}\ldots\iota_{X_k}\Omega_h=0
 \; , \quad
\iota_{\bf  X}\omega=
 \iota_{X_1}\ldots\iota_{X_k}\omega\not=0 \ ,
 \label{hameq1}
\end{equation}
(we denote by $\omega=\d^kx$ the volume form in $M$ as well as its pull-backs to all the manifolds).

In a system of natural coordinates, the components of  ${\bf  X}$
are given by $$
X_\alpha=(X_\alpha)_\beta\derpar{}{x^\alpha}+ (X_\alpha)^i\derpar{}{q^i} +(X_\alpha)^\beta \derpar{}{v^i_\beta}\,.
$$

Then,
in order to assure the so-called  ``transversal condition'' $\iota_{\bf X}\omega\not=0$,
we can take $(  X_\alpha)^\beta=\delta^\beta_\alpha$, which leads to $\iota_{\bf X}\omega=1$,
and hence the other equation (\ref{hameq1}) give become
 \begin{equation}
\derpar{{\mathcal H}}{q^i}=-\ds\sum_{\alpha=1}^k ( X_\alpha)^\alpha_i
 \quad , \quad
 \derpar{{\mathcal H}}{p_i^\alpha}= ( X_\alpha)^i \  .
 \label{eqsG2}
 \end{equation}

 Let us observe that these equations coincide with equations (\ref{k-cosymp condvf}). Thus we obtain

 \begin{theorem}
A $k$-vector field ${\bf   X}=(  X_1,\ldots,  X_k)$ on $J^1\pi^*\simeq\rk\times (T^1_k)^*Q$
   is a solution of the equations
(\ref{hameq1}) if, and only if, it is also a solution of the equations (\ref{geonah});
that is, we have that
$\vf^k_{h}(\rk\times (T^1_k)^*Q)=\vf^k_{\mathcal{H}}(\rk\times (T^1_k)^*Q)$.
\end{theorem}

Let us observe that when $E=\rk\times Q$, then,  if the section
$\psi: \rk=M\to \rk\times (T^1_k)^*Q=J^1\pi^*$ is a integral section  of the $k$-vector field ${\bf X}$, $\psi$ is a solution of the Hamilton-De Donder-Weyl equations (\ref{HDWeqs}), (as a consequence of (\ref{eqsG2})).

\section{Multisymplectic Lagrangian formalism}

\subsection{Multisymplectic Lagrangian systems}
\protect\label{mls}

  A {\sl Lagrangian density} is a
$\pi$-semibasic $k$-form on $J^1\pi$, and hence it can be
expressed as ${\mathbb{L}} =L\, \omega$, where ${L}\in
C^{\infty}(J^1\pi)$ is the {\sl Lagrangian function} associated with
$\mathbb{L}$ and $\omega$, where $\omega$ is a volume form on $M$. Using the canonical structures of $J^1\pi$, we
can define the  Poincar\'e-Cartan $k$-form $\Theta_{\mathbb{L}}$ and  Poincar\'e-Cartan   $(k+1)$-form $ \Omega_{\mathbb{L}}=-d\Theta_{\mathbb{L}}$,
which have the following local expressions \cite{EMR-1996}:
\beann
\Theta_{\mathbb{L}}&=&\derpar{L}{v^i_\alpha}\d q^i\wedge\d^{k-1}x_\alpha -
\left(\derpar{L}{v^i_\alpha}v^i_\alpha -L\right)\d^kx
\\
\Omega_{\mathbb{L}}&=& -d\left( \derpar{L}{v^i_\alpha}  \right)\wedge\d q^i\wedge\d^{k-1}x_\alpha
+d\left(\derpar{L}{v^i_\alpha}v^i_\alpha -L\right)\wedge\d^kx
%\\
%-\frac{\partial^2L}{\partial v^j_\beta\partial v^i_\alpha}
%\d v^j_\beta\wedge\d q^i\wedge\d^{k-1}x_\alpha
%-\frac{\partial^2L}{\partial q^j\partial v^i_\alpha}\d q^j\wedge\d q^i\wedge\d^{k-1}x_\alpha +
% \\  & &
%\frac{\partial^2L}{\partial v^j_\beta\partial v^i_\alpha}v^i_\alpha\d v^j_\beta\wedge\d^kx  +
%\left(\frac{\partial^2L}{\partial q^i\partial v^j_\beta}v^j_\beta
% -\derpar{L}{q^j}+\frac{\partial^2L}{\partial x^\alpha\partial v^i_\alpha}\right)\d q^i\wedge\d^kx
\eeann
$(J^1\pi,\mathbb{L})$ is said to be a {\rm Lagrangian system}.

The Lagrangian system and the Lagrangian function are
{\sl regular} if $\Omega_{\mathbb{L}}$ is a multisymplectic
$(k+1)$-form. The regularity condition is locally equivalent to demand that the matrix
$\ds\left(\frac{\partial^2L}{\partial v_\alpha^i\partial v_\beta^j}\right)$
is regular at every point in $J^1\pi$.

The Lagrangian field equations can be derived from a variational principle.
In fact:

\begin{definition} \label{hvp}
 Let $(J^1\pi,\mathbb{L})$  be a Lagrangian system.
 Let $\Gamma(M,E)$ be the set of sections of $\pi$. Consider the map
$$ \begin{array}{cccl}
 {\bf L} \colon & \Gamma(M,E) & \longrightarrow & \r,
 \\ \noalign{\medskip}
   &  \phi &\mapsto & \int_M(j^1\phi)^*\Theta_\mathbb{L}   ,
  \end{array}$$
where the convergence of the integral is assumed.
 The {\rm variational problem} for this Lagrangian system
 is the search of the critical (or
 stationary) sections of the functional ${\bf L}$,
 with respect to the variations of $\phi$ given
 by $\phi_t =\sigma_t\circ\phi$, where $\{\sigma_t\}$ is a
 local one-parameter group of any compact-supported
 $Z\in\vf^{{\rm V}(\pi)}(E)$
 (the module of $\pi$-vertical vector fields in $E$), that is:
 \[
  \frac{\d}{\d t}\Big\vert_{t=0}\int_M\big(j^1\phi_t\big)^*\Theta_\mathbb{L} = 0   .
  \]
   This is the {\it Hamilton principle} of the Lagrangian formalism.
\end{definition}

\begin{theorem}\label{equics}
 The following assertions on a
 section $\phi\in\Gamma(M,E)$ are equivalent:
 \begin{enumerate}\itemsep=0pt
 \item[$1.$]
 $\phi$ is a critical section for the variational problem posed by
the Hamilton principle.
 \item[$2.$] $
(j^1 \phi)^*\iota_X\Omega_{\mathbb{L}}= 0 \,\, $ for
every $ X\in\vf (J^1\pi)$,
where $j^1\phi:M\to  J^1\pi$ is the section defined by $j^1\phi(x)=j_x^1\phi$.

\end{enumerate}
 \end{theorem}

If $\phi(x^\alpha)=(x^\alpha,\phi^i(x^\alpha))$ is a critical section then
$$j^1\phi (x^\alpha)=\big(x^\alpha ,\phi^i(x^\alpha),\derpar{\phi^i}{x^\alpha}\big)$$
   satisfies the Euler-Lagrange field equations
  \begin{equation}
\derpar{}{x^\alpha}\left(\derpar{L}{v^i_\alpha}\circ j^1\phi\right)- \derpar{L}{q^i}\circ j^1\phi
 = 0   .
 \label{ELeqs}
  \end{equation}

\index{Extended Legendre transformation}
\index{Restricted Legendre transformation}
Finally, $\Theta_{\mathbb{L}}\in\df^1(J^1\pi)$ being $\pi$-semibasic,
we have a natural map
$\widetilde{FL}\colon$ $ J^1\pi\to{\mathcal M}\pi$,
given by
   $$
  \widetilde{FL}({  y})=\Theta_{\mathbb{L}}({  y}) \quad ; \quad   y\in J^1\pi
  $$
which is called the \emph{extended Legendre transformation} associated to the
Lagrangian $L$.
The \emph{restricted Legendre transformation} is
$F{L} =\mu\circ\widetilde{F{L}}\colon J^1\pi\to J^1\pi^*$. Their local expressions are
\bea
\widetilde{F{L}} &\colon& (x^\alpha, q^i,v^i_\alpha)\mapsto
\left( x^\alpha, q^i,\derpar{{L}}{v^i_\alpha},{L}-v^i_\alpha\derpar {{L}}{v^i_\alpha}\right)
 \nonumber \\
F{L} &\colon& (x^\alpha, q^i,v^i_\alpha)\mapsto \left( x^\alpha, q^i,\derpar{{L}}{v^i_\alpha}\right)
\label{legmulti}
\eea
Moreover, we have $\widetilde{F{L}}^*\Theta=\Theta_{\mathbb{L}}$,
and $\widetilde{F{L}}^*\Omega=\Omega_{\mathbb{L}}$.
Observe that the Legendre transformations $F{L}$
defined for the $k$-cosymplectic and the multisymplectic formalisms
are the same, as their local expressions (\ref{locfl1}) and (\ref{legmulti}) show.

\subsection{Relation with the  $k$-cosymplectic Lagrangian formalism}
\protect\label{rmk}
Like in the Hamiltonian case, in order to compare the multisymplectic Lagrangian formalism and the $k$-cosymplectic Lagrangian formalism for field theories, we consider the case when $\pi\colon \mathbb{R}\to M$ is the trivial bundle $\mathbb{R}^k\times Q\to \mathbb{R}^k$.
We can define the energy  function $E_L$ as
$$
E_L=\Theta_{\mathbb{L}}\left(\derpar{}{x^1},\dots,\derpar{}{x^k}\right)
$$
whose local expression is $\ds E_L =v^i_\alpha\frac{\partial
L}{\partial v^i_\alpha}-L$. Then we can write
$$
\begin{array}{lcl}
\Theta_{\mathbb{L}}&=&\derpar{L}{v^i_\alpha}\d q^i\wedge d^{k-1}x_\alpha - E_L\d^kx \, , \\\noalign{\medskip} \Omega_{\mathbb{L}} &=&
-\d\left(\derpar{L}{v^i_\alpha}\right)\wedge\d q^i\wedge d^{k-1}x_\alpha +
dE_L\wedge\d^kx\,.
\end{array}
$$

In this particular case, like in the Hamiltonian case, we can relate
the non-canonical Lagrangian multisymplectic (or
pre-multisymplectic) form $\Omega_{\mathbb{L}}$ with the non-canonical
Lagrangian $k$-cosymplectic (or $k$-precosymplectic) structure in
$\rk\times T^1_kQ$ constructed in Section \ref{kcolag} as follows:
starting from the forms $\Theta_{\mathbb{L}}$ and $\Omega_{\mathbb{L}}$ in
$J^1\pi \simeq\rk\times T^1_kQ$, we can define the forms
$\Theta_{L}^\alpha$ and $\Omega_{L}^\alpha=-\d\Theta_{L}^\alpha$ on $\rk\times T^1_kQ$,
as follows:  for $X,Y\in\vf(\rk\times T^1_kQ)$,
\bea
 \Theta_{L}^\alpha(X) &=& -\left(\iota_{\derpar{}{x^k}}\ldots
\iota_{\derpar{}{x^1}}(\Theta_{\mathbb{L}}\wedge\d x^\alpha)\right)(X)
\nonumber \\
\Omega_{L}^\alpha(X,Y)&=&
(-1)^{k+1}\left(\iota_{\derpar{}{x^k}}\ldots
 \iota_{\derpar{}{x^1}}(\Omega_{\mathbb{L}}\wedge\d x^\alpha)\right)(X,Y)
 \label{relatomegal}
\eea
and the $1$-forms $\eta^\alpha=\d x^\alpha$ are canonically defined.

Conversely, starting from the Lagrangian $k$-cosymplectic (or
$k$-preco\-symplectic) structure on $\rk\times T^1_kQ$, and from
$E_L$, we can construct on $ \rk\times T^1_kQ\simeq
J^1\pi$
\begin{equation}
\begin{array}{c}
\Theta_{\mathbb{L}}= -E_L\d^k t+\Theta^\alpha_{L}\wedge d^{k-1}x_\alpha, \\\noalign{\medskip}
\Omega_{\mathbb{L}}=-d\Theta_{\mathbb{L}}=dE_L\wedge d^k x+\Omega_{L}^\alpha\wedge d^{k-1}x_\alpha
\,.\end{array}
 \label{relatomegal2}
\end{equation}

So we have proved the following theorem, which allows us to relate
the non-canonical Lagrangian multisymplectic (or
pre-multisymplectic) forms  with the non-canonical
Lagrangian $k$-cosymplectic (or $k$-precosymplectic) structure in
$\rk\times T^1_kQ$

\begin{theorem}
The Lagrangian multisymplectic (or pre-multisymplectic) form and the
Lagrangian $2$-forms of the $k$-cosymplectic (or
$k$-precosymplectic) structure on $J^1\pi\equiv\rk\times  T^1_kQ$
are related by (\ref{relatomegal}) and (\ref{relatomegal2}).
\end{theorem}

Let $\vf^k_{\mathbb{L}} (J^1\pi)$ be the set
of $k$-vector fields ${\bf
 \Gamma}=( \Gamma_1,\dots, \Gamma_k)$ in $J^1\pi$,
 that are solutions of the equations
 \begin{equation}
 \iota_{\bf  \Gamma}\Omega_{\mathbb{L}}=0 \quad , \quad
 \iota_{\bf  \Gamma}\omega\not=0  \ .
 \label{lageq1}
 \end{equation}
In a system of natural coordinates the components of  ${\bf  \Gamma}$
are given by
$$
\Gamma_\alpha=(\Gamma_\alpha)^\beta \derpar{}{x^\beta}
+ (\Gamma_\alpha)^i \derpar{}{q^i}+ (\Gamma_\alpha)^i_\beta \derpar{}{v^i_\beta}
$$

 Then,
in order to assure the condition $\iota_{\bf  \Gamma}\omega\not=0$,
we can take $( \Gamma_\alpha)^\beta=\delta^\beta_\alpha$, which leads to $\iota_{\bf \Gamma}\omega=1$,
and  thus ${\bf  \Gamma}$ is a solution of (\ref{lageq1}) if, and only if,
  $( \Gamma_\alpha)^i$ and $( \Gamma _\alpha)^i_\beta$ satisfy the equations (\ref{lform}).
When $L$ is regular, we obtain that $( \Gamma_\alpha)^i=v^i_\alpha$, and the
equations (\ref{xal}) hold.
%, and hence,
%if it is integrable, its integral sections are holonomic and
%they are solutions of the Euler-Lagrange equation for $L$ (\ref{Elem}).

%Conversely, if we have
%holonomic sections which are solutions of the Euler-Lagrange equation (\ref{Elem})
%and are integral sections of integrable $k$-vector fields, then these $k$-vector fields are
% solution of the equations (\ref{lageq1}) at the points of the image of those sections.

%If $L $ is not regular, the existence of solutions of the equations (\ref{lageq1})
% is not assured, in general, except in a submanifold of $J^1\pi$
%(in the most favourable situations). Moreover,
%solutions of (\ref{lageq1}) are not {\sc sopde} necessarily,
%and this condition must be added
%to the solutions of this equation.
Then we can assert the following.
\begin{theorem}
A $k$-vector field ${\bf \Gamma}=(\Gamma_1,\ldots,\Gamma_k)$
 in $J^1\pi\simeq\rk\times T^1_kQ$ is a solution of the equations
(\ref{lageq1}) if, and only if, it is also a solution of the equations (\ref{lageq0});
that is, we have that
$\vf^k_{\mathbb{L}}(\rk\times T^1_kQ)=\vf^k_{L}(\rk\times T^1_kQ)$.
\end{theorem}

Observe also that, when $E=\rk\times Q$ and $J^1\pi\simeq\rk\times T^1_kQ$,
we have that $j^1\phi=\phi^{[1]}$, and hence, if $j^1\phi$ is an integral section of
${\bf \Gamma}=(\Gamma_1,\ldots,\Gamma_k)$,
then $\phi$ is a solution  to the Euler-Lagrange equations.

\section[Correspondences]{Correspondences between the $k$-sym\-plec\-tic, $k$-cosymplectic and multisymplectic formalisms}

%In the following tables we summarizes the correspondences between the $k$-symplectic, $k$-cosymplectic and multisymplectic approaches.
%
%In the Hamiltonian case we have
\begin{table}[!h]
\begin{center}
\scalebox{1}{
    \begin{tabular}{|c|c|c|c|}
        \hline
        \multicolumn{4}{|c|}{\sc Hamiltonian approach}\\
        \hline
         & $k$-symplectic & $k$-cosymplectic & Multisymplectic\\
         \hline
         Phase space & $(T^1_k)^*Q$ & $\rk\times  (T^1_k)^*Q$ &  ${\Lambda^k _2 E}\to J^1\pi^*$ \\
         \hline
         \multirow{2}{*}{Canonical forms} & $\theta^\alpha\in\df^1((T^1_k)^*Q )$ & $\Theta^\alpha\in\df^1(\rk\times(T^1_k)^*Q )$  &  $\Theta\in\df^k({\Lambda^k _2 E})$\\
         \cline{2-4}
          & $\omega^\alpha=-d\theta^\alpha$ & $\Omega^\alpha= -d\Theta^\alpha$ & $\Omega=-d\Theta$\\\hline
         Hamiltonian & $H:(T^1_k)^*Q\to \r$  & ${\mathcal H}:\rk\times (T^1_k)^*Q\to \r$  & $h:J^1\pi^* \to {\Lambda^k _2 E}$ \\
         \hline
         \multirow{2}{*}{\begin{tabular}{c}Geometric\\ equations\end{tabular}} & ${\ds\sum_{\alpha=1}^k\iota_{X_\alpha}\omega^\alpha=dH}$ & $\begin{array}{c}    \ds\sum_{\alpha=1}^k\iota_{{X}_\alpha}  \Omega^\alpha=\d H-\frac{\partial   H}{\partial  x^\alpha}\d x^\alpha   \\ \d x^\alpha( {X}_\beta)=\delta^\alpha_\beta \end{array}$ & $ \begin{array}{c}  \iota_{\bf   X}\Omega_h=0 \\  \iota_{\bf   X}\omega=1 \end{array}$\\
         & $\begin{array}{c}{\bf  X}\in\vf^k((T^1_k)^*Q)    \end{array} $& $\begin{array}{c}{\bf  X}\in\vf^k(\rk\times (T^1_k)^*Q)    \end{array}$  & ${\bf  X}\in\vf^k(J^1\pi^*)$\\
%        \hline
%    \end{tabular}
%    }
%    \caption{$k$-symplectic, $k$-cosymplectic and multisymplectic Hamiltonian formalism.}
%\end{center}
%\end{table}
%
%%In the Lagrangian case the corresponding table is the following:
%\begin{table}[h]
%\begin{center}
%\scalebox{0.9}{
%    \begin{tabular}{|c|c|c|c|}
        \hline
        \multicolumn{4}{|c|}{\sc Lagrangian approach}\\
        \hline
         & $k$-symplectic & $k$-cosymplectic & Multisymplectic\\
         \hline
         Phase space & $T^1_kQ$ & $\rk\times  T^1_kQ$ &  $J^1\pi$ \\
         \hline
         Lagrangian & $L:T^1_kQ\to \r$  & ${\mathcal L}:\rk\times T^1_kQ\to \r$  & $\mathcal{L}:J^1\pi \to \mathcal{R}, \mathbb{L}=\mathcal{L}\omega$ \\
         \hline
         \multirow{2}{*}{Cartan forms} & $\theta^\alpha_L\in\df^1(T^1_kQ )$ & $\Theta^\alpha_{\mathcal{L}}\in\df^1(\rk\times T^1_kQ )$  &  $\Theta_{\mathbb{L}}\in\df^k(J^1\pi)$\\
         \cline{2-4}
          & $\omega^\alpha_L=-d\theta^\alpha_L$ & $\Omega^\alpha_\mathcal{L}= -d\Theta^\alpha_{\mathcal{L}}$ & $\Omega_{\mathcal{L}}=-d\Theta_{\mathcal{L}}$\\\hline
         \multirow{2}{*}{\begin{tabular}{c}Geometric\\ equations\end{tabular}} & ${\ds\sum_{\alpha=1}^k\iota_{\Gamma_\alpha}\omega^\alpha_L=dE_L}$ & $\begin{array}{c}    \ds\sum_{\alpha=1}^k\iota_{\Gamma_\alpha}  \Omega^\alpha_{\mathbb{L}}=\d \mathcal{E}_L-\frac{\partial   \mathcal{L}}{\partial  x^\alpha}\d x^\alpha   \\ \d x^\alpha( \Gamma_\beta)=\delta^\alpha_\beta \end{array}$ & $ \begin{array}{c}  \iota_\Gamma\Omega_{\mathbb{L}}=0 \\  \iota_\Gamma\omega=1 \end{array}$\\
         & $\begin{array}{c}\Gamma\in\vf^k(T^1_kQ)    \end{array} $& $\begin{array}{c}\Gamma\in\vf^k(\rk\times T^1_kQ)    \end{array}$  & $\Gamma\in\vf^k(J^1\pi)$\\
        \hline
    \end{tabular}
    }
    \caption{$k$-symplectic, $k$-cosymplectic and multisymplectic formalism.}
\end{center}
\end{table}

%\newpage
%\mbox{}
%\thispagestyle{empty} % para que no se numere esta p\'{a}gina

%\input{./chapters/Para_multy.tex}

            %%%%%%%%%%%%%%%%%%%%%%%%%%%%%%%%%%%%%%%%%%%%%%%%%%%%%%%%%%%%%%
            %%%             APENDICES                               %%%%%
            %%%%%%%%%%%%%%%%%%%%%%%%%%%%%%%%%%%%%%%%%%%%%%%%%%%%%%%%%%%%%%
            \appendix

            \cleardoublepage
            \addappheadtotoc
            \appendixpage

\chapter{Symplectic  manifolds}\label{symma}

In chapter \ref{chapter: Mechanics} we have presented a review of the Hamiltonian Mechanics on the cotangent bundle, using the canonical symplectic form on $T^*Q$ and also the time-dependent counterpart. This approach can be extended, in a similar way, to the case of an arbitrary symplectic manifolds and cosymplectic manifolds, respectively.

Here  we recall the formal definition of symplectic and cosymplectic manifolds.

The canonical model of symplectic structure is the cotangent bundle $T^*Q$ with its canonical symplectic form.
\index{Symplectic}
\index{Symplectic!Manifold}
\index{Symplectic!Structure}
\index{Presymplectic structure}
\index{Almost symplectic structure}
\begin{definition}
    Let $\omega$ be an arbitrary $2$-form on a manifold $M$. Then
    \begin{enumerate}
        \item $\omega$ is called a \emph{presymplectic structure on $M$} if $\omega$ is a closed $2$-form, that is, $d\omega=0$.
        \item $\omega$ is called an \emph{almost symplectic structure on $M$} if it is non-degenerate.
        \item $\omega$ is called a \emph{symplectic structure} if it is a closed non-degenerated $2$-form.
    \end{enumerate}
\end{definition}

  Let us observe that if $\omega$ is an almost symplectic structure on $M$, then $M$ has even dimension, say $2n$, and we have an  isomorphism of $\mathcal{C}^\infty(M)$-modules
$$\flat : \mathfrak{X}(M) \longrightarrow \textstyle\bigwedge^1(M) \, ,  \quad \flat(Z)=\iota_Z\omega\,.$$

 Let $(x^1,\ldots,x^n,y^1,\ldots,y^n)$ be the standard coordinates on $\r^{2n}$. The canonical
 symplectic form on $\r^{2n}$ is
 $$\omega_0=dx^1\wedge dy^1+\ldots + dx^n\wedge dy^n.$$

 The most important theorem in Symplectic Geometry is the following.
\index{Darboux Theorem!Symplectic}
 \begin{theorem}[Darboux Theorem] Let $\omega$ be an almost symplectic $2$-form on a $2n$-dimensional manifold $M$. Then $d\omega=0$ (that is, $\omega$ is symplectic), if and only if for each $x\in M$ there exists a coordinate neighborhood $U$ with local coordinates $(x^1,\ldots,x^{n},y^1,\ldots,y^n)$ such that
 $$
 \omega=\ds\sum_{i=1}^n dx^i\wedge dy^i\,.
 $$
 \end{theorem}
 Taking into account this result, one could develop the Hamiltonian formulation described in section \ref{hamquaions} substituting the cotangent bundle $T^*Q$ by an arbitrary symplectic manifold, see \cite{AM-1978,Arnold-1978, god, Goldstein, HSS-2009, Holm-2008}.

\newpage
\mbox{}
\thispagestyle{empty} % para que no se numere esta p\'{a}gina

\chapter{Cosymplectic manifolds}\label{cosymma}.

\index{Cosymplectic manifold}
\begin{definition}
A \emph{cosymplectic manifold} is a triple $(M, \eta, \omega)$ consisting of a smooth $(2n+1)$-dimen\-sio\-nal manifold $M$ equipped with a closed $1$-form $\eta$ and a closed $2$-form  $\omega$, such that $\eta\wedge \omega^n\neq 0$.

In particular,  $\eta\wedge \omega^n$ is a volume form on $M$.
\end{definition}

The standard example of a cosymplectic manifold is provided by the extended cotangent bundle $(\mathbb{R}\times T^*N,dt,\pi^*\omega_N)$ where $t\colon \mathbb{R}\times T^*N\to N$ and $\pi\colon \mathbb{R}\times T^*N\to T^*N$ are the canonical projections and $\omega_N$ is the canonical symplectic form on $T^*N$.

Consider the vector bundle homomorphism
\[
    \begin{array}{rcl}
        \flat\colon  TM & \to & T^*M\\\noalign{\medskip}
             v & \mapsto & \flat(v)=\iota_v\omega + (\iota_v\eta)\eta\,.
    \end{array}
\]

Then $\flat$ is a vector bundle isomorphism with inverse $\sharp$. Of course, the linear homomorphism
\[
    \flat_x\colon T_xM\to T_x^*M
\]
induced by $\flat$ is an isomorphism too, for all $x\in M$.

Given a cosymplectic manifold $(M,\eta,\omega)$, then there exists a distinguished vector field $\mathcal{R}$ (called the \emph{Reeb vector field}) such that
\[
    \iota_\mathcal{R}\eta= 1,\quad \iota_\mathcal{R}\omega=0\,,
\]
or, in other form,
\[
    \mathcal{R}=\sharp(\eta)\,.
\]

\index{Darboux theorem!cosymplectic manifold}
\begin{theorem}[Darboux theorem for cosymplectic manifolds]
    Given a cosymplectic manifold \newline $(M,\eta,\omega)$, there exists, around each point $x$ of $M$, a coordinate neighborhood with coordinates $(t,q^i, p_i)$, $\,1\leq i\leq n$, such that
    \[
        \eta=dt,\quad \omega=dq^i\wedge dp_i\,.
    \]
    These coordinates are called \emph{Darboux coordinates}.
\end{theorem}

In Darboux coordinates, we have $\mathcal{R}=\nicefrac{\partial}{\partial t}.$

Using the isomorphisms $\flat$ and $\sharp$ one can associate with every function $f\in\mathcal{C}^\infty(M)$ these  following vector fields:
\begin{itemize}
\index{Gradient vector field}
    \item The \emph{gradient vector field}, ${\rm grad}\, f$  defined by
        \[
            {\rm grad}\, f= \sharp(df)\,,
        \]
         or, equivalently,
         \[
            \iota_{{\rm grad}\, f}\eta = \mathcal{R}(f),\quad \iota_{{\rm grad}\, f}\omega=df-\mathcal{R}(f)\eta\,.
         \]
\index{Hamiltonian!vector field on cosymplectic manifolds}
\index{Cosymplectic manifold!Hamiltonian vector field}
    \item The \emph{Hamiltonian vector field} $X_f$ defined by
        \[
            X_f = \sharp(df-\mathcal{R}(f)\eta)\,,
        \]
        or, equivalently,
        \[
            \iota_{X_f}\eta=0,\quad \iota_{X_f}\omega=df-\mathcal{R}(f)\eta\,.
        \]
    \index{Evolution vector field}
    \index{Cosymplectic manifold!Evolution vector field}
    \item The \emph{evolution vector field} $E_f$ defined by
        \[
            E_f=\mathcal{R}+ X_f\,,
        \]
        or, equivalently,
        \[
            \iota_{E_f}\eta=1,\quad \iota_{E_f}\omega = df- \mathcal{R}(f)\eta\,.
        \]
\end{itemize}

In Darboux coordinates we have
\begin{align*}
    {\rm grad}\,f = & \ds\frac{\partial f}{\partial t}\ds\frac{\partial}{\partial t} + \ds\frac{\partial f}{\partial p_ i}\ds\frac{\partial}{\partial q^i}- \ds\frac{\partial f}{\partial q^i}\ds\frac{\partial}{\partial p_i}\,,\\\noalign{\medskip}
    X_f =& \ds\frac{\partial f}{\partial p_ i}\ds\frac{\partial}{\partial q^i}- \ds\frac{\partial f}{\partial q^i}\ds\frac{\partial}{\partial p_i}\,,\\\noalign{\medskip}
    E_f =& \ds\frac{\partial}{\partial t} + \ds\frac{\partial f}{\partial p_ i}\ds\frac{\partial}{\partial q^i}- \ds\frac{\partial f}{\partial q^i}\ds\frac{\partial}{\partial p_i}\,.
\end{align*}

Consider now an integral curve $c(s)=(t(s),q^i(s), p_i(s))$ of the evolution vector field $E_f$: this implies that $c(s)$ should satisfy the following system of differential equations
\[
    \ds\frac{dt}{ds}=1,\quad \ds\frac{dq^i}{ds}=\ds\frac{\partial f}{\partial p_i},\quad \ds\frac{dp_i}{ds}= -\ds\frac{\partial f}{\partial q^i}\,.
\]

Since $\ds\frac{dt}{ds}=1$ implies $t(s)=s+ constant$, we deduce that
\[
    \ds\frac{dq^i}{dt} = \ds\frac{\partial f}{\partial p_i},\quad \ds\frac{dp_i}{dt}= -\ds\frac{\partial f}{\partial q^i}\,,
\]
since $t$ is an affine transformation of $s$.

As in symplectic hamiltonian mechanics, we can define a Poisson bracket. Indeed, if $f,g\in\mathcal{C}^\infty(M)$, then
\[
    \{f,q\}=\omega({\rm grad}\, f,{\rm grad}\, g)
\]
such that we obtain the usual expression for this Poisson bracket
\[
    \{f,g\}= \ds\frac{\partial f}{\partial p_i}\ds\frac{\partial g}{\partial q^i}- \ds\frac{\partial f}{\partial q^i}\ds\frac{\partial g}{\partial p_i}\,.
\]

Observe that a cosymplectic manifold is again a Poisson manifold when it is equipped with this bracket $\{\cdot,\cdot\}$.

For more details of cosymplectic manifolds see, for instance \cite{CLL-1992,CLM-91}.

%\newpage
%\mbox{}
%\thispagestyle{empty} % para que no se numere esta p\'{a}gina

           % \input{./chapters/Comments.tex}

    %%%%%%%%%%%%%%%%%%%%%%%%%%%%%%%%%%%%%%%%%%%%%%%%%%%%%%%%%%%%%%%%%%%%%%%%%%%%%%%%%%%%%%%%%%%%%%%%%%%%%%%%%%%%%%%%%%%%%%%%%%
    %%%%            Backmatter
    %%%%%%%%%%%%%%%%%%%%%%%%%%%%%%%%%%%%%%%%%%%%%%%%%%%%%%%%%%%%%%%%%%%%%%%%%%%%%%%%%%%%%%%%%%%%%%%%%%%%%%%%%%%%%%%%%%%%%%%%%%

    \backmatter

            \cleardoublepage
            \addcontentsline{toc}{chapter}{Bibliography}
            \markboth{{\rm Bibliography}}{{\rm Bibliography}}

            \markboth{{\rm Index}}{{\rm Index}}
            \printindex

\begin{thebibliography}{99}

\bibitem{AM-1978}
R. Abraham and J.E. Marsden. \textit{Foundations of Mechanics}. Second edition. \textit{Benjamin/Cummings Publishing Co.}, Inc., Advanced Book Program, Reading, Mass., 1978.

\bibitem{Albert-1994}
C. Albert. Le th\'{e}oreme de r\'{e}duction de Marsden-Weinstein en g\'{e}ometrie cosymplectique et de contact. \textit{J. Geom. Phys. 6} (1989), no. 4, 627–-649.

%\bibitem{Ant95}
%S.S. Antman. \textit{Nonlinear Problems of Elasticity}. Second edition. \textit{Applied Mathematical Sciences, 107}. Springer, New York, 2005.

\bibitem{Arnold-1978}
V.I. Arnold. \textit{Mathematical Methods of Classical Mechanics}. \textit{Graduate Texts in Mathematics, 60}. Springer-Verlag, New York-Heidelberg, 1978.

\bibitem{Arnold-1998}
L. Arnold.  \textit{Random Dynamical Systems}, \textit{Springer Monographs in Mathematics}.
Springer-Verlag,  Berlin, 1998.

\bibitem{Awane-1992}
A. Awane. $k$-symplectic structures.
\textit{J. Math. Phys. 33} (1992), no. 12, 4046–-4052.

\bibitem{Awane-1994}
A. Awane.  G-espaces $k$-symplectiques homog\`{e}nes.
\textit{J. Geom. Phys. 13} (1994), no. 2, 139–-157.

\bibitem{Awane-2000}
A. Awane and  M. Goze. \textit{Pfaffian systems,
$k$-symplectic systems}. Kluwer Academic Publishers,
Dordrecht, 2000.

\bibitem{BLM-2012}
M. Barbero-Li\~{n}an, M. de Le\'{o}n and D. Mart\'{\i}n de Diego. Lagrangian submanifolds and Hamilton-Jacobi equation. \textit{Monatsh. Math. 171} (2013), no. 3-4, 269–-290.

\bibitem{BEMS-1971}
A. Barone, F. Esposito, C.J.  Magee and  A.C. Scott.  Theory and applications of the Sine-Gordon equation. \textit{Riv. Nuovo Cim. 1} (1971), no. 2, 227--267.

\bibitem{BS-1993}
L. Bates and J. Sniatycki. Nonholonomic reduction. \textit{Rep. Math. Phys. 32} (1993), no. 1, 99–-115.

\bibitem{BPP-2008}
M.C. Bertin, B.M. Pimentel and P.J. Pompeia.
Hamilton-Jacobi approach for first order actions and theories with higher order derivatives. \textit{Annals Phys. 323} (2008), 527--547.

\bibitem{BLMS-2002}
E. Binz, M. de Le\'{o}n, D. Mart\'{\i}n de Diego and D. Socolescu.
Nonholonomic Constraints in Classical Field Theories. XXXIII Symposium on Mathematical Physics (Tor\'{u}n, 2001).
\textit{Rep. Math. Phys. 49} (2002), no. 2-3, 151–-166.

\bibitem{BS-1981}
A.R. Bishop and T. Schneider. \textit{Solitons and Condensed Matter Physics}. Proceedings of the Symposium on Nonlinear (Soliton) Structure and Dynamics in Condensed Matter held in Oxford, June 27–29, 1978. Edited by Alan R. Bishop and Toni Schneider. \textit{Springer Series in Solid-State Sciences, 8}. Springer-Verlag, Berlin-New York, 1978.



\bibitem{Bloch}
A.M. Bloch. \textit{Nonholonomic Mechanics and Control}, volume 24 of
\textit{Interdisciplinary Applied Mathematics Series, 24}.
Systems and Control. Springer-Verlag, New York, 2003.

\bibitem{BKMM-1996}
A.M. Bloch, P.S. Krishnaprasad, J.E. Marsden and  R.M.
Murray. Nonholonomic mechanical systems with symmetry.  \textit{Arch. Rational Mech. Anal. 136} (1996), no. 1, 21–-99.


\bibitem{Bruno-2007}
D. Bruno. Constructing a class of solutions for the Hamilton-Jacobi equations in field theory. \textit{J. Math. Phys. 48} (2007), no. 11, 112902, 12 pp.


%\bibitem{Bucataru-2005}
%I. Bucataru. Linear connections for systems of higher
%order differential equations. \textit{Houston J. Math. 31} (2005), no. 2, 315–-332.

\bibitem{CW-1999}
A. Cannas da Silva and A. Weinstein. \textit{Geometric models for
noncommutative algebras}. \textit{Berkeley Mathematics Lecture Notes, 10}. American Mathematical Society, Providence, RI; Berkeley Center for Pure and Applied Mathematics, Berkeley, CA, 1999.

\bibitem{CIL-1996}
F. Cantrijn,  A. Ibort and  M. de Le\'on.
Hamiltonian structures on multisymplectic manifolds. Geometrical structures for physical theories, I (Vietri, 1996).
\textit{Rend. Sem. Mat. Univ. Politec. Torino 54} (1996), no. 3, 225-–236.

\bibitem{CIL-1999}
F. Cantrijn, A. Ibort and  M. de Le\'on.
On the geometry of multisymplectic manifolds.
\textit{J. Austral. Math. Soc. Ser. A 66} (1999), no. 3, 303–-330.

\bibitem{CLL-1992}
F. Cantrijn,  M. de Le\'on and E. Lacomba.
Gradient vector fields on cosymplectic manifolds.
\textit{J. Phys. A 25} (1992), no. 1, 175–-188.

\bibitem{Carmeli}
M. Carmeli.
\textit{Classical fields: general relativity and gauge theory}. World Scientific Publishing Co., Inc., River Edge, New Jersey, 2001.

\bibitem{CCI-91}
J.F. Cari\~nena, M. Crampin and  L.A. Ibort.
On the multisymplectic formalism for first order field theories.
\textit{Differential Geom. Appl. 1} (1991), no. 4, 345–-374.


\bibitem{CGMMR}
J.F. Cari\~{n}ena, X. Gr\`{a}cia, G. Marmo, E. Mart\'{\i}nez,
M.C. Mu\~{n}oz-Lecanda and N. Rom\'{a}n-Roy.
Geometric Hamilton-Jacobi theory. \textit{Int. J. Geom. Methods Mod. Phys. 3} (2006), no. 7, 1417-–1458.

\bibitem{pepin2}
J.F. Cari\~{n}ena, X. Gr\`{a}cia, G. Marmo, E. Mart\'{\i}nez,
M.C. Mu\~{n}oz-Lecanda and N. Rom\'{a}n-Roy.
Geometric Hamilton-Jacobi theory for nonholonomic dynamical systems. \textit{Int. J. Geom. Methods Mod. Phys. 7} (2010), no. 3, 431-–454.

%\bibitem{CLM-1989}
%J.F. Cari\~{n}ena, C. L\'{o}pez and E. Mart\'{\i}nez.
%A new approach to the converse of Noether's theorem.
%\textit{J. Phys. A 22} (1989), no. 22, 4777-–4786.

\bibitem{CGR-2001}
M. Castrill\'{o}n L\'{o}pez, P.L. Garc\'{\i}a P\'{e}rez and T.S. Ratiu.
Euler-Poincar\'{e} reduction on principal bundles.
\textit{Lett. Math. Phys. 58} (2001), no. 2, 167-–180.

\bibitem{CM-2008}
M. Castrill\'{o}n L\'{o}pez and J.E. Marsden.
Covariant and dynamical reduction for principal bundle field theories. \textit{Ann. Global Anal. Geom. 34} (2008), no. 3, 263–-285.

%\bibitem{CR-2003}
%M. Castrill\'{o}n L\'{o}pez and T.S. Ratiu.
%Reduction in principal bundles: covariant Lagrange-Poincar\'{e} equations. \textit{Comm. Math. Phys. 236} (2003), no. 2, 223–-250.

\bibitem{CLM-91}
D. Chinea, M. de Le\'on and J.C. Marrero.
Locally conformal cosymplectic manifolds and time-dependent Hamiltonian systems.
\textit{Comment. Math. Univ. Carolin. 32} (1991), no. 2, 383–-387.

%\bibitem{CenMarRat}
%H. Cendra, J.E. Marsden and T.S. Ratiu.
%Lagrangian reduction by stages.
%\textit{Mem. Amer. Math. Soc. 152} (2001), no. 722, x+108 pp.

\bibitem{Cortes-2002}
J. Cort\'{e}s.
\textit{Geometric, control and numerical aspects of nonholonomic systems}. \textit{Lecture Notes in Mathematics, 1793}. Springer-Verlag, Berlin, 2002. xvi+219 pp.

\bibitem{CLMMM-2006}
J. Cort\'{e}s, M. de Le\'{o}n, J.C. Marrero, D. Mart\'{\i}n de Diego and E.
Martinez.
A survey of Lagrangian mechanics and control on
Lie algebroids and groupoids.
\textit{Int. J. Geom. Methods Mod. Phys. 3} (2006), no. 3, 509–558.

\bibitem{CLMM}
J. Cort\'es, M. de Le\'on, D. Mart\'{\i}n de Diego and
S. Mart\'{\i}nez.
Geometric description of vakonomic and
nonholonomic dynamics. Comparison of soluciones. \textit{SIAM J. Control Optim. 41} (2003), no. 5, 1389--1412.


\bibitem{CMC}
J. Cort\'es, S. Mart\'\i nez and F. Cantrijn.  Skinner-Rusk
approach to time-dependent mechanics.
\textit{Phys. Lett. A 300} (2002), no. 2-3, 250–-258.

%\bibitem{crampin-1971}
%M. Crampin.
%On horizontal distributions on the tangent bundle of a
%differentiable manifold.
%\textit{J. London Math. Soc. 2} (1971), no 3, 178–-182.


\bibitem{Crampin-1983}
M. Crampin.
Tangent bundle geometry for Lagrangian dynamics.
\textit{J. Phys. A 16} (1983), no. 16, 3755–-3772.



\bibitem{cram2}
M. Crampin.
Defining Euler-Lagrange fields in terms of
almost tangent structures.
\textit{Phys. Lett. A 95} (1983), no. 9, 466-–468.

%\bibitem{cram3}
%M. Crampin.
%Generalized Bianchi identities for horizontal distributions.
%\textit{Math. Proc. Cambridge Philos. Soc. 94} (1983), no. 1, 125-–132.
%
%\bibitem{cram5}
%M. Crampin.
%Connections of Berwald type.
%\textit{Publ. Math. Debrecen 57} (2000), no. 3-4, 455–-473.


\bibitem{Davydov-1985}
A.S. Davydov.
\textit{Solitons in Molecular Systems}.
\textit{Mathematics and its Applications (Soviet Series), 4}. D. Reidel Publishing Co., Dordrecht, 1985.


\bibitem{Die}
J. Dieudonn\'e.
\textit{Foundations of modern analysis}. 2nd ed. \textit{Pure and Applied Mathematics, Vol. 10-I}. Academic Press, New York-London, 1969.

\bibitem{du}
E. Durand.
\textit{\'{E}lectrostatique, les distributions}, Mason
et Cie Editeurs, Paris 1964.


\bibitem{ELMR-2005}
A. Echeverr\'\i a-Enr\'\i quez, M. De Le\'on, M.C. Mu\~noz-Lecanda and N. Rom\'an-Roy.
Extended Hamiltonian Systems in Multisymplectic Field Theories.
\textit{J. Math. Phys. 48} (2007), no. 11, 112901, 30 pp.

\bibitem{ELMMR-2004}
A. Echeverr\'{\i}a-Enr\'{\i}quez, C. L\'{o}pez, J. Mar\'{i}n-Solano, M.C.
Mu\~{n}oz-Lecanda and N. Rom\'{a}n-Roy.
Lagrangian-Hamiltonian unified formalism for field theory. \textit{J. Math. Phys. 45} (2004), no. 1, 360–380.

\bibitem{EM-92}
A. Echeverr\'{\i}a-Enr\'{\i}quez and M.C. Mu\~{n}oz Lecanda.
Variational calculus in several variables: a Hamiltonian approach.
\textit{Ann. Inst. H. Poincar\'{e} Phys. Th\'{e}or. 56} (1992), no. 1, 27-–47.

\bibitem{EMR-1991}
A. Echeverr\'{\i}a-Enr\'{\i}quez, M.C. Mu\~{n}oz-Lecanda and N. Rom\'{a}n-Roy.
Geometrical setting of time-dependent regular systems: Alternative
models. \textit{Rev. Math. Phys. 3} (1991), no. 3, 301–-330.

\bibitem{EMR-1995}
A. Echeverr\'{\i}a-Enr\'{\i}quez, M.C. Mu\~{n}oz-Lecanda and N. Rom\'{a}n-Roy.
Non-standard connections in classical mechanics. \textit{J. Phys. A 28} (1995), no. 19, 5553-–5567.

\bibitem{EMR-1996}
A. Echeverr\'{\i}a-Enr\'{\i}quez, M.C. Mu\~{n}oz-Lecanda and N. Rom\'{a}n-Roy.
Geometry of Lagrangian first-order Classical Field Theories.
\textit{Fortschr. Phys. 44} (1996), no. 3, 235–-280.

 \bibitem{EMR-1998}
A. Echeverr\'\i a-Enr\'\i quez, M.C. Mu\~noz-Lecanda and N.
 Rom\'an-Roy. Multivector fields and connections: setting Lagrangian equations in field theories. \textit{J. Math. Phys. 39} (1998), no. 9, 4578-–4603.


\bibitem{EMR-1999}
A. Echeverr\'\i a-Enr\'\i quez, M.C. Mu\~noz-Lecanda and N.
Rom\'an-Roy. Multivector field formulation of Hamiltonian field theories: equations and symmetries.
\textit{J. Phys. A 32} (1999), no. 48, 8461–-8484.


\bibitem{EMR-00}
A. Echeverr\'{\i}a-Enr\'{\i}quez, M.C. Mu\~{n}oz-Lecanda and N. Rom\'{a}n-Roy. \textit{Geometry of multisymplectic Hamiltonian first-order field theories}.
\textit{J. Math. Phys. 41} (2000), no. 11, 7402–-7444

\bibitem{Ehresmann}
Ch. Ehresmann.
Les prolongements d'une vari\'{e}t\'{e} diff\'{e}rentiable. I. Calcul des jets, prolongement principal.
\textit{C. R. Acad. Sci. Paris 233}, (1951). 598–-600.

\bibitem{EL-1978}
J. Eells and L. Lemaire.
A report on harmonic maps
\textit{Bull. London Math. Soc. 10} (1978), no. 1, 1-–68.

\bibitem{Frankel}
T. Frankel.
Maxwell's equations.
\textit{Amer. Math. Monthly 81} (1974), 343–-349.

\bibitem{GP1}
P.L. Garc\'{\i}a and A. P\'{e}rez-Rend\'{o}n.
 Symplectic approach to the theory of quantized fields. I. \textit{Comm. Math. Phys. 13} (1969), 24-–44.

\bibitem{GP2}
P.L. Garc\'{\i}a and A. P\'{e}rez-Rend\'{o}n.
 Symplectic approach to the theory of quantized fields. II. \textit{Arch. Rational Mech. Anal. 43} (1971), 101–-124.

\bibitem{Sarda2}
G. Giachetta, L. Mangiarotti and G. Sardanashvily.
New Lagrangian and Hamiltonian Methods in Field Theory. Scientific Publishing Co., Inc., River Edge, NJ, 1997.

\bibitem{Sd-95}
G. Giachetta, L.  Mangiarotti and G. Sardanashvily.
Covariant Hamilton equations for field theory. \textit{J. Phys. A 32} (1999), no. 38, 6629-–6642.


\bibitem{GJM-1979}
J.D. Gibbon, I.N. James and I.M. Moroz.
The Sine-Gordon Equation as a Model for a Rapidly Rotating Baroclinic Fluid. \textit{Phys. Scripta 20} (1979), no. 3-4, 402–-408.

\bibitem{GL-1950}
V.L. Ginzburg and L.D. Landau. On the theory of superconductivity, \textit{Zh. Eksp. Teor. Fiz. 20} (1950), 1064--1082.

\bibitem{god}
C. Godbillon. \textit{G\'{e}om\'{e}trie diff\'{e}rentielle et m\'{e}canique analytique.} (French) Hermann, Paris, 1969, 183 pp.

\bibitem{Goldstein}
H. Godstein, C.P. Poole Jr. and J.L. Safko. \textit{ Classical Mechanics} (3rd Edition). Addison-Wesley Pub. Co., 2001.


\bibitem{gs}
H. Goldschmidt and S. Sternberg. The Hamilton-Cartan
Formalism in the Calculus of Variations.
\textit{Ann. Inst. Fourier (Grenoble) 23} (1973), no. 1, 203-–267.


\bibitem{Go1}
M.J. Gotay.
\textit{An exterior differential systems approach to the Cartan form.}
Symplectic geometry and mathematical physics (Aix-en-Provence, 1990), 160–188,
\textit{Progr. Math., 99}, Birkh\"{a}user Boston, Boston, MA, 1991.

\bibitem{Go2}
M.J. Gotay.
\textit{A multisymplectic framework for classical field theory and the calculus of variations. I. Covariant Hamiltonian formalism}. Mechanics, analysis and geometry: 200 years after Lagrange, 203–-235,
\textit{North-Holland Delta Ser.}, North-Holland, Amsterdam, 1991.

\bibitem{Go3}
M.J. Gotay.
A multisymplectic framework for classical field theory and the calculus of variations. II. Space + time decomposition.
\textit{Differential Geom. Appl. 1} (1991), no. 4, 375-–390.

\bibitem{Gymmsy}
M.J. Gotay, J. Isenberg,  J.E. Marsden and R. Montgomery.  Momentum Maps and Classical Relativistic Fields. Part I: Covariant Field Theory.\url{http://arxiv.org/abs/physics/9801019v2} [math-ph] (2004).

\bibitem{Gymmsy2}
M.J.  Gotay, J. Isenberg and J.E. Marsden.
Momentum Maps and Classical Relativistic Fields. Part II: Canonical Analysis of Field Theories.
\url{http://arxiv.org/abs/math-ph/0411032} (2004).



\bibitem{grif1}
J. Grifone.
Structure presque-tangente et connexions. I.
\textit{Ann. Inst. Fourier (Grenoble) 22} (1972), no. 1, 287–-334.

\bibitem{grif2}
J. Grifone.
Structure presque-tangente et connexions. II.
\textit{Ann. Inst. Fourier (Grenoble) 22} (1972), no. 3, 291–-338.


\bibitem{grif3}
J.Grifone and M. Mehdi.
On the geometry of Lagrangian mechanics with non-holonomic constraints.
\textit{J. Geom. Phys. 30} (1999), no. 3, 187-–203.



\bibitem{Gu-1987}
C. G\"{u}nther.
The polysymplectic Hamiltonian formalism in field theory and calculus of variations. I. The local case.
\textit{J. Differential Geom. 25} (1987), no. 1, 23–-53.

\bibitem{HK-04}
F. H\'elein and J. Kouneiher.
 Covariant Hamiltonian formalism for the calculus of variations with several variables: Lepage-Dedecker versus De Donder-Weyl.
\textit{Adv. Theor. Math. Phys. 8} (2004), no. 3, 565–-601.

\bibitem{HM-1990}
P.J. Higgins and K. Mackenzie.
Algebraic constructions in the category of Lie algebroids.
\textit{J. Algebra 129} (1990), no. 1, 194–-230.

\bibitem{HSS-2009}
D.D. Holm, T. Schmah and C. Stoica.
\textit{Geometric mechanics and symmetry. From finite to infinite dimensions}. With solutions to selected exercises by David C. P. Ellis. \textit{Oxford Texts in Applied and Engineering Mathematics, 12}. Oxford University Press, Oxford, (2009)

\bibitem{Holm-2008}
D.D. Holm.
\textit{Geometric mechanics. Part I. Dynamics and symmetry}. Imperial College Press, London; distributed by World Scientific Publishing Co. Pte. Ltd., Hackensack, NJ, 2008.

%\bibitem{ikeda}
%N. Ikeda.
%Two dimensional gravity and nonlinear gauge
%theory.
%\textit{Ann. Physics 235} (1994), no. 2, 435-–464.

\bibitem{IR-2000}
E. Infeld and G.  Rowlands.
Nonlinear waves, solitons and chaos.
Cambridge University Press, Cambridge, 1990.

\bibitem{Kana}
I.V. Kanatchikov.
Canonical structure of classical field theory in the polymomentum phase space.
\textit{Rep. Math. Phys. 41} (1998), no. 1, 49-–90.

\bibitem{Ki-73}
J. Kijowski.
A finite-dimensional canonical formalism in the classical field
theory.
\textit{Comm. Math. Phys. 30} (1973), 99–-128.


\bibitem{KS-75}
J. Kijowski and W. Szczyrba.
\textit{Multisymplectic manifolds and the geometrical construction of
the Poisson brackets in the classical field theory}.
G\'{e}om\'{e}trie symplectique et physique math\'{e}matique (Colloq. Internat. C.N.R.S., Aix-en-Provence, 1974), pp. 347–349. \'{E}ditions Centre Nat. Recherche Sci., Paris, 1975.


\bibitem{KT-79}
J. Kijowski and W.M. Tulczyjew.
\textit{A symplectic framework for
field theories}. \textit{Lecture Notes in Physics, 107}. Springer-Verlag, Berlin-New York, 1979.

\bibitem{klein}
J. Klein.
Espaces variationelles et m\'{e}canique. \textit{Ann.
Inst. Fourier  12} (1962), 1--124.


\bibitem{kms}
I. Kol\'{a}\v{r}, P. Michor and J. Slov\'{a}k. \textit{Natural operations
in Differential Geometry}. Springer-Verlag, Berlin, 1993.

\bibitem{KN}
S. Kobayashi and K. Nomizu.
\textit{Foundations of differential geometry. Vol I.} Interscience Publishers, a division of John Wiley $\&$ Sons, New York-Lond on 1963.

%\bibitem{kv-1993}
%O. Krupkova and A. Vondra.
%On some integration methods for connections on fibered manifolds.
%Differential geometry and its applications (Opava, 1992), 89–101,
%\textit{Math. Publ., 1}, Silesian Univ. Opava, Opava, 1993.

\bibitem{Lee}
J.M. Lee.
\textit{Introduction to smooth manifolds}.
\textit{Graduate Texts in Mathematics, 218}. Springer-Verlag, New York, 2003.

\bibitem{lmm1}
M. de Le\'on, D. Iglesias-Ponte and D. Mart\'{\i}n de Diego. Towards a Hamilton-Jacobi theory for nonholonomic mechanical systems.
\textit{J. Phys. A 41} (2008), no. 1, 015205, 14 pp.

%\bibitem{LL-89}
%M. de Le\'{o}n and E.A. Lacomba.
%Lagrangian submanifolds and higher-order mechanical systems.
%\textit{J. Phys. A 22} (1989), no. 18, 3809–-3820.
%
%\bibitem{LLR-91}
%M. de Le\'{o}n, E.A. Lacomba and P.R.
%Rodrigues.
%\textit{Special presymplectic manifolds, Lagrangian
%submanifolds and the Lagrangian-Hamiltonian systems on jet
%bundles.} Proceedings of the First "Dr. Antonio A. R. Monteiro'' Congress on Mathematics (Spanish) (Bah\'{\i}a Blanca, 1991), 103-–122, Univ. Nac. del Sur, Bah\'{\i}a Blanca, 1991.

\bibitem{lmm1}
M. de Le\'on, J. Mar\'{\i}n and J. C. Marrero. Ehresmann
Connections in Classical Field Theories. Differential geometry and its applications (Granada, 1994), 73–89,
\textit{An. F\'{\i}s. Monogr. 2,} CIEMAT, Madrid, 1995.

\bibitem{LMM-1996b}
M. de Le\'on, J. Mar\'{\i}n and J.C. Marrero. A geometrical approach to Classical Field Theories: a constraint algorithm for singular theories.
\textit{New Developments in Differential Geometry
Mathematics and Its Applications 350} (1996),  291--312.


\bibitem{LMM-1997}
M. de Le\'{o}n, J.C. Marrero and D. Mart\'{\i}n de Diego.
Mechanical systems with nonlinear constraints.
\textit{International Journal of Theoretical Physics
36} (1997), no 4, 979--995.

\bibitem{LMM-1996}
M. de Le\'{o}n, J.C. Marrero and D. Mart\'{\i}n de
Diego.
Non-holonomic Lagrangian systems in jet manifolds.
\textit{J. Phys. A 30} (1997), no. 4, 1167–-1190.

\bibitem{LMM-2002}
M. de Le\'on, J.C. Marrero and D. Mart\'\i n de Diego. A new
geometrical setting for Classical Field Theories. \textit{Banach Center Publications Institute of Mathematics, Polish Academy of Sciences) Vol. 59 } (2003), 189--209



\bibitem{LMM-2009}
M. de Le\'on, J.C. Marrero and D. Mart\'{\i}n de Diego \textit{A geometric Hamilton-Jacobi theory for classical field theories}.
Variations, geometry and physics, 129–-140, Nova Sci. Publ., New York, 2009.

\bibitem{LMM-2010}
M. de Le\'on, J.C. Marrero and D. Mart\'{\i}n de Diego.
Linear almost Poisson structures and Hamilton-Jacobi equation. Applications to nonholonomic mechanics.
\textit{J. Geom. Mech. 2} (2010), no. 2, 159–-198.


\bibitem{LMMSV-2010}
M. de Le\'on, J.C. Marrero, D. Mart\'{\i}n de Diego, M. Salgado and S. Vilari\~{n}o.
Hamilton-Jacobi theory in k-symplectic field theories.
\textit{Int. J. Geom. Methods Mod. Phys. 7} (2010), no. 8, 1491–-1507.

\bibitem{LMM-2005}
M. de Le\'{o}n M, J.C.  Marrero and E. Mart\'{\i}nez.
Lagrangian submanifolds and dynamics on Lie algebroids.
\textit{J. Phys. A 38} (2005), no. 24, 241–-308.

\bibitem{LM-1996}
M. de Le\'{o}n and D. Mart\'{\i}n de Diego.
On the geometry of non-holonomic Lagrangian systems.
\textit{J. Math. Phys. 37} (1996), no. 7, 3389–-3414.

 \bibitem{LM-96}
M. de Le\'on and D. Mart\'\i n de Diego.
Symmetries and constants of the motion for singular Lagrangian systems.
\textit{Internat. J. Theoret. Phys. 35} (1996), no. 5, 975-–1011.

\bibitem{LMS-03}
M. de Le\'{o}n, D.Mart\'{\i}n de Diego and A. Santamar\'{\i}a.
Tulczyjew triples and Lagrangian submanifolds in Classical Field Theories.
\textit{Applied Differential Geometry and Mechanics} (2003), 21--47.


\bibitem{LMS-2004}
M. de Le\'on, D. Mart\'\i n de Diego and A. Santamar\'\i
a-Merino.
Symmetries in classical field theory.
\textit{Int. J. Geom. Methods Mod. Phys. 1} (2004), no. 5, 651-–710.

\bibitem{LMS-2003}
M. de Le\'{o}n, D. Mart\'{\i}n de Diego and A. Santamar\'{\i}a-Merino.
Geometric integrators and nonholonomic mechanics.
\textit{J. Math. Phys. 45} (2004), no. 3, 1042–-1064.


\bibitem{LMSV-2008}
M. de Le\'{o}n, D. Mart\'{\i}n de Diego, M. Salgado and S. Vilari\~{n}o.
Nonholonomic constraints in k-symplectic classical field theories.
\textit{Int. J. Geom. Methods Mod. Phys. 5} (2008), no. 5, 799–-830.

\bibitem{LMSV-2009}
M. de Le\'{o}n, D. Mart\'{\i}n de Diego, M. Salgado and S. Vilari\~{n}o.
$k$-symplectic formalism on Lie algebroids.
\textit{J. Phys. A: Math. Theor. 42} (2009), 385209, 31pp.

\bibitem{LMcNRS-2002}
M. de Le\'on, M. McLean, L.K. Norris, A.M. Rey and M. Salgado.
Geometric Structures in Field Theory. \url{http://arxiv.org/abs/math-ph/0208036} (2002).

\bibitem{LMS-88}
M. de Le\'{o}n, I. M\'{e}ndez and M. Salgado. $p$-almost tangent
structures.
\textit{Rend. Circ. Mat. Palermo (2) 37} (1988), no. 2, 282-–294.

\bibitem{LMS-91}
M. de Le\'{o}n, I. M\'{e}ndez and M. Salgado.
Integrable $p$-almost tangent manifolds and tangent bundles of $p^1$-velocities.
\textit{Acta Math. Hungar. 58} (1991), no. 1-2, 45–-54.

\bibitem{LMS-93}
M. de Le\'{o}n, I. M\'{e}ndez and M. Salgado.
Regular p-almost cotangent structures. \textit{J. Korean Math. Soc. 25} (1988), no. 2, 273–-287.

\bibitem{LMeS-97}
M. de Le\'on, E. Merino, J. A. Oubi\~{n}a and M. Salgado.
Stable almost cotangent structures.
\textit{Boll. Un. Mat. Ital. B (7) 11} (1997), no. 3, 509–-529.

\bibitem{LMORS-1998}
M. de Le\'on, E. Merino, J. A. Oubi\~{n}a, P. R. Rodrigues and M. Salgado.
Hamiltonian systems on k-cosymplectic manifolds.
\textit{J. Math. Phys. 39} (1998), no. 2, 876–-893.

\bibitem{LMeS-2001}
M. de Le\'on, E. Merino and M. Salgado.
$k$-cosymplectic manifolds and Lagrangian field theories.
\textit{J. Math. Phys. 42} (2001), no. 5, 2092–-2104.

\bibitem{lr1}
M. de Le\'on and P. Rodrigues. \textit{Generalized Classical
Mechanics and Field Theory}. A geometrical approach of Lagrangian and Hamiltonian formalisms involving higher order derivatives. \textit{North-Holland Mathematics Studies, 112}. Notes on Pure Mathematics, 102. North-Holland Publishing Co., Amsterdam, 1985.


\bibitem{lr} M. de Le\'{o}n, P.R.  Rodrigues. \textit{Methods
 of differential geometry in analytical mechanics}.
\textit{North-Holland Mathematics Studies, 158}. North-Holland Publishing Co., Amsterdam, 1989.

\bibitem{LV-2012}
M. de Le\'on and S. Vilari\~{n}o.
Lagrangian submanifolds in k-symplectic settings.
\textit{Monatsh. Math. 170} (2013), no. 3-4, 381–-404.

\bibitem{LV-2014}
M. de Le\'{o}n and S. Vilari\~{n}o.
Hamilton-Jacobi theroy in $k$-cosymplectic field theories.
\textit{Int. J. Geom. Methods Mod. Phys. 11} (2014), no. 1, 1450007 (18 pages).

\bibitem{lm}
P. Libermann, C.M.  Marle.
\textit{Symplectic geometry and analytical mechanics}.
\textit{Mathematics and its Applications, 35}. D. Reidel Publishing Co., Dordrecht, 1987.
%
%\bibitem{LMR-99}
%C. L\'opez, E. Mart\'\i nez and M.F. Ra\~nada.
%Dynamical Symmetries, non-Cartan Symmetries and Superintegrability of the
%$n$-Dimensional Harmonic Oscillator.
%\textit{J. Phys. A 32} (1999), no. 7, 1241-–1249.

\bibitem{LV-2014}
J. de Lucas and S. Vilari\~{n}o.
$k$-symplectic Lie systems: theory and applications.
Arxiv:1404.1596v2. (2014)

\bibitem{Mack-1987}
K. Mackenzie. \textit{Lie groupoids and Lie algebroids in
differential geometry}. \textit{London Mathematical Society Lecture Note Series, 124}. Cambridge University Press, Cambridge, 1987.

\bibitem{Mack-1995}
K. Mackenzie.
Lie algebroids and Lie pseudoalgebras.
\textit{Bull. London Math. Soc. 27} (1995), no. 2, 97–-147.



\bibitem{ms}
L.  Mangiarotti and G. Sardanashvily.
  Gauge mechanics. World Scientific Publishing Co., Inc., River Edge, NJ, 1998.

\bibitem{mssv}
G. Marmo, E.J. Saletan, A. Simoni and B. Vitale. \textit{Dynamical
systems. A differential geometric approach to symmetry and
reduction}.
  Wiley-Interscience Publication, John Wiley $\&$ Sons, Ltd., Chichester,
 (1985).

%\bibitem{MM96} R. S. Marnning and J.H. Maddocks.
%A continuum rod model of sequence dependent dna structure. \textit{J. Chem. Phys. 105} (1996), 5626–-5646.

\bibitem{MRSV-2013}
J.C. Marrero, N. Rom\'{a}n-Roy, M. Salgado and S. Vilari\~{n}o. Reduction of polysymplectic manifolds.
\url{http://arxiv.org/abs/math-ph/1306.0337v1} (2013).

\bibitem{MS-99}
J.E. Marsden and S. Shkoller. Multisymplectic geometry, covariant Hamiltonians, and water waves.
\textit{Math. Proc. Cambridge Philos. Soc. 125} (1999), no. 3, 553-–575.


\bibitem{mar}
G. Martin. Dynamical structures for $k$-vector fields.
\textit{Internat. J. Theoret. Phys. 27} (1988), no. 5, 571-–585.

\bibitem{mar2} G. Martin. A Darboux theorem for multi-symplectic manifolds.
   \textit{Lett. Math. Phys. 16} (1988), no. 2, 133-–138.

\bibitem{MV-2010}
D. Mart\'{\i}n de Diego and S. Vilari\~{n}o.
Reduced classical field theories: $k$-cosymplectic formalism on Lie algebroids.
\textit{J. Phys. A: Math. Theor. 43} (2010), 325204, 32pp.

\bibitem{MAR-01}
E. Mart\'{\i}nez.
Lagrangian mechanics on Lie algebroids.
\textit{Acta Appl. Math. 67} (2001), no. 3, 295–-320.



\bibitem{MAR-04}
E. Mart\'{\i}nez. Classical Field theory on Lie algebroids: multisymplectic formalism.
\url{http://arxiv.org/abs/math/0411352} (2004)


\bibitem{Mart-2001}
E. Mart\'{\i}nez E. \textit{Geometric formulation of Mechanics on Lie
algebroids}. Proceedings of the VIII Fall Workshop on Geometry and Physics (Spanish) (Medina del Campo, 1999), 209–-222,
\textit{Publ. R. Soc. Mat. Esp., 2}, R. Soc. Mat. Esp., Madrid, 2001.


\bibitem{Mart-2005} E.  Mart\'{\i}nez.
Classical field theory on Lie algebroids: variational aspects.
\textit{J. Phys. A 38} (2005), no. 32, 7145–-7160.

%\bibitem{Marti1}
%E. Mart\'{\i}nez and J.F. Cari\~{n}ena.
%Geometric characterization of linearisable second-order differential equations.
%\textit{Math. Proc. Cambridge Philos. Soc. 119} (1996), no. 2, 373–-381.
%
%
%\bibitem{Marti2}
%E. Mart\'{\i}nez, J.F. Cari\~{n}ena and W. Sarlet. Derivations of differential forms along the tangent bundle projection. \textit{Differential Geom. Appl. 2} (1992), no. 1, 17-–43.
%
%\bibitem{Marti3}
%E. Mart\'{\i}nez, J.F. Cari\~{n}ena and W. Sarlet. Derivations of differential forms along the tangent bundle projection. II. \textit{Differential Geom. Appl. 3 } (1993), no. 1, 1–-29.

\bibitem{MN-2000} M. McLean and L.K. Norris. Covariant field theory on frame bundles of fibered manifolds.
\textit{    J. Math. Phys. 41} (2000), no. 10, 6808–-6823.


\bibitem{Tesis merino}
E.  Merino. \textit{Geometr\'{\i}a $k$-simpl\'{e}ctica y
$k$-cosimpl\'{e}ctica. Aplicaciones a las teor\'{\i}as cl\'{a}sicas de campos},
Publicaciones del Dpto. de Geometr\'{\i}a y Topolog\'{\i}a, {\bf 87}, Tesis
Doctoral, Universidad de Santiago de Compostela, Santiago de
Compostela, 1997.

\bibitem{mor3}
A. Morimoto. \textit{Prolongations of Geometric Structures}.
Mathematical Institute, Nagoya University, Nagoya 1969

\bibitem{mor}
A. Morimoto. Liftings of some types of tensor fields and connections to tangent bundles of $p^r$-velocities. \textit{Nagoya Math. J. 40} (1970), 13–-31.

 \bibitem{MRS-2004}
F. Munteanu, A. M. Rey and M. Salgado. The G\"{u}nther's
formalism in classical field theory: momentum map and reduction. \textit{J. Math. Phys. 45} (2004), no. 5, 1730–-1751.


\bibitem{MR-2001}
M.C. Mu\~{n}oz-Lecanda, N. Rom\'{a}n-Roy and F.J. Y\'{a}niz. Time-dependent Lagrangians invariant by a vector field.
\textit{Lett. Math. Phys. 57} (2001), no. 2, 107–-121.


\bibitem{MSV-2005}
M.C. Mu\~{n}oz-Lecanda, M. Salgado and S. Vilari\~{n}o. Nonstandard
connections in $k$-cosymplectic field theory.
\textit{J. Math. Phys. 46} (2005), no. 12, 122901, 25 pp.


\bibitem{MSV-2009}
M.C. Mu\~noz-Lecanda,  M. Salgado and S. Vilari\~no.
$k$-symplectic and $k$-cosymplectic Lagrangian field theories:
some interesting examples and applications. \textit{Int. J. Geom. Methods Mod. Phys. 7} (2010), no. 4, 669-–692.

\bibitem{No2}
L.K. Norris.
\textit{Generalized symplectic geometry on the frame bundle of a manifold}.
Differential geometry: geometry in mathematical physics and related topics (Los Angeles, CA, 1990), 435-–465,
\textit{Proc. Sympos. Pure Math., 54, Part 2}, Amer. Math. Soc., Providence, RI, 1993.


\bibitem{No3} L.K. Norris. Symplectic geometry on $T^*M$ derived from
$n$-symplectic geometry on $LM$.  \textit{J. Geom. Phys. 13} (1994), no. 1, 51-–78.

\bibitem{No4} L.K. Norris. Schouten-Nijenhuis Brackets.
\textit{J. Math. Phys. 38} (1997), no. 5, 2694–2709.

\bibitem{No5} L.K. Norris, $n$-symplectic algebra of observables in
covariant Lagrangian field theory. \textit{J. Math. Phys. 42} (2001), no. 10, 4827-–4845.

\bibitem{Olver}
P.J.  Olver. \textit{Applications of Lie groups to differential equations}. \textit{Graduate Texts in Mathematics, 107}. Springer-Verlag, New York, 1986.

\bibitem{olver}
P.J.  Olver.  \textit{Applied Mathematics Lecture Notes}. \url{http://www.math.umn.edu/~olver/appl.html}, (2007).

\bibitem{PR-2002-b}
C. Paufler and H. R\"omer.
Geometry of Hamiltonian $n$-vector fields in multisymplectic
field theory.  \textit{J. Geom. Phys. 44} (2002), no. 1, 52-–69.

\bibitem{PR-2002} C. Paufler and H. R\"{o}mer. De Donder-Weyl equations and multisymplectic geometry. XXXIII Symposium on Mathematical Physics (Tor\'{u}n, 2001).
\textit{Rep. Math. Phys. 49} (2002), no. 2-3, 325-–334.

\bibitem{Poor}
W. A. Poor. \textit{Differential geometric structures}. McGraw-Hill Book Co., New York, 1981.

\bibitem{PS-1962}
J.K. Perring and T.H.R. Skyrme. A model unified field equation.
\textit{Nuclear Phys. 31} (1962) 550–-555.

%\bibitem{Ra-95}
%M.F. Ra\~nada.  Integrable three-particle systems, hidden symmetries and deformations of the Calogero-Moser system. \textit{J. Math. Phys. 36} (1995), no. 7, 3541–-3558.
%
%\bibitem{Ra-97}
%M.F. Ra\~nada. Superintegrable $n=2$ systems, quadratic constants of motion, and potentials of Drach. \textit{J. Math. Phys. 38} (1997), no. 8, 4165–-4178.

\bibitem{RRS}
A.M. Rey, N. Rom\'{a}n-Roy and M. Salgado.  G\"{u}nther's formalism
($k$-symplectic formalism) in classical field theory: Skinner-Rusk
approach and the evolution operator.  \textit{J. Math. Phys. 46} (2005), no. 5, 052901, 24 pp.

\bibitem{RR-2009}
N. Rom\'an-Roy.
Multisymplectic Lagrangian and Hamiltonian Formalisms of Classical Field Theories.
\textit{SIGMA Symmetry Integrability Geom. Methods Appl. 5} (2009), Paper 100, 25 pp.



 \bibitem{RRSV-2011} N. Rom\'{a}n-Roy, M.A. Rey, M. Salgado and S. Vilari\~{n}o. On the k-symplectic, k-cosymplectic and multisymplectic formalisms of classical field theories.  \textit{J. Geom. Mech. 3} (2011), no. 1, 113–-137.

\bibitem{RRSV-2012} M.A. Rey, N. Rom\'{a}n-Roy, M. Salgado and S. Vilari\~{n}o. $k$-cosymplectic classical field theories: Tulczyjew and Skinner-Rusk formulations. \textit{Math. Phys. Anal. Geom. 15} (2012), no 2, 85--119.

\bibitem{RSV-2007}
N. Rom\'{a}n-Roy, M. Salgado and S. Vilari\~{n}o. Symmetries and
conservation laws in the G\"{u}nther $k$-symplectic formalism of field
theory. \textit{Rev. Math. Phys. 19} (2007), no. 10, 1117-–1147.

\bibitem{RSV-2013}
N. Rom\'{a}n-Roy, M. Salgado and S. Vilari\~{n}o.
Higher-order Noether symmetries in $k$-symplectic Hamiltonian field theory.
{Int. J. Geom. Methods Mod. Phys. 10} (2013), no. 8, 1360013.

\bibitem{RSS-2010}
N. Rom\'{a}n-Roy, M. Salgado and S. Vilari\~no.
On a kind of Noether symmetries and conservation laws in $k$-symplectic
field theory.
\textit{ J. Math. Phys. 52} (2011), no. 2, 022901, 20 pp.

\bibitem{rosen} G. Rosen. Hamilton-Jacobi functional theory for the integration of classical field equations.
\textit{Internat. J Theoret. Phys. 4} (1971), 281–-285.


\bibitem{Rund}
 H. Rund. \textit{The Hamilton-Jacobi theory in the calculus of variations. Its role in mathematics and physics}. Reprinted edition, with corrections. Robert E. Krieger Publishing Co., Huntington, N.Y., 1973.




\bibitem{saletan}
J.V. Jos\'{e} and E.J. Saletan. \textit{Classical Dynamics}. A contemporary approach. Cambridge University Press, Cambridge, 1998.

\bibitem{Sarda} G. Sardanashvily. \textit{Gauge theory in jet manifolds}. \textit{Hadronic Press Monographs in Applied Mathematics}. Hadronic Press, Inc., Palm Harbor, FL, 1993.

\bibitem{Sd-95b} G. Sardanashvily. Generalized Hamiltonian formalism for field theory. Constraint systems.
    \textit{Class. Quantum Grav. 13} (1996), no. 12.



\bibitem{SC-81}
W. Sarlet and F. Cantrijn. Higher-order Noether symmetries
and constants of the motion. \textit{J. Phys. A 14} (1981), no. 2, 479–-492.


\bibitem{s1} D.J. Saunders. An Alternative Approach
to the Cartan Form in Lagrangian Field Theories.
\textit{J. Phys. A 20} (1987), no. 2, 339–-349.

\bibitem{s2}
D.J. Saunders. Jet Fields, Connections and Second-Order
Differential Equations. \textit{J. Phys. A 20} (1987), no. 11, 3261–-3270.

\bibitem{Saunders-89}
D.J. Saunders. \textit{The Geometry of Jet Bundles}.  \textit{London Mathematical Society Lecture Note Series, 142.} Cambridge University Press, Cambridge, 1989.

\bibitem{skinner2} R. Skinner and R. Rusk. Generalized Hamiltonian dynamics. I. Formulation on $T^*Q\oplus TQ$. \textit{J. Math. Phys. 24} (1983), no. 11, 2589–-2594.

%\bibitem{SS-1994} P. Schaller and T. Strobl. Poisson structure induced (topological) field theories. \textit{Modern Phys. Lett. A 9} (1994), no. 33, 3129–-3136.

\bibitem{SCYYQLY-2008}
S. Du, C. Hao, Yueke Hu, Y. Hui, Q. Shi, L. Wang and Y. Wu. Maxwell electromagnetic theory from a viewpoint of differential forms. \url{http://arxiv.org/abs/0809.0102v4} (2011).


\bibitem{Snia} J. Sniatycki.
On the geometric structure of classical field theory in Lagrangian formulation.
\textit{Proc. Cambridge Philos. Soc. 68} (1970), 475-–484.

%\bibitem{ST-2007} J. Spillmann and M. Teschner. CORDE: Cosserat Rod elements for the dynamic simulation
%of one-dimensional elastic objects. Proceeding
%SCA '07 Proceedings of the 2007 ACM SIGGRAPH/Eurographics symposium on Computer animation, 63--72.
%
%\bibitem{Strobl-2004}
%T. Strobl. Gravity from Lie algebroid morphisms. \textit{Comm. Math. Phys. 246} (2004), no. 3, 475-–502.


\bibitem{szilasi}
J. Szilasi. A setting for spray and Finsler geometry. \textit{Handbook of Finsler geometry. Vol. 1, 2}, 1183-–1426, Kluwer Acad. Publ., Dordrecht, 2003.


\bibitem{Tulczy1}
W.M. Tulczyjew. Hamiltonian systems, Lagrangian systems and the Legendre transformation. \textit{Symposia Mathematica, Vol. XIV} (Convegno di Geometria Simplettica e Fisica Matematica, INDAM, Rome, 1973), pp. 247-–258. Academic Press, London, 1974.

\bibitem{T1}W.M. Tulczyjew. Les sous-vari\'{e}t\'{e}s lagrangiennes
et la dynamique lagrangienne.  \textit{C. R. Acad. Sci. Paris S\'{e}r. A-B 283} (1976), no. 8, Av, A675-–A678.


\bibitem{T2} W.M. Tulczyjew. Les sous-vari\'{e}t\'{e}s lagrangiennes
et la dynamique hamiltonienne.  \textit{C. R. Acad. Sci. Paris S\'{e}r. A-B 283} (1976), no. 1, Ai, A15-–A18.



\bibitem{Van-2007}
J. Vankerschaver.  Euler-Poincar\'{e} reduction for discrete
field theories. \textit{J. Math. Phys. 48} (2007), no. 3, 032902, 17 pp.

\bibitem{Vanthesis}
 J. Vankerschaver. \textit{Continuous and discrete aspects of
Classical Field Theories with nonholonomic constraints}. PhD diss.,
Ghent University.
\url{http://hdl.handle.net/1854/LU-470403}

\bibitem{VC-06}  J. Vankerschaver and F. Cantrijn. Discrete Lagrangian field theories on Lie groupoids. \textit{J. Geom. Phys. 57} (2007), no. 2, 665–-689.

\bibitem{VCLM-2005}
J. Vankerschaver, F. Cantrijn, M. de Le\'{o}n and D. Mart\'{\i}n de
Diego. Geometric aspects of nonholonomic field theories. \textit{Rep. Math. Phys. 56} (2005), no. 3, 387-–411.



\bibitem{MV-2008} J. Vankerschaver and D. Mart\'{\i}n de
Diego. Symmetry aspects of nonholonomic field theories. \textit{J. Phys. A 41} (2008), no. 3, 035401, 17 pp.

\bibitem{vita} L. Vitagliano. The Hamilton-Jacobi formalism for higher-order field theories. \textit{Int. J. Geom. Methods Mod. Phys. 7} (2010), no. 8, 1413–-1436.


\bibitem{we-1996}
A. Weinstein. Lagrangian mechanics and groupoids. Mechanics day (Waterloo, ON, 1992), 207–231, \textit{Fields Inst. Commun., 7}, Amer. Math. Soc., Providence, RI, 1996.

\bibitem{WR-2006}
K.F. Warnick and P. Russer. Two, three and four-dimensional electromagnetics using differential forms.
\textit{ Turk. J. Elec. Engin.  14} (2006), no. 1, 153--172.


\bibitem{ZS}
Z. Guo and I. Schmidt. Converting Classical Theories to Quantum Theories by Solutions of the Hamilton-Jacobi Equation. \textit{Phys. Rev. D 86} (2012), 045012.

\end{thebibliography}
\end{document}